\documentclass[12pt,twoside,english]{mitthesis}
\usepackage[latin9]{inputenc}
\usepackage{geometry}
\usepackage{babel}
\usepackage{verbatim}
\usepackage{longtable}
\usepackage{float}
\usepackage{calc}
\usepackage{textcomp}
\usepackage{amsthm}
\usepackage{amsmath}
\usepackage{amssymb}
\usepackage{graphicx}
\usepackage{setspace}
\usepackage{esint}
\PassOptionsToPackage{normalem}{ulem}
\usepackage{ulem}
\makeatletter
\newcommand{\lyxmathsym}[1]{\ifmmode\begingroup\def\b@ld{bold}
  \text{\ifx\math@version\b@ld\bfseries\fi#1}\endgroup\else#1\fi}
\providecommand{\tabularnewline}{\\}

\floatstyle{ruled}
\newfloat{algorithm}{tbp}{loa}
\providecommand{\algorithmname}{Algorithm}
\floatname{algorithm}{\protect\algorithmname}
\numberwithin{equation}{section}
\numberwithin{figure}{section}
\theoremstyle{plain}
\newtheorem{thm}{\protect\theoremname}
 \theoremstyle{definition}
 \newtheorem*{defn*}{\protect\definitionname}
  \theoremstyle{plain}
  \newtheorem{lem}{\protect\lemmaname}
  \theoremstyle{plain}
  \newtheorem*{thm*}{\protect\theoremname}
  \theoremstyle{plain}
  \newtheorem*{lem*}{\protect\lemmaname}
  \theoremstyle{definition}
  \newtheorem{defn}{\protect\definitionname}
  \theoremstyle{plain}
  \newtheorem{cor}{\protect\corollaryname}
  \theoremstyle{plain}
  \newtheorem{conjecture}{\protect\conjecturename}
  \theoremstyle{plain}
  \newtheorem*{prop*}{\protect\propositionname}

\newtheorem{lemma}{Lemma}

\newtheorem{theorem}{Theorem}

  \providecommand{\conjecturename}{Conjecture}
  \providecommand{\definitionname}{Definition}
  \providecommand{\lemmaname}{Lemma}
  \providecommand{\propositionname}{Proposition}
  \providecommand{\theoremname}{Theorem}
\providecommand{\corollaryname}{Corollary}
\providecommand{\theoremname}{Theorem}
\newcommand{\be}{\begin{equation}}
\newcommand{\ee}{\end{equation}}
\newcommand{\ba}{\begin{array}}
\newcommand{\ea}{\end{array}}
\newcommand{\bea}{\begin{eqnarray}}
\newcommand{\eea}{\end{eqnarray}}

\newcommand{\calH}{{\cal H }}

\newcommand{\calM}{{\cal M }}

\newcommand{\calT}{{\cal T }}

\newcommand{\calP}{{\cal P }}

\newcommand{\calD}{{\cal D }}

\newcommand{\CC}{\mathbb{C}}

\newcommand{\la}{\langle}
\newcommand{\ra}{\rangle}
\newcommand{\eff}{\mathrm{eff}}
\newcommand{\nn}{\nonumber}
\newcommand{\trace}{\mathop{\mathrm{Tr}}\nolimits}
\makeatother
  
\begin{document}

\title{Eigenvalues and Low Energy Eigenvectors of Quantum Many-Body Systems}

\author{Ramis Movassagh}
       \prevdegrees{B.Sc., Applied and Engineering Physics, Cornell University}
\department{Department of Mathematics}

 \degree{Doctor of Philosophy}

\degreemonth{June}
\degreeyear{2012}
\thesisdate{April 23, 2012}


\supervisor{Peter W. Shor}{Professor}

\chairman{Michel Goemans}{Chairman, Department Committee on Graduate Theses}

\maketitle



\cleardoublepage
\setcounter{savepage}{\thepage}
\begin{abstractpage}
%
%
%
I first give an overview of the thesis and Matrix Product States (MPS) representation of quantum spin systems on a line with an improvement on the notation.\\

The rest of this thesis is divided into two parts. The first part is devoted to eigenvalues of quantum many-body systems (QMBS). I introduce Isotropic Entanglement (IE), which draws from various tools in random matrix theory and free probability theory to accurately approximate the eigenvalue distribution of QMBS on a line with generic interactions. We then list some open related problems. Next, I discuss the eigenvalue distribution of one particle hopping random Schr\"{o}dinger operator  in one dimension from free probability theory in context of the Anderson model. \\

The second part is devoted to ground states and gap of QMBS. I first give the necessary background on frustration free Hamiltonians, real and imaginary time evolution of quantum spin systems on a line within MPS representation and a numerical implementation. I then prove the degeneracy and unfrustration condition for quantum spin chains with generic local interactions, including corrections to our earlier assertions. Following this, I summarize my efforts in proving lower bounds for the entanglement of the ground states, which includes partial new results, with the hope that they inspire future work resulting in solving the conjecture given therein. Next I discuss two interesting measure zero examples where the Hamiltonians are carefully constructed to give unique ground states with high entanglement. One of the examples  (i.e., $d=4$) has not appeared elsewhere. In particular, we calculate the Schmidt numbers exactly, entanglement entropies and introduce a novel technique for calculating the gap which may be of independent interest. The last chapter elaborates on one of the measure zero examples (i.e., $d=3$) which is the first example of a frustration free translation-invariant spin-1 chain that has a unique highly entangled ground state and exhibits  signatures of a critical behavior.

\end{abstractpage}


\cleardoublepage

\section*{Acknowledgments}
First and foremost I like to thank my advisor Peter W. Shor for his unconditional support, trust in my decisions and having provided an umbrella, as much as it was possible for him, under which I could work happily and freely. Interactions with him energize and inspire me.

I thank Alan Edelman, from whom I learned everything I know about random matrix theory, free probability theory and a good deal of linear algebra and Matlab. I thought I knew the latter two, till I met Alan. Above all I thank him for his friendship. I like to thank Jeffrey Goldstone for his time, many exciting discussions and his critical reading of my work when I asked him to.

I thank Daniel Nagaj for the discussions the first summer I started, Salman Beigi for various discussions, Bernie Wang for helping with Latex issues when I was writing my thesis and Eduardo Cuervo-Reyes for exciting physics related discussions during my Z\"urich years.

There have been many great scientists who influenced my scientific life trajectory in positive ways. I am very grateful to Reinhard Nesper, Roald Hoffmann, Mehran Kardar, and Richard V. Lovelace. I also like to thank Otto E. R\"ossler, J\"urg Fr\"ohlich, John McGreevy, Jack Wisdom, Alexei Borodin and John Bush.

I have been fortunate to have met so many wonderful people and made wonderful friendships during my PhD at MIT; too many to name here. I owe a good share of my happiness, balance in life, and the fun I had, to them. \\

Last but certainly not least, I thank my dad, Javad Movassagh, and mom, Mahin Shalchi, for the biological existence and all they did for my sister and I to have a worthwhile future and an education. I dedicate this thesis to them. 

\tableofcontents
\newpage
\listoffigures
\newpage
\listoftables

\chapter{Many-Body Physics and an Overview}

Physics is concerned with specification and evolution in time of the
state of particles given the laws of interactions they are subjected
to. This thesis is devoted to better understanding of quantum aspects
of Many-Body Systems.

\section{Phase Space}

The phase space of a classical system of $N$ particles is a $6N$
dimensional space where each particle contributes $3$ spatial coordinates
and $3$ momenta. A point in this space fully specifies the state
of the system at any given time; the motion of this point in time
specifies the time evolution of the system. Even if the laws of interaction
are precisely formulated, the analytical solution of the equations
of motion, given the exact initial state, can be daunting %
\footnote{Uncertainties in the initial conditions can give rise to further complications
caused by chaotic behavior for positive Lyapunov exponents%
}. The complexity is due to interactions. For example general solution
of the $3-$body problem has been an open problem for roughly $350$
years\cite{three-body}. Despite the interaction being $2-$body,
the correlation in time, of the distances between pairs of masses
makes the analytical solution hard to obtain. In such a case, one
can resort to computational methods to simulate and make predictions
with controllable accuracy by keeping track of $6N$ real parameters. 

The state of a quantum system of $N$ interacting particles is defined
by a \textit{ray} in the Hilbert space, the dimension of which is
a multiplicative function of the dimension of individual Hilbert spaces
(see \cite[section 2.1]{Dirac,weinbergQFTI} for nice expositions
of quantum theory). Mathematically, 

\begin{equation}
\psi\rangle\in\bigotimes_{i=1}^{N}\mathcal{H}_{i},\label{eq:QMBS_Hspace}
\end{equation}
where we denote the Hilbert space of the $i^{\mbox{th}}$ particle
by $\mathcal{H}_{i}$ and, following Dirac's notation \cite[section 20]{Dirac},
the pure state of the system by the vector $\psi\rangle$. We restore
$|$ only when ambiguity of the label of the vector with linear operators
preceding it may arise. To fully specify the system one needs to specify
$d^{N}$ complex numbers (assuming each particle has $d$ degrees
of freedom), which makes the study of quantum many-body systems (QMBS)
computationally intractable. This, compared to simulation of classical
many-body systems, is an additional obstacle we face. Hence, the complexity
in studying QMBS is in interactions as well as state specification. 

It is worth mentioning that there are quantum systems that are particularly
simple, without having as simple a classical analogue. An example
of which is a quantum bit (\textit{qu}bit), whose state consists of
only two points in the phase space. 

In this thesis we confine ourselves to finite dimensional Hilbert
spaces where a general pure state of an $N-$body problem each of
which having $d$ states is

\begin{equation}
\psi\rangle=\psi^{i_{1}\cdots i_{N}}\quad i_{1}\rangle\otimes i_{2}\rangle\otimes\cdots\otimes i_{N}\rangle.\label{eq:QMBS_PureState}
\end{equation}
where repeated indices are summed over. %
\footnote{One could further simplify the notation by denoting the state by $\rangle\equiv\psi\rangle=\psi^{i_{1}\cdots i_{N}}\quad\rangle_{i_{1}}\rangle_{i_{2}}\cdots\rangle_{i_{N}}$,
but I do not think it is worth the trade offs.%
}

In classical physics one can pick a single particle among $N$ interacting
particles, specify its state at some time and predict its evolution
subject to the \textit{fields} of the other particles. If after some
time the fields impressed on this particle by the remaining particles
diminish, the particle becomes free and uncorrelated from the remaining
$N-1$ particles. In contrast, in quantum mechanics, interaction can
lead to quantum correlations (entanglement) that persist even if the
interaction strength diminishes, by say taking the particles far apart.
For example two electrons can interact for some time and end up in
the entangled state $\psi_{1,2}\rangle=00\rangle+11\rangle$ \cite[see discussions on EPR pairs]{NC}.
This peculiar feature of quantum mechanics, as of yet unexplained
by classical physics, is a radical departure from the latter. Quantum
computation and quantum information science make use of entanglement
as a resource to do tasks that are classically difficult or impossible
to do in reasonable time such as \cite{shorFac,grover}. 

Whether, for a given real linear operator, any eigenvalues and eigenvectors
exist, and if so, how to find them is in general very difficult to
answer \cite[p. 32]{Dirac}.

In this thesis I mainly focus on QMBS systems with Hamiltonians (except
in chapter 3)

\begin{equation}
H=\sum_{l=1}^{N-1}\mathbb{I}_{d^{l-1}}\otimes H_{l,l+1}\otimes\mathbb{I}_{d^{N-l-1}}.\label{eq:Hamiltonian-1}
\end{equation}

\begin{figure}
\centering{}\includegraphics[scale=0.35]{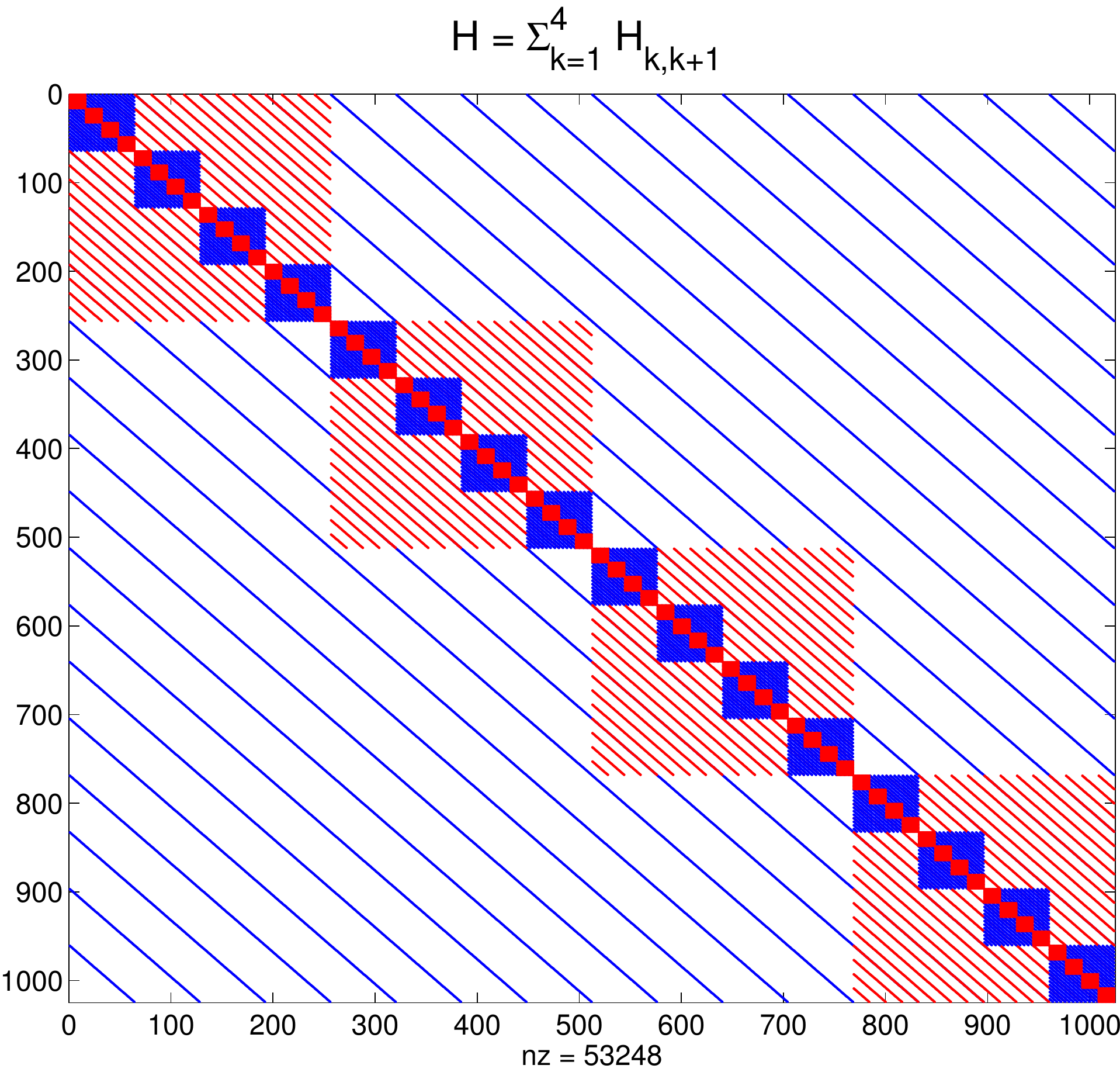}\caption{\label{fig:FS}Sparsity pattern of $H$. The nonzero elements for
$k$ odd are shown in blue and $k$ even in red.}
\end{figure}
 See Figure \ref{fig:FS} for the sparsity pattern of $H$.

Eq. \ref{eq:QMBS_PureState} is the most general equation for a pure
state of $N$ particles irrespective of the type of interactions or
configuration (see Fig. \ref{fig:Nbody_interacting}). However, the
state of a system for any given problem has inherent constraints such
as an underlying lattice that reduces the spatial degrees of freedom,
a range of interaction, or underlying symmetries that the system has
to obey. It is therefore, reasonable to suspect that the physical
properties of any given problem could be well approximated by a reduced
number of parameter say $\sim\mbox{poly}\left(N\right)$ if one understood
the effective degrees of freedom well enough. 

Non-commutativity of the interaction in QMBS is responsible for the
richness of possibilities such as various phases of matter and quantum
computing. Simultaneously, it is accountable for formidability of
finding eigenvalues and corresponding eigenvectors. The exact computation
of eigenvalues alone, on a line, has shown to be QMA-Hard \cite{schuch}.
Energy\textit{ eigenvalue distributions} are needed for calculating
the partition function and calculating other transport properties
(Part I). The eigenvectors specify the corresponding states of the
system. Bulk of matter usually finds itself in its \textit{lowest
energy }\cite[p. 48]{kadanoff}; hence, the urge to comprehend the
low lying states in condensed matter research (Part II).

\section{Part I: Eigenvalues }

Consider the general problem of predicting the eigenvalue distribution
of sums of matrices from the known distribution of summands. In general
this is impossible to do without further information about the eigenvectors.
However, any progress in this direction is extremely desirable as
many problems are modeled by non-commuting matrices. For example,
the Schrödinger equation has a kinetic term and a potential term;
often the former is diagonalizable in the Fourier and the latter in
position space. However, the sum does not have an obvious global basis
nor a distribution that can trivially be inferred from the known pieces. 

For the sake of concreteness suppose we are interested in the distribution
of the random matrix $M=M_{1}+M_{2}$ where the distributions of $M_{1}$
and $M_{2}$ are known. There are two special cases worth considering:
\begin{enumerate}
\item The summands commute. In this case, one can find a simultaneous set
of eigenvectors that diagonalizes each of the summands. In the language
of linear algebra of diagonal matrices, the eigenvalues add. When
the eigenvalues are random, this connects us to the familiar classical
probability theory where the distribution of the sum is obtained by
a convolution of the distribution of the summands. 
\item The summands are in generic positions. In this case, the eigenvectors
of the summands are in generic positions. It is a fascinating fact
that this case also has a known analytical answer given by Free Probability
Theory (FPT) \cite{Voiculescu1985}\cite{speicher}. The eigenvalue
distribution of $M$ is given by the ``free convolution'' of the
distributions of $M_{1}$ and $M_{2}$.
\end{enumerate}
There are many interesting questions that one can ask. How ``free''
are general non-commuting matrices? What is the relationship between
commutation relation and freeness of two matrices? To what extent
does the Fourier matrix act generic? Can a large class of non-commuting
matrices be analyzed using a convex combination of the two extreme
cases discussed above (see Isotropic Entanglement)? 

Suppose the local terms in Eq. \ref{eq:Hamiltonian-1} are generic
(i.e., random matrices), can we utilize the existing tools of random
matrix theory to capture the eigenvalue distribution of $H$ given
the distribution of $H_{l,l+1}$'s? Despite, the local terms being
generic, $H$ is non-generic. The number of random parameters grow
polynomially with $N$ whereas $H$ is $d^{N}$ dimensional. Fraction
of sparsity of $H$ is $\le\left(N-1\right)d^{-\left(N-2\right)}$
(Figure \ref{fig:FS}).

Since the exact evaluation of the density of $H$ is very difficult
\cite{schuch}, one can use two known approximations. As far as parameter
counting is concerned the quantum problem falls nicely in between
the two extreme case (Figure \ref{fig:Parameter-counting} and Isotropic
Entanglement). 

\begin{figure}
\begin{centering}
\includegraphics[scale=0.40]{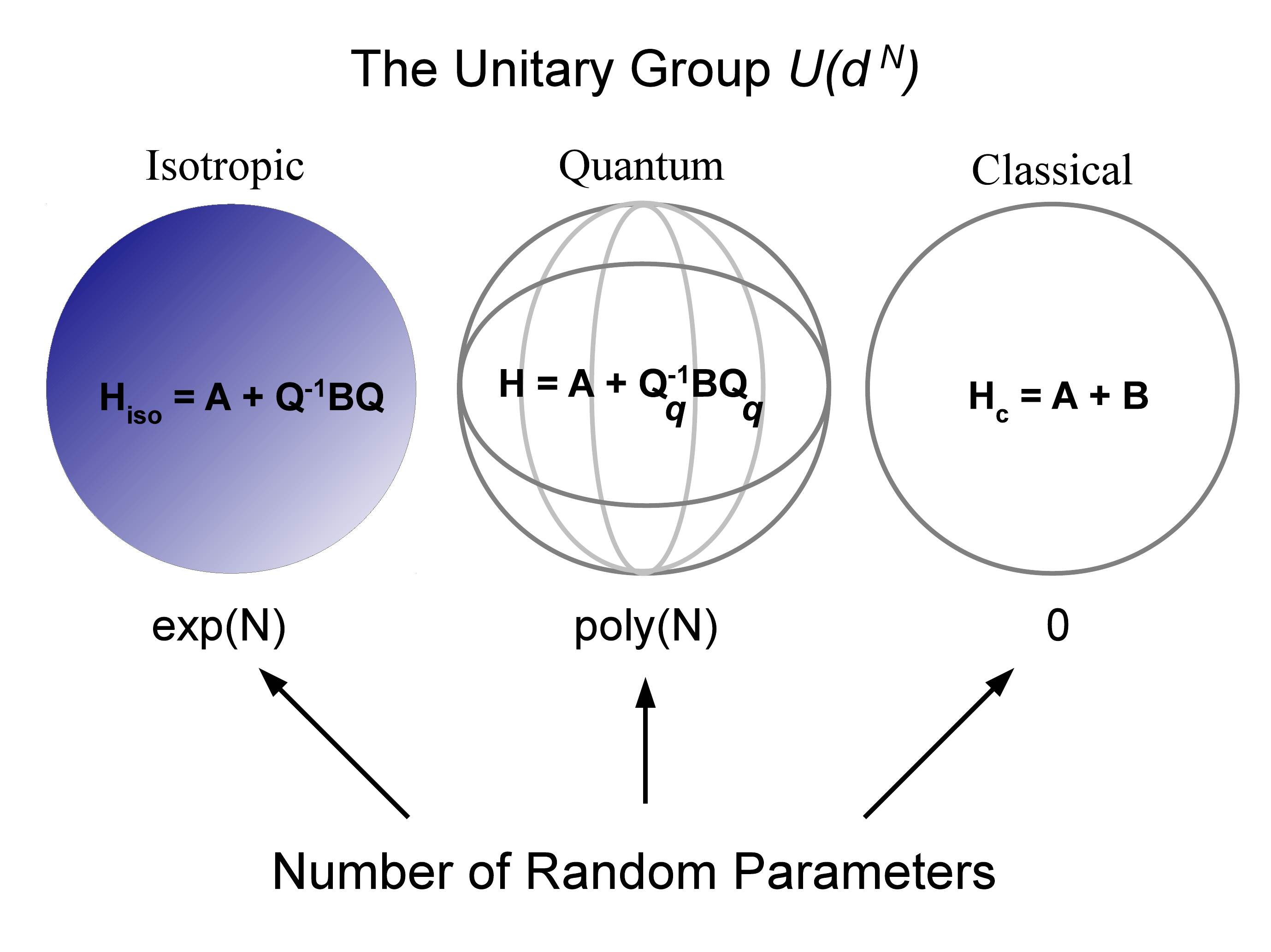}
\par\end{centering}

\caption{\label{fig:Parameter-counting}Parameter counting}
\end{figure}

Chapter 2 gives a detailed description of Isotropic Entanglement which
gives the distribution of Eq. \ref{eq:Hamiltonian-1} and elaborates
on the discussion of this section.

Chapter 3 is on the distribution of one particle hopping random Schr\"{o}dinger
equation and Anderson model.

\section{Part II: Eigenvectors}

Is it possible to capture the essential physics of the system accurately
enough with an efficient simulation with a much smaller $\chi\thicksim\textrm{poly}(N)$,
spanning only a small part of the full Hilbert space of the system?
In our case, the qudits are arranged on a 1D lattice and only have
nearest-neighbor interactions. We could thus expect that a reduced
space might suffice for our needs. This concept is common for the
various approaches proposed for efficient (tractable on a classical
computer) numerical investigation of quantum many body systems such
as the Density Matrix Renormalization Group \cite{white}, Matrix
Product States \cite{frank}, Tensor Product States \cite{wen} and
projected entangled pair states (PEPSs) \cite{peps}. For gapped one
dimensional systems MPS ansatz is proved to suffice \cite{Hastings07}.
Consider a general Hamiltonian for $1D$ open chain with generic local
interactions as given by Eq. \ref{eq:Hamiltonian-1}, where each $H_{l,l+1}$
is a $d^{2}\times d^{2}$ matrix of rank $r$ and the total number
of particles is $N$. 

It is interesting to ask under what circumstances can there be a degeneracy
of ground states? Moreover, when is the ground state of the whole
system ($H$ in Eq. \ref{eq:Hamiltonian-1}) also the local ground
state of all $H_{l,l+1}$'s ; i.e., the system is ``unfrustrated''?
We answered these questions for $1D$ spin systems with \textit{generic}
interactions \cite{MFGNOS} (Chapter 6). We found that in the regime
where $r\leq d^{2}/4$ the system is unfrustrated with many ground
states; moreover for $r<d$ there can be product states among the
ground states. For sufficiently large $N$ the system is frustrated
for $r>d^{2}/4$. 

The next natural question is: how entangled are the ground states
in the regime $d\le r\le d^{2}/4$? The entanglement can be quantified
by the Schmidt rank (see MPS below). We call a state highly entangled
if its Schmidt rank is exponentially large in the number of particles.
In this regime, it is straightforward to show that among the many
ground states, there are no product states and that there exist states
with high amount of entanglement with probability one. The latter
can be shown for example using results of algebraic geometry \cite{cubitt}\cite{Hayden_Generic}.
I have been trying to show that \textit{all }the ground states, in
this regime, have Schmidt ranks that are exponentially large with
probability one. Despite some partial results, the goal has not been
fulfilled.

Using a genericity argument one can show that, to prove results in
the generic case, it is sufficient to find an example of local terms
whose ground states all have large Schmidt ranks. We have not yet
succeeded in finding such examples in the regime of interest $d\le r\le d^{2}/4$.
There are, however, interesting examples for which there is a unique
ground state with exponentially large entanglement in the \textit{frustrated}
regime, i.e., probability zero case (Chapters 8 and 9). 

Below I give a background on Matrix Product States on with an improvement
on the notation.

Chapter 4 discusses the unfrustration condition and a numerical code
I developed to study spin systems on a line with local interactions
and without translational invariance. I find the ground states using
imaginary time evolution. I then provide the proofs and corrections
of our previous work regarding the unfrustration condition and degeneracy
of quantum spin chains with generic local interactions.

Chapter 5 summarizes various attempts I made in proving a lower bound
on the Schmidt rank of the ground states of generic spin chains. It
includes unpublished results and two ways one can potentially prove
the conjecture given there. 

Chapter 6 discusses two interesting measure zero examples ($d=3$
and $d=4$). I include the combinatorial techniques for calculating
the entanglement entropies. The $d=4$ example has not been published
elsewhere.

Chapter 7 elaborates on the $d=3$ example of Chapter 6 and has a
novel technique for calculating the gap that may be of independent
interest.

\subsection{SVD and Matrix Product States (MPS) on a Line}

In this section, in order to avoid confusion, we restore the summation
symbols. As stated earlier the state of a composite system is a vector
in the tensor product of the Hilbert spaces of the constituents. Suppose
we have the pure state of a composite system, how can we express it
in terms of the pure states of the constituents? This is answered
by singular value decomposition (SVD). This application of SVD in
quantum information theory is called Schmidt Decomposition \cite{wikipedia_SD},\cite{NC}.

\begin{thm}
(Schmidt Decomposition) Let $H_{1}$ and $H_{2}$ be Hilbert spaces
of dimensions $n$ and $m$ respectively. Assume $m\le n$. For any
vector $\psi_{12}\rangle\in H_{1}\otimes H_{2}$, there exists orthonormal
sets $\left\{ \eta_{1}\rangle,\cdots,\eta_{n}\rangle\right\} \subset H_{1}$
and $\left\{ \omega_{1}\rangle,\cdots,\omega_{m}\rangle\right\} \subset H_{2}$
such that $\psi\rangle=\sum_{i=1}^{m}\lambda_{i}\;\eta_{i}\rangle\otimes\omega_{i}\rangle$,
where $\lambda_{i}$ are non-negative and , as a set, uniquely determined
by $\psi\rangle$.
\end{thm}
The number of nonzero $\lambda$'s are called the Schmidt rank denoted
by $\chi$.

\begin{flushleft}
Comment: $\chi\le m\le n$; the Schmidt rank is no greater than the
minimum of the dimensions of the two Hilbert spaces. 
\par\end{flushleft}

Comment: Schmidt Decomposition can be thought of as an expansion of
a vector in bases of the subsystems. The need for having more than
one expansion coefficient (i.e., $\chi>1$) indicates that the state
is not separable in any basis (i.e., subsystems are entangled). A
quantifier for entanglement is therefore $\chi$.

\begin{figure}
\begin{centering}
\includegraphics[scale=1.4]{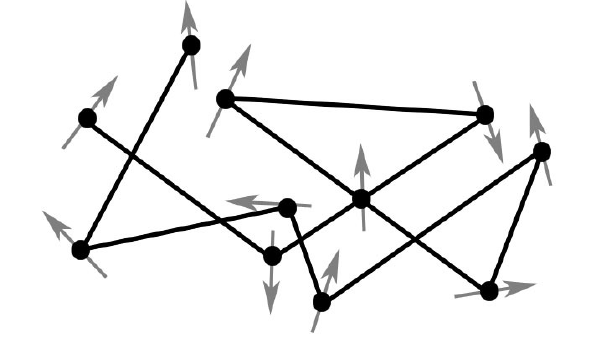}\caption{\label{fig:Nbody_interacting}Some general $N$ interaction quantum
spins whose pure state is represented by Eq. \ref{eq:QMBS_PureState}.}

\par\end{centering}

\end{figure}

DMRG, natural representation of which is MPS, beautifully utilizes
Schmidt Decomposition to capture low energy properties of QMBS on
a line or trees \cite{cirac,vidal}\cite{vidalTree}. Though any state
can be expresses as a MPS, the time evolution is only naturally implemented
for the line or a tree. Some attractive features are 1. MPS gives
a local description of QMBS on a line 2. MPS allows for a systematic
truncation of the Hilbert space to capture the low energy properties
with controllable accuracy. Below I give a derivation of MPS similar
to \cite{daniel} and improve the notation by making it more compact
(Eq. \ref{eq:Ramis}).

Suppose we have a chain with $N$ sites and take $1\le n\le N$. First,
perform a Schmidt decomposition of the chain between sites $n-1$
and $n$ as 

\begin{equation}
\psi\rangle=\sum_{\alpha_{n-1}}^{\chi_{n-1}}\lambda_{\alpha_{n-1}}^{\left(n-1\right)}\quad|\phi_{\alpha_{n-1}}\rangle_{1,\cdots n-1}\otimes\;\phi_{\alpha_{n-1}}\rangle_{n,\cdots,N},\label{eq:breatAtn-1}
\end{equation}
where the states on the left and on the right of the division form
an orthonormal bases for the respective subsystems of the state of
$\psi\rangle$. The Schmidt rank $\chi_{n-1}$ is the minimum number
of terms required in this decomposition. Recall that $\chi_{n-1}$
is at most equal to the minimum dimension of the Hilbert spaces of
the two subsystems.

Next, the Schmidt decomposition for a split between $n$ and $n+1$
gives

\begin{equation}
\psi\rangle=\sum_{\alpha_{n}}^{\chi_{n}}\lambda_{\alpha_{n}}^{\left(n\right)}\quad|\theta_{\alpha_{n}}\rangle_{1,\cdots n}\otimes\;\theta_{\alpha_{n}}\rangle_{n+1,\cdots,N}.\label{eq:breakAtn}
\end{equation}
These two decompositions describe the same state Eq. \ref{eq:QMBS_PureState},
allowing us to combine them by expressing the basis of the subsystem
$n,\cdots,N$ as

\begin{equation}
\phi_{\alpha_{n-1}}\rangle_{n,\cdots,N}=\sum_{i_{n}=1}^{d}\sum_{\alpha_{n}=1}^{\chi_{n}}\Gamma_{\alpha_{n-1},\alpha_{n}}^{i_{n},\left[n\right]}\lambda_{\alpha_{n}}^{\left(n\right)}\quad|i_{n}\rangle\otimes\theta_{\alpha_{n}}\rangle_{n+1,\cdots,N},\label{eq:MPS_Gamman}
\end{equation}
where we inserted $\lambda_{\alpha_{n}}^{\left(n\right)}$ for convenience.
This gives us the tensor $\Gamma^{\left[n\right]}$ which carries
an index $i_{n}$ corresponding to the physical states $i_{n}\rangle$
of the $n^{\mbox{th}}$ spin, and indices $\alpha_{n-1}$ and $\alpha_{n}$
corresponding to the two consecutive divisions of the system. Since
$\phi_{\alpha_{n-1}}\rangle$ and $\theta_{\alpha_{n}}\rangle$ are
orthonormal states, the vectors $\lambda$ and tensors $\Gamma$ obey
the following normalization conditions. From Eq. \ref{eq:breakAtn}
we have

\begin{equation}
\sum_{\alpha_{n}=1}^{\chi_{n}}\lambda_{\alpha_{n}}^{\left[n\right]2}=1,\label{eq:norm0}
\end{equation}
while \ref{eq:MPS_Gamman} implies 

\begin{equation}
\langle\phi_{\alpha'_{n-1}}|\phi_{\alpha_{n-1}}\rangle=\sum_{i_{n}=1}^{d}\sum_{\alpha_{n}=1}^{\chi_{n}}\Gamma_{\alpha'_{n-1},\alpha_{n}}^{\left[n\right],i_{n}*}\lambda_{\alpha_{n}}^{\left[n\right]}\Gamma_{\alpha_{n-1},\alpha_{n}}^{\left[n\right],i_{n}}\lambda_{\alpha_{n}}^{\left[n\right]}=\delta_{\alpha_{n-1},\alpha_{n-1}'},\label{eq:norm1}
\end{equation}
and 

\begin{equation}
\langle\theta_{\alpha'_{n}}|\theta_{\alpha_{n}}\rangle=\sum_{i_{n}=1}^{d}\sum_{\alpha_{n-1}=1}^{\chi_{n-1}}\lambda_{\alpha_{n-1}}^{\left[n-1\right]}\Gamma_{\alpha_{n-1},\alpha'_{n}}^{\left[n\right],i_{n}*}\lambda_{\alpha_{n-1}}^{\left[n-1\right]}\Gamma_{\alpha_{n-1},\alpha_{n}}^{\left[n\right],i_{n}}=\delta_{\alpha_{n},\alpha_{n}'}.\label{eq:norm2}
\end{equation}
\begin{figure}
\begin{centering}
\includegraphics[scale=0.3]{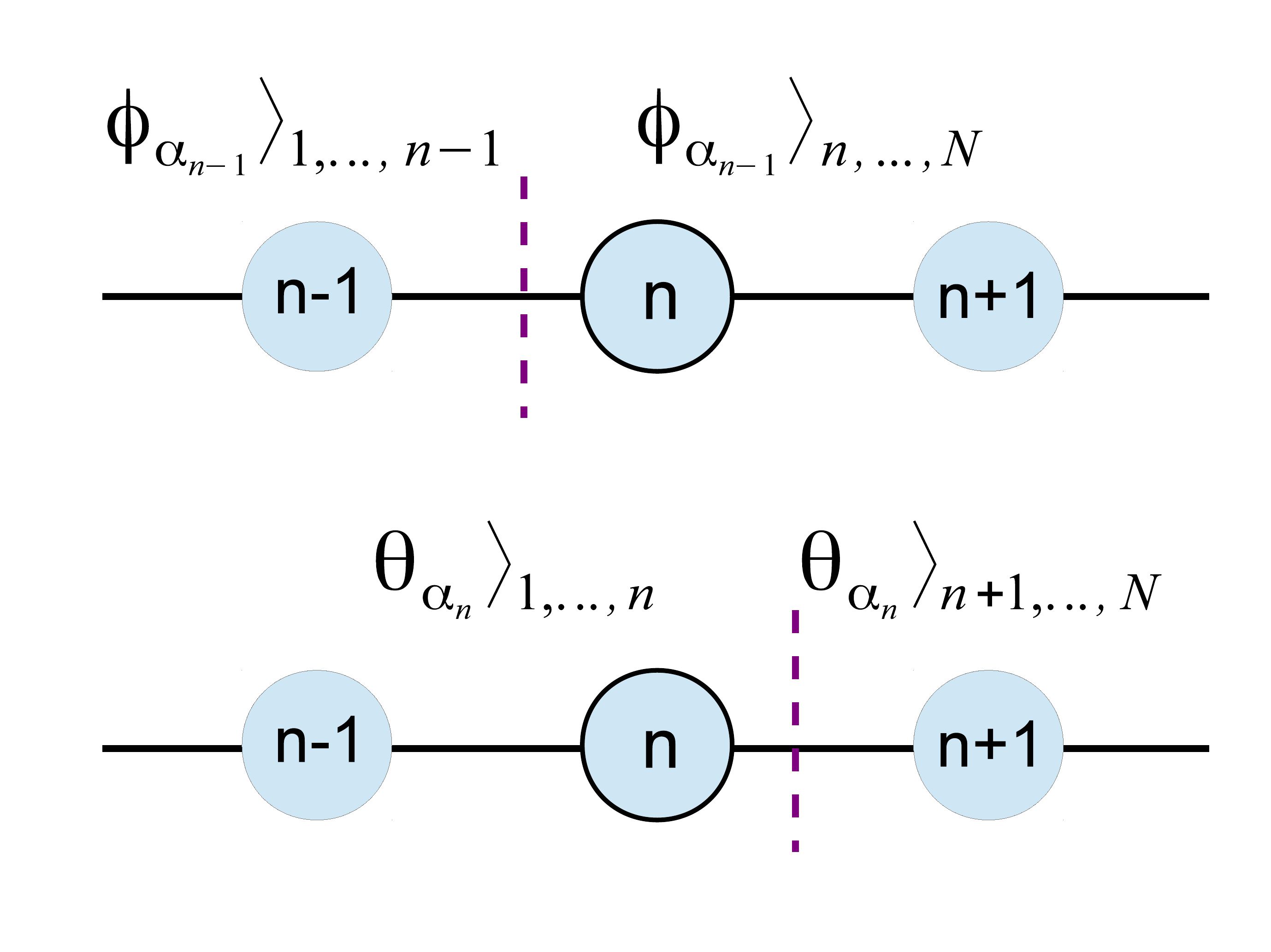}
\par\end{centering}

\caption{\label{fig:MPS_derivation}Decompositions that define $\Gamma^{[n]}$
in the derivation of MPS.}
\end{figure}

One can do what we just did for every site $1\le n\le N$ and get
the MPS representation of the spin chain, denoting open boundary conditions
and periodic boundary condition by OBC and PBC respectively,

\begin{eqnarray}
\psi^{i_{1}i_{2}\dots i_{N}} & = & \sum_{\alpha_{1},\dots,\alpha_{N-1}=1}^{\chi}\Gamma_{1,\alpha_{1}}^{i_{1},[1]}\Gamma_{\alpha_{1},\alpha_{2}}^{i_{2},[2]}\cdots\Gamma_{\alpha_{N-1},1}^{i_{N},[N]}\quad\quad\quad\mbox{OBC}\label{eq:mpsdefinition}\\
\psi^{i_{1}i_{2}\dots i_{N}} & = & \sum_{\alpha_{0},\alpha_{1},\dots,\alpha_{N-1}=1}^{\chi}\Gamma_{\alpha_{0},\alpha_{1}}^{i_{1},[1]}\Gamma_{\alpha_{1},\alpha_{2}}^{i_{2},[2]}\cdots\Gamma_{\alpha_{N-1},\alpha_{0}}^{i_{N},[N]}\quad\mbox{PBC}\nonumber 
\end{eqnarray}

Comments: It is customary to absorb the $\lambda$'s into $\Gamma$'s
and omit them as we did in the foregoing equation. The upper limit
$\chi=\max_{1\le n\le N}\left\{ \chi_{n}\right\} $.

I believe the MPS notation given by Eq. \ref{eq:mpsdefinition} can
be improved 

\begin{eqnarray}
\psi\rangle & = & {\displaystyle \mathcal{P}}\left\{ \bigotimes_{p=1}^{N}\sum_{i_{p}=1}^{d}\Gamma\left(i_{p}\right)\sigma^{i_{p}}\right\} \;0\rangle^{\otimes N}\quad\quad\quad\mbox{OBC}\label{eq:Ramis}\\
\psi\rangle & = & \mbox{Tr}_{\chi}{\displaystyle \mathcal{P}}\left\{ \bigotimes_{p=1}^{N}\sum_{i_{p}=1}^{d}\Gamma\left(i_{p}\right)\sigma^{i_{p}}\right\} \;0\rangle^{\otimes N}\quad\quad\mbox{PBC}\nonumber 
\end{eqnarray}
where $\sigma^{i_{p}}$ defined by $\sigma^{i_{p}}\;|0\rangle\equiv i_{p}\rangle$
are the generalized Pauli operators, each $\Gamma\left(i_{p}\right)\equiv\Gamma^{[p],i_{p}}\rightarrow\Gamma_{\alpha_{p-1},\alpha_{p}}^{[p],i_{p}}$,
for a given $i_{p}$, is a $\chi\times\chi$ matrix and $\mathcal{P}$
denotes an ordering with respect to $p$ of the tensor products. The
subscript $\chi$ on the trace reminds us that the trace is on $\chi\times\chi$
part of $\Gamma$'s and not the physical indices $i_{p}$. Lastly,
we can simplify notation by assuming repeated indices are summed over
to get

\begin{eqnarray}
\psi\rangle & = & {\displaystyle \mathcal{P}}\left\{ \bigotimes_{p=1}^{N}\Gamma\left(i_{p}\right)\sigma^{i_{p}}\right\} \;0\rangle^{\otimes N}\quad\quad\quad\mbox{OBC}\nonumber \\
\psi\rangle & = & \mbox{Tr}_{\chi}{\displaystyle \mathcal{P}}\left\{ \bigotimes_{p=1}^{N}\Gamma\left(i_{p}\right)\sigma^{i_{p}}\right\} \;0\rangle^{\otimes N}\quad\quad\mbox{PBC}\label{eq:RamisBrief}
\end{eqnarray}
Note that we now need at most $\sim Nd\chi^{2}$ parameters to specify
any state. There are orthogonality conditions on $\Gamma$'s that
further reduce the number of independent parameters needed. 

\begin{figure}
\centering{}\includegraphics[scale=0.3]{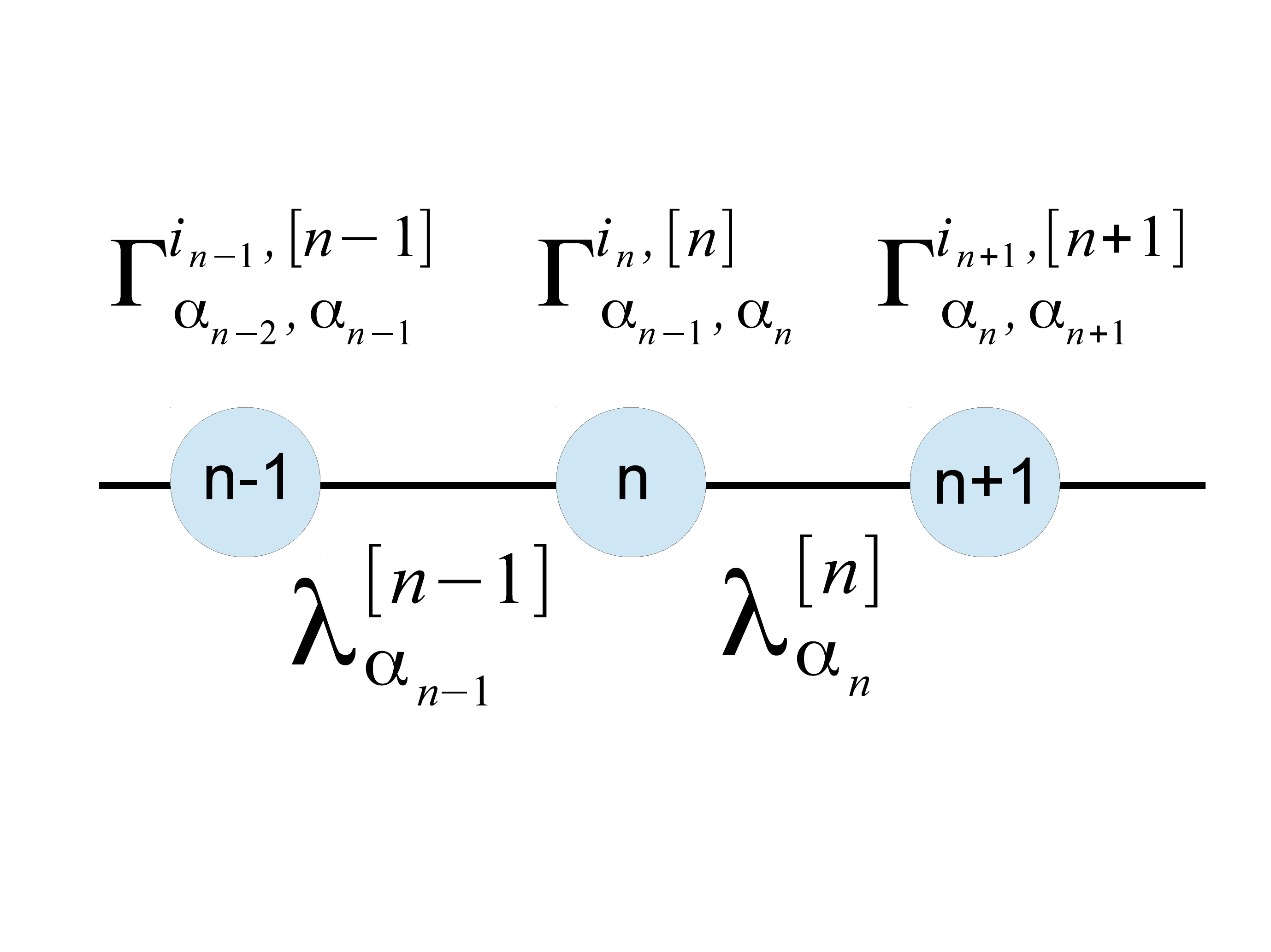}\caption{MPS representation}
\end{figure}

\subsection{MPS on a Tree}

The MPS representation relies on splitting the system into two subsystems
and making use of Schmidt Decomposition which applies for loop-less
configuration of spins such as chains and trees. The derivation above
can be generalized (see \cite{daniel}) to spins on a $k-$child tree.

\begin{equation}
\psi\rangle=\left(\Pi_{p\in\mbox{bonds}}\sum_{\alpha_{p}=1}^{\chi_{p}}\lambda_{\alpha_{p}}^{\left[p\right]}\right)\left(\Pi_{n\in\mbox{sites}}\sum_{i_{n}}\Gamma_{\alpha_{n_{1}},\cdots,\alpha_{n_{k}}}^{i_{n},\left[n\right]}\right)\quad\cdots\rangle\otimes i_{n}\rangle\otimes\cdots\rangle,\label{eq:MPStree}
\end{equation}
where $\alpha_{n_{1}},\cdots,\alpha_{n_{k}}$ are indices corresponding
to the $k$ bonds $n_{1},\cdots,n_{k}$ coming out of site $n$. Each
index $\alpha_{n_{j}}$ appears in $k$ $\Gamma$ tensors and one
$\lambda$ vector. The normalization conditions for a MPS description
of a state on a tree are analogous to Eqs. \ref{eq:norm0},\ref{eq:norm1},
and \ref{eq:norm2}. We have

\begin{eqnarray}
\sum_{\alpha_{p}}\lambda_{\alpha_{p}}^{\left[p\right]2} & = & 1,\label{eq:normTree}\\
\sum_{i_{n}}\sum_{\alpha_{2},\cdots,\alpha_{n}}\Gamma_{\alpha'_{n_{1}},\cdots,\alpha_{n_{k}}}^{i_{n},\left[n\right]*}\lambda_{\alpha_{n_{2}}}^{\left[2\right]2}\cdots\lambda_{n_{k}}^{\left[k\right]2} & \Gamma_{\alpha{}_{n_{1}},\cdots,\alpha_{n_{k}}}^{i_{n},\left[n\right]*}= & \delta_{\alpha'_{n_{1}},\alpha_{n_{1}},}\label{eq:norm2Tree}
\end{eqnarray}
and the other variations with $n_{2},\cdots,n_{k}$.

\part{Eigenvalues}
\chapter{Isotropic Entanglement}

In this first part of the thesis I focus on the density of states
or QMBS. This chapter treats the eigenvalue distribution of spin chains,
though some of the theorems apply in higher dimensions. We treat generic
local interactions, where by local I mean every interaction term acts
nontrivially on $L$ consecutive spins. This chapter also appears
in \cite{RamisIE,RamisPRL}.

\section{\label{sec:The-Elusive-Spectra} Elusive Spectra of Hamiltonians}

\begin{onehalfspace}
Random matrix techniques have proved to be fruitful in addressing
various problems in quantum information theory (QIT) \cite{Hayden_Generic,hastings,cubitt}
and theoretical physics \cite{RMTinCMP,Beenakker}. In condensed matter
physics (CMP), quantum spins with random exchange have been extensively
studied with surprising phase transition behaviors \cite{DFisher1994,Fisher2001,Tu,Schuettler}. 

Moreover, constraint satisfaction problems are at the root of complexity
theory \cite{cook1971,krzakala_Classical}. The quantum analogue of
the satisfiability problem (QSAT) encodes the constraints by a Hamiltonians
acting on quantum spins (or qu\textit{d}its) \cite{Bravyi}. Interesting
phase transition behaviors have been seen for random QSAT with important
implications for the Hardness of typical instances\cite{chris,chris2,QLLL}. 

While the application of the aggregate of eigenvalues may be different
in CMP, we point out a few examples where they are studied \cite{HiT_HeatC,FullDOS,WolfReview,PhotoEmissionDOS}.
By comparison, in QIT and CMP the ground state and the first few excited
states have been well studied to date \cite{cirac,vidal,white,wen,frank,hastings,ramis}.
To the best of our knowledge the aggregate has not been as well studied
perhaps in part because of the complexity \cite{schuch} and perhaps
in part because the {\it low} energy states have captivated so much
interest to date. 

Energy \textit{eigenvalue distributions} or the \textit{density of
states (DOS)} are needed for calculating the partition function\cite[p. 14]{kadanoff}.
The DOS plays an important role in the theory of solids, where it
is used to calculate various physical properties such the internal
energy, the density of particles, specific heat capacity, and thermal
conductivity. Quantum Many-Body Systems (QMBS) spectra have been elusive
for two reasons: 1. The terms that represent the interactions are
generally non-commuting. This is pronounced for systems with random
interactions (e.g., quantum spin glasses \cite{spinGlass,sachdev2}).
2. Standard numerical diagonalization is limited by memory and computer
speed. The exact calculation of the spectrum of interacting QMBS has
been shown to be difficult \cite{schuch}. 

An accurate description of tails of distributions are desirable for
CMP. Isotropic Entanglement (IE) provides a direct method for obtaining
eigenvalue distributions of quantum spin systems with generic local
interactions and does remarkably well in capturing the tails. Indeed
interaction is the very source of entanglement generation \cite[Section 2.4.1]{NC}\cite{Preskill}
which makes QMBS a resource for quantum computation \cite{gottesman}
but their study a formidable task on a classical computer. 

Suppose we are interested in the eigenvalue distribution of a sum
of Hermitian matrices $M=\sum_{i=1}^{N}M_{i}.$ In general, $M_{i}$
cannot be simultaneously diagonalized, consequently the spectrum of
the sum is not the sum of the spectra. Summation mixes the entries
in a very complicated manner that depends on eigenvectors. Nevertheless,
it seems possible that a one-parameter approximation might suffice. 

Though we are not restricted to one dimensional chains, for sake of
concreteness, we investigate $N$ interacting $d$-dimensional quantum
spins (qudits) on a line with generic interactions. The Hamiltonian
is

\vspace{-0.15in}

\end{onehalfspace}

\begin{onehalfspace}
\begin{equation}
H=\sum_{l=1}^{N-1}\mathbb{I}_{d^{l-1}}\otimes H_{l,\cdots,l+L-1}\otimes\mathbb{I}_{d^{N-l-\left(L-1\right)}},\label{eq:Hamiltonian}
\end{equation}
where the local terms $H_{l,\cdots,l+L-1}$ are finite $d^{L}\times d^{L}$
random matrices. We take the case of nearest neighbors interactions,
$L=2$, unless otherwise specified. 
\end{onehalfspace}

\begin{onehalfspace}
The eigenvalue distribution of any commuting subset of $H$ such as
the terms with $l$ odd (the ``odds'') or $l$ even (the ``evens'')
can be obtained using local diagonalization. However, the difficulty
in approximating the full spectrum of $H\equiv H_{\mbox{odd}\vphantom{\mbox{even}}}+H_{\mbox{even}\vphantom{\mbox{odd}}}$
is in summing the odds and the evens because of their overlap at every
site. 

The intuition behind IE is that terms with an overlap, such as $H_{l,l+1}$
and $H_{l+1,l+2}$, introduce randomness and mixing through sharing
of a site. Namely, the process of entanglement generation introduces
an \textit{isotropicity }between the eigenvectors of evens and odds
that can be harnessed to capture the spectrum.

\begin{figure}
\centering{}\includegraphics[scale=1.5]{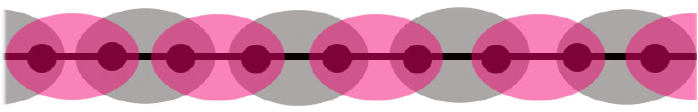}\caption{\label{fig:Odd-and-even}Odd and even summands can separately be locally
diagonalized, but not the sum. The overlap of the two subsets at every
site generally requires a global diagonalization.}
\end{figure}

Random Matrix Theory (RMT) often takes advantage of eigenvectors with
Haar measure, the uniform measure on the orthogonal/unitary group.
However, the eigenvectors of QMBS have a more special structure (see
Eq. \ref{eq:OddEven}).

Therefore we created a \textit{hybrid theory}, where we used a finite
version of Free Probability Theory (FPT) and Classical Probability
Theory to capture the eigenvalue distribution of Eq. \ref{eq:Hamiltonian}.
Though such problems can be QMA-complete, our examples show that IE
provides an accurate picture well beyond what one expects from the
first four moments alone.\textit{ }The \textit{Slider} (bottom of
Figure \ref{fig:RoadMap}) displays the proposed mixture $p$.
\end{onehalfspace}

\begin{onehalfspace}

\section{\label{sec:The-Method}The Method of Isotropic Entanglement}
\end{onehalfspace}

\begin{onehalfspace}

\subsection{\label{sub:Road-Map} Overview}
\end{onehalfspace}

\begin{onehalfspace}
\begin{figure}
\centering{}\includegraphics[scale=0.6]{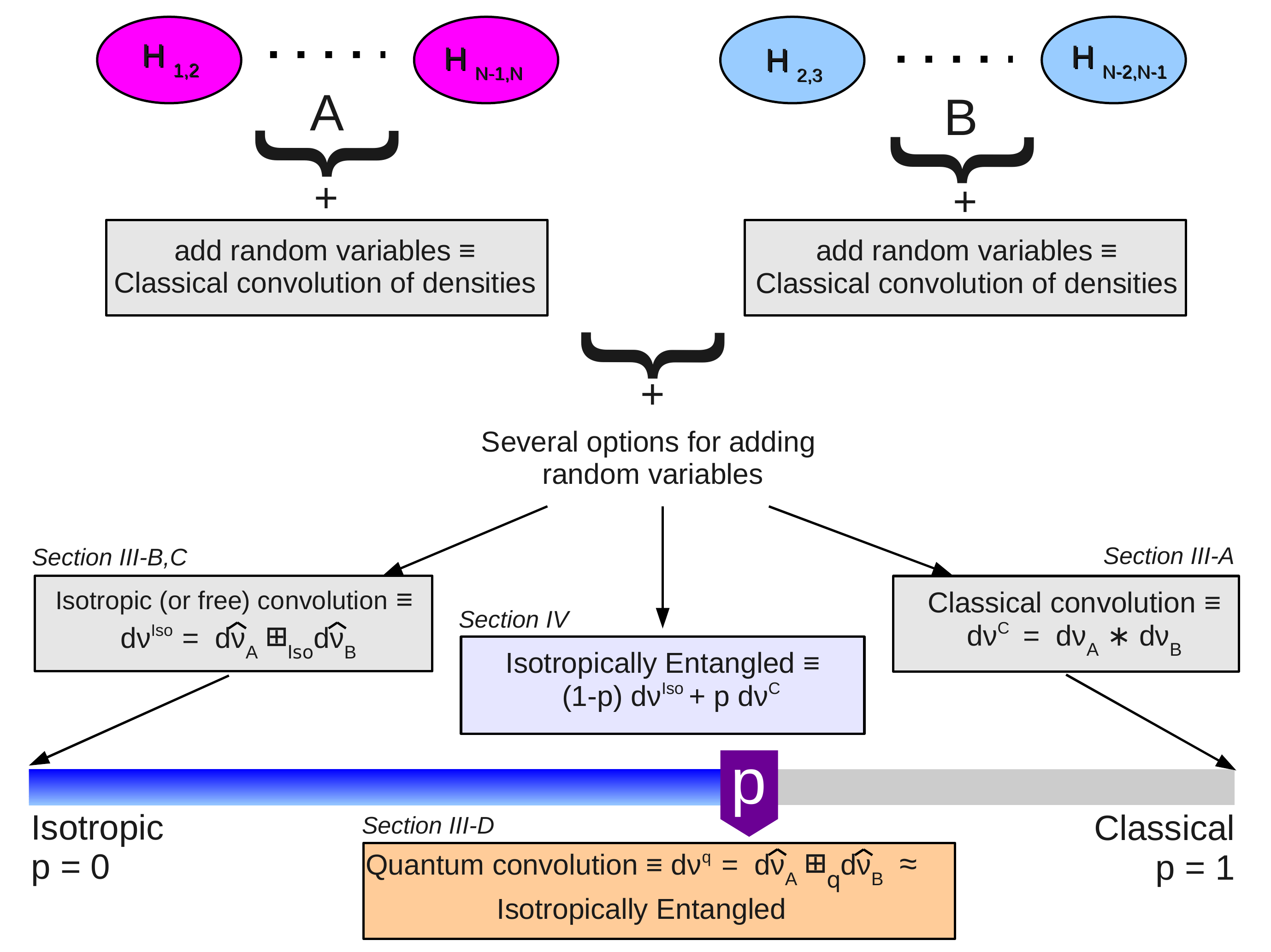}\caption{\label{fig:RoadMap}The method of Isotropic Entanglement: Quantum
spectra as a convex combination of isotropic and classical distributions.
The Slider (bottom) indicates the $p$ that matches the quantum kurtosis
as a function of classical ($p=1$) and isotropic ($p=0$) kurtoses.
To simplify we drop the tensor products (Eq. \ref{eq:HevenHodd})
in the local terms (ellipses on top). Note that isotropic and quantum
convolution depend on multivariate densities for the eigenvalues. }
\end{figure}

We propose a method to compute the ``density of states'' (DOS) or
``eigenvalue density'' of quantum spin systems with generic local
interactions. More generally one wishes to compute the DOS of the
sum of non-commuting random matrices from their, individually known,
DOS's. 

We begin with an example in Figure \ref{fig:The_intro}, where we
compare exact diagonalization against two approximations:

\begin{figure}
\centering{}\includegraphics[width=12cm]{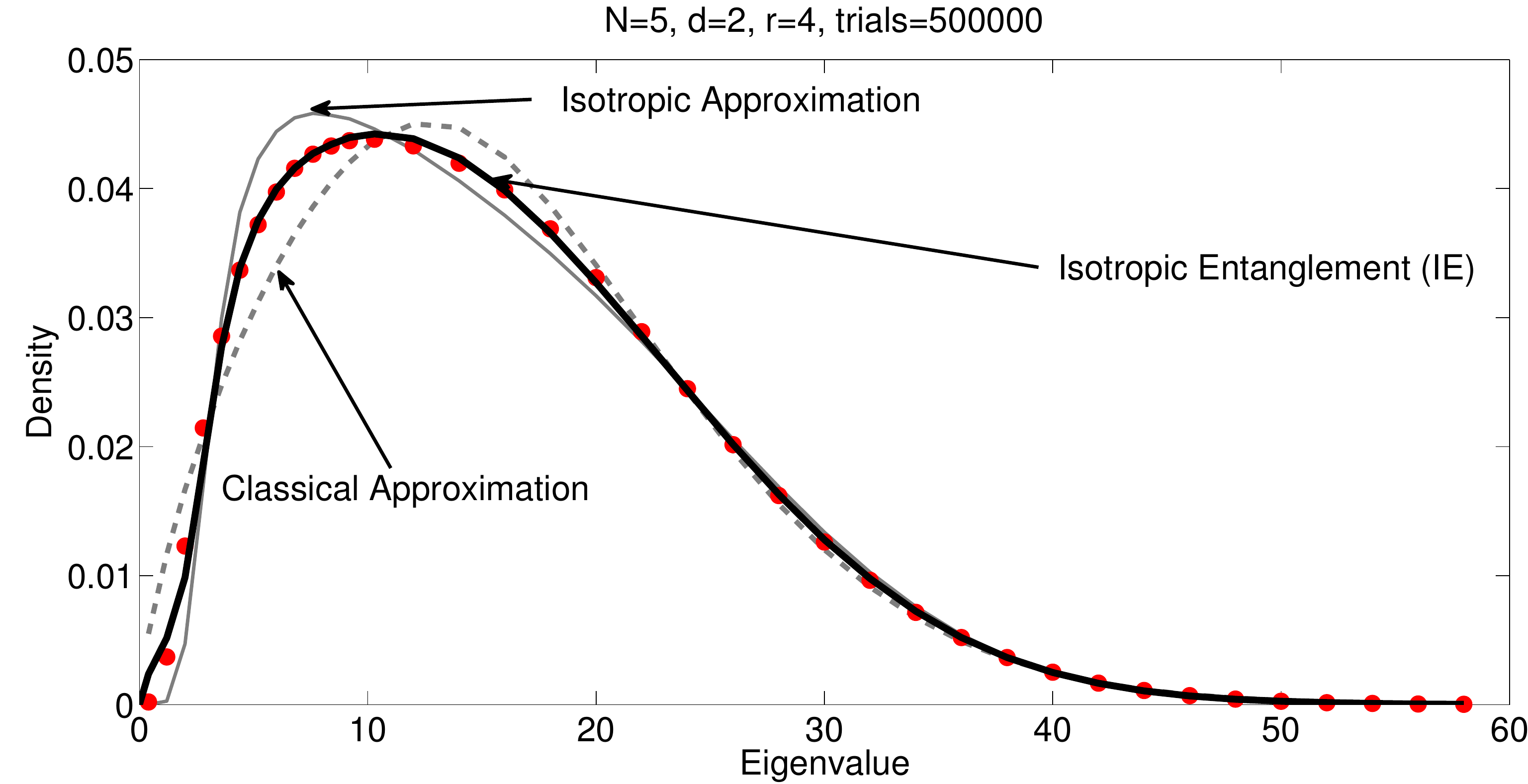}\caption{\label{fig:The_intro}The exact diagonalization in dots and IE compared
to the two approximations. The title parameters are explained in the
section on numerical results.}
\end{figure}

\end{onehalfspace}
\begin{itemize}
\begin{onehalfspace}
\item Dashed grey curve: \textit{classical} approximation. Notice that it
overshoots to the right.
\item Solid grey curve: \textit{isotropic} approximation. Notice that it
overshoots to the left.
\item Solid black curve: \textit{isotropic entanglement (IE).}
\item Dots: \textit{exact diagonalization} of the quantum problem given
in Eq. \ref{eq:Hamiltonian}.\end{onehalfspace}

\end{itemize}
\begin{onehalfspace}
The \textit{classical approximation} ignores eigenvector structure
by summing random eigenvalues uniformly from non-commuting matrices.
The dashed curve is the convolution of the probability densities of
the eigenvalues of each matrix. 

The\textit{ isotropic approximation} assumes that the eigenvectors
are in ``general position''; that is, we add the two matrices with
correct eigenvalue densities but choose the eigenvectors from Haar
measure. As the matrix size goes to infinity, the resulting distribution
is the free convolution of the individual distributions \cite{speicher}. 

The exact diagonalization given by red dots, the dashed and solid
grey curves have exactly the same first three moments, but differing
fourth moments. 

\textit{Isotropic Entanglement (IE)} is a linear combination of the
two approximations that is obtained by matching the fourth moments.
We show that 1) the fit is better than what might be expected by four
moments alone, 2) the combination is always convex for the problems
of interest, given by $0\leq p\leq1$ and 3) this convex combination
is universal depending on the parameter counts of the problem but
not the eigenvalue densities of the local terms.

\textit{Parameter counts: exponential, polynomial and zero.} Because
of the \textit{locality} of generic interactions, the complete set
of eigenstates has parameter count equal to a polynomial in the number
of spins, though the dimensionality is exponential. The classical
and isotropic approximations have zero and exponentially many random
parameters respectively. This suggests that the problem of interest
somehow lies in between the two approximations. 

Our work supports a very general principle that one can obtain an
accurate representation of inherently exponential problems by approximating
them with less complexity. This realization is at the heart of other
recent developments in QMBS research such as Matrix Product States
\cite{cirac,vidal}, and Density Matrix Renormalization Group \cite{white},
where the \textit{state} (usually the ground state of $1D$ chains)
can be adequately represented by a Matrix Product State (MPS) ansatz
whose parameters grow \textit{linearly} with the number of quantum
particles. Future work includes explicit treatment of fermionic systems
and numerical exploration of higher dimensional systems.
\end{onehalfspace}

\begin{onehalfspace}

\subsection{Inputs and Outputs of the Theory}
\end{onehalfspace}

\begin{onehalfspace}
In general we consider Hamiltonians $H=H_{\mbox{odd}\mbox{\ensuremath{\vphantom{even}}}}+H_{\mbox{even}\mbox{\ensuremath{\vphantom{odd}}}}$,
where the local terms that add up to $H_{\mbox{odd}\vphantom{\mbox{even}}}$
(or $H_{\mbox{even}\vphantom{\mbox{odd}}}$) form a commuting subset.
All the physically relevant quantities such as the lattice structure,
$N$, dimension of the spin $d$ and the rank $r$ are encoded in
the eigenvalue densities. The output of the theory is a $0\leq p\leq1$
by which the IE distribution is obtained and $d\nu^{IE}$ serves as
an approximation to the spectral measure. The inputs can succinctly
be expressed in terms of the dimension of the quantum spins, and the
nature of the lattice (Figure \ref{tab:Inputs-and-Outputs}).

\begin{figure}
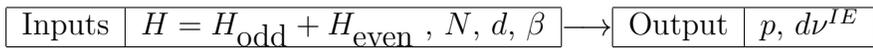

\begin{centering}
\begin{tabular}{|c|c|}
\hline 
Inputs & $H=H_{\mbox{odd}\vphantom{\mbox{even}}}+H_{\mbox{even}\vphantom{\mbox{odd}}}$
, $N$, $d$, $\beta$\tabularnewline
\hline 
\end{tabular}$\longrightarrow$%
\begin{tabular}{|c|c|}
\hline 
Output & $p$, $d\nu^{IE}$\tabularnewline
\hline 
\end{tabular}
\par\end{centering}

\caption{\label{tab:Inputs-and-Outputs}Inputs and outputs of the IE theory.
See section \ref{sec:Spectra-Sums-from-Prob} for the definition of
$d\nu^{IE}$.}
\end{figure}

\end{onehalfspace}

\begin{onehalfspace}

\subsection{\label{sub:More-Than-Four}More Than Four Moments of Accuracy?}
\end{onehalfspace}

\begin{onehalfspace}
Alternatives to IE worth considering are 1) Pearson and 2) Gram-Charlier
moment fits. 

We illustrate in Figure \ref{fig:IEvsOthersZERO} how the IE fit is
better than expected when matching four moments. We used the first
four moments to approximate the density using the Pearson fit as implemented
in MATLAB and also the well-known Gram-Charlier fit \cite{Cramer}.
In \cite{Gram} it was demonstrated that the statistical mechanics
methods for obtaining the DOS, when applied to a finite dimensional
vector space, lead to a Gaussian distribution in the lowest order.
Further, they discovered that successive approximations lead naturally
to the Gram-Charlier series \cite{Cramer}. Comparing these against
the accuracy of IE leads us to view IE as more than a moment matching
methodology. 

\begin{figure}
\begin{centering}
\includegraphics[width=8cm]{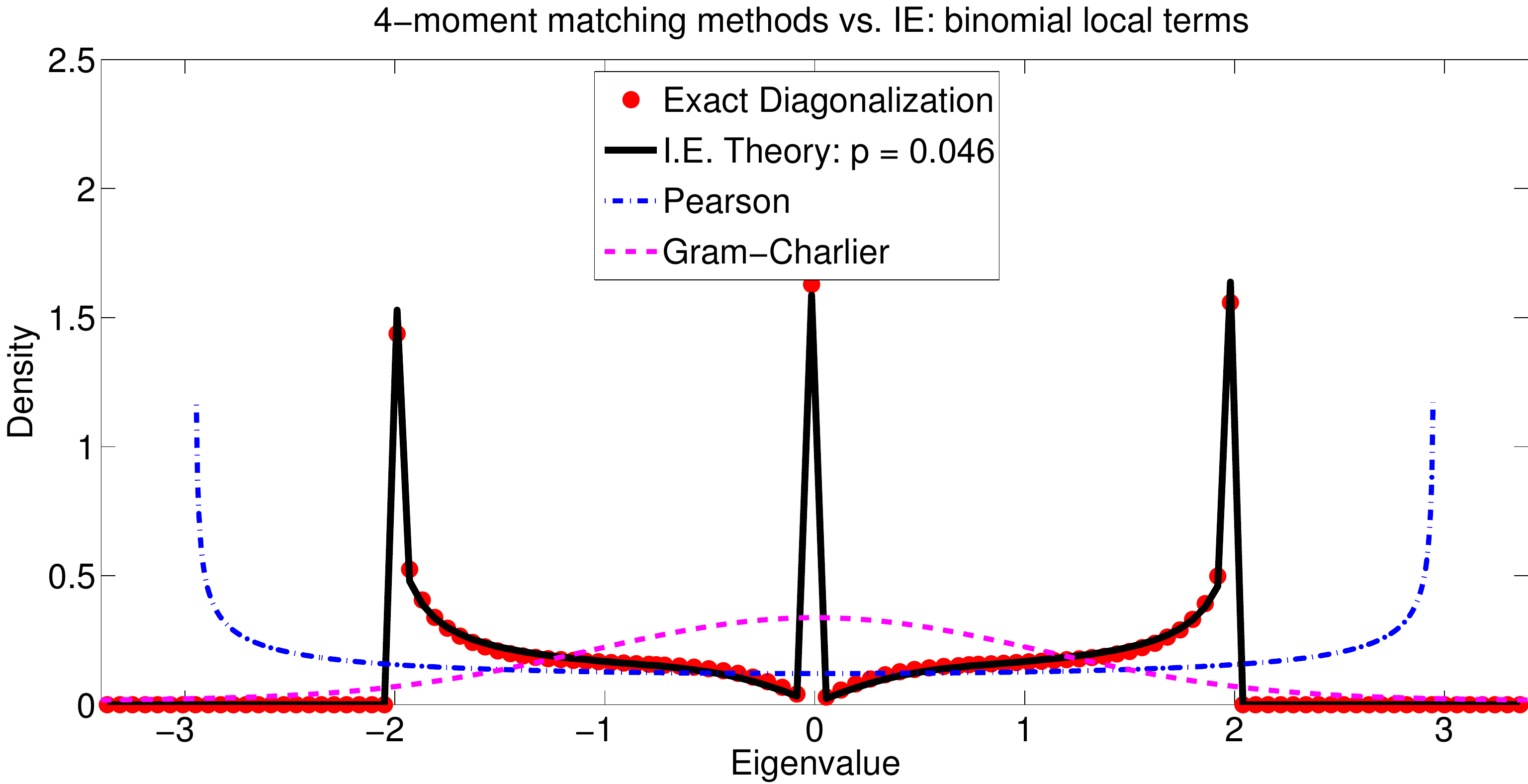}\includegraphics[width=8cm]{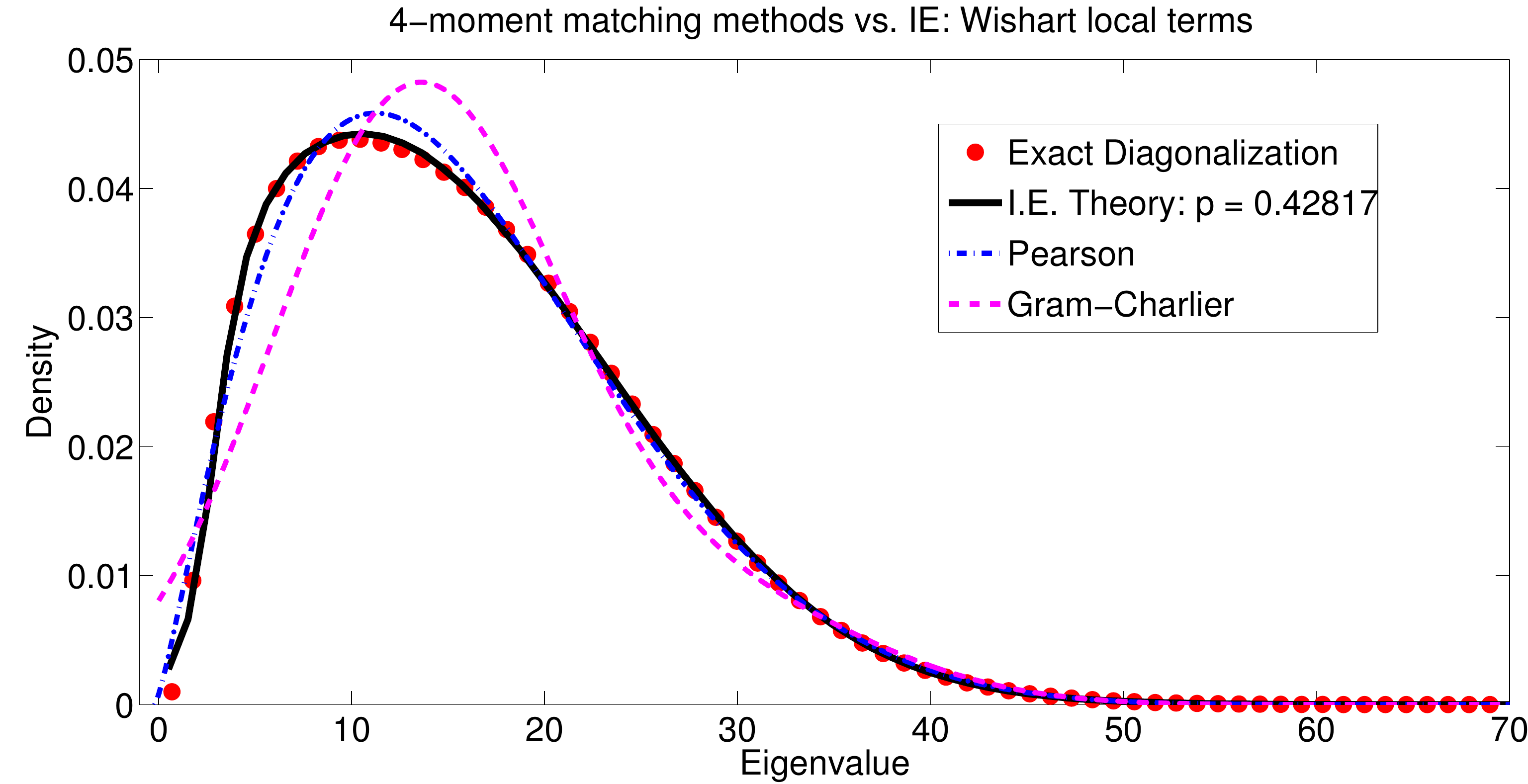}
\par\end{centering}

\centering{}\caption{\label{fig:IEvsOthersZERO}IE vs. Pearson and Gram-Charlier}
\end{figure}

The departure theorem (Section \ref{sub:departure}) shows that in
any of the higher moments ($>4$) there are many terms in the quantum
case that match IE exactly. Further, we conjecture that the effect
of the remaining terms are generally less significant. 
\end{onehalfspace}

\begin{onehalfspace}

\section{\label{sec:Spectra-Sums-from-Prob}Spectra Sums in Terms of Probability
Theory}
\end{onehalfspace}

\begin{onehalfspace}
The density of eigenvalues may be thought of as a histogram. Formally
for an $m\times m$ matrix $M$ the\textit{ eigenvalue distribution}
is \cite[p. 4]{zeitouni1}\cite[p. 101]{Hiai} 
\begin{equation}
d\nu_{M}(x)=\frac{1}{m}\sum_{i=1}^{m}\delta(x-\lambda_{i}\left(M\right)).
\end{equation}

For a random matrix, there is the expected eigenvalue distribution
\cite{alan}, \cite[p. 362]{speicher}
\begin{equation}
d\nu_{M}(x)=\frac{1}{m}\mathbb{E}\left[\sum_{i=1}^{m}\delta(x-\lambda_{i}\left(M\right))\right],
\end{equation}
 which is typically a smooth curve \cite[p. 101]{zeitouni1}\cite[p. 115]{Hiai}. 

The eigenvalue distributions above are measures on one variable. We
will also need the multivariate measure on the $m$ eigenvalues of
$M$: 
\[
d\hat{\nu}_{M}\left(x\right)\mbox{= The symmetrized joint density of the eigenvalues. }
\]

Given the densities for $M$ and $M',$ the question arises: What
kind of summation of densities might represent the density for $M+M'$?
This question is unanswerable without further information.

One might try to answer this using various assumptions based on probability
theory. The first assumption is the familiar ``classical'' probability
theory where the distribution of the sum is obtained by convolution
of the density of summands. Another assumption is the modern ``free''
probability theory; we introduce a finite version to obtain the ``isotropic''
theory. Our target problem of interest, the ``quantum'' problem,
we will demonstrate, practically falls nicely in between the two.
The ``Slider'' quantifies to what extent the quantum problem falls
in between (Figure \ref{fig:RoadMap} bottom). 
\end{onehalfspace}

\begin{onehalfspace}

\subsection{\label{sub:Classical}Classical}
\end{onehalfspace}

\begin{onehalfspace}
Consider random diagonal matrices $A$ and $B$ of size $m$, the
only randomness is in a uniform choice among the $m!$ possible orders.
Then there is no difference between the density of eigenvalue sums
and the familiar convolution of densities of random variables,
\begin{equation}
d\nu^{c}=d\nu_{A}*d\nu_{B}.
\end{equation}

Comment: From this point of view, the diagonal elements of $A,$ say,
are identically distributed random variables that need not be independent.
Consider Wishart matrices \cite{wishart}, where there are dependencies
among the eigenvalues. To be precise let $\mathbf{a}\in\mathbb{R}^{m}$
be a symmetric random variable, i.e., $P\mathbf{a}$ has the same
distribution as $\mathbf{a}$ for all permutation matrices $P$. We
write, $A=\left(\begin{array}{ccc}
a_{1}\\
 & \ddots\\
 &  & a_{m}
\end{array}\right)\equiv\textrm{diag}(\mathbf{a})$. Similarly for $B$.

Comment: The classical convolution appears in Figure \ref{fig:RoadMap}
in two different ways. Firstly, in the definition of $A$ (or $B$)
, the eigenvalues of the odd (or even) terms are added classically.
Secondly, $A$ and $B$ are added classically to form one end of the
Slider. 
\end{onehalfspace}

\begin{onehalfspace}

\subsection{\label{sub:Free}Free and Isotropic}
\end{onehalfspace}

\begin{onehalfspace}
Free probability \cite[is recommended]{speicher} provides a new natural
mathematical ``sum'' of random variables. This sum is computed ``free
convolution'' denoted 

\begin{equation}
d\nu^{f}=d\nu_{A}\boxplus d\nu_{B}.
\end{equation}
Here we assume the random matrices $A$ and $B$, representing the
eigenvalues, have densities $d\nu_{A}$ and $d\nu_{B}$. In the large
$m$ limit, we can compute the DOS of $A+Q^{T}BQ$, where $Q$ is
a $\beta-$Haar distributed matrix (see Table \ref{tab:field}). 

Comment: In this paper we will not explore the free approach strictly
other than observing that it is the infinite limit of the isotropic
approach (i.e., $t\rightarrow\infty$ in Eq. \ref{eq:HIso}). This
infinite limit is independent of the choice of $\beta$ (see Table
\ref{tab:field}).

\begin{table}[H]
\begin{centering}
\begin{tabular}{|c|c|c|c|c|}
\hline 
 & Real $\mathbb{R}$ & Complex $\mathbb{C}$ & Quaternions $\mathbb{H}$ & ``Ghosts'' \tabularnewline
\hline 
\hline 
$\beta$ & $1$ & $2$ & $4$ & general $\beta$\tabularnewline
\hline 
Notation & $Q$ & $U$ & $S$ & $\mathcal{Q_{\beta}}$\tabularnewline
\hline 
Haar matrices & orthogonal & unitary & symplectic & $\beta-$orthogonal\tabularnewline
\hline 
\end{tabular}
\par\end{centering}

\centering{}\caption{\label{tab:field}Various $\beta-$Haar matrices.}
\end{table}

We define an isotropic convolution. The isotropic sum depends on a
copying parameter $t$ and $\beta$ (Table \ref{tab:field}). The
new Hamiltonian is the \textit{isotropic Hamiltonian} (\textit{``iso}''):

\begin{equation}
H_{iso}\equiv\left(A'\otimes\mathbb{I}_{t}\right)+Q_{\beta}^{-1}\left(\mathbb{I}_{t}\otimes B'\right)Q_{\beta},\label{eq:HIso}
\end{equation}
where $Q_{\beta}$ is a $\beta-$Haar distributed matrix, $A=A'\otimes\mathbb{I}_{t}$
and $B=\mathbb{I}_{t}\otimes B'$. For the copying parameter $t=d$,
$H_{iso}$ has the same dimension as $H$ in Eq. \ref{eq:Hamiltonian};
however, $t>d$ allows us to formally treat problems of growing size.
We can recover the free convolution by taking the limit: $\lim_{t\rightarrow\infty}d\nu^{iso\left(\beta,t\right)}=d\nu^{f}$.
The effect of $Q_{\beta}$ is to spin the eigenvectors of $\mathbb{I}_{t}\otimes B$
to point isotropically with respect to the eigenvectors of $A$. We
denote the isotropic eigenvalue distribution by

\begin{equation}
d\nu^{iso\left(\beta,t\right)}=d\hat{\nu}_{A}\boxplus_{iso\left(\beta,t\right)}d\hat{\nu}_{B}
\end{equation}
omitting $t$ and $\beta$ when it is clear from the context. 

Comment: In Eq. \ref{eq:HIso}, the $\mathbb{I}_{t}$ and $B$ in
$\mathbb{I}_{t}\otimes B$, can appear in any order. We chose this
presentation in anticipation of the quantum problem.

Comment: In this paper we primarily consider $t$ to match the dimension
of $H$. 
\end{onehalfspace}

\begin{onehalfspace}

\subsection{\label{sub:Quantum}Quantum}
\end{onehalfspace}

\begin{onehalfspace}
Let $d\nu^{q}$ denote the eigenvalue distribution for the Hamiltonian
in Eq. \ref{eq:Hamiltonian}. This is the distribution that we will
approximate by $d\nu^{IE}$. In connection to Figure \ref{fig:Odd-and-even}
the Hamiltonian can be written as

\begin{equation}
H\equiv H_{\mbox{odd}\vphantom{\mbox{even}}}+H_{\mbox{even}\mbox{\ensuremath{\vphantom{odd}}}}=\sum_{l=1,3,5,\cdots}\mathbb{I}\otimes H_{l,l+1}\otimes\mathbb{I}+\sum_{l=2,4,6,\cdots}\mathbb{I}\otimes H_{l,l+1}\otimes\mathbb{I}.\label{eq:HevenHodd}
\end{equation}

We proceed to define a ``quantum convolution'' on the distributions
$d\hat{\nu}_{A}$ and $d\hat{\nu}_{B}$, which is $\beta$-dependent

\begin{equation}
d\nu^{q\left(\beta\right)}=d\hat{\nu}_{A}\boxplus_{q}d\hat{\nu}_{B}.
\end{equation}

In general, without any connection to a Hamiltonian, let $d\hat{\nu}_{A}$
and $d\hat{\nu}_{B}$ be symmetric measures on $\mathbb{R}^{d^{N}}$.
We define $d\nu^{q\left(\beta\right)}$ to be the eigenvalue distribution
of 

\begin{equation}
H=A+Q_{q}^{-1}BQ_{q},
\end{equation}

where $Q_{q}=\left(Q_{q}^{(A)}\right)^{-1}Q_{q}^{(B)}$ with

\begin{equation}
\begin{array}{c}
Q_{q}^{(A)}=\left[\bigotimes_{i=1}^{\left(N-1\right)/2}Q_{i}^{(O)}\right]\otimes\mathbb{I}_{d}\;\textrm{and}\; Q_{q}^{(B)}=\mathbb{I}_{d}\otimes\left[\bigotimes_{i=1}^{\left(N-1\right)/2}Q_{i}^{(E)}\right]\qquad N\;\textrm{odd}\\
\\
\qquad Q_{q}^{(A)}=\left[\bigotimes_{i=1}^{N/2}Q_{i}^{(O)}\right]\;\textrm{and}\; Q_{q}^{(B)}=\mathbb{I}_{d}\otimes\left[\bigotimes_{i=1}^{N/2-1}Q_{i}^{(E)}\right]\otimes\mathbb{I}_{d}\qquad N\;\textrm{even}
\end{array}\label{eq:OddEven}
\end{equation}

and each $Q_{i}^{\left(\bullet\right)}$ is a $\beta-$Haar measure
orthogonal matrix of size $d^{2}$ and $\mathbb{I}_{d}$ is an identity
matrix of size $d$.

Comment: $A$, $B$ and $Q_{q}$ are $d^{N}\times d^{N}.$

Comment: In our examples given in this paper, we assume the local
terms are independent and identically distributed (iid) random matrices,
each of which has eigenvectors distributed with $\beta-$Haar measure. 

The tensor product in (\ref{eq:OddEven}) succinctly summarizes the
departure of the quantum case from a generic matrix as well as from
the classical case. First of all the number of parameters in $Q_{q}$
grows linearly with $N$ whereas in $Q$ it grows exponentially with
$N$. Second, the quantum case possesses isotropicity that makes it
different from the classical, whose eigenvectors are a point on the
orthogonal group (i.e., the identity matrix).

Comment: General $\beta$'s can be treated formally \cite{alanGhost}.
In particular, for quantum mechanical problems $\beta$ is taken to
be $1$ or $2$ corresponding to real and complex entries in the local
terms. $\beta=4$ corresponds to quaternions.
\end{onehalfspace}
\begin{defn*}
\begin{onehalfspace}
\label{def:The-Hadamard-product}The \textbf{Hadamard product} of
two matrices $M_{1}$ and $M_{2}$ of the same size, denoted by $M_{1}\circ M_{2}$,
is the product of the corresponding elements. \end{onehalfspace}
\end{defn*}
\begin{lem}
\begin{onehalfspace}
\label{lem:var}The elements of $Q_{q}$ defined in Eq. \ref{eq:OddEven}
are (dependent) random variables with mean zero and variance $d^{-N}$.\end{onehalfspace}
\end{lem}
\begin{proof}
\begin{onehalfspace}
Here expectations are taken with respect to the random matrix $Q_{q}$
which is built from local Haar measure matrices by Eq. \ref{eq:OddEven}.
The fact that $\mathbb{E}\left(Q_{q}^{\left(A\right)}\right)=\mathbb{E}\left(Q_{q}^{\left(B\right)}\right)=0_{d^{N}}$
follows from the Haar distribution of local terms. Thus $\mathbb{E}\left(Q_{q}\right)=0$
by independence of $Q_{q}^{\left(A\right)}$ and $Q_{q}^{\left(B\right)}$.
Further, each element in $Q_{q}$ involves a dot product between columns
of $Q_{q}^{\left(A\right)}$ and $Q_{q}^{\left(B\right)}$. In every
given column of $Q_{q}^{\left(A\right)}$ any nonzero entry is a distinct
product of entries of local $Q's$ (see Eq.\ref{eq:OddEven}). For
example the expectation value of the $1,1$ entry is $\mathbb{E}\left(q_{i,1}^{\left(A\right)}q_{j,1}^{\left(A\right)}q_{i,1}^{\left(B\right)}q_{j,1}^{\left(B\right)}\right)=\mathbb{E}\left(q_{i,1}^{\left(A\right)}q_{j,1}^{\left(A\right)}\right)\mathbb{E}\left(q_{i,1}^{\left(B\right)}q_{j,1}^{\left(B\right)}\right)$.
Because of the Haar measure of the local terms, this expectation is
zero unless $i=j$. We then have that

\begin{equation}
\begin{array}{c}
\mathbb{E}\left(Q_{q}\circ Q_{q}\right)=\mathbb{E}\left(Q_{q}^{\left(A\right)}\circ Q_{q}^{\left(A\right)}\right)^{T}\mathbb{E}\left(Q_{q}^{\left(B\right)}\circ Q_{q}^{\left(B\right)}\right)=\\
\begin{cases}
\begin{array}{c}
\left(\left[\bigotimes_{i=1}^{\left(N-1\right)/2}d^{-2}J_{d^{2}}\right]\otimes\mathbb{I}_{d}\right)\left(\mathbb{I}_{d}\otimes\left[\bigotimes_{i=1}^{\left(N-1\right)/2}d^{-2}J_{d^{2}}\right]\right)\qquad N\;\textrm{odd}\\
\left(\bigotimes_{i=1}^{N/2}d^{-2}J_{d^{2}}\right)\left(\mathbb{I}_{d}\otimes\left[\bigotimes_{i=1}^{N/2-1}d^{-2}J_{d^{2}}\right]\otimes\mathbb{I}_{d}\right)\qquad\qquad\quad N\;\mbox{even}
\end{array}\end{cases}\\
=d^{-N}J_{d^{N}},
\end{array}
\end{equation}
where $J_{i}=i\times i$ matrix of all ones. We use facts such as
$\left(J_{i}/i\right)^{2}=\left(J_{i}/i\right)$, $\left(J_{i}/i\right)\otimes\left(J_{i}/i\right)=\left(J_{i^{2}}/i^{2}\right)$
and the variance of the elements of an $i\times i$ $\beta-$Haar
matrix is $1/i$.\end{onehalfspace}

\end{proof}
\begin{onehalfspace}

\section{\label{sec:Theory-of-Isotropic}Theory of Isotropic Entanglement}
\end{onehalfspace}

\begin{onehalfspace}

\subsection{\label{sub:convex}Isotropic Entanglement as the Combination of Classical
and Isotropic}
\end{onehalfspace}

\begin{onehalfspace}
We create a ``Slider'' based on the fourth moment. The moment $m_{k}$
of a random variable defined in terms of its density is $m_{k}=\int x^{k}d\nu_{M}.$
For the eigenvalues of an $m\times m$ random matrix, this is $\frac{1}{m}\mathbb{E}\mbox{Tr}M^{k}.$
In general, the moments of the classical sum and the free sum are
different, but the first three moments, $m_{1},\ m_{2},$ and $m_{3}$
are theoretically equal \cite[p. 191]{speicher}. Further, to anticipate
our target problem, the first three moments of the quantum eigenvalues
are also equal to that of the iso and the classical (The Departure
and the Three Moments Matching theorems in Section \ref{sub:departure}).
These moments are usually encoded as the mean, variance, and skewness.

We propose to use the fourth moment (or the excess kurtosis) to choose
a correct $p$ from a sliding hybrid sum:

\begin{equation}
d\nu^{q}\approx d\nu^{IE}=pd\nu^{c}+(1-p)d\nu^{iso}\label{eq:dNuq}
\end{equation}

Therefore, we find $p$ that expresses the kurtosis of the quantum
sum $(\gamma_{2}^{q}$) in terms of the kurtoses of the classical
($\gamma_{2}^{c}$) and isotropic ($\gamma_{2}^{iso}$) sums:

\begin{equation}
\gamma_{2}^{q}=p\gamma_{2}^{c}+\left(1-p\right)\gamma_{2}^{iso}\Rightarrow\qquad p=\frac{\gamma_{2}^{q}-\gamma_{2}^{iso}}{\gamma_{2}^{c}-\gamma_{2}^{iso}}.\label{eq:convex}
\end{equation}
Recall that the kurtosis $\gamma_{2}\equiv\frac{m_{4}}{\sigma^{4}}$,
where $\sigma^{2}$ is the variance. Hence kurtosis is the correct
statistical quantity that encodes the fourth moments:

\begin{equation}
m_{4}^{c}=\frac{1}{d^{N}}\mathbb{E}\textrm{Tr}\left(A+\Pi^{T}B\Pi\right)^{4},\; m_{4}^{iso}=\frac{1}{d^{N}}\mathbb{E}\textrm{Tr}\left(A+Q^{T}BQ\right)^{4},\; m_{4}^{q}=\frac{1}{d^{N}}\mathbb{E}\textrm{Tr}\left(A+Q_{q}^{T}BQ_{q}\right)^{4},\label{eq:moments}
\end{equation}
where $\Pi$ is a random uniformly distributed permutation matrix,
$Q$ is a $\beta-$Haar measure orthogonal matrix of size $d^{N}$,
and $Q_{q}$ is given by Eq. \ref{eq:OddEven}. \textit{Unless stated
otherwise, in the following the expectation values are taken with
respect to random eigenvalues $A$ and $B$ and eigenvectors. The
expectation values over the eigenvectors are taken with respect to
random permutation $\Pi$, $\beta-$Haar $Q$ or $Q_{q}$ matrices
for classical, isotropic or quantum cases respectively. }
\end{onehalfspace}

\begin{onehalfspace}

\subsection{\label{sub:departure}The Departure and The Matching Three Moments
Theorems}
\end{onehalfspace}

\begin{onehalfspace}
In general we have the $i^{\textrm{th}}$ moments:

\begin{eqnarray*}
m_{i}^{iso} & = & \frac{1}{m}\mathbb{E}\textrm{Tr}\left(A+Q^{T}BQ\right)^{i}\\
m_{i}^{q} & = & \frac{1}{m}\mathbb{E}\textrm{Tr}\left(A+Q_{q}^{T}BQ_{q}\right)^{i},\mbox{ and }\\
m_{i}^{c} & = & \frac{1}{m}\mathbb{E}\textrm{Tr}\left(A+\Pi^{T}B\Pi\right)^{i}.
\end{eqnarray*}
where $m\equiv d^{N}$. If we expand the moments above we find some
terms can be put in the form $\mathbb{E}\textrm{Tr}\left(A^{m_{1}}Q_{\bullet}^{T}B^{m_{2}}Q_{\bullet}\right)$
and the remaining terms can be put in the form $\mathbb{E}\textrm{Tr}\left\{ \ldots Q_{\bullet}^{T}B^{\ge1}Q_{\bullet}A^{\ge1}Q_{\bullet}^{T}B^{\ge1}Q_{\bullet}\ldots\right\} .$
The former terms we denote \textit{non-departing}; the remaining terms
we denote \textit{departing}.

For example, when $i=4$, 

\begin{eqnarray}
m_{4}^{iso} & = & \frac{1}{m}\mathbb{E}\left\{ \textrm{Tr}\left[A^{4}+4A^{3}Q^{T}BQ+4A^{2}Q^{T}B^{2}Q+4AQ^{T}B^{3}Q+\mathbf{\underline{2\left(\mathbf{AQ^{T}BQ}\right)^{2}}}+B^{4}\right]\right\} \label{eq:fourthMoments}\\
m_{4}^{q} & = & \frac{1}{m}\mathbb{E}\left\{ \textrm{Tr}\left[A^{4}+4A^{3}Q_{q}^{T}BQ_{q}+4A^{2}Q_{q}^{T}B^{2}Q_{q}+4AQ_{q}^{T}B^{3}Q_{q}+\underline{\mathbf{2\left(AQ_{q}^{T}BQ_{q}\right)^{2}}}+B^{4}\right]\right\} \nonumber \\
m_{4}^{c} & = & \frac{1}{m}\mathbb{E}\left\{ \textrm{Tr}\left[A^{4}+4A^{3}\Pi^{T}B\Pi+4A^{2}\Pi^{T}B^{2}\Pi+4A\Pi^{T}B^{3}\Pi+\mathbf{\underline{2\left(A\Pi^{T}B\Pi\right)^{2}}}+B^{4}\right]\right\} ,\nonumber 
\end{eqnarray}
where the only departing terms and the corresponding classical term
are shown as underlined and bold faced.
\end{onehalfspace}
\begin{thm*}
\begin{onehalfspace}
\textbf{\textup{\label{thm:(The-Departure-Theorem)}(The Departure
Theorem)}}\textbf{ }The moments of the quantum, isotropic and classical
sums differ only in the departing terms: $\mathbb{E}\textrm{Tr}\left\{ \ldots Q_{\bullet}^{T}B^{\ge1}Q_{\bullet}A^{\ge1}Q_{\bullet}^{T}B^{\ge1}Q_{\bullet}\ldots\right\} .$ \end{onehalfspace}
\end{thm*}
\begin{proof}
\begin{onehalfspace}
Below the repeated indices are summed over. If $A$ and $B$ are any
diagonal matrices, and $Q_{\bullet}$ is $Q$ or $Q_{q}$ or $\Pi$
of size $m\times m$ then $\mathbb{E}\left(q_{ij}^{2}\right)=1/m$
, by symmetry and by Lemma \ref{lem:var} respectively. Since the
$\mathbb{E}\textrm{Tr}\left(AQ_{\bullet}^{T}BQ_{\bullet}\right)=\mathbb{E}\left(q_{ij}^{2}a_{i}b_{j}\right)$,
where expectation is taken over randomly ordered eigenvalues and eigenvectors;
the expected value is $m^{2}\left(\frac{1}{m}\right)\mathbb{E}\left(a_{i}b_{j}\right)$
for any $i$ or $j$. Hence, $\frac{1}{m}\mathbb{E}\textrm{Tr}\left(AQ_{\bullet}^{T}BQ_{\bullet}\right)=\mathbb{E}\left(a_{i}b_{j}\right)=\mathbb{E}\left(a_{i}\right)\mathbb{E}\left(b_{j}\right)$,
which is equal to the classical value. The first equality is implied
by permutation invariance of entries in $A$ and $B$ and the second
equality follows from the independence of $A$ and $B$. \end{onehalfspace}

\end{proof}
\begin{onehalfspace}
Therefore, the three cases differ only in the terms $\frac{2}{m}\mathbb{E}\textrm{Tr}\left(AQ^{T}BQ\right)^{2}$,
$\frac{2}{m}\mathbb{E}\textrm{Tr}\left(AQ_{q}^{T}BQ_{q}\right)^{2}$
and $\frac{2}{m}\mathbb{E}\textrm{Tr}\left(A\Pi^{T}B\Pi\right)^{2}$
in Eq. \ref{eq:fourthMoments}.
\end{onehalfspace}
\begin{thm*}
\begin{onehalfspace}
\textbf{\textup{\label{thm:(The-Matching-Three}(The Matching Three
Moments Theorem) }}The first three moments of the quantum, iso and
classical sums are equal.\end{onehalfspace}
\end{thm*}
\begin{proof}
\begin{onehalfspace}
The first three moments are
\begin{equation}
\begin{array}{c}
m_{1}^{\left(\bullet\right)}=\frac{1}{m}\mathbb{E}\textrm{Tr}\left(A+B\right)\\
m_{2}^{\left(\bullet\right)}=\frac{1}{m}\mathbb{E}\textrm{Tr}\left(A+Q_{\bullet}^{T}BQ_{\bullet}\right)^{2}=\frac{1}{m}\mathbb{E}\textrm{Tr}\left(A^{2}+2AQ_{\bullet}^{T}BQ_{\bullet}+B^{2}\right)\\
m_{3}^{\left(\bullet\right)}=\frac{1}{m}\mathbb{E}\textrm{Tr}\left(A+Q_{\bullet}^{T}BQ_{\bullet}\right)^{3}=\frac{1}{m}\mathbb{E}\textrm{Tr}\left(A^{3}+3A^{2}Q_{\bullet}^{T}BQ_{\bullet}+3AQ_{\bullet}^{T}B^{2}Q_{\bullet}+B^{3}\right),
\end{array}
\end{equation}
where $Q_{\bullet}$ is $Q$ and $Q_{q}$ for the iso and the quantum
sums respectively and we used the familiar trace property $\textrm{Tr}(M_{1}M_{2})=\textrm{Tr}(M_{2}M_{1})$.
The equality of the first three moments of the iso and quantum with
the classical follows from The Departure Theorem.\end{onehalfspace}

\end{proof}
\begin{onehalfspace}
Furthermore, in the expansion of any of the moments $>4$ all the
non-departing terms are exactly captured by IE. These terms are equal
to the corresponding terms in the classical and the isotropic and
therefore equal to any linear combination of them. The departing terms
in higher moments (i.e.,$>4$) that are approximated by IE, we conjecture
are of little relevance. For example, the fifth moment has only two
terms (shown in bold) in its expansion that are departing: 

\begin{equation}
\begin{array}{c}
m_{5}=\frac{1}{m}\mathbb{E}\textrm{Tr}\left(A^{5}+5A^{4}Q_{\bullet}^{T}BQ_{\bullet}+\mathit{5}A^{3}Q_{\bullet}^{T}B^{2}Q_{\bullet}+\mathit{5}A^{2}Q_{\bullet}^{T}B^{3}Q_{\bullet}+\mathbf{\underline{5A\left(AQ_{\bullet}^{T}BQ_{\bullet}\right)^{2}}+}\right.\\
\left.\mathbf{\underline{5\left(AQ_{\bullet}^{T}BQ_{\bullet}\right)^{2}Q_{\bullet}^{T}BQ_{\bullet}}}+\mathit{5}AQ_{\bullet}^{T}B^{4}Q_{\bullet}+B^{5}\right)
\end{array}\label{eq:fifthmoments}
\end{equation}

\begin{table}
\begin{centering}
\begin{tabular}{|c|c|c|c|c|c|}
\hline 
\noalign{\vskip2sp}
number of & number of odds  & number of odds  & size of & Number of  & Dimension of\tabularnewline
\noalign{\vskip2sp}
 sites & or evens ($N$ odd) & or evens ($N$ even) &  $H$ & copies & quidits\tabularnewline
\hline 
$N$ & $k=\frac{N-1}{2}$ & $k_{\mbox{odd}\vphantom{\mbox{even}}}=\frac{N}{2},\; k_{\mbox{even}\mbox{\ensuremath{\vphantom{odd}}}}=\frac{N-2}{2}$ & $m=d^{N}$ & $t$ & $d$\tabularnewline
\hline 
\end{tabular}
\par\end{centering}

\vspace{0.3in}

\begin{centering}
\begin{tabular}{|c|c|c|c|c|c|c|}
\hline 
size of local terms & $l^{th}$ moment & $l^{th}$ cumulant & mean & variance & skewness & kurtosis\tabularnewline
\hline 
\hline 
$n=d^{2}$ & $m_{l}$ & $\kappa_{l}$ & $\mu$ & $\sigma^{2}$ & $\gamma_{1}$ & $\gamma_{2}$\tabularnewline
\hline 
\end{tabular}
\par\end{centering}

\caption{\label{tab:parameters}Notation}
\end{table}

By the Departure Theorem the numerator in Eq. \ref{eq:convex} becomes,

\begin{equation}
\gamma_{2}^{q}-\gamma_{2}^{iso}=\frac{\kappa_{4}^{q}-\kappa_{4}^{iso}}{\sigma^{4}}=\frac{2}{m}\frac{\mathbb{E}\left\{ \textrm{Tr}\left[\left(AQ_{q}^{T}BQ_{q}\right)^{2}-\left(AQ^{T}BQ\right)^{2}\right]\right\} }{\sigma^{4}}\label{eq:numer}
\end{equation}
and the denominator in Eq. \ref{eq:convex} becomes, 

\begin{equation}
\gamma_{2}^{c}-\gamma_{2}^{iso}=\frac{\kappa_{4}^{c}-\kappa_{4}^{iso}}{\sigma^{4}}=\frac{2}{m}\frac{\mathbb{E}\left\{ \textrm{Tr}\left[\left(A\Pi^{T}B\Pi\right)^{2}-\left(AQ^{T}BQ\right)^{2}\right]\right\} }{\sigma^{4}},\label{eq:denom}
\end{equation}
where as before, $Q$ is a $\beta-$Haar measure orthogonal matrix
of size $m=d^{N}$, $Q_{q}=\left(Q_{q}^{(A)}\right)^{T}Q_{q}^{(B)}$
given by Eq. \ref{eq:OddEven} and $\kappa_{4}^{\bullet}$ denote
the fourth cumulants. Therefore, evaluation of $p$ reduces to the
evaluation of the right hand sides of Eqs. \ref{eq:numer} and \ref{eq:denom}.

Below we do not want to restrict ourselves to only chains with odd
number of sites and we need to take into account the multiplicity
of the eigenvalues as a result of taking the tensor product with identity.
It is convenient to denote the size of the matrices involved by $m=d^{N}=tn^{k}$,
where $n=d^{2}$ and $k=\frac{N-1}{2}$ and $t$ is the number of
copies (Section \ref{sub:Free} and Table \ref{tab:parameters}). 
\end{onehalfspace}

\begin{onehalfspace}

\subsection{Distribution of $A$ and $B$ }
\end{onehalfspace}

\begin{onehalfspace}
The goal of this section is to express the moments of the entries
of $A$ and $B$ (e.g., $m_{2}^{A}$ and $m_{1,1}^{A}$) in terms
of the moments of the local terms (e.g for odd local terms $m_{2}^{\mbox{odd}},m_{11}^{\mbox{odd}}$).
Note that $A$ and $B$ are independent. The odd summands that make
up $A$ all commute and therefore can be locally diagonalized to give
the diagonal matrix $A$ (similarly for $B$),

\begin{eqnarray}
A & = & \sum_{i=1,3,\cdots}^{N-2}\mathbb{I}\otimes\Lambda_{i}\otimes\mathbb{I}\label{eq:A_and_B}\\
B & = & \sum_{i=2,4,\cdots}^{N-1}\mathbb{I}\otimes\Lambda_{i}\otimes\mathbb{I},\nonumber 
\end{eqnarray}
where $\Lambda_{i}$ are of size $d^{2}$ and are the diagonal matrices
of the local eigenvalues. 

The diagonal matrices $A$ and $B$ are formed by a direct sum of
the local eigenvalues of odds and evens respectively. For open boundary
conditions (OBC) each entry has a multiplicity given by Table \ref{tab:multiplicity}.

\begin{table}[H]
\begin{centering}
\begin{tabular}{|c|c|c|}
\hline 
OBC & $N$ odd & $N$ even\tabularnewline
\hline 
\hline 
$A$ & $d$ & $1$\tabularnewline
\hline 
$B$ & $d$ & $d^{2}$\tabularnewline
\hline 
\end{tabular}
\par\end{centering}

\centering{}\caption{\label{tab:multiplicity}The multiplicity of terms in $A$ and $B$
for OBC. For closed boundary conditions there is no repetition.}
\end{table}

Comment: We emphasize that $A$ and $B$ are independent of the eigenvector
structures. In particular, $A$ and $B$ are the same among the three
cases of isotropic, quantum and classical. 

We calculate the moments of $A$ and $B$. Let us treat the second
moment of $A$ ($B$ is done the same way). By the permutation invariance
of entries in $A$

\begin{eqnarray}
m_{2}^{A}\equiv\mathbb{E}\left(a_{1}^{2}\right) & = & \mathbb{E}\left(\lambda_{i_{1}}^{\left(1\right)}+\cdots+\lambda_{i_{k}}^{\left(k\right)}\right)^{2}\nonumber \\
 & = & \mathbb{E}\left[k\left(\lambda^{2}\right)+k\left(k-1\right)\lambda^{\left(1\right)}\lambda^{\left(2\right)}\right]\label{eq:m_2}\\
 & = & km_{2}^{\mbox{odd}}+k\left(k-1\right)m_{1,1}^{\mbox{odd}}\nonumber 
\end{eqnarray}
where expectation is taken over randomly chosen local eigenvalues,
$m_{2}^{\mbox{odd}}\equiv\mathbb{E}\left(\lambda_{i}^{2}\right)$
and $m_{1,1}^{\mbox{odd}}\equiv\mathbb{E}\left(\lambda_{i}\lambda_{j}\right)$
for some uniformly chosen $i$ and $j$ with $i\ne j$. The permutation
invariance assumption implies $\mathbb{E}\left(a_{i}^{2}\right)=\mathbb{E}\left(a_{1}^{2}\right)$
for all $i=1\cdots m$.

Comment: The key to this argument giving $m_{2}^{A}$ is that the
indices are not sensitive to the copying that results from the tensor
product with $\mathbb{I}_{d}$ at the boundaries.

Next we calculate the correlation between two diagonal terms, namely
$m_{1,1}^{A}\equiv\mathbb{E}\left(a_{i}a_{j}\right)$ for $i\neq j$.
We need to incorporate the multiplicity, denoted by $t$, due to the
tensor product with an identity matrix at the end of the chain,

\begin{eqnarray}
m_{1,1}^{A} & = & \frac{1}{m\left(m-1\right)}\mathbb{E}\left\{ \left(\sum_{i_{1},\cdots,i_{k}=1}^{n}\lambda_{i_{1}}^{\left(1\right)}+\cdots+\lambda_{i_{k}}^{\left(k\right)}\right)^{2}-\sum_{i_{1},\cdots,i_{k}=1}^{n}\left(\lambda_{i_{1}}^{\left(1\right)}+\cdots+\lambda_{i_{k}}^{\left(k\right)}\right)^{2}\right\} \label{eq:elemGeneral}\\
 & = & k\left(k-1\right)\mathbb{E}\left(\lambda\right)^{2}+k\left\{ \textrm{prob}\left(\lambda^{2}\right)\mathbb{E}\left(\lambda^{2}\right)+\textrm{prob}\left(\lambda_{1}\lambda_{2}\right)\mathbb{E}\left(\lambda_{1}\lambda_{2}\right)\right\} \nonumber \\
 & = & k\left(k-1\right)m_{2}^{\mbox{odd}}+\frac{k}{m-1}\left\{ \left(tn^{k-1}-1\right)m_{2}^{\mbox{odd}}+\left(tn^{k-1}\left(n-1\right)\right)m_{1,1}^{\mbox{odd}}\right\} \nonumber 
\end{eqnarray}
where, $\textrm{prob}\left(\lambda^{2}\right)=\frac{tn^{k-1}-1}{m-1}$
and $\textrm{prob}\left(\lambda_{1}\lambda_{2}\right)=\frac{tn^{k-1}\left(n-1\right)}{m-1}$.
Similarly for $B$.\bigskip{}

\end{onehalfspace}

\begin{onehalfspace}

\subsection{\label{sub:Isotropic-theory} Evaluation and Universality of $p=\frac{\gamma_{2}^{q}-\gamma_{2}^{iso}}{\gamma_{2}^{c}-\gamma_{2}^{iso}}$}
\end{onehalfspace}

\begin{onehalfspace}
Recall the definition of $p$; from Eqs. \ref{eq:convex}, \ref{eq:numer}
and \ref{eq:denom} we have,

\begin{equation}
1-p=\frac{\mathbb{E}\mbox{Tr}\left(A\Pi^{T}B\Pi\right)^{2}-\mathbb{E}\textrm{Tr}\left(AQ_{q}^{T}BQ_{q}\right)^{2}}{\mathbb{E}\mbox{Tr}\left(A\Pi^{T}B\Pi\right)^{2}-\mathbb{E}\textrm{Tr}\left(AQ^{T}BQ\right)^{2}}.\label{eq:1-p}
\end{equation}

The classical case

\begin{equation}
\frac{1}{m}\mathbb{E}\mbox{Tr}\left(A\Pi^{T}B\Pi\right)^{2}=\frac{1}{m}\mathbb{E}\sum_{i=1}^{m}a_{i}^{2}b_{i}^{2}=\mathbb{E}\left(a_{i}^{2}\right)\mathbb{E}\left(b_{i}^{2}\right)=m_{2}^{A}m_{2}^{B}.\label{eq:classical_depart}
\end{equation}
\begin{table}
\begin{centering}
\begin{tabular}{|c|c|c|}
\hline 
moments & expectation values & count\tabularnewline
\hline 
\hline 
$m_{2}^{2}$ & $\mathbb{E}\left(\left|q_{i,j}\right|^{4}\right)=\frac{\beta+2}{m\left(m\beta+2\right)}$ & $m^{2}$\tabularnewline
\hline 
$m_{2}m_{11}$ & $\mathbb{E}\left(\left|q_{1,1}q_{1,2}\right|^{2}\right)=\frac{\beta}{m\left(m\beta+2\right)}$ & $2m^{2}\left(m-1\right)$\tabularnewline
\hline 
$\left(m_{11}\right)^{2}$ & $\mathbb{E}\left(q_{1,1}\overline{q_{1,2}}\overline{q_{2,1}}q_{2,2}\right)=-\frac{\beta}{m\left(m\beta+2\right)\left(m-1\right)}$ & $m^{2}\left(m-1\right)^{2}$\tabularnewline
\hline 
 & $\mathbb{E}\left(q_{13}^{2}q_{24}^{2}\right)=\frac{\beta\left(n-1\right)+2}{n\left(n\beta+2\right)\left(n-1\right)}$ & \tabularnewline
\hline 
\end{tabular}
\par\end{centering}

\centering{}\caption{\label{tab:HaarExp} The expectation values and counts of colliding
terms in $Q$ when it is $\beta-$Haar distributed. In this section
we use the first row; we include the last three rows for the calculations
in the appendix.}
\end{table}

Comment: Strictly speaking after the first equality we must have used
$b_{\pi_{i}}$ instead of $b_{i}$ but we simplified the notation
as they are the same in an expectation sense.

The general form for the denominator of Eq. \ref{eq:1-p} is

\begin{equation}
\frac{1}{m}\mathbb{E}\textrm{Tr}\left[\left(A\Pi^{T}B\Pi\right)^{2}-\left(AQ^{T}BQ\right)^{2}\right]=\frac{1}{m}\mathbb{E}\left\{ a_{l}^{2}b_{l}^{2}-a_{i}a_{k}b_{j}b_{p}\left(q_{ji}q_{jk}q_{pk}q_{pi}\right)\right\} .\label{eq:isotropic}
\end{equation}
It's worth noting that the arguments leading to Eq. \ref{eq:Iso-Classical_FINAL}
hold even if one fixes $A$ and $B$ and takes expectation values
over $\Pi$ and a permutation invariant $Q$ whose entries have the
same expectation value. The right hand side of Eq. \ref{eq:Iso-Classical_FINAL}
is a homogeneous polynomial of order two in the entries of $A$ and
$B$; consequently it necessarily has the form

\[
\frac{1}{m}\mathbb{E}\textrm{Tr}\left[\left(A\Pi^{T}B\Pi\right)^{2}-\left(AQ^{T}BQ\right)^{2}\right]=c_{1}\left(B,Q\right)m_{2}^{A}+c_{2}\left(B,Q\right)m_{1,1}^{A}
\]
but Eq. \ref{eq:isotropic} must be zero for $A=I$, for which $m_{2}^{A}=m_{1,1}^{A}=1$.
This implies that $c_{1}=-c_{2}$, allowing us to factor out $\left(m_{2}^{A}-m_{1,1}^{A}\right)$.
Similarly, the homogeneity and permutation invariance of $B$ implies,

\[
\frac{1}{m}\mathbb{E}\textrm{Tr}\left[\left(A\Pi^{T}B\Pi\right)^{2}-\left(AQ^{T}BQ\right)^{2}\right]=\left(m_{2}^{A}-m_{1,1}^{A}\right)\left(D_{1}\left(Q\right)m_{2}^{B}+D_{2}\left(Q\right)m_{1,1}^{B}\right).
\]
The right hand side should be zero for $B=I$, whereby we can factor
out $\left(m_{2}^{B}-m_{1,1}^{B}\right)$

\begin{equation}
\frac{1}{m}\mathbb{E}\textrm{Tr}\left[\left(A\Pi^{T}B\Pi\right)^{2}-\left(AQ^{T}BQ\right)^{2}\right]=\left(m_{2}^{A}-m_{1,1}^{A}\right)\left(m_{2}^{B}-m_{1,1}^{B}\right)f\left(Q\right),\label{eq:FreeExpect}
\end{equation}
where $m_{2}^{A}=\mathbb{E}\left(a_{i}^{2}\right)$, $ $ $m_{2}^{B}=\mathbb{E}\left(b_{j}^{2}\right)$,
and $m_{1,1}^{A}=\mathbb{E}\left(a_{i},a_{j}\right)$ , $m_{1,1}^{B}=\mathbb{E}\left(b_{i},b_{j}\right)$.
Moreover $f\left(Q\right)$ is a homogeneous function of order four
in the entries of $Q$. To evaluate $f\left(Q\right)$, it suffices
to let $A$ and $B$ be projectors of rank one where $A$ would have
only one nonzero entry on the $i^{\mbox{th }}$ position on its diagonal
and $B$ only one nonzero entry on the $j^{\mbox{th }}$ position
on its diagonal. Further take those nonzero entries to be ones, giving
$m_{1,1}^{A}=m_{1,1}^{B}=0$ and $m_{2}^{A}=m_{2}^{B}=1/m$,

\begin{equation}
\frac{1}{m}\mathbb{E}\textrm{Tr}\left[\left(A\Pi^{T}B\Pi\right)^{2}-\left(AQ^{T}BQ\right)^{2}\right]=\frac{1}{m^{2}}f\left(Q\right)
\end{equation}
But the left hand side is

\begin{eqnarray*}
\frac{1}{m}\mathbb{E}\textrm{Tr}\left[\left(A\Pi^{T}B\Pi\right)^{2}-\left(AQ^{T}BQ\right)^{2}\right] & = & \frac{1}{m}\mathbb{E}\left[\delta_{ij}-q_{ij}^{4}\right]\\
 & = & \frac{1}{m}\left\{ \frac{1}{m^{2}}\sum_{ij}\delta_{ij}-\frac{1}{m^{2}}\sum_{ij}\mathbb{E}\left(q_{ij}^{4}\right)\right\} \\
 & = & \frac{1}{m}\left\{ \frac{1}{m}-\mathbb{E}\left(q_{ij}^{4}\right)\right\} ,
\end{eqnarray*}
where, we used the homogeneity of $Q$. Consequently, by equating
this to $f\left(Q\right)/m^{2}$, we get the desired quantity$f\left(Q\right)=\left\{ 1-m\mathbb{E}\left(q_{ij}^{4}\right)\right\} .$

Our final result Eq. \ref{eq:FreeExpect} now reads

\begin{equation}
\frac{1}{m}\mathbb{E}\textrm{Tr}\left[\left(A\Pi^{T}B\Pi\right)^{2}-\left(AQ^{T}BQ\right)^{2}\right]=\left(m_{2}^{A}-m_{1,1}^{A}\right)\left(m_{2}^{B}-m_{1,1}^{B}\right)\left\{ 1-m\mathbb{E}\left(q_{ij}^{4}\right)\right\} .\label{eq:isotropicBook}
\end{equation}
The same calculation where each of the terms is obtained separately
yields the same result (Appendix). In this paper $p$ is formed by
taking $Q$ to have a $\beta-$Haar measure. Expectation values of
the entries of $Q$ are listed in the Table \ref{tab:HaarExp}.

We wish to express everything in terms of the local terms; using Eqs.
\ref{eq:m_2} and \ref{eq:elemGeneral} as well as $tn^{k}=m$,

\begin{eqnarray*}
m_{2}^{A}-m_{1,1}^{A} & = & \frac{tk\left(n-1\right)n^{k-1}}{m-1}\left(m_{2}^{\mbox{odd}}-m_{1,1}^{\mbox{odd}}\right)\\
m_{2}^{B}-m_{1,1}^{B} & = & \frac{tk\left(n-1\right)n^{k-1}}{m-1}\left(m_{2}^{\mbox{even}}-m_{1,1}^{\mbox{even}}\right),
\end{eqnarray*}
giving 

\begin{eqnarray}
\frac{1}{m}\mathbb{E}\left[\mbox{Tr}\left(A\Pi^{T}B\Pi\right)^{2}-\textrm{Tr}\left(AQ^{T}BQ\right)^{2}\right] & = & \left(m_{2}^{\mbox{odd}}-m_{1,1}^{\mbox{odd}}\right)\left(m_{2}^{\mbox{even}}-m_{1,1}^{\mbox{even}}\right)\times\nonumber \\
 &  & \left(\frac{km\left(n-1\right)}{n\left(m-1\right)}\right)^{2}\left\{ 1-m\mathbb{E}\left(q_{ij}^{4}\right)\right\} .\label{eq:Iso-Classical_FINAL}
\end{eqnarray}

We now proceed to the quantum case where we need to evaluate 
\[
\frac{1}{m}\mathbb{E}\left[\left(A\Pi^{T}B\Pi\right)^{2}-\textrm{Tr}\left(AQ_{q}^{T}BQ_{q}\right)^{2}\right].
\]
 
\end{onehalfspace}

In this case, we cannot directly use the techniques that we used to
get Eq. \ref{eq:Iso-Classical_FINAL} because $Q_{q}$ is not permutation
invariant despite local eigenvectors being so. Before proceeding further
we like to prove a useful lemma (Lemma \ref{lem:diamonds}). Let us
simplify the notation and denote the local terms that are drawn randomly
from a known distribution by $H_{l,l+1}\equiv H^{\left(l\right)}$
whose eigenvalues are $\Lambda_{l}$ as discussed above. 

\begin{onehalfspace}
Recall that $A$ represents the \textit{sum} of all the odds and $Q_{q}^{-1}BQ_{q}$
the \textit{sum} of all the evens,

\[
H_{\mbox{odd}\vphantom{\mbox{even}}}=\sum_{l=1,3,5,\cdots}\mathbb{I}\otimes H^{\left(l\right)}\otimes\mathbb{I},\mbox{ and}\quad H_{\mbox{even}\vphantom{\mbox{odd}}}=\sum_{l=2,4,6,\cdots}\mathbb{I}\otimes H^{\left(l\right)}\otimes\mathbb{I},
\]

Hence, the expansion of $\frac{1}{m}\mathbb{E}\left[\textrm{Tr}\left(AQ_{q}^{T}BQ_{q}\right)^{2}\right]$
amounts to picking an odd term, an even term, then another odd term
and another even term, multiplying them together and taking the expectation
value of the trace of the product (Figure \ref{fig:racksMoment}).
Therefore, each term in the expansion can have four, three or two
different local terms, whose expectation values along with their counts
are needed. These expectations are taken with respect to the local
terms (dense $d^{2}\times d^{2}$ random matrices).

\begin{figure}
\begin{centering}
\includegraphics[scale=1.5]{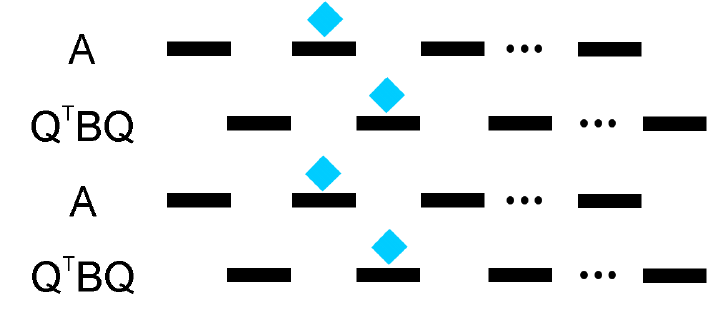}
\par\end{centering}

\caption{\label{fig:racksMoment}The terms in the expansion of $\frac{1}{m}\mathbb{E}\left[\textrm{Tr}\left(AQ_{q}^{T}BQ_{q}\right)^{2}\right]$
can be visualized as picking an element from each row from top to
bottom and multiplying. Each row has $k$ of the local terms corresponding
to a chain with odd number of terms. Among $k^{4}$ terms roughly
$k^{2}$ of them differ among the classical, isotropic and quantum
cases (See Eqs. \ref{eq:QExpCountOdd} and \ref{eq:QExpCountEven}).
An example of such a choice is shown by diamonds. }
\end{figure}

The expectation values depend on the type of random matrix distribution
from which the local terms are drawn. The counting however, depends
on the configuration of the lattice only. We show the counting of
the number of terms, taking care of the boundary terms for an open
chain, along with the type of expectation values by which they need
to be weighted:

\uline{For $N$ odd ($k$ odd terms and $k$ even terms)}

\begin{equation}
\begin{array}{c}
\textrm{Four}\: H^{\left(\centerdot\right)}\textrm{'s}:k^{2}\left(k-1\right)^{2}\Rightarrow d^{N-u_{1}}\mathbb{E}\textrm{Tr}\left(H^{\left(l\right)}\right)^{4},\mbox{ }u_{1}\in\left\{ 5,\cdots,8\right\} \\
\textrm{Three}\: H^{\left(\centerdot\right)}\textrm{'s}:2k^{2}\left(k-1\right)\Rightarrow d^{N-u_{2}}\mathbb{E}\textrm{Tr}\left(\left[H^{\left(l\right)}\right]^{2}\right)\mathbb{E}\textrm{Tr}\left(H^{\left(l\right)}\right)^{2},\mbox{ }u_{2}\in\left\{ 4,5,6\right\} \\
\textrm{Two}\: H^{\left(\centerdot\right)}\textrm{'s:}\left(k-1\right)^{2}\;\textrm{Not\;\ Entangled}\Rightarrow d^{N-4}\left\{ \mathbb{E}\textrm{Tr}\left(\left[H^{\left(l\right)}\right]^{2}\right)\right\} ^{2}\\
\textrm{Two}\: H^{\left(\centerdot\right)}\textrm{'s}:\left(2k-1\right)\;\textrm{ Entangled}\Rightarrow d^{N-3}\mathbb{E}\textrm{Tr}\left[\left(H^{\left(l\right)}\otimes\mathbb{I}\right)\left(\mathbb{I}\otimes H^{\left(l+1\right)}\right)\left(H^{\left(l\right)}\otimes\mathbb{I}\right)\left(\mathbb{I}\otimes H^{\left(l+1\right)}\right)\right]
\end{array}\label{eq:QExpCountOdd}
\end{equation}

\uline{For $N$ even ($k$ odd terms and $k-1$ even terms)}

\begin{equation}
\begin{array}{c}
\textrm{Four}\: H^{\left(\centerdot\right)}\textrm{'s}:k\left(k-1\right)^{2}\left(k-2\right)\Rightarrow d^{N-u_{1}}\mathbb{E}\textrm{Tr}\left(H^{\left(l\right)}\right)^{4},\mbox{ }u_{1}\in\left\{ 5,\cdots,8\right\} \\
\textrm{Three}\: H^{\left(\centerdot\right)}\textrm{'s}:k\left(k-1\right)\left(2k-3\right)\Rightarrow d^{N-u_{2}}\mathbb{E}\textrm{Tr}\left(\left[H^{\left(l\right)}\right]^{2}\right)\mathbb{E}\textrm{Tr}\left(H^{\left(l\right)}\right)^{2},\mbox{ }u_{2}\in\left\{ 4,5,6\right\} \\
\textrm{Two}\: H^{\left(\centerdot\right)}\textrm{'s}:\left(k-1\right)\left(k-2\right)\;\textrm{Not\;\ Entangled}\Rightarrow d^{N-4}\left\{ \mathbb{E}\textrm{Tr}\left(\left[H^{\left(l\right)}\right]^{2}\right)\right\} ^{2}\\
\textrm{Two}\: H^{\left(\centerdot\right)}\textrm{'s}:2\left(k-1\right)\;\textrm{ Entangled}\Rightarrow d^{N-3}\mathbb{E}\textrm{Tr}\left[\left(H^{\left(l\right)}\otimes\mathbb{I}\right)\left(\mathbb{I}\otimes H^{\left(l+1\right)}\right)\left(H^{\left(l\right)}\otimes\mathbb{I}\right)\left(\mathbb{I}\otimes H^{\left(l+1\right)}\right)\right]
\end{array}\label{eq:QExpCountEven}
\end{equation}
Here $u_{1}$ and $u_{2}$ indicate the number of sites that the local
terms act on (i.e., occupy). Therefore, $\frac{1}{m}\mathbb{E}\left[\textrm{Tr}\left(AQ_{q}^{T}BQ_{q}\right)^{2}\right]$
is obtained by multiplying each type of terms, weighted by the counts
and summing. For example for $u_{1}=5$ and $u_{2}=3$, when $N$
is odd,

\begin{equation}
\begin{array}{c}
\frac{1}{m}\mathbb{E}\left[\textrm{Tr}\left(AQ_{q}^{T}BQ_{q}\right)^{2}\right]=\frac{1}{m}\left\{ d^{N-5}k^{2}\left(k-1\right)^{2}\mathbb{E}\textrm{Tr}\left(H^{\left(l\right)}\right)^{4}+\right.\\
2k^{2}\left(k-1\right)d^{N-4}\mathbb{E}\textrm{Tr}\left(\left[H^{\left(l\right)}\right]^{2}\right)\mathbb{E}\textrm{Tr}\left(H^{\left(l\right)}\right)^{2}+\left(k-1\right)^{2}d^{N-4}\left\{ \mathbb{E}\textrm{Tr}\left(\left[H^{\left(l\right)}\right]^{2}\right)\right\} ^{2}+\\
\left.\left(2k-1\right)d^{N-3}\mathbb{E}\textrm{Tr}\left[\left(H^{\left(l\right)}\otimes\mathbb{I}\right)\left(\mathbb{I}\otimes H^{\left(l+1\right)}\right)\left(H^{\left(l\right)}\otimes\mathbb{I}\right)\left(\mathbb{I}\otimes H^{\left(l+1\right)}\right)\right]\right\} 
\end{array}\label{eq:quantumOdd}
\end{equation}
and similarly for $N$ even,

\begin{equation}
\begin{array}{c}
\frac{1}{m}\mathbb{E}\left[\textrm{Tr}\left(AQ_{q}^{T}BQ_{q}\right)^{2}\right]=\frac{\left(k-1\right)}{m}\left\{ k\left(k-1\right)\left(k-2\right)d^{N-5}\mathbb{E}\textrm{Tr}\left(H^{\left(l\right)}\right)^{4}+\right.\\
k\left(2k-3\right)d^{N-4}\mathbb{E}\textrm{Tr}\left(\left[H^{\left(l\right)}\right]^{2}\right)\mathbb{E}\textrm{Tr}\left(H^{\left(l\right)}\right)^{2}+\left(k-2\right)d^{N-4}\left\{ \mathbb{E}\textrm{Tr}\left(\left[H^{\left(l\right)}\right]^{2}\right)\right\} ^{2}+\\
\left.2d^{N-3}\mathbb{E}\textrm{Tr}\left[\left(H^{\left(l\right)}\otimes\mathbb{I}\right)\left(\mathbb{I}\otimes H^{\left(l+1\right)}\right)\left(H^{\left(l\right)}\otimes\mathbb{I}\right)\left(\mathbb{I}\otimes H^{\left(l+1\right)}\right)\right]\right\} .
\end{array}\label{eq:quantumEven}
\end{equation}

The expectation values depend on the type of random matrix distribution
from which the local terms are drawn. We will give explicit examples
in the following sections. In the following lemma, we use $\mathbb{E\left(\mathit{H^{\left(l\right)}}\right)=}\mu\mathbb{I}_{d^{2}}$
and $\mathbb{E\left(\mathit{H^{\left(l\right)}}\right)^{\mathit{2}}=}m_{2}\mathbb{I}_{d^{2}}$.
\end{onehalfspace}
\begin{lem}
\begin{onehalfspace}
\label{lem:diamonds}In calculating the $\mathbb{E}\mbox{Tr}\left(AQ_{q}^{T}BQ_{q}\right)^{2}$
if at least one of the odds (evens) commutes with one of the evens
(odds) then the expectation value is the same as the classical expectation
value. Further if the local terms have permutation invariance of eigenvalues
then the only quantum expectation value that differs from classical
is of Type II (see the proof and the diamonds in figure \ref{fig:racksMoment}).\end{onehalfspace}
\end{lem}
\begin{proof}
\begin{onehalfspace}
This can be shown using the trace property $\mbox{Tr}\left(MP\right)=\mbox{Tr}\left(PM\right)$.
In calculating $\mathbb{E}\mbox{Tr}\left(H_{l}^{odd}H_{p}^{even}H_{j}^{odd}H_{k}^{even}\right)$;
if any of the odd (even) terms commutes with any of the even (odd)
terms to its left or right then they can be swapped. For example one
gets $\mathbb{E}\mbox{Tr}\left(H_{l}^{odd}H_{p}^{even}H_{k}^{even}H_{j}^{odd}\right)=\mathbb{E}\mbox{Tr}\left(H_{j}^{odd}H_{l}^{odd}H_{p}^{even}H_{k}^{even}\right)$
which is just the classical value. Hence the only types of expectations
that we need to worry about are

\[
\begin{array}{ccc}
 & \underset{\_\_}{H^{\left(l\right)}}\\
 &  & \underset{\_\_}{H^{\left(l+1\right)}}\\
 & \underset{\_\_}{H^{\left(l\right)}}\\
\underset{\_\_}{H^{\left(l-1\right)}}\\
 & \mbox{Type I}
\end{array}\qquad\mbox{ and }\qquad\begin{array}{cc}
\underset{\_\_}{H^{\left(l\right)}}\\
 & \underset{\_\_}{H^{\left(l+1\right)}}\\
\underset{\_\_}{H^{\left(l\right)}}\\
 & \underset{\_\_}{H^{\left(l+1\right)}}\\
\mbox{Type II}
\end{array}
\]
now we show that with permutation invariance of the local eigenvalues
the first type are also classical leaving us with the ``diamond terms''
alone (Fig. \ref{fig:racksMoment}). Consider a Type I term, which
involves three independent local terms,

\[
\begin{array}{c}
{\scriptstyle \frac{1}{m}\mathbb{E}\textrm{Tr}\left[\left(\mathbb{I}_{d^{2}}\otimes H^{\left(3\right)}\otimes\mathbb{I}_{d^{N-4}}\right)\left(\mathbb{I}\otimes H^{\left(2\right)}\otimes\mathbb{I}_{d^{N-3}}\right)\left(\mathbb{I}_{d^{2}}\otimes H^{\left(3\right)}\otimes\mathbb{I}_{d^{N-4}}\right)\left(\mathbb{I}_{d^{3}}\otimes H^{\left(4\right)}\otimes\mathbb{I}_{d^{N-5}}\right)\right]}\\
=\mu^{2}m_{2}.
\end{array}
\]
This follows immediately from the independence of $H^{\left(4\right)}$
, which allows us to take its expectation value separately giving
a $\mu$ and leaving us with 
\[
\frac{\mu}{m}\mathbb{E}\textrm{Tr}\left[\left(\mathbb{I}_{d^{2}}\otimes H^{\left(3\right)}\otimes\mathbb{I}_{d^{N-4}}\right)^{2}\left(\mathbb{I}\otimes H^{\left(2\right)}\otimes\mathbb{I}_{d^{N-3}}\right)\right]=\mu^{2}m_{2}.
\]

Therefore the only relevant terms, shown by diamonds in Fig. \ref{fig:racksMoment},
are of Type II. As an example of such terms consider (here on repeated
indices are summed over)

\begin{equation}
\begin{array}{c}
{\scriptstyle \frac{1}{m}\mathbb{E}\textrm{Tr}\left[\left(H^{\left(1\right)}\otimes\mathbb{I}_{d^{N-2}}\right)\left(\mathbb{I}\otimes H^{\left(2\right)}\otimes\mathbb{I}_{d^{N-3}}\right)\left(H^{\left(1\right)}\otimes\mathbb{I}_{d^{N-2}}\right)\left(\mathbb{I}\otimes H^{\left(2\right)}\otimes\mathbb{I}_{d^{N-3}}\right)\right]}\\
=\frac{1}{d^{3}}\left\{ \mathbb{E}\left(H_{i_{1}i_{2},j_{1}j_{2}}^{\left(1\right)}H_{i_{1}p_{2},j_{1}k_{2}}^{\left(1\right)}\right)\mathbb{E}\left(H_{j_{2}i_{3},k_{2}k_{3}}^{\left(2\right)}H_{i_{2}i_{3},p_{2}k_{3}}^{\left(2\right)}\right)\right\} ,
\end{array}
\end{equation}
where the indices with subscript $2$ prevent us from treating the
two expectation values independently: $H^{\left(1\right)}$ and $H^{\left(2\right)}$
overlap at the second site. The number of such terms is $2k-1$, where
$k=\frac{N-1}{2}$.\end{onehalfspace}

\end{proof}
\begin{onehalfspace}
Therefore, we have found a further reduction of the terms from the
departure theorem, that distinguishes the quantum problem from the
other two. Luckily and interestingly the kurtosis of the quantum case
lies in between the classical and the iso. We emphasize that \textit{the
only inputs to the theory are the geometry of the lattice (e.g., the
number of summands and the inter-connectivity of the local terms)
and the moments} that characterizes the type of the local terms. 

Comment: The most general treatment would consider Type I terms as
well, i.e., there is no assumption of permutation invariance of the
eigenvalues of the local terms. This allows one to treat all types
of local terms. Here we are confining to random local interactions,
where the local eigenvectors are generic or the eigenvalues locally
are permutation invariant in the expectation value sense. 

The goal is to find $p$ by matching fourth moments

\[
1-p=\frac{\mathbb{E}\mbox{Tr}\left(A\Pi^{T}B\Pi\right)^{2}-\mathbb{E}\mbox{Tr}\left(AQ_{q}^{T}BQ_{q}\right)^{2}}{\mathbb{E}\mbox{Tr}\left(A\Pi^{T}B\Pi\right)^{2}-\mathbb{E}\mbox{Tr}\left(AQ^{T}BQ\right)^{2}}
\]
for which we calculated the denominator resulting in Eq. \ref{eq:Iso-Classical_FINAL},
where $\mathbb{E}\left(\left|q_{i,j}\right|^{4}\right)=\frac{\beta+2}{m\left(m\beta+2\right)}$
for $\beta-$Haar $Q$$ $ (Table \ref{tab:HaarExp}). If the numerator
allows a factorization of the moments of the local terms as in Eq.
\ref{eq:Iso-Classical_FINAL}, then the value of $p$ will be independent
of the covariance matrix (i.e., eigenvalues of the local terms). 
\end{onehalfspace}
\begin{lem*}
\begin{onehalfspace}
\textbf{\textup{(Universality)}} $p\mapsto p\left(N,d,\beta\right)$,
namely, it is independent of the distribution of the local terms.\end{onehalfspace}
\end{lem*}
\begin{proof}
\begin{onehalfspace}
We use a similar techniques as we did in the isotropic case. The general
form for the numerator of Eq. \ref{eq:1-p} is (denoting Lemma \ref{lem:diamonds}
by L3)

\begin{eqnarray}
\frac{1}{m}\mathbb{E}\textrm{Tr}\left[\left(A\Pi^{T}B\Pi\right)^{2}-\left(AQ_{q}^{T}BQ_{q}\right)^{2}\right] & \overset{\mbox{L3}}{=} & \frac{\left(2k-1\right)}{d^{3}}\mbox{\ensuremath{\mathbb{E}}Tr}\left\{ \left(H^{\left(l\right)}\otimes\mathbb{I}_{d}\right)^{2}\left(\mathbb{I}_{d}\otimes H^{\left(l+1\right)}\right)^{2}\right.\nonumber \\
 &  & -\left.\left[\left(H^{\left(l\right)}\otimes\mathbb{I}_{d}\right)\left(\mathbb{I}_{d}\otimes H^{\left(l+1\right)}\right)\right]^{2}\right\} \nonumber \\
 & = & \frac{\left(2k-1\right)}{d^{3}}\mbox{\ensuremath{\mathbb{E}}}\mbox{Tr}\left\{ \left(Q_{l}^{-1}\Lambda_{l}Q_{l}\otimes\mathbb{I}_{d}\right)^{2}\left(\mathbb{I}_{d}\otimes Q_{l+1}^{-1}\Lambda_{l+1}Q_{l+1}\right)^{2}\right.\nonumber \\
 &  & -\left.\left[\left(Q_{l}^{-1}\Lambda_{l}Q_{l}\otimes\mathbb{I}_{d}\right)\left(\mathbb{I}_{d}\otimes Q_{l+1}^{-1}\Lambda_{l+1}Q_{l+1}\right)\right]^{2}\right\} \label{eq:quantumBook-1}
\end{eqnarray}
where the expectation on the right hand side is taken with respect
to the local terms $H^{\left(l\right)}$ and $H^{\left(l+1\right)}$
. The right hand side is a homogeneous polynomial of order two in
the entries of $\Lambda_{l}$, as well as, in the entries of $\Lambda_{l+1}$;
consequently Eq. \ref{eq:quantumBook-1} necessarily has the form

\[
c_{1}\left(\Lambda^{\mbox{even}},Q_{\mbox{odd}},Q_{\mbox{even}}\right)m_{2}^{\mbox{odd}}+c_{2}\left(H^{\mbox{even}},Q_{\mbox{odd}},Q_{\mbox{even}}\right)m_{1,1}^{\mbox{odd}}
\]
but Eq. \ref{eq:quantumBook-1} must be zero for $\Lambda_{l}=I$,
for which $m_{2}^{\mbox{odd}}=m_{1,1}^{\mbox{odd}}=1$. This implies
that $c_{1}=-c_{2}$. By permutation invariance of the local terms
we can factor out $\left(m_{2}^{\mbox{odd}}-m_{1,1}^{\mbox{odd}}\right)$.
Similarly, the homogeneity and permutation invariance of $H^{\left(l+1\right)}$
implies,

\[
\left(m_{2}^{\mbox{odd}}-m_{1,1}^{\mbox{odd}}\right)\left[D_{1}\left(Q_{\mbox{odd}},Q_{\mbox{even}}\right)m_{2}^{\mbox{even}}+D_{2}\left(Q_{\mbox{odd}},Q_{\mbox{even}}\right)m_{1,1}^{\mbox{even}}\right].
\]
The right hand side should be zero for $\Lambda_{l+1}=I$, whereby
we can factor out $\left(m_{2}^{\mbox{even}}-m_{1,1}^{\mbox{even}}\right)$;
hence the right hand side of Eq. \ref{eq:quantumBook-1} becomes

\begin{equation}
\frac{\left(2k-1\right)}{d^{3}}\left(m_{2}^{\mbox{odd}}-m_{1,1}^{\mbox{odd}}\right)\left(m_{2}^{\mbox{even}}-m_{1,1}^{\mbox{even}}\right)f_{q}\left(Q_{\mbox{odd}},Q_{\mbox{even}}\right)\label{eq:Q-C_final}
\end{equation}
where $f_{q}\left(Q_{\mbox{odd}},Q_{\mbox{even}}\right)$ is a homogeneous
function of order four in the entries of $Q_{\mbox{odd}}$ as well
as $Q_{\mbox{even}}$. To evaluate $f_{q}$, it suffices to let $\Lambda_{l}$
and $\Lambda_{l+1}$ be projectors of rank one where $\Lambda_{l}$
would have only one nonzero entry on the $i^{\mbox{th }}$ position
on its diagonal and $\Lambda_{l+1}$ only one nonzero entry on the
$j^{\mbox{th }}$ position on its diagonal. Further take those nonzero
entries to be ones, giving $m_{1,1}^{A}=m_{1,1}^{B}=0$ and $m_{2}^{A}=m_{2}^{B}=1/n$.
Using this choice of local terms the right hand side of Eq. \ref{eq:quantumBook-1}
now reads 

\begin{eqnarray}
\frac{\left(2k-1\right)}{d^{3}} & \mbox{\ensuremath{\mathbb{E}}}\mbox{Tr} & \left\{ \left(|q_{i}^{\left(l\right)}\rangle\langle q_{i}^{\left(l\right)}|\otimes I_{d}\right)^{2}\left(I_{d}\otimes|q_{j}^{\left(l+1\right)}\rangle\langle q_{j}^{\left(l+1\right)}|\right)^{2}\right.\nonumber \\
 &  & -\left.\left[\left(|q_{i}^{\left(l\right)}\rangle\langle q_{i}^{\left(l\right)}|\otimes I_{d}\right)\left(I_{d}\otimes|q_{j}^{\left(l+1\right)}\rangle\langle q_{j}^{\left(l+1\right)}|\right)\right]^{2}\right\} \label{eq:Qdepart}
\end{eqnarray}
where here the expectation value is taken with respect to random choices
of local eigenvectors. Equating this and Eq. \ref{eq:Q-C_final} 

\begin{eqnarray}
f_{q}\left(Q_{\mbox{odd}},Q_{\mbox{even}}\right) & \mbox{=} & n^{2}\mbox{\ensuremath{\mathbb{E}}}\mbox{Tr}\left\{ \left(|q_{i}^{\left(l\right)}\rangle\langle q_{i}^{\left(l\right)}|\otimes I_{d}\right)^{2}\left(I_{d}\otimes|q_{j}^{\left(l+1\right)}\rangle\langle q_{j}^{\left(l+1\right)}|\right)^{2}\right.\label{eq:Qdepart-1}\\
 &  & -\left.\left[\left(|q_{i}^{\left(l\right)}\rangle\langle q_{i}^{\left(l\right)}|\otimes I_{d}\right)\left(I_{d}\otimes|q_{j}^{\left(l+1\right)}\rangle\langle q_{j}^{\left(l+1\right)}|\right)\right]^{2}\right\} \nonumber 
\end{eqnarray}

To simplify notation let us expand these vectors in the computational
basis $|q_{i}^{\left(l\right)}\rangle=u_{i_{1}i_{2}}|i_{1}\rangle|i_{2}\rangle$
and $|q_{j}^{\left(l+1\right)}\rangle=v_{i_{2}i_{3}}|i_{2}\rangle|i_{3}\rangle.$
The first term on the right hand side of Eq. \ref{eq:Qdepart}, the
classical term, is obtained by assuming commutativity and using the
projector properties,

\begin{eqnarray}
\mbox{Tr}\left[\left(|q_{i}^{\left(l\right)}\rangle\langle q_{i}^{\left(l\right)}|\otimes I_{d}\right)^{2}\left(I_{d}\otimes|q_{j}^{\left(l+1\right)}\rangle\langle q_{j}^{\left(l+1\right)}|\right)^{2}\right] & =\nonumber \\
\mbox{Tr}\left[\left(|q_{i}^{\left(l\right)}\rangle\langle q_{i}^{\left(l\right)}|\otimes I_{d}\right)\left(I_{d}\otimes|q_{j}^{\left(l+1\right)}\rangle\langle q_{j}^{\left(l+1\right)}|\right)\right] & =\nonumber \\
\mbox{Tr}\left[u_{i_{1},i_{2}}\overline{u_{j_{1}j_{2}}}v_{j_{2}i_{3}}\overline{v_{k_{2}k_{3}}}u_{j_{1}k_{2}}|i_{1}i_{2}i_{3}\rangle\langle j_{1}k_{2}k_{3}|\right] & =\nonumber \\
\left[u_{i_{1},i_{2}}\overline{u_{i_{1}j_{2}}}v_{j_{2}i_{3}}\overline{v_{i_{2}i_{3}}}\right]=\left(u^{\dagger}u\right)_{j_{2}i_{2}}\left(vv^{\dagger}\right)_{j_{2}i_{2}} & =\nonumber \\
\mbox{Tr}\left[\left(u^{\dagger}u\right)\left(vv^{\dagger}\right)\right]=\mbox{Tr}\left[uv\left(uv\right)^{\dagger}\right] & =\nonumber \\
\left\Vert uv\right\Vert _{\mbox{F}}^{2} & = & \sum_{i=1}^{d}\sigma_{i}^{2}.\label{eq:CDepart}
\end{eqnarray}
where $\left\Vert \centerdot\right\Vert _{\mbox{F}}$ denotes the
Frobenius norm and $\sigma_{i}$ are the singular values of $uv$.
The second term, the quantum term, is

\begin{eqnarray}
\mbox{Tr}\left[\left(|q_{i}^{\left(l\right)}\rangle\langle q_{i}^{\left(l\right)}|\otimes I_{d}\right)\left(I_{d}\otimes|q_{j}^{\left(l+1\right)}\rangle\langle q_{j}^{\left(l+1\right)}|\right)\right]^{2} & =\label{eq:QdepartFinal}\\
\mbox{Tr}\left[u_{i_{1}i_{2}}\overline{u_{j_{1}j_{2}}}v_{j_{2}i_{3}}\overline{v_{k_{2}k_{3}}}u_{j_{1}k_{2}}\overline{u_{m_{1}m_{2}}}v_{m_{2}k_{3}}\overline{v_{i_{2}i_{3}}}|i_{1}i_{2}i_{3}\rangle\langle p_{1}p_{2}p_{3}|\right] & =\nonumber \\
\left(u^{\dagger}u\right)_{j_{2}k_{2}}\left(vv^{\dagger}\right)_{m_{2}k_{2}}\left(u^{\dagger}u\right)_{m_{2}i_{2}}\left(vv^{\dagger}\right)_{j_{2}i_{2}} & =\nonumber \\
\left(u^{\dagger}uvv^{\dagger}\right)_{j_{2}m_{2}}\left(u^{\dagger}uvv^{\dagger}\right)_{m_{2}j_{2}}=\mbox{Tr}\left\{ \left[uv\left(uv\right)^{\dagger}\right]^{2}\right\}  & =\nonumber \\
\left\Vert uv\left(uv\right)^{\dagger}\right\Vert _{\mbox{F}}^{2} & = & \sum_{i=1}^{d}\sigma_{i}^{4}.\nonumber 
\end{eqnarray}
where we used the symmetry of $\left(uv\left(uv\right)^{\dagger}\right)^{2}=uv\left(uv\right)^{\dagger}\left[uv\left(uv\right)^{\dagger}\right]^{\dagger}$.

Now we can calculate 

\begin{equation}
f_{q}\left(Q_{\mbox{odd}},Q_{\mbox{even}}\right)=n^{2}\mbox{\ensuremath{\mathbb{E}}}\left\{ \left\Vert uv\right\Vert _{\mbox{F}}^{2}-\left\Vert uv\left(uv\right)^{\dagger}\right\Vert _{\mbox{F}}^{2}\right\} \label{eq:d4_introduced}
\end{equation}
giving us the desired result

\begin{eqnarray}
\frac{1}{m}\mathbb{E}\textrm{Tr}\left[\left(A\Pi^{T}B\Pi\right)^{2}-\left(AQ_{q}^{T}BQ_{q}\right)^{2}\right] & = & d\left(2k-1\right)\left(m_{2}^{\mbox{odd}}-m_{1,1}^{\mbox{odd}}\right)\left(m_{2}^{\mbox{even}}-m_{1,1}^{\mbox{even}}\right)\nonumber \\
 & \times & \mathbb{E}\left(\left\Vert uv\right\Vert _{\mbox{F}}^{2}-\left\Vert uv\left(uv\right)^{\dagger}\right\Vert _{\mbox{F}}^{2}\right),\label{eq:Quantum-Classical}
\end{eqnarray}
from which 

\begin{eqnarray}
1-p & = & \frac{\mbox{ETr}\left(A\Pi^{T}B\Pi\right)^{2}-\mbox{ETr}\left(AQ_{q}^{-1}BQ_{q}\right)^{2}}{\mbox{ETr}\left(A\Pi^{T}B\Pi\right)^{2}-\mbox{ETr}\left(AQ^{-1}BQ\right)^{2}}\label{eq:1-pFINAL}\\
 & = & \frac{d\left(2k-1\right)\mathbb{E}\left(\left\Vert uv\right\Vert _{\mbox{F}}^{2}-\left\Vert uv\left(uv\right)^{\dagger}\right\Vert _{\mbox{F}}^{2}\right)}{\left(\frac{km\left(n-1\right)}{n\left(m-1\right)}\right)^{2}\left\{ 1-m\mathbb{E}\left(q_{ij}^{4}\right)\right\} }.\nonumber 
\end{eqnarray}
The dependence on the covariance matrix has cancelled- a covariance
matrix is one whose element in the $i,j$ position is the covariance
between the $i^{th}$ and $j^{th}$ eigenvalue. This shows that $p$
is independent of eigenvalues of the local terms which proves the
universality lemma.

Comment: To get the numerator we used permutation invariance of $A$
and $B$ and local terms, to get the denominator we used permutation
invariance of $Q$. \end{onehalfspace}

\end{proof}
\begin{onehalfspace}
Comment: It is interesting that the amount of mixture of the two extremes
needed to capture the quantum spectrum is independent of the actual
types of local terms. It only depends on the physical parameters of
the lattice.
\end{onehalfspace}

\begin{onehalfspace}

\subsection{\label{sub:The-Slider-Theorem}The Slider Theorem and a Summary}
\end{onehalfspace}

\begin{onehalfspace}
In this section we make explicit use of $\beta-$Haar properties of
$Q$ and local terms. To prove that there exists a $0\le p\le1$ such
that the combination in Eq. \ref{eq:convex} is convex we need to
evaluate the expected Frobenius norms in Eq. \ref{eq:Quantum-Classical}.
\end{onehalfspace}
\begin{lem}
\begin{onehalfspace}
$\mathbb{E}\left\Vert uv\right\Vert _{F}^{2}=1/d$ and $\mathbb{E}\left\Vert uv\left(uv\right)^{\dagger}\right\Vert _{F}^{2}=\frac{\beta^{2}\left[3d\left(d-1\right)+1\right]+2\beta\left(3d-1\right)+4}{d\left(\beta d^{2}+2\right)^{2}}$,
when local terms have $\beta-$Haar eigenvectors. \end{onehalfspace}
\end{lem}
\begin{proof}
\begin{onehalfspace}
It is a fact that $G=u\chi_{\beta d^{2}}$, when $u$ is uniform on
a sphere, $G$ is a $d\times d$ $\beta-$Gaussian matrix whose expected
Frobenius norm has a $\chi-$distribution denoted here by $\chi_{\beta d^{2}}$
(similarly for $v$). Recall that $\mathbb{E}\left(\chi_{h}^{2}\right)=h$
and $\mathbb{E}\left(\chi_{h}^{4}\right)=h\left(h+2\right)$.

\begin{eqnarray}
\mathbb{E}\left\Vert uv\right\Vert _{\mbox{F}}^{2} & \mathbb{E}\left(\chi_{\beta d^{2}}\right)^{2}= & \mathbb{E}\left\Vert \left(G_{1}G_{2}\right)\right\Vert _{\mbox{F}}^{2}\label{eq:SVclass}\\
\Rightarrow\mathbb{E}\left\Vert uv\right\Vert _{\mbox{F}}^{2} & = & \frac{1}{\left(\beta d^{2}\right)^{2}}\mathbb{E}\left\Vert G_{1}G_{2}\right\Vert _{\mbox{F}}^{2}=\frac{d^{2}}{\left(\beta d^{2}\right)^{2}}\mathbb{E}\sum_{k=1}^{d}\left(g_{i,k}^{(1)}g_{kj}^{(2)}\right)^{2}\nonumber \\
 & = & \frac{d^{2}}{\left(\beta d^{2}\right)^{2}}d\left(\beta\right)^{2}=\frac{1}{d}.\nonumber 
\end{eqnarray}
The quantum case, $u^{\dagger}u=\frac{G_{1}^{\dagger}G_{1}}{\left\Vert G_{1}\right\Vert _{\mbox{F}}^{2}}\equiv\frac{W_{1}}{\left\Vert G_{1}\right\Vert _{\mbox{F}}^{2}}$
, similarly $v^{\dagger}v=\frac{G_{2}^{\dagger}G_{2}}{\left\Vert G_{2}\right\Vert _{\mbox{F}}^{2}}\equiv\frac{W_{2}}{\left\Vert G_{2}\right\Vert _{\mbox{F}}^{2}}$,
where $W_{1}$ and $W_{2}$ are Wishart matrices. 

\begin{equation}
\mathbb{E}\left\Vert uv\left(uv\right)^{\dagger}\right\Vert _{\mbox{F}}^{2}=\frac{\mathbb{E}\mbox{Tr}\left(W_{1}W_{2}\right)^{2}}{\mathbb{E}\left(\chi_{d^{2}\beta}^{4}\right)\mathbb{E}\left(\chi_{d^{2}\beta}^{4}\right)}=\frac{\mathbb{E}\mbox{Tr}\left(W_{1}W_{2}\right)^{2}}{\left[d^{2}\beta\left(d^{2}\beta+2\right)\right]^{2}}\label{eq:QuantumFrob}
\end{equation}
hence the complexity of the problem is reduced to finding the expectation
of the trace of a product of Wishart matrices. 

\begin{equation}
\mathbb{E}\mbox{Tr}\left(W_{1}W_{2}\right)^{2}=\mathbb{E}\mbox{Tr}\left(W_{1}W_{2}W_{1}W_{2}\right)=\mathbb{E}\sum_{1\le ijkl\le d}x_{i}x_{i}^{\dagger}y_{j}y_{j}^{\dagger}x_{k}x_{k}^{\dagger}y_{l}y_{l}^{\dagger}\equiv\Pi\left[\begin{array}{cc}
x_{i}^{\dagger}y_{j} & y_{l}^{\dagger}x_{i}\\
y_{j}^{\dagger}x_{k} & x_{k}^{\dagger}y_{l}
\end{array}\right],\label{eq:PROD}
\end{equation}
where $\Pi$ denotes the product of the elements of the matrix. There
are three types of expectations summarized in Table \ref{tab:Expectation-values.}.

\begin{table}
\noindent \begin{centering}
\begin{tabular}{|c|c|c|}
\hline 
Notation & Type & Count\tabularnewline
\hline 
\hline 
$X$ & $\begin{array}{ccc}
i\neq k & \& & j\neq l\end{array}$ & $d^{2}\left(d-1\right)^{2}$\tabularnewline
\hline 
$Y$ & $\begin{array}{ccc}
i=k & \& & j\neq l\\
 & \mbox{or}\\
i\neq k & \& & j=l
\end{array}$ & $2d^{2}\left(d-1\right)$\tabularnewline
\hline 
$Z$ & $\begin{array}{ccc}
i=k & \& & j=l\end{array}$ & $d^{2}$\tabularnewline
\hline 
\end{tabular}
\par\end{centering}

\caption{\label{tab:Expectation-values.}Expectation values.}
\end{table}

In Table \ref{tab:Expectation-values.}

\begin{eqnarray*}
X & \equiv & \mathbb{E}\left[\Pi\left(x_{i}x_{k}\right)\left(y_{i}y_{l}\right)\right]\\
Y & \equiv & \mathbb{E}\left[\left(x_{i}^{\dagger}y_{j}\right)^{2}\left(x_{i}^{\dagger}y_{l}\right)^{2}\right]\\
Z & \equiv & \mathbb{E}\left[\left(x_{i}^{\dagger}y_{i}\right)^{4}\right].
\end{eqnarray*}
We now evaluate these expectation values. We have

\[
X=\Pi\left(\begin{array}{cc}
\chi_{\beta d} & g_{\beta}\\
0 & \chi_{\beta\left(d-1\right)}\\
0 & 0\\
\vdots & \vdots
\end{array}\right)^{\dagger}\left(\begin{array}{cc}
g_{\beta} & g_{\beta}\\
g_{\beta} & g_{\beta}\\
\mbox{DC} & \mbox{DC}\\
\vdots & \vdots
\end{array}\right)
\]
by $QR$ decomposition, where $g_{\beta}$ and $\chi_{h}$ denote
an element with a $\beta-$Gaussian and $\chi_{h}$ distribution respectively;
DC means ``Don't Care''. Consequently

\begin{eqnarray*}
X & = & \Pi\left(\begin{array}{cc}
\chi_{\beta d} & g_{\beta}\\
0 & \chi_{\beta\left(d-1\right)}
\end{array}\right)\left(\begin{array}{cc}
a & b\\
c & d
\end{array}\right)\\
 & =\Pi & \left[\begin{array}{cc}
a\chi_{\beta d} & g_{\beta}a+\chi_{\beta\left(d-1\right)}c\\
b\chi_{\beta d} & g_{\beta}b+\chi_{\beta\left(d-1\right)}d
\end{array}\right]=\Pi\left[\begin{array}{cc}
a\chi_{\beta d} & g_{\beta}a\\
b\chi_{\beta d} & g_{\beta}b
\end{array}\right]\\
 & = & \chi_{\beta d}^{2}a^{2}b^{2}g_{\beta}^{2}=\beta^{4}d.
\end{eqnarray*}
where we denoted the four independent Gaussian entries by $a,b,c,d$
to not confuse them as one number. From Eq. \ref{eq:PROD} we have

\begin{eqnarray*}
Y & = & \mathbb{E}\left[\left(x_{i}^{\dagger}y_{j}\right)^{2}\left(x_{i}^{\dagger}y_{l}\right)^{2}\right]=\mathbb{E}\left(\chi_{d\beta}g_{\beta}^{\left(1\right)}\right)^{2}\left(\chi_{d\beta}g_{\beta}^{\left(2\right)}\right)^{2}=\beta d\left(\beta d+2\right)\beta^{2}\\
Z & = & \mathbb{E}\left(x^{\dagger}y\right)^{4}=\mathbb{E}\left(\chi_{\beta d}^{4}\right)\mathbb{E}\left(\chi_{\beta}^{4}\right)=\beta d\left(\beta d+2\right)\beta\left(\beta+2\right).
\end{eqnarray*}
Eq. \ref{eq:QuantumFrob} now reads

\begin{equation}
\mathbb{E}\left\Vert uv\left(uv\right)^{\dagger}\right\Vert _{F}^{2}=\frac{\beta^{2}\left[3d\left(d-1\right)+1\right]+2\beta\left(3d-1\right)+4}{d\left(\beta d^{2}+2\right)^{2}}.\label{eq:SVq}
\end{equation}
\end{onehalfspace}
\end{proof}
\begin{thm*}
\begin{onehalfspace}
\textbf{\textup{(The Slider Theorem)\label{thm:The-Slider-Theorem}}}\textbf{
}The quantum kurtosis lies in between the classical and the iso kurtoses,
$\gamma_{2}^{iso}\leq\gamma_{2}^{q}\leq\gamma_{2}^{c}$. Therefore
there exists a $0\leq p\leq1$ such that $\gamma_{2}^{q}=p\gamma_{2}^{c}+\left(1-p\right)\gamma_{2}^{iso}$.
Further, $\lim_{N\rightarrow\infty}p=1$.\end{onehalfspace}
\end{thm*}
\begin{proof}
\begin{onehalfspace}
We have $\left\{ 1-\frac{1}{m}\sum_{ij=1}^{m}q_{ij}^{4}\right\} \geq0$,
since $\sum_{ij}q_{ij}^{4}\leq\sum_{ij}q_{ij}^{2}=m.$ The last inequality
follows from $q_{ij}^{2}\le1$ . Therefore, Eq. \ref{eq:denom} is

\begin{eqnarray*}
\gamma_{2}^{iso}-\gamma_{2}^{c} & = & \frac{2}{\sigma^{4}}\left(m_{2}^{\mbox{odd}}-m_{11}^{\mbox{odd}}\right)\left(m_{2}^{\mbox{even}}-m_{11}^{\mbox{even}}\right)\times\\
 &  & \left(\frac{km\left(n-1\right)}{n\left(m-1\right)}\right)^{2}\left\{ m\mathbb{E}\left(q_{11}^{4}\right)-1\right\} \le0.
\end{eqnarray*}

From Eqs. \ref{eq:QdepartFinal} and \ref{eq:CDepart} and using the
fact that the singular values $\sigma_{i}\le1$ we have

\[
\left\Vert uv\left(uv\right)^{\dagger}\right\Vert _{\mbox{F}}^{2}=\sum_{i=1}^{d}\sigma_{i}^{4}\le\sum_{i=1}^{d}\sigma_{i}^{2}=\left\Vert uv\right\Vert _{\mbox{F}}^{2}
\]
which proves $\gamma_{2}^{q}-\gamma_{2}^{c}\leq0$. In order to establish
$\gamma_{2}^{iso}\leq\gamma_{2}^{q}\leq\gamma_{2}^{c}$, we need to
show that $\gamma_{2}^{c}-\gamma_{2}^{q}\le\gamma_{2}^{c}-\gamma_{2}^{iso}$.
Eq. \ref{eq:1-pFINAL} after substituting $m\mathbb{E}\left(q_{ij}^{4}\right)=\frac{\beta+2}{\left(m\beta+2\right)}$
from Table \ref{tab:HaarExp} and Eqs. \ref{eq:SVclass}, \ref{eq:SVq}
reads

\begin{eqnarray}
1-p & = & \left(1-d^{-2k-1}\right)\left[1-\left(\frac{k-1}{k}\right)^{2}\right]\left\{ \left(1-\frac{1-d^{-2k+1}}{1+\beta d^{2}/2}\right)\left(\frac{d}{d+1}\right)^{2}\right.\nonumber \\
 &  & \left.\left[\frac{\beta\left(d^{3}+d^{2}-2d+1\right)+4d-2}{\left(d-1\right)\left(\beta d^{2}+2\right)}\right]\right\} \label{eq:1-pSlider}
\end{eqnarray}

We want to show that $0\le1-p\le1$ for any integer $k\ge1,\mbox{ }d\ge2$
and $\mbox{ }\beta\ge1$. All the factors are manifestly $\ge0$,
therefore $1-p\ge0$. The first two factors are clearly $\le1$ so
we need to prove that the term in the braces is too. Further, $k=1$
provides an upper bound as $\left(1-\frac{1-d^{-2k+1}}{1+\beta d^{2}/2}\right)\le\left(1-\frac{1-d^{-3}}{1+\beta d^{2}/2}\right)$.
We rewrite the term in the braces

\begin{equation}
\frac{d\left(\beta d^{3}+2\right)\left[\beta\left(d^{3}+d^{2}-2d+1\right)+4d-2\right]}{\left(\beta d^{2}+2\right)^{2}\left(d+1\right)^{2}\left(d-1\right)},\label{eq:num-denom}
\end{equation}
but we can subtract the denominator from the numerator to get\vspace{-1cm}

\[
\left(\beta d+2\right)\left[\beta\left(d^{4}-2d^{3}\right)+2\left(d^{3}-d^{2}-1\right)\right]\ge0\quad\forall\; d\ge2.
\]

This proves that ($\ref{eq:num-denom}$) is less than one. Therefore,
the term in the braces is less than one and hence $0\le p\le1$. Let
us note the following limits of interest (recall $N-1=2k$)\vspace{-2cm}

\begin{eqnarray*}
\lim_{d\rightarrow\infty}\left(1-p\right) & = & \frac{2k-1}{k^{2}}\overset{k=1}{=}1\\
\lim_{N\rightarrow\infty}\left(1-p\right) & \sim & \frac{1}{N}\rightarrow0
\end{eqnarray*}
the first limit tells us that if we consider having two local terms
and take the local dimension to infinity we have essentially free
probability theory as expected. The second limit shows that in the
thermodynamical limit (i.e., $N\rightarrow\infty$) the convex combination
slowly approaches the classical end. In the limit where $\beta\rightarrow\infty$
the $\beta$ dependence in $\left(1-p\right)$ cancels out. This is
a reconfirmation of the fact that in free probability theory, for
$\beta\rightarrow\infty$, the result should be independent of $\beta$.
We see that the bounds are tight. 

\begin{figure}[H]
\begin{centering}
\includegraphics[scale=0.35]{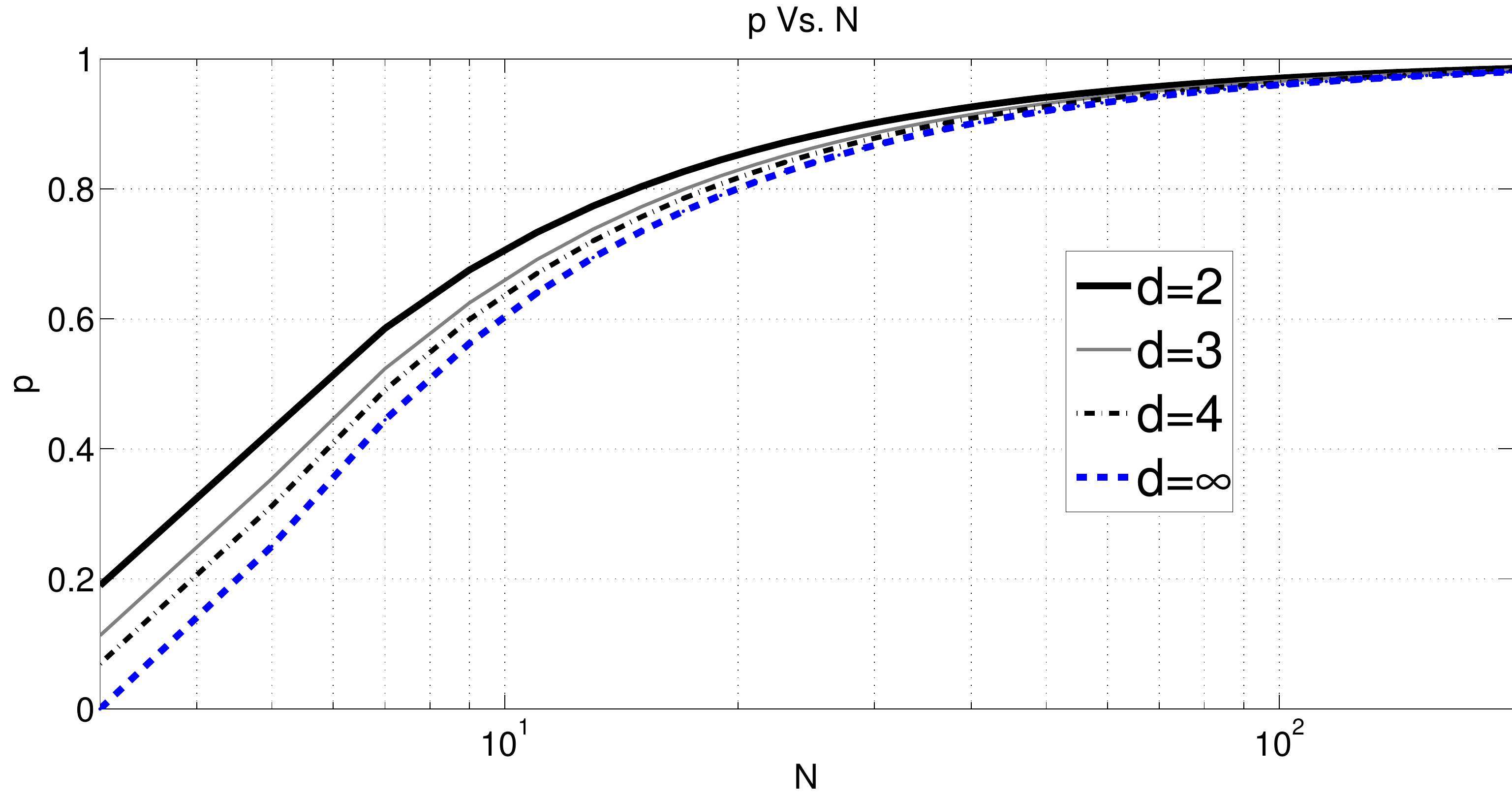}
\par\end{centering}

\centering{}\caption{\label{fig:pVsN}An example: $\beta=1$: the quantum problem for all
$d$ lies in between the iso $(p=0)$ and the classical $(p=1)$. }
\end{figure}
\end{onehalfspace}

\end{proof}
\begin{onehalfspace}
Comment: Entanglement shows itself starting at the fourth moment;
further, in the expansion of the fourth moments only the terms that
involve \textit{a pair} of local terms \textit{sharing a site} differ.
Note that when the QMBS possesses a translational symmetry, there
is an additional complication introduced by the dependence of the
local terms. Though, in this case, the non-iid nature of the local
terms complicates the matter theoretically, we have not seen a practical
limitation of IE in our numerical experiments. 

Comment: One could from the beginning use free approximation instead
of isotropic ($m\rightarrow\infty$), in which case the proofs are
simplified.

We now \textit{\uline{summarize}} the main thesis of this work.
We are interested in the eigenvalue distribution of 

\[
H\equiv H_{\mbox{odd}\vphantom{\mbox{even}}}+H_{\mbox{even}\mbox{\ensuremath{\vphantom{odd}}}}=\sum_{l=1,3,5,\cdots}\mathbb{I}\otimes H_{l,l+1}\otimes\mathbb{I}+\sum_{l=2,4,6,\cdots}\mathbb{I}\otimes H_{l,l+1}\otimes\mathbb{I},
\]
 which in a basis in that $H_{\mbox{odd}\vphantom{\mbox{even}}}$
is diagonal reads $H=A+Q_{q}^{-1}BQ_{q}$. Since this problem has
little hope in being solved exactly we consider along with it two
known approximations:

\begin{eqnarray*}
H_{c} & = & A+\Pi^{-1}B\Pi\\
H & = & A+Q_{q}^{-1}BQ_{q}\\
H_{iso} & = & A+Q^{-1}BQ.
\end{eqnarray*}

We proved that the first three moments of the three foregoing equations
are equal. We then calculated their fourth moments as encoded by their
kurtoses ($\gamma_{2}$'s) analytically and proved that there exists
a $0\le p\le$1 such that 

\[
\gamma_{2}^{q}=p\gamma_{2}^{c}+\left(1-p\right)\gamma_{2}^{iso}.
\]
It turned out that the only terms in the expansion of the fourth moments
that were relevant were

\begin{equation}
1-p=\frac{\mathbb{E}\mbox{Tr}\left\{ \left(A\Pi^{-1}B\Pi\right)^{2}-\left(AQ_{q}^{-1}BQ_{q}\right)^{2}\right\} }{\mathbb{E}\mbox{Tr}\left\{ \left(A\Pi^{-1}B\Pi\right)^{2}-\left(AQ^{-1}BQ\right)^{2}\right\} }.\label{eq:1-P_Effective}
\end{equation}
Through direct calculation we found that the numerator $\mathbb{E}\mbox{Tr}\left\{ \left(A\Pi^{-1}B\Pi\right)^{2}-\left(AQ_{q}^{-1}BQ_{q}\right)^{2}\right\} $
evaluates to be

\[
d\left(2k-1\right)\left(m_{2}^{\mbox{odd}}-m_{1,1}^{\mbox{odd}}\right)\left(m_{2}^{\mbox{even}}-m_{1,1}^{\mbox{even}}\right)\mathbb{E}\left(\left\Vert uv\right\Vert _{\mbox{F}}^{2}-\left\Vert uv\left(uv\right)^{\dagger}\right\Vert _{\mbox{F}}^{2}\right),
\]
and the denominator $\mathbb{E}\mbox{Tr}\left\{ \left(A\Pi^{-1}B\Pi\right)^{2}-\left(AQ^{-1}BQ\right)^{2}\right\} $

\[
\left(m_{2}^{\mbox{odd}}-m_{1,1}^{\mbox{odd}}\right)\left(m_{2}^{\mbox{even}}-m_{1,1}^{\mbox{even}}\right)\left(\frac{km\left(n-1\right)}{n\left(m-1\right)}\right)^{2}\left\{ 1-m\mathbb{E}\left(q_{ij}^{4}\right)\right\} .
\]

Therefore $1-p$ does not depend on the local distribution and can
generally be expressed as

\framebox{\begin{minipage}[t]{1\columnwidth}%
\[
1-p=\frac{d\left(2k-1\right)\mathbb{E}\left(\left\Vert uv\right\Vert _{\mbox{F}}^{2}-\left\Vert uv\left(uv\right)^{\dagger}\right\Vert _{\mbox{F}}^{2}\right)}{\left(\frac{km\left(n-1\right)}{n\left(m-1\right)}\right)^{2}\left\{ 1-m\mathbb{E}\left(q_{ij}^{4}\right)\right\} }.
\]
\end{minipage}}

\bigskip{}

If we further assume that the local eigenvectors are $\beta-\mbox{Haar}$
distributed we get 

\begin{eqnarray*}
1-p & = & \left(1-d^{-2k-1}\right)\left[1-\left(\frac{k-1}{k}\right)^{2}\right]\left(1-\frac{1-d^{-2k+1}}{1+\beta d^{2}/2}\right)\left(\frac{d}{d+1}\right)^{2}\\
 & \times & \left[\frac{\beta\left(d^{3}+d^{2}-2d+1\right)+4d-2}{\left(d-1\right)\left(\beta d^{2}+2\right)}\right].
\end{eqnarray*}

\medskip{}

Next we asserted that this $p$ can be used to approximate the distribution

\[
d\nu^{q}\approx d\nu^{IE}=pd\nu^{c}+\left(1-p\right)d\nu^{iso}.
\]
We argued that the spectra obtained using Isotropic Entanglement (IE)
are accurate well beyond four moments. 

For illustration, we apply IE theory in full detail to a chain with
Wishart matrices as local terms. Other types of local terms (e.g.
GOE, random $\pm1$ eigenvalues) can be treated similarly; therefore
in Section \ref{sec:Other-Examples} we show the plots comparing IE
with exact diagonalization for these cases. 
\end{onehalfspace}

\begin{onehalfspace}

\section{\label{sec:First-Example:-Wishart}A Detailed Example: Wishart Matrices
as Local Terms }
\end{onehalfspace}

\begin{onehalfspace}
As an example take a chain with odd number of sites and for the local
terms in Eq. \ref{eq:Hamiltonian} pick $H^{\left(l\right)}=W^{T}W$,
where $W$ is a rank $r$ matrix whose elements are picked randomly
from a Gaussian distribution ($\beta=1$); these matrices $W^{T}W$
are known as \textit{Wishart matrices. }Clearly the maximum possible
rank is $r=d^{2}$ for each of the local terms.

Any cumulant is equal to the corresponding cumulant of one local term,
denoted by $\kappa$, times the number of summands in Eq. \ref{eq:Hamiltonian}.
In particular, the fourth cumulant of $H$ is $\kappa_{4}^{\left(N-1\right)}=\left(N-1\right)\kappa_{4}$.
Below we drop the superscripts when the quantity pertains to the whole
chain. Next we recall the definitions in terms of cumulants of the
mean $(\mu)$, the variance $(\sigma^{2})$, the skewness $(\gamma_{1})$,
and the kurtosis $(\gamma_{2})$ 

\begin{equation}
\begin{array}{cccccccc}
\mu\equiv\kappa_{1} & \quad & \sigma^{2}\equiv\kappa_{2} & \quad & \gamma_{1}\equiv\frac{\kappa_{3}}{\sigma^{3}} & \quad & \gamma_{2}\equiv\frac{\kappa_{4}}{\sigma^{4}}=\frac{m_{4}}{\sigma^{4}}-3 & .\end{array}\label{eq:momCumul}
\end{equation}

\end{onehalfspace}

\begin{onehalfspace}

\subsection{\label{sub:Wishart-Classical-Case:-Calculation}Evaluation of $p=\frac{\gamma_{2}^{q}-\gamma_{2}^{iso}}{\gamma_{2}^{c}-\gamma_{2}^{iso}}$}
\end{onehalfspace}

\begin{onehalfspace}
The moments of the local terms are obtained from MOPS \cite{mops}; 

\begin{equation}
\begin{array}{c}
m_{1}=\beta r\\
m_{2}=\beta r\left[\beta\left(r+n-1\right)+2\right]\\
m_{3}=\beta r\left\{ \beta^{2}\left[n^{2}+\left(r-1\right)\left(3n+r-2\right)\right]+6\beta\left(n+r-1\right)+8\right\} \\
m_{4}=\beta r\left\{ 48+\beta^{3}\left[n^{3}+6n^{2}(r-1)+n\left(6r-11\right)\left(r-1\right)-6\left(r^{2}+1\right)+r^{3}+11r\right]\right.\\
\left.+2\beta^{2}\left[6\left(n^{2}+r^{2}\right)+17\left(n(r-1)-r\right)+11\right]+44\beta\left(n+r-1\right)\right\} \\
m_{1,1}=\beta^{2}r\left(r-1\right)
\end{array}
\end{equation}
which for real matrices $\beta=1$ yields

\begin{equation}
\begin{array}{c}
m_{1}=r\\
m_{2}=r\left(r+n+1\right)\\
m_{3}=r\left(n^{2}+3n+3rn+3r+r^{2}+4\right)\\
m_{4}=r\left(6n^{2}+21n+6rn^{2}+17rn+21r+6nr^{2}+6r^{2}+n^{3}+r^{3}+20\right)\\
m_{1,1}=r\left(r-1\right).
\end{array}
\end{equation}

The mean, variance, skewness, and kurtosis are obtained from the foregoing
relations, through the cumulants Eq. \ref{eq:momCumul}. We drop the
superscripts when the quantity pertains to the whole chain. Therefore,
using Eq. \ref{eq:momCumul}, we have

\begin{equation}
\begin{array}{ccc}
\mu\equiv\left(N-1\right)r & \quad & \sigma^{2}\equiv r\left(N-1\right)\left(n+1\right)\\
\gamma_{1}\equiv\frac{n^{2}+3n+4}{\left(n+1\right)^{3/2}\sqrt{r\left(N-1\right)}} & \quad & \gamma_{2}^{\left(c\right)}\equiv\frac{n^{2}\left(n+6\right)-rn\left(n+1\right)+21n+2r+20}{r\left(N-1\right)\left(n+1\right)^{2}}.
\end{array}\label{eq:AllClassical}
\end{equation}

From Eq. \ref{eq:m_2} we readily obtain 

\begin{equation}
\frac{1}{m}\mathbb{E}\textrm{Tr}\left(A\Pi^{T}B\Pi\right)^{2}=r^{2}k^{2}\left(rk+n+1\right)^{2}.\label{eq:Classical}
\end{equation}

By The Matching Three Moments theorem we immediately have the mean,
the variance and the skewness for the isotropic case 

\[
\begin{array}{ccc}
\mu=\left(N-1\right)r &  & \sigma^{2}=r\left(N-1\right)\left(n+1\right)\\
 & \gamma_{1}=\frac{n^{2}+3n+4}{\left(n+1\right)^{3/2}\sqrt{r\left(N-1\right)}}.
\end{array}
\]

Note that the denominator in Eq. \ref{eq:convex} becomes, 

\begin{equation}
\gamma_{2}^{c}-\gamma_{2}^{iso}=\frac{\kappa_{4}^{(c)}-\kappa_{4}^{(iso)}}{\sigma^{4}}=\frac{2}{m}\frac{\mathbb{E}\left\{ \textrm{Tr}\left[\left(A\Pi^{T}B\Pi\right)^{2}-\left(AQ^{T}BQ\right)^{2}\right]\right\} }{r^{2}\left(N-1\right)^{2}\left(n+1\right)^{2}}.\label{eq:denomEx}
\end{equation}

In the case of Wishart matrices, $m_{1}^{\mbox{odd}}=m_{1}^{\mbox{even}}=r$,
and $m_{2}^{\mbox{odd}}=m_{2}^{\mbox{even}}=r\left(r+n+1\right)$,
$m_{11}^{\mbox{odd}}=m_{11}^{\mbox{even}}=r\left(r-1\right)$ given
by Eqs. \ref{eq:m_2} and \ref{eq:elemGeneral} respectively. Therefore
we can substitute these into Eq. \ref{eq:Iso-Classical_FINAL} 

\begin{eqnarray}
\frac{1}{m}\mathbb{E}\left\{ \textrm{Tr}\left[\left(A\Pi^{T}B\Pi\right)^{2}-\left(AQ^{T}BQ\right)^{2}\right]\right\}  & = & \left(m_{2}^{(A)}-m_{1,1}^{(A)}\right)\left(m_{2}^{(B)}-m_{1,1}^{(B)}\right)\left\{ 1-m\mathbb{E}\left(q_{ij}^{4}\right)\right\} \nonumber \\
 & = & \frac{\beta\left(m-1\right)}{\left(m\beta+2\right)}\left(\frac{km\left(n-1\right)}{n\left(m-1\right)}\right)^{2}\left(m_{2}-m_{1,1}\right)^{2}\nonumber \\
 & = & \frac{\beta k^{2}m^{2}\left(n-1\right)^{2}}{\left(m\beta+2\right)\left(m-1\right)n^{2}}\left(m_{2}-m_{1,1}\right)^{2}\label{eq:WishartDenom}
\end{eqnarray}
 One can also calculate each of the terms separately and obtain the
same results (see Appendix for the alternative).

From Eq. \ref{eq:Quantum-Classical} we have 

\begin{eqnarray}
\frac{1}{m}\mathbb{E}\left[\mbox{Tr}\left(A\Pi^{T}B\Pi\right)^{2}-\textrm{Tr}\left(AQ_{q}^{T}BQ_{q}\right)^{2}\right] & = & d\left(2k-1\right)\left(m_{2}-m_{1,1}\right)^{2}\nonumber \\
 & \times & \mathbb{E}\left(\left\Vert uv\right\Vert _{\mbox{F}}^{2}-\left\Vert uv\left(uv\right)^{\dagger}\right\Vert _{\mbox{F}}^{2}\right)\nonumber \\
 & = & \left(2k-1\right)\left(m_{2}-m_{1,1}\right)^{2}\nonumber \\
 & \times & \left\{ \frac{1+d\left(d^{3}+d-3\right)}{\left(d^{2}+2\right)^{2}}\right\} .\label{eq:WishartNum}
\end{eqnarray}
We can divide Eq. \ref{eq:WishartNum} by Eq. \ref{eq:WishartDenom}
to evaluate the parameter $p$. 
\end{onehalfspace}

\begin{onehalfspace}

\subsection{\label{sub:Wishart-Summary}Summary of Our Findings and Further Numerical
Results}
\end{onehalfspace}

\begin{onehalfspace}
We summarize the results along with numerical experiments in Tables
\ref{tab:summaryTheory} and \ref{tab:SummaryNumerical} to show the
equivalence of the first three moments and the departure of the three
cases in their fourth moment. As said above,

\begin{equation}
Q_{c}=\mathbb{I}_{d^{N}}\qquad Q_{iso}\equiv Q\;\textrm{Haar}\; d^{N}\times d^{N}\qquad Q_{q}=\left(Q_{q}^{(A)}\right)^{T}Q_{q}^{(B)}
\end{equation}

where, $\left(Q_{q}^{(A)}\right)^{T}Q_{q}^{(B)}$ is given by Eq.
\ref{eq:OddEven}. In addition, from Eq. \ref{eq:AllClassical} we
can define $\Delta$ to be the part of the kurtosis that is equal
among the three cases

$\Delta=$ $\gamma_{2}^{c}$ - $\frac{2m_{2}^{A}m_{2}^{B}}{\sigma^{4}}=$
$\gamma_{2}^{c}-\frac{1}{2}\frac{\left(rk+n+1\right)^{2}}{\left(n+1\right)^{2}}$.

Using $\Delta$ we can obtain the full kurtosis for the iso and quantum
case, and therefore (see Table \ref{tab:summaryTheory} for a theoretical
summary): 

\begin{equation}
p=\frac{\gamma_{2}^{q}-\gamma_{2}^{iso}}{\gamma_{2}^{c}-\gamma_{2}^{iso}}.
\end{equation}

\begin{table}
\begin{centering}
\begin{tabular}{|c||c||c||c|}
\hline 
$\beta=1$ Wishart & Iso & Quantum & Classical \tabularnewline
\hline 
\hline 
Mean $\mu$ & \multicolumn{3}{c|}{$r\left(N-1\right)$}\tabularnewline
\hline 
Variance $\sigma^{2}$ & \multicolumn{3}{c|}{$r\left(N-1\right)\left(d^{2}+1\right)$}\tabularnewline
\hline 
Skewness $\gamma_{1}$ & \multicolumn{3}{c|}{$\frac{d^{4}+3d^{2}+4}{\sqrt{r\left(N-1\right)\left(d^{2}+1\right)^{3}}}$}\tabularnewline
\hline 
$\frac{1}{m}\mathbb{E}\left[\textrm{Tr}\left(AQ_{\bullet}^{T}BQ_{\bullet}\right)^{2}\right]$ & $m_{2}^{A}m_{2}^{B}-$ Eq. \ref{eq:WishartDenom} & $m_{2}^{A}m_{2}^{B}-$ Eq. \ref{eq:WishartNum} & $r^{2}k^{2}\left(rk+n+1\right)^{2}$\tabularnewline
\hline 
Kurtosis $\gamma_{2}^{\left(\bullet\right)}$ & $\frac{2}{m\sigma^{4}}\mathbb{E}\left[\textrm{Tr}\left(AQ^{T}BQ\right)^{2}\right]$$+\Delta$ & $\frac{2}{m\sigma^{4}}\mathbb{E}\left[\textrm{Tr}\left(AQ_{q}^{T}BQ_{q}\right)^{2}\right]$$+\Delta$ & Eq. \ref{eq:AllClassical}\tabularnewline
\hline 
\end{tabular}
\par\end{centering}

\centering{}\caption{\label{tab:summaryTheory}Summary of the results when the local terms
are Wishart matrices. The fourth moment is where the three cases differ.}
\end{table}

\begin{table}
\begin{centering}
\begin{longtable}{|c||c||c||c|c||c||c|c|c|c|}
\hline 
\noalign{\vskip\doublerulesep}
\multicolumn{10}{|c|}{Experiments based on $500000$ trials}\tabularnewline[1sp]
\hline 
\noalign{\vskip\doublerulesep}
\multicolumn{4}{|c|}{$N=3$ } & \multicolumn{3}{c|}{Theoretical value} & \multicolumn{3}{c|}{Numerical Experiment }\tabularnewline[1sp]
\hline 
\noalign{\vskip\doublerulesep}
\multicolumn{4}{|c|}{} & Iso & Quantum & Classical & \multicolumn{1}{c||}{Iso} & \multicolumn{1}{c||}{Quantum} & Classical\tabularnewline[1sp]
\hline 
\hline 
\noalign{\vskip\doublerulesep}
\multicolumn{4}{|c|}{Mean $\mu$} & \multicolumn{3}{c|}{$8$} & $8.007$ & $8.007$ & $7.999$\tabularnewline[1sp]
\hline 
\noalign{\vskip\doublerulesep}
\multicolumn{4}{|c|}{Variance $\sigma^{2}$} & \multicolumn{3}{c|}{$40$} & $40.041$ & $40.031$ & $39.976$\tabularnewline[1sp]
\hline 
\noalign{\vskip\doublerulesep}
\multicolumn{4}{|c|}{Skewness $\gamma_{1}$} & \multicolumn{3}{c|}{$\frac{8}{25}\sqrt{10}=1.01192$} & $1.009$ & $1.009$ & $1.011$\tabularnewline[1sp]
\hline 
\noalign{\vskip\doublerulesep}
\multicolumn{4}{|c|}{Kurtosis $\gamma_{2}$} & \multicolumn{1}{c|}{$\frac{516}{875}=0.590$} & \multicolumn{1}{c|}{$\frac{33}{50}=0.660$} & $\frac{24}{25}=0.960$  & $0.575$ & $0.645$ & $0.953$\tabularnewline[1sp]
\hline 
\noalign{\vskip\doublerulesep}
\multicolumn{10}{|c|}{Experiments based on $500000$ trials}\tabularnewline[1sp]
\hline 
\noalign{\vskip\doublerulesep}
\multicolumn{4}{|c|}{$N=5$ } & \multicolumn{3}{c|}{Theoretical value} & \multicolumn{3}{c|}{Numerical Experiment }\tabularnewline[1sp]
\hline 
\noalign{\vskip\doublerulesep}
\multicolumn{4}{|c|}{} & Iso & Quantum & Classical & \multicolumn{1}{c||}{Iso} & \multicolumn{1}{c||}{Quantum} & Classical\tabularnewline[1sp]
\hline 
\hline 
\noalign{\vskip\doublerulesep}
\multicolumn{4}{|c|}{Mean $\mu$} & \multicolumn{3}{c|}{$16$} & $15.999$ & $15.999$ & $16.004$\tabularnewline[1sp]
\hline 
\noalign{\vskip\doublerulesep}
\multicolumn{4}{|c|}{Variance $\sigma^{2}$} & \multicolumn{3}{c|}{$80$} & $79.993$ & $80.005$ & $80.066$\tabularnewline[1sp]
\hline 
\noalign{\vskip\doublerulesep}
\multicolumn{4}{|c|}{Skewness $\gamma_{1}$} & \multicolumn{3}{c|}{$\frac{8}{25}\sqrt{5}=0.716$} & $0.715$ & $0.715$ & $0.717$\tabularnewline[1sp]
\hline 
\noalign{\vskip\doublerulesep}
\multicolumn{4}{|c|}{Kurtosis $\gamma_{2}$} & \multicolumn{1}{c|}{$\frac{228}{2635}=0.087$} & \multicolumn{1}{c|}{$\frac{51}{200}=0.255$} & $\frac{12}{25}=0.480$  & $0.085$ & $0.255$ & $0.485$\tabularnewline[1sp]
\hline 
\noalign{\vskip\doublerulesep}
\multicolumn{10}{|c|}{Experiments based on $300000$ trials}\tabularnewline[1sp]
\hline 
\noalign{\vskip\doublerulesep}
\multicolumn{4}{|c|}{$N=7$ } & \multicolumn{3}{c|}{Theoretical value} & \multicolumn{3}{c|}{Numerical Experiment }\tabularnewline[1sp]
\hline 
\noalign{\vskip\doublerulesep}
\multicolumn{4}{|c|}{} & Iso & Quantum & Classical & \multicolumn{1}{c||}{Iso} & \multicolumn{1}{c||}{Quantum} & Classical\tabularnewline[1sp]
\hline 
\hline 
\noalign{\vskip\doublerulesep}
\multicolumn{4}{|c|}{Mean $\mu$} & \multicolumn{3}{c|}{$24$} & $23.000$ & $23.000$ & $24.095$\tabularnewline[1sp]
\hline 
\noalign{\vskip\doublerulesep}
\multicolumn{4}{|c|}{Variance $\sigma^{2}$} & \multicolumn{3}{c|}{$120$} & $120.008$ & $120.015$ & $120.573$\tabularnewline[1sp]
\hline 
\noalign{\vskip\doublerulesep}
\multicolumn{4}{|c|}{Skewness $\gamma_{1}$} & \multicolumn{3}{c|}{$\frac{8}{75}\sqrt{30}=0.584$} & $0.585$ & $0.585$ & $0.588$\tabularnewline[1sp]
\hline 
\noalign{\vskip\doublerulesep}
\multicolumn{4}{|c|}{Kurtosis $\gamma_{2}$} & \multicolumn{1}{c|}{$-\frac{16904}{206375}=-0.082$} & \multicolumn{1}{c|}{$\frac{23}{150}=0.153$} & $\frac{8}{25}=0.320$  & $-0.079$ & $0.156$ & $0.331$\tabularnewline[1sp]
\newpage
\hline 
\noalign{\vskip\doublerulesep}
\multicolumn{10}{|c|}{Experiments based on $40000$ trials}\tabularnewline[1sp]
\hline 
\noalign{\vskip\doublerulesep}
\multicolumn{4}{|c|}{$N=9$ } & \multicolumn{3}{c|}{Theoretical value} & \multicolumn{3}{c|}{Numerical Experiment }\tabularnewline[1sp]
\hline 
\noalign{\vskip\doublerulesep}
\multicolumn{4}{|c|}{} & Iso & Quantum & Classical & \multicolumn{1}{c||}{Iso} & \multicolumn{1}{c||}{Quantum} & Classical\tabularnewline[1sp]
\hline 
\hline 
\noalign{\vskip\doublerulesep}
\multicolumn{4}{|c|}{Mean $\mu$} & \multicolumn{3}{c|}{$32$} & $32.027$ & $32.027$ & $31.777$\tabularnewline[1sp]
\hline 
\noalign{\vskip\doublerulesep}
\multicolumn{4}{|c|}{Variance $\sigma^{2}$} & \multicolumn{3}{c|}{$160$} & $160.074$ & $160.049$ & $157.480$\tabularnewline[1sp]
\hline 
\noalign{\vskip\doublerulesep}
\multicolumn{4}{|c|}{Skewness $\gamma_{1}$} & \multicolumn{3}{c|}{$\frac{4}{25}\sqrt{10}=0.506$} & $0.505$ & $0.506$ & $0.500$\tabularnewline[1sp]
\hline 
\noalign{\vskip\doublerulesep}
\multicolumn{4}{|c|}{Kurtosis $\gamma_{2}$} & \multicolumn{1}{c|}{$-\frac{539142}{3283175}=-0.164$} & \multicolumn{1}{c|}{$\frac{87}{800}=0.109$} & $\frac{6}{25}=0.240$  & $-0.165$ & $0.109$ & $0.213$\tabularnewline[1sp]
\hline 
\noalign{\vskip\doublerulesep}
\multicolumn{10}{|c|}{Experiments based on $2000$ trials}\tabularnewline[1sp]
\hline 
\noalign{\vskip\doublerulesep}
\multicolumn{4}{|c|}{$N=11$} & \multicolumn{3}{c|}{Theoretical value} & \multicolumn{3}{c|}{Numerical Experiment }\tabularnewline[1sp]
\hline 
\noalign{\vskip\doublerulesep}
\multicolumn{4}{|c|}{} & \multicolumn{1}{c||}{Iso} & Quantum & Classical & \multicolumn{1}{c||}{Iso} & \multicolumn{1}{c||}{Quantum} & Classical\tabularnewline[1sp]
\hline 
\noalign{\vskip\doublerulesep}
\multicolumn{4}{|c|}{Mean $\mu$} & \multicolumn{3}{c|}{$40$} & $39.973$ & $39.973$ & $39.974$\tabularnewline[1sp]
\hline 
\noalign{\vskip\doublerulesep}
\multicolumn{4}{|c|}{Variance $\sigma^{2}$} & \multicolumn{3}{c|}{$200$} & $200.822$ & $200.876$ & $197.350$\tabularnewline[1sp]
\hline 
\noalign{\vskip\doublerulesep}
\multicolumn{4}{|c|}{Skewness $\gamma_{1}$} & \multicolumn{3}{c|}{$\frac{8}{25}\sqrt{2}=0.452548$} & $0.4618$ & $0.4538$ & $0.407$\tabularnewline[1sp]
\hline 
\noalign{\vskip\doublerulesep}
\multicolumn{4}{|c|}{Kurtosis $\gamma_{2}$} & \multicolumn{1}{c|}{$-\frac{11162424}{52454375}=-0.213$} & \multicolumn{1}{c|}{$\frac{21}{250}=0.084$} & $\frac{24}{125}=0.192$  & $-0.189$ & $0.093$ & $0.102$\tabularnewline[1sp]
\hline 
\end{longtable}\medskip{}

\par\end{centering}

\caption{\label{tab:SummaryNumerical}The mean, variance and skewness of classical,
iso and quantum results match. However, the fourth moments (kurtoses)
differ. Here we are showing results for $d=2,\; r=4$ with an accuracy
of three decimal points.}
\end{table}

\pagebreak{}

The numerical convergence of the kurtoses to the theoretical values
were rather slow. To make sure the results are consistent we did a
large run with $500$ million trials for $N=5$, $d=2$, $r=3$ and
$\beta=1$ and obtained four digits of accuracy

\[
\begin{array}{c}
\qquad\gamma_{2}^{c}-\gamma_{2}^{iso}=0.39340\qquad\textrm{Numerical\;\ experiment}\\
\gamma_{2}^{c}-\gamma_{2}^{iso}=0.39347\qquad\textrm{Theoretical\;\ value.}
\end{array}
\]

Convergence is faster if one calculates $p$ based on the departing
terms alone (Eq. \ref{eq:1-P_Effective}). In this case, for full
rank Wishart matrices with $N=5$ and $d=2$

\begin{tabular}{|c|c|}
\hline 
$\beta=1$, trials: $5$ Million & $1-p$\tabularnewline
\hline 
\hline 
Numerical Experiment & $0.57189$\tabularnewline
\hline 
Theoretical Value & $0.57183$\tabularnewline
\hline 
\end{tabular}\hspace{1cc}%
\begin{tabular}{|c|c|}
\hline 
$\beta=2$, trials: $10$ Million & $1-p$\tabularnewline
\hline 
\hline 
Numerical Experiment & $0.63912$\tabularnewline
\hline 
Theoretical Value & $0.63938$\tabularnewline
\hline 
\end{tabular}

\[
\]

Below we compare our theory against exact diagonalization for various
number of sites $N$, local ranks $r$, and site dimensionality $d$
(Figures \ref{fig:N=00003D3}-\ref{fig:N=00003D11}).

\begin{figure}[H]
\begin{centering}
\includegraphics[scale=0.4]{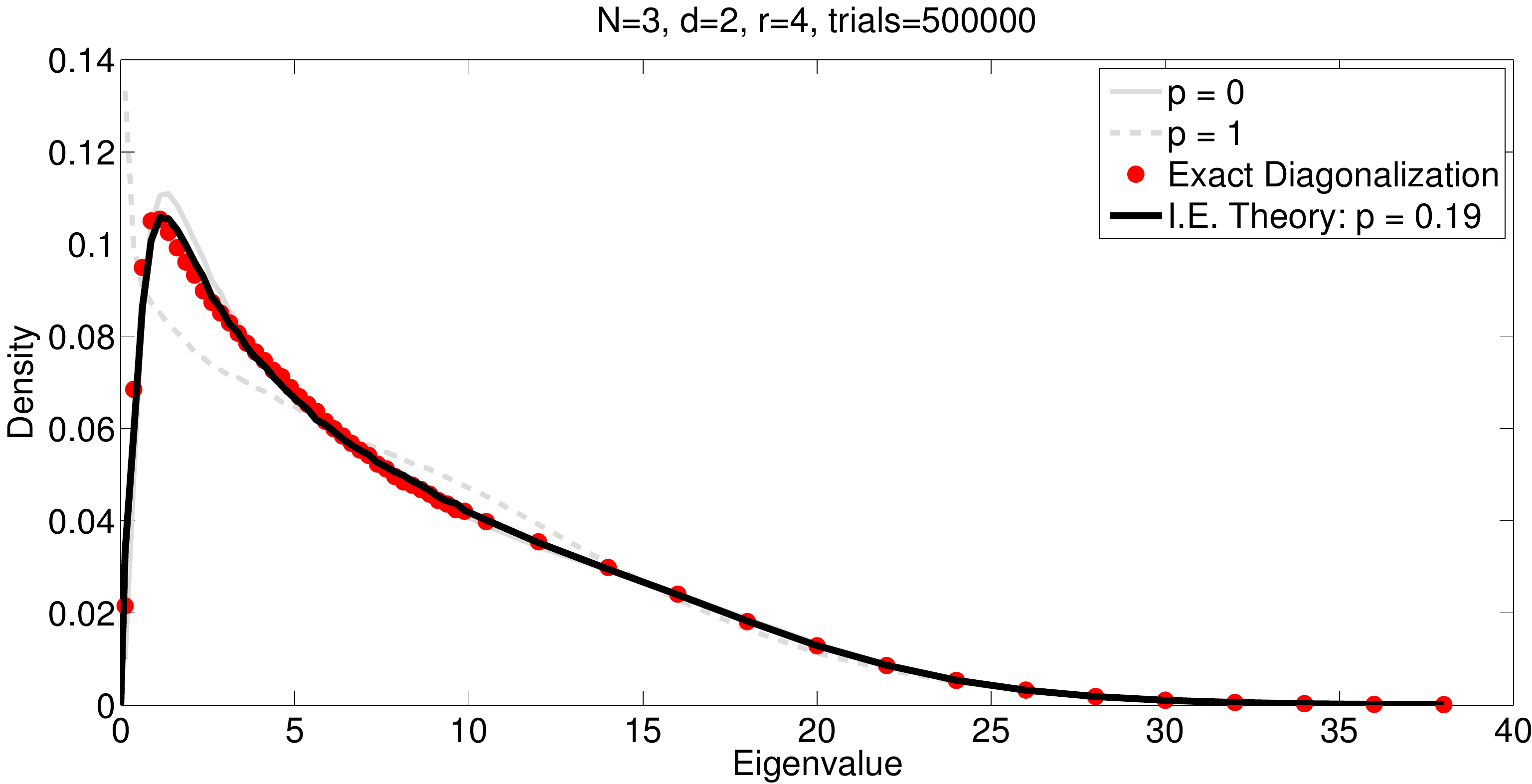}\vspace{0.2in}

\par\end{centering}

\begin{centering}
\includegraphics[scale=0.4]{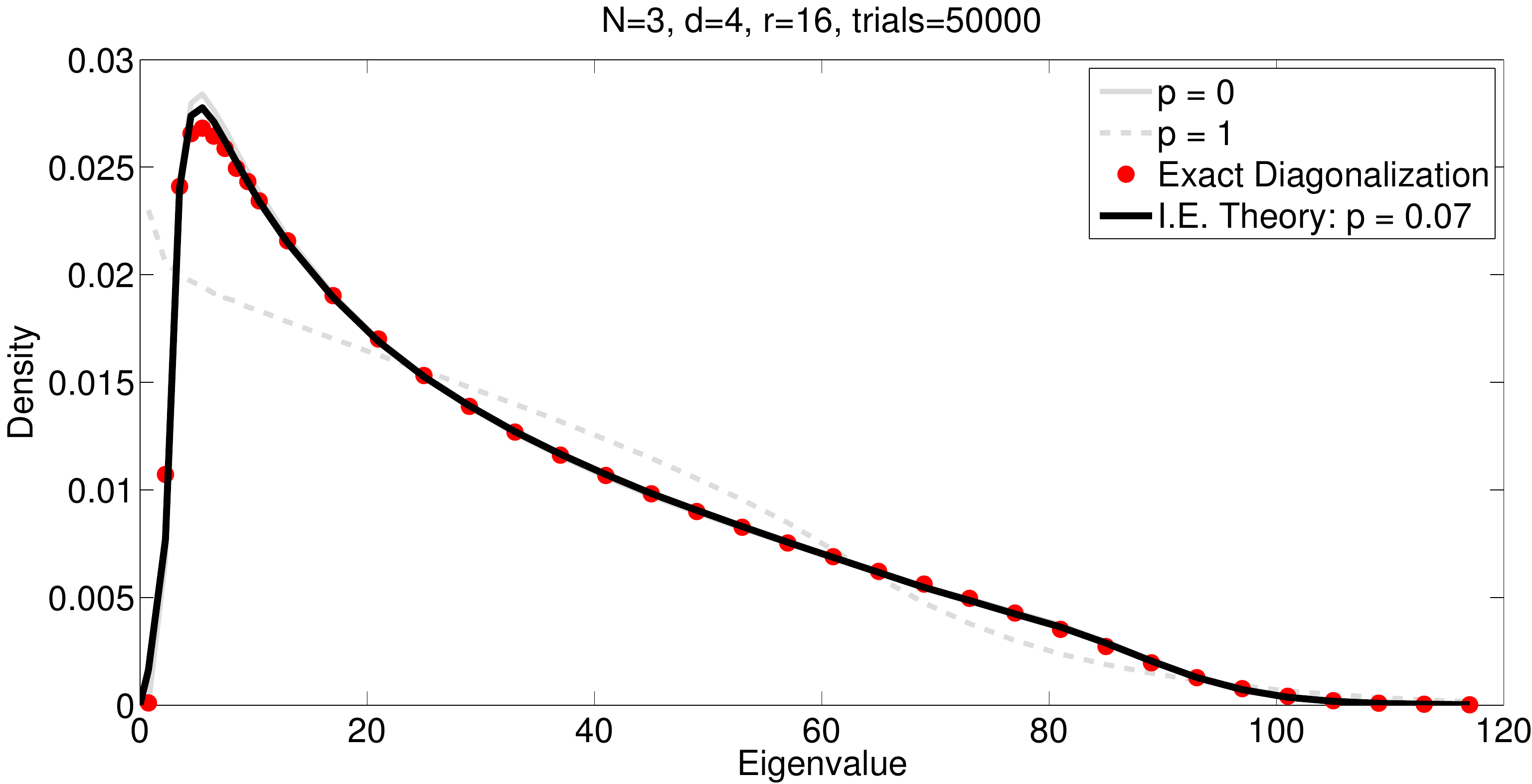}
\par\end{centering}

\begin{centering}
\vspace{0.2in}

\par\end{centering}

\begin{centering}
\includegraphics[scale=0.4]{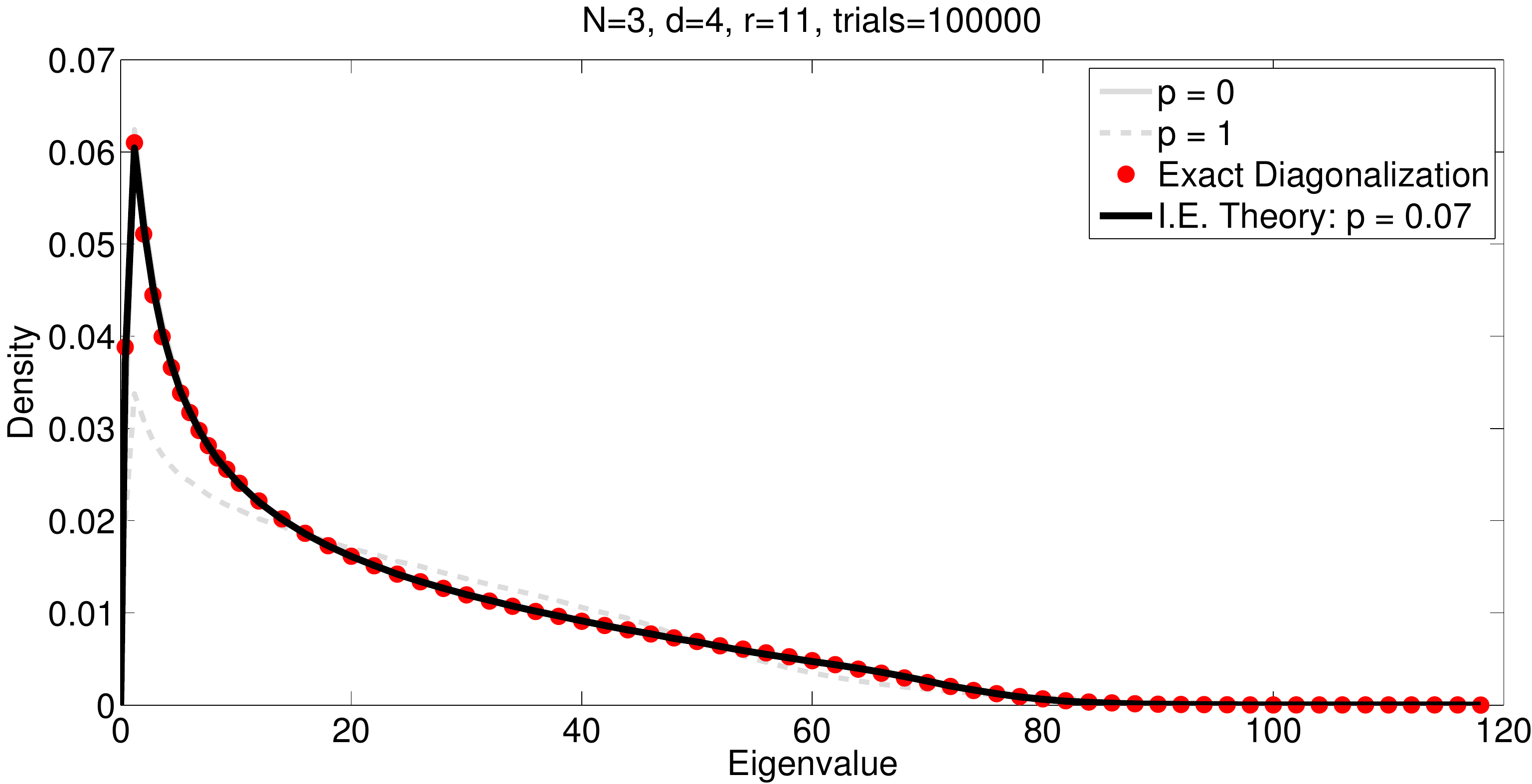}
\par\end{centering}

\centering{}\caption{\label{fig:N=00003D3}$N=3$ examples. Note that the last two plots
have the same $p$ despite having different ranks $r$. This is a
consequence of the Universality Lemma since they have the same $N$
and $d$. }
\end{figure}

\begin{figure}[H]
\begin{centering}
\includegraphics[scale=0.4]{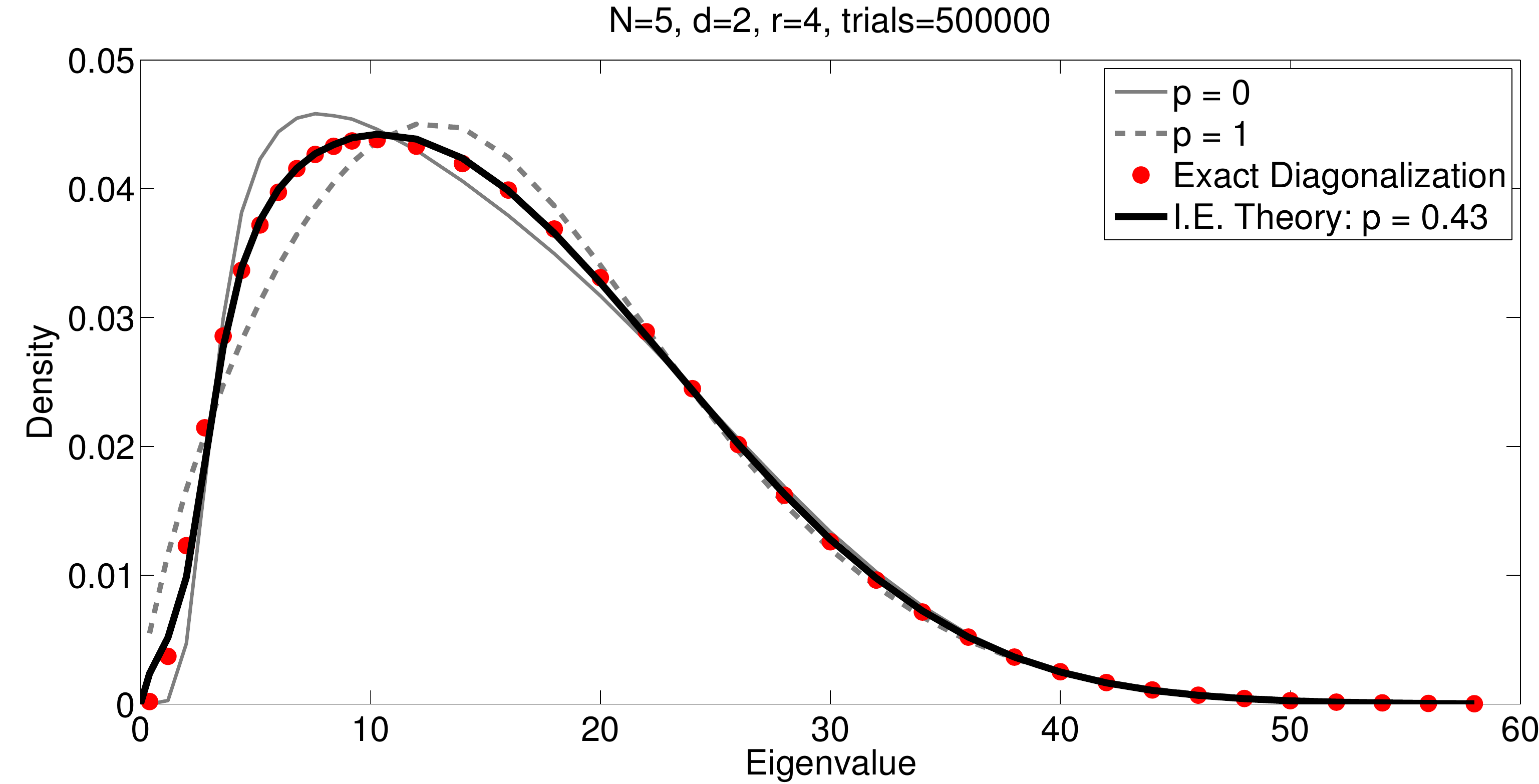}
\par\end{centering}

\begin{centering}
\vspace{0.2in}

\par\end{centering}

\begin{centering}
\includegraphics[scale=0.4]{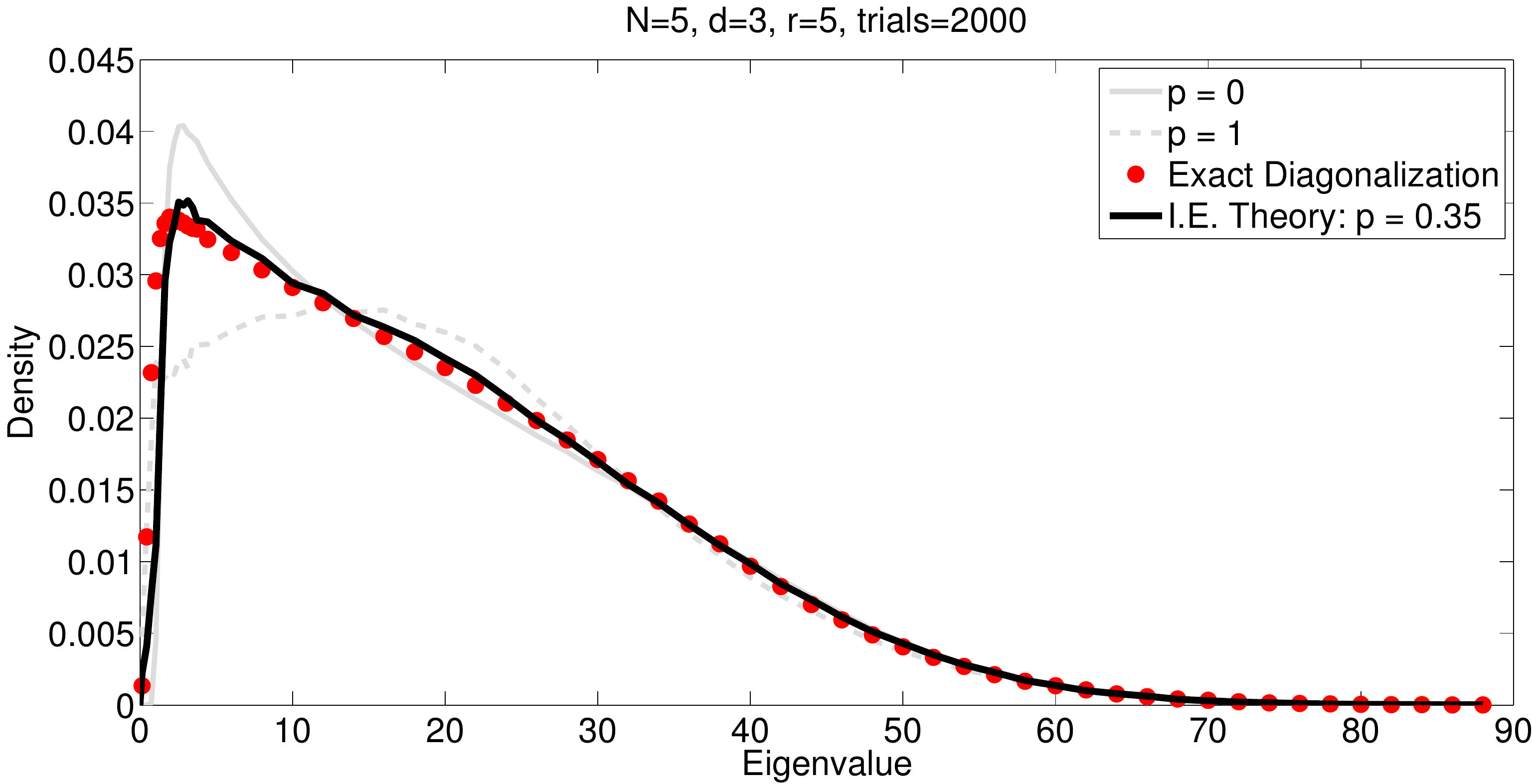}
\par\end{centering}

\centering{}\caption{$N=5$}
\end{figure}

\begin{figure}[H]
\begin{centering}
\includegraphics[scale=0.4]{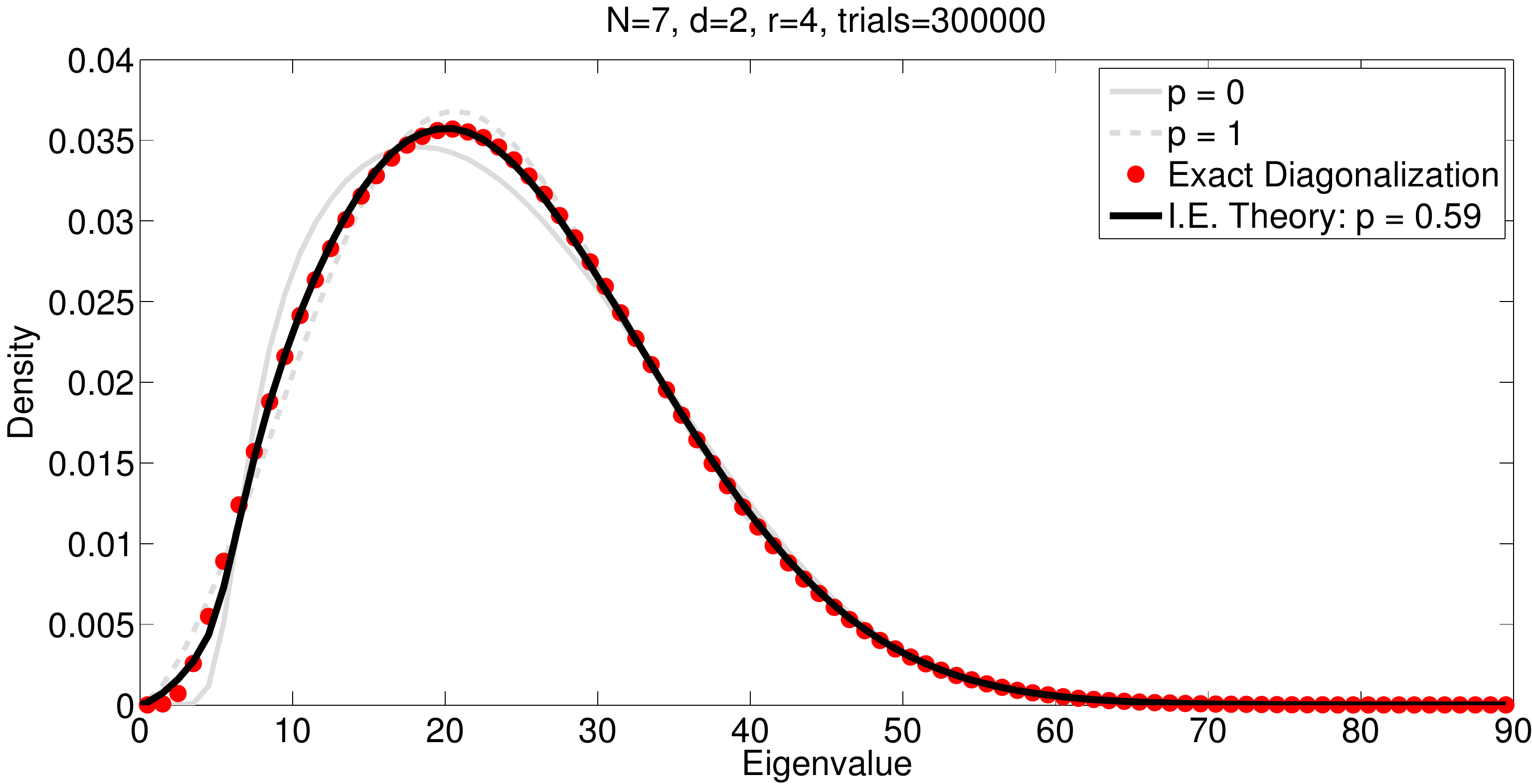}
\par\end{centering}

\centering{}\caption{$N=7$}
\end{figure}

\begin{figure}[H]
\begin{centering}
\includegraphics[scale=0.4]{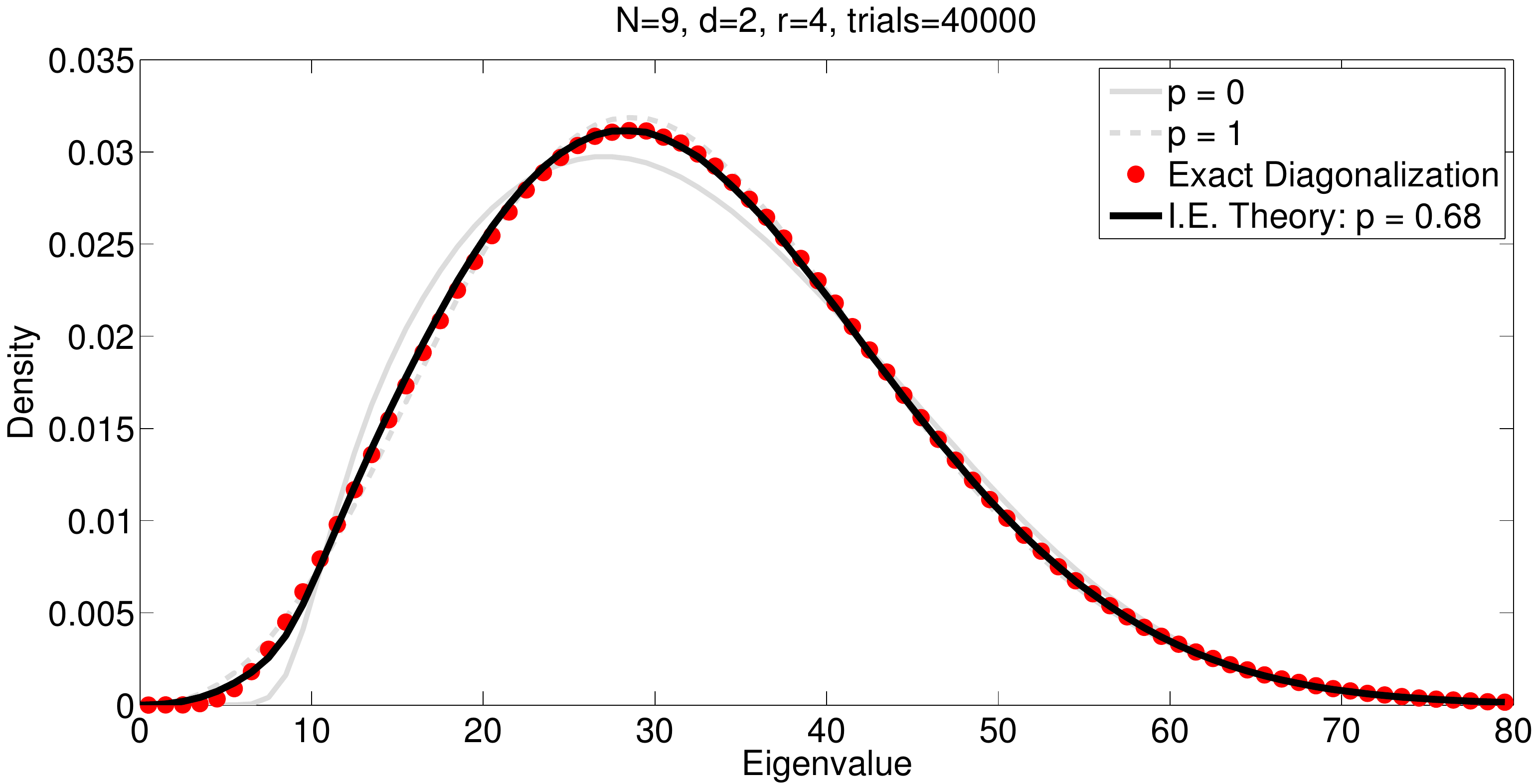}
\par\end{centering}

\begin{centering}
\includegraphics[scale=0.4]{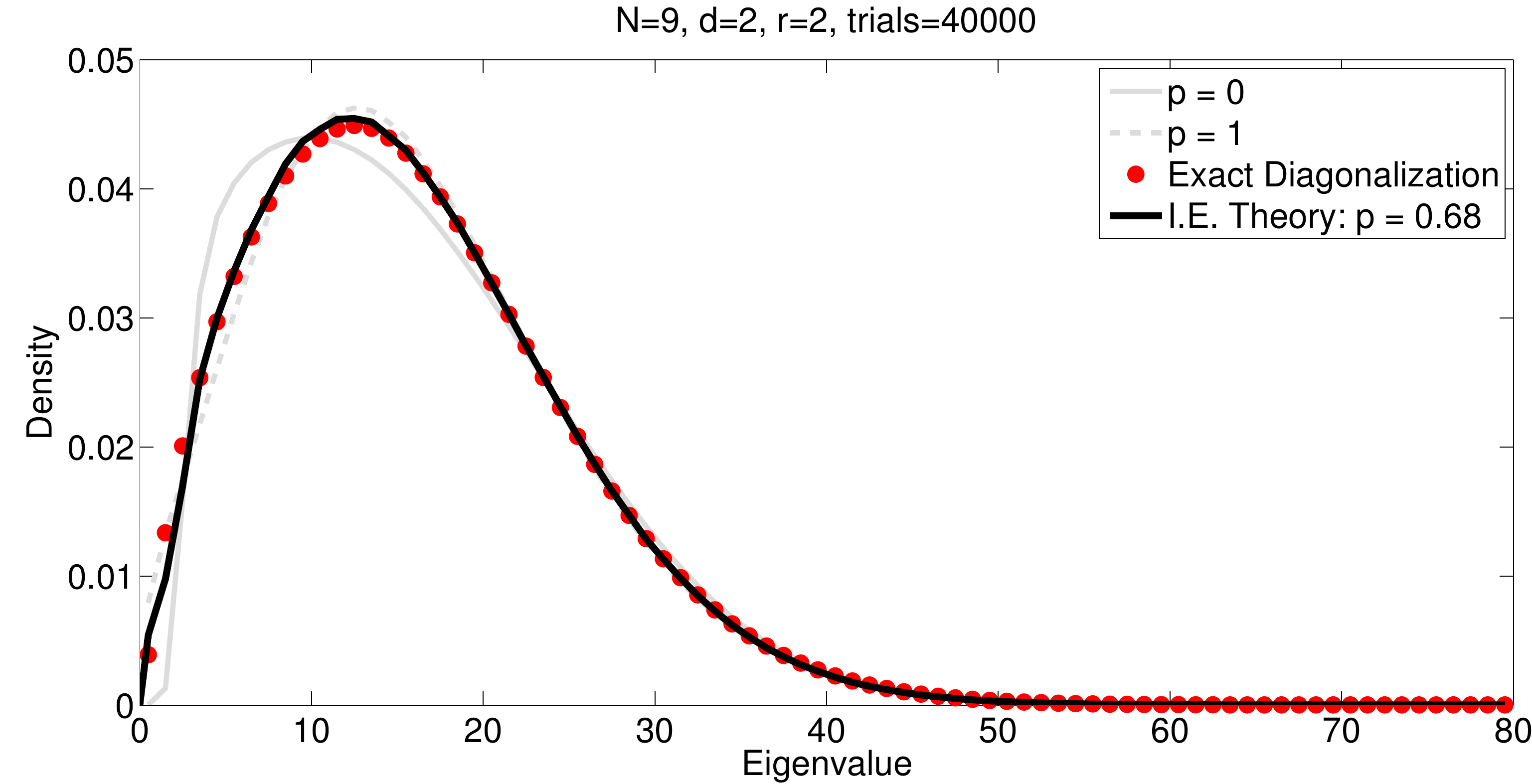}
\par\end{centering}

\centering{}\caption{$N=9$. Note that the two plots have the same $p$ despite having
different local ranks.}
\end{figure}

\begin{figure}[H]
\centering{}\includegraphics[scale=0.4]{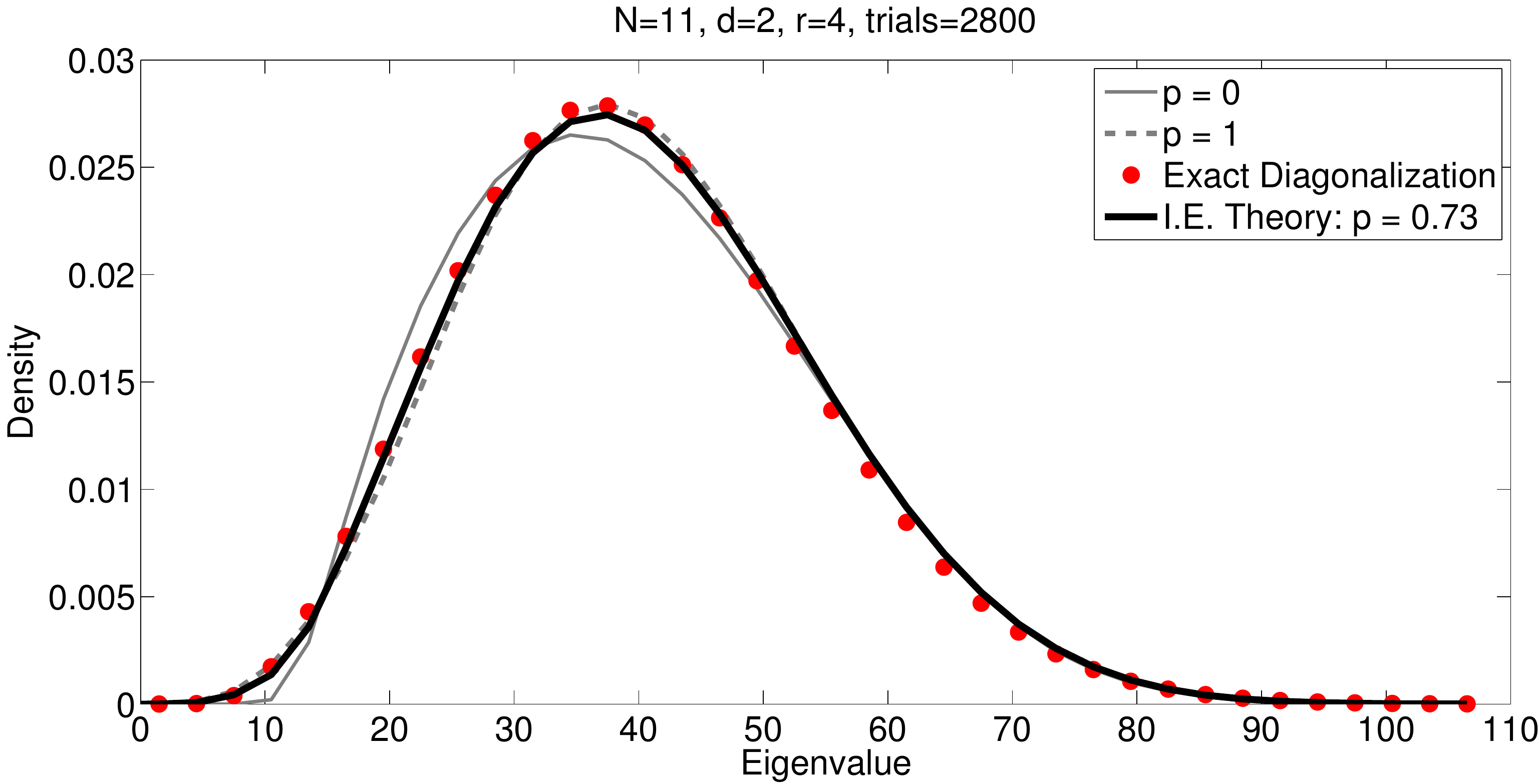}\caption{\label{fig:N=00003D11}$N=11$}
\end{figure}

\end{onehalfspace}

\begin{onehalfspace}

\section{\label{sec:Other-Examples}Other Examples of Local Terms}
\end{onehalfspace}

\begin{onehalfspace}
Because of the Universality lemma, $p$ is independent of the type
of local distribution. Furthermore, as discussed above, the application
of the theory for other types of local terms is entirely similar to
the Wishart case. Therefore, we only show the results in this section.
As a second example consider GOE's as local terms, i.e., $H_{l,l+1}=\frac{G^{T}+G}{2}$,
where $G$ is a full rank matrix whose elements are real Gaussian
random numbers. 

\begin{figure}[H]
\begin{centering}
\includegraphics[scale=0.4]{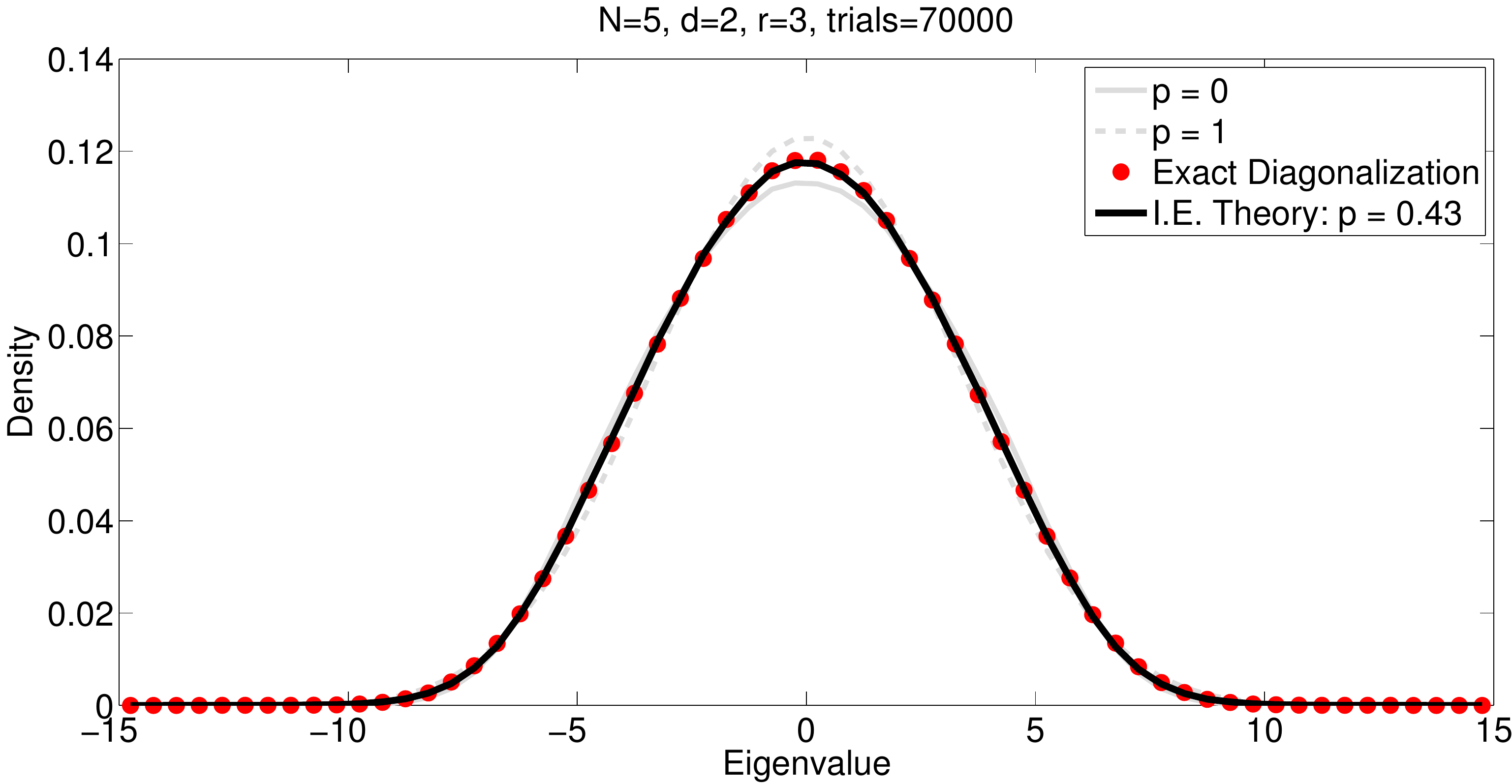}
\par\end{centering}

\begin{centering}
\includegraphics[scale=0.4]{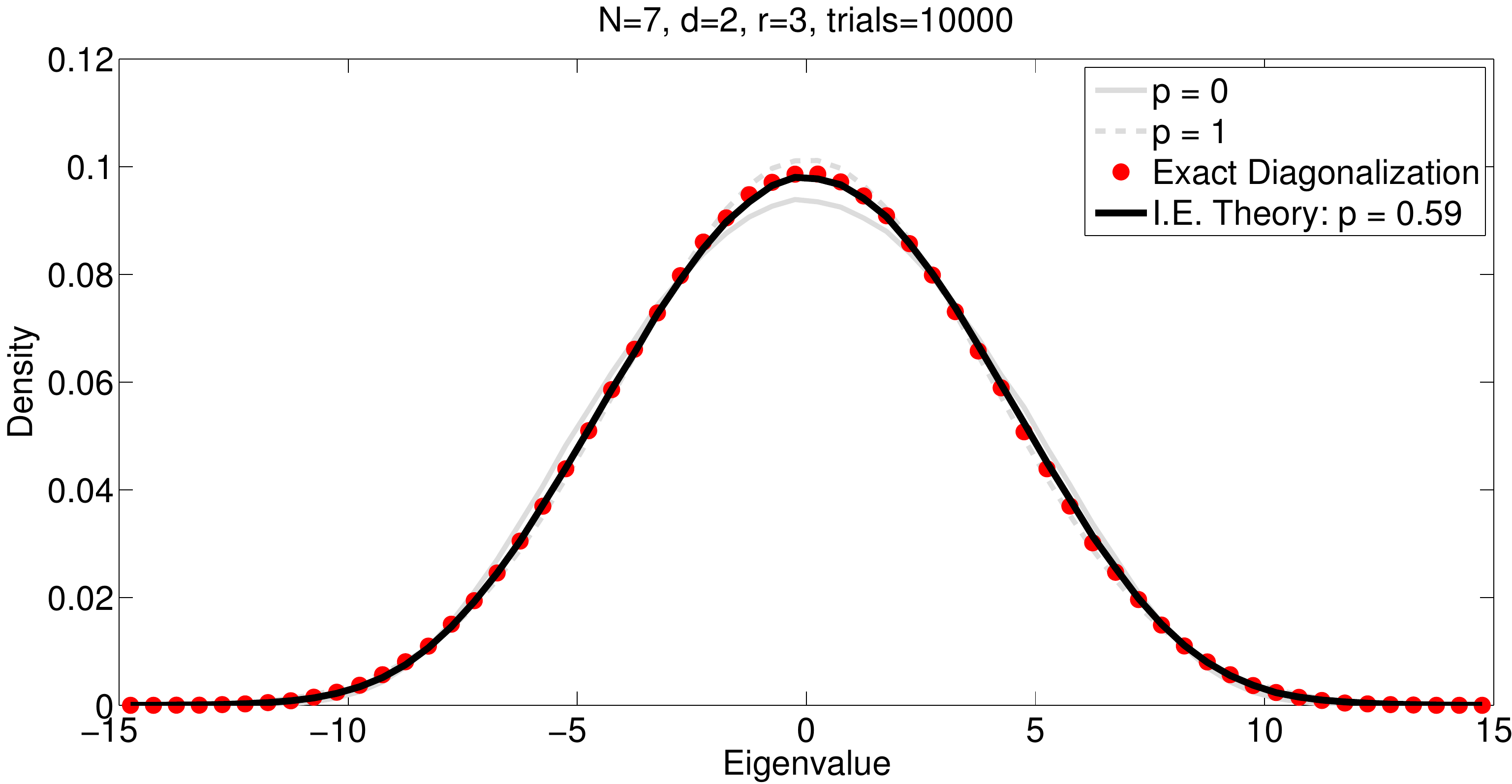}
\par\end{centering}

\begin{centering}
\includegraphics[scale=0.4]{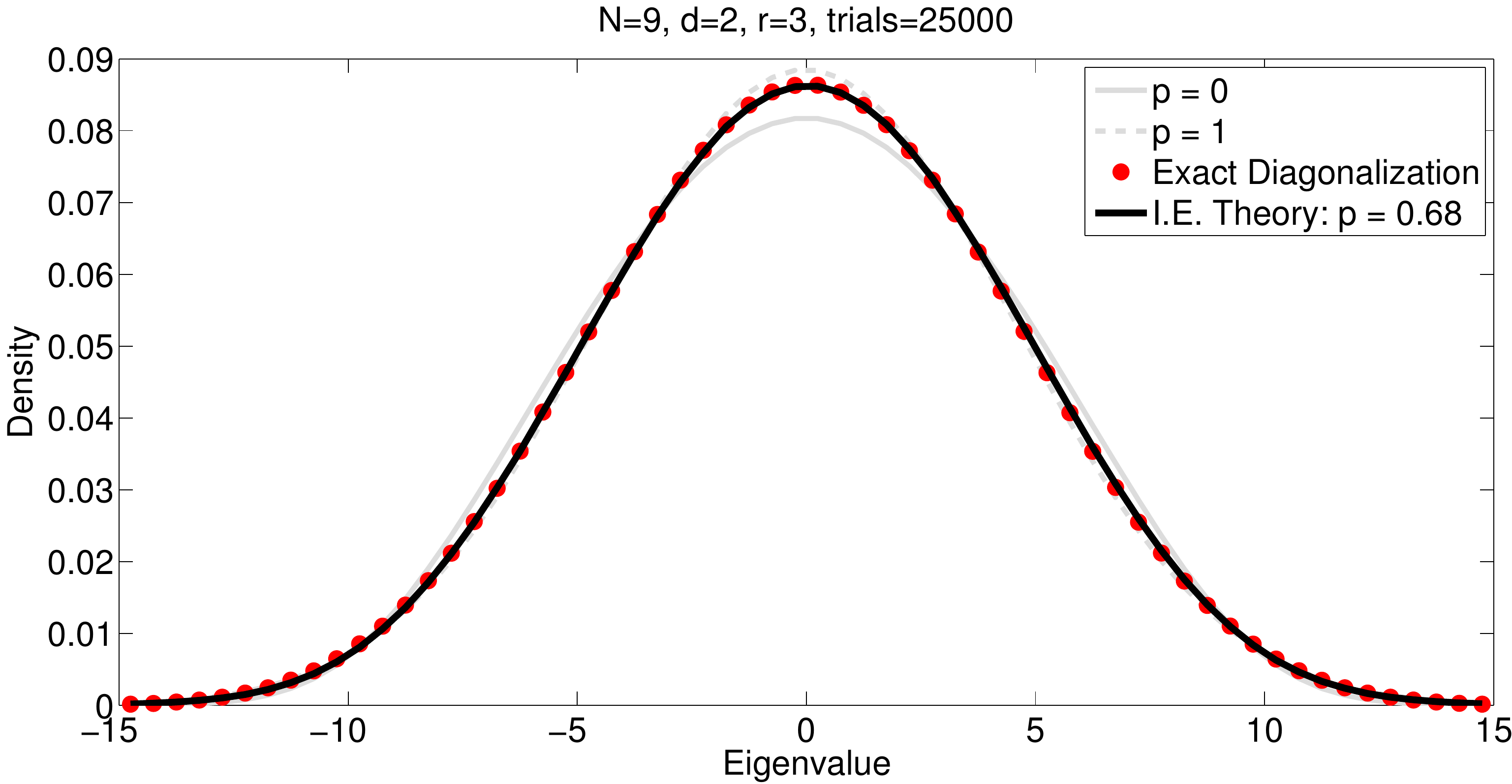}
\par\end{centering}

\caption{\label{fig:GOE's}GOE's as local terms}
\end{figure}

Lastly take the local terms to have Haar eigenvectors but with random
eigenvalues $\pm1$, i.e., $H_{l,l+1}=Q_{l}^{T}\Lambda_{l}Q_{l}$,
where $\Lambda_{l}$ is a diagonal matrix whose elements are binary
random variables $\pm1$ (Figure \ref{fig:binomial}).

\begin{figure}
\begin{centering}
\includegraphics[scale=0.4]{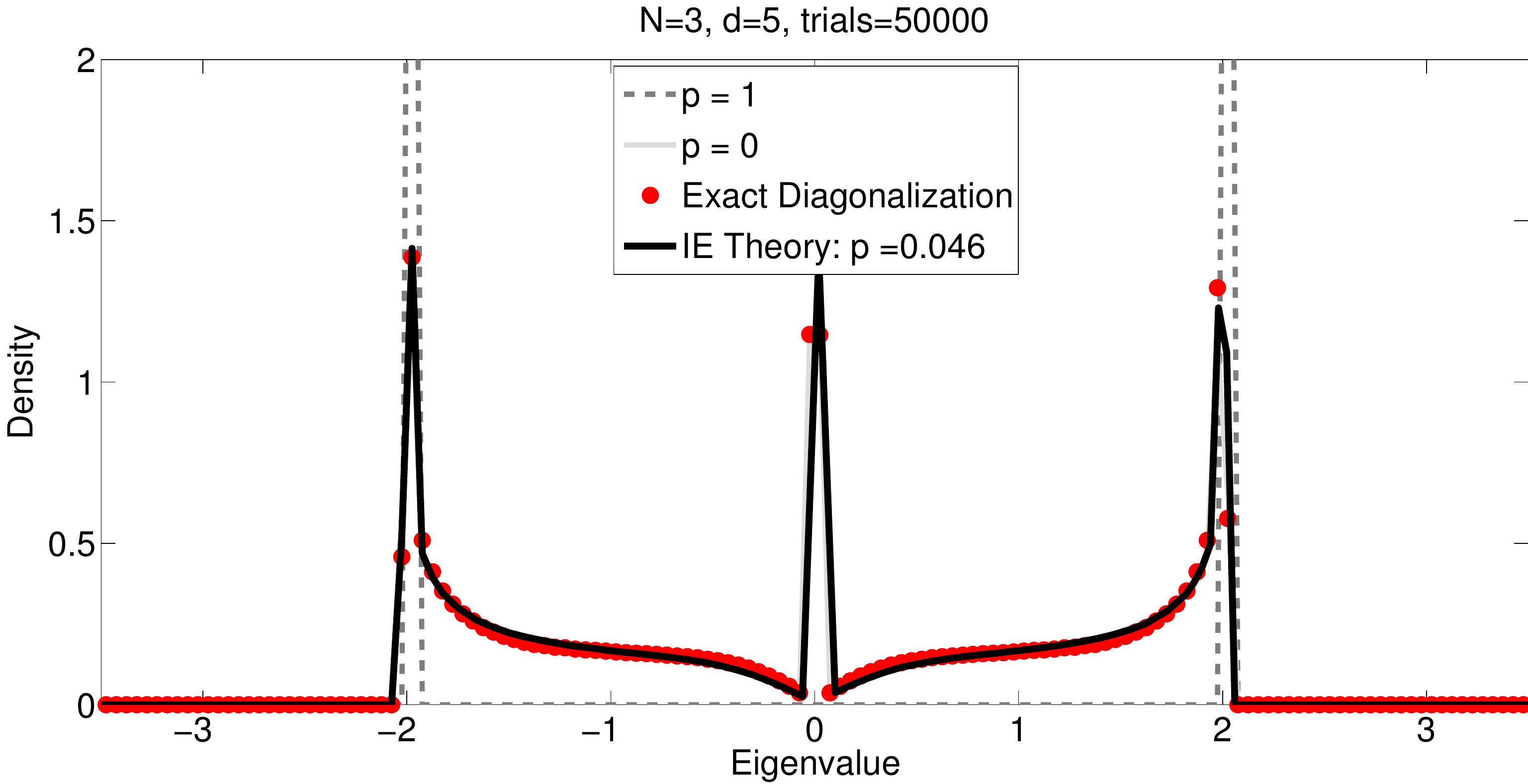}
\par\end{centering}

\caption{\label{fig:binomial}Local terms have a random binomial distribution.}
\end{figure}

In this case the classical treatment of the local terms leads to a
binomial distribution. As expected $p=1$ in Figure \ref{fig:binomial}
has three atoms at $-2,0,2$ corresponding to the randomized sum of
the eigenvalues from the two local terms. The exact diagonalization,
however, shows that the quantum chain has a much richer structure
closer to iso; i.e, $p=0$. This is captured quite well by IE with
$p=0.046$. 
\end{onehalfspace}

\begin{onehalfspace}

\section{\label{sub:Wishart-Beyond-Nearest-Neighbors}Beyond Nearest Neighbors
Interaction: $L>2$}
\end{onehalfspace}

\begin{onehalfspace}
If one fixes all the parameters in the problem and compares $L>2$
with nearest neighbor interactions, then one expects the former to
act more isotropic as the number of random parameters in Eq. \ref{eq:Hamiltonian}
are more. When the number of random parameters introduced by the local
terms, i.e., $\left(N-L+1\right)d^{L}$ and $d^{N}$ are comparable,
we find that we can approximate the spectrum with a high accuracy
by taking the summands to be all isotropic\cite{ramisPI} (See Figures
\ref{fig:L=00003D3}-\ref{fig:L=00003D5}). 

\begin{figure}[H]
\begin{centering}
\includegraphics[scale=0.4]{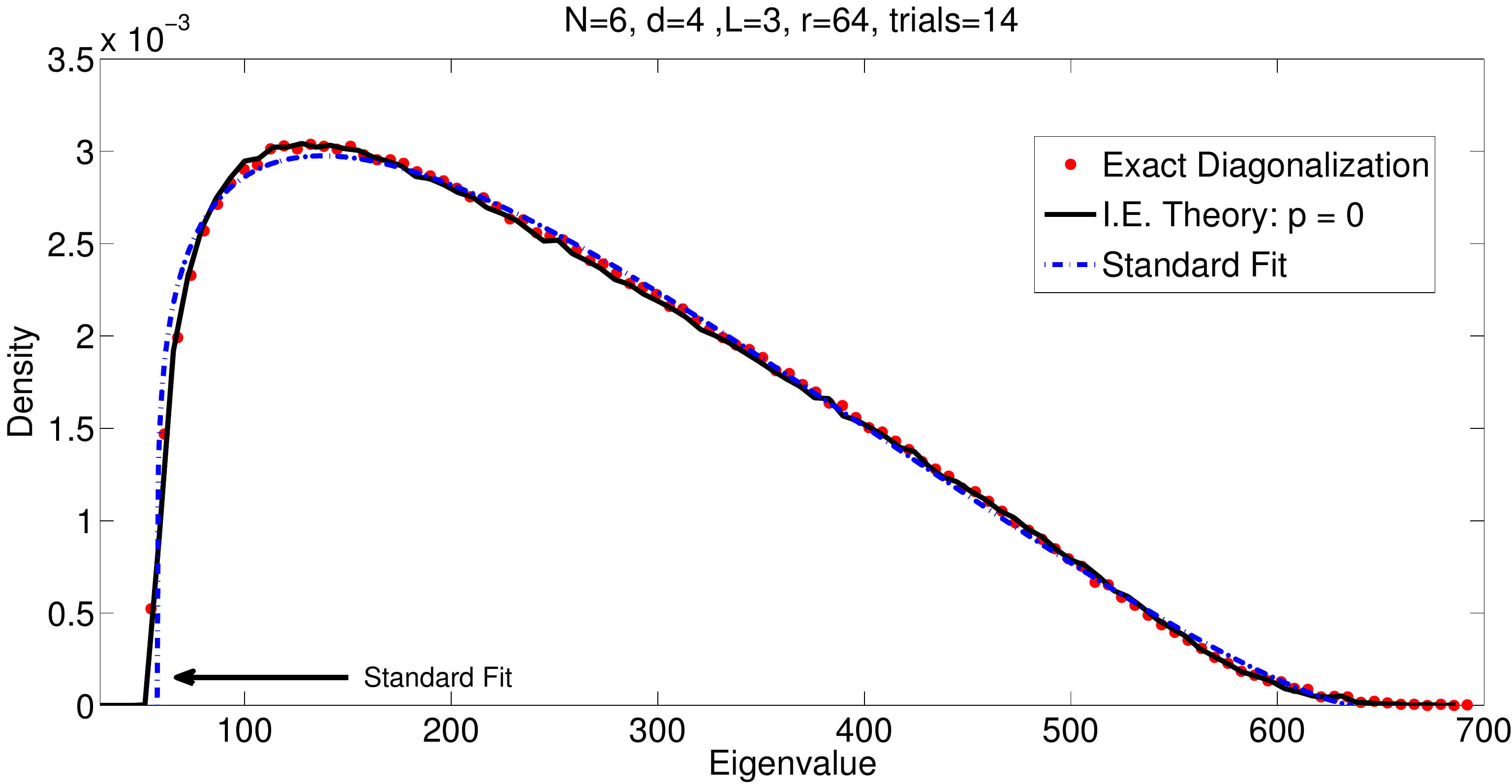}
\par\end{centering}

\caption{\label{fig:L=00003D3}IE method approximates the quantum spectrum
by $H^{IE}=\sum_{l=1}^{4}Q_{l}^{T}H_{l,\cdots,l+2}Q_{l}$}
\end{figure}

\begin{figure}[H]
\begin{centering}
\includegraphics[scale=0.4]{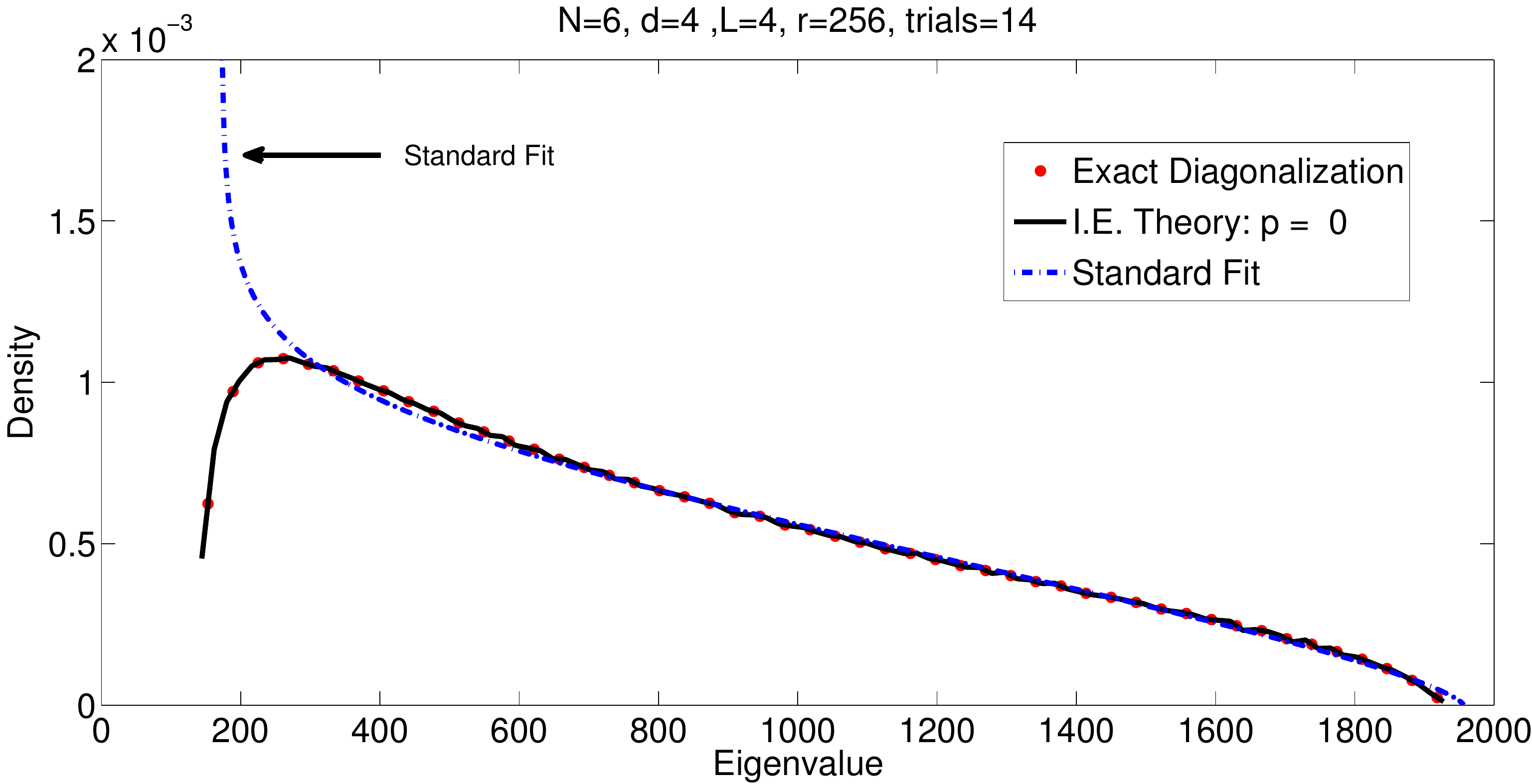}
\par\end{centering}

\caption{\label{fig:L=00003D4}IE method approximates the quantum spectrum
by $H^{IE}=\sum_{l=1}^{3}Q_{l}^{T}H_{l,\cdots,l+3}Q_{l}$. }
\end{figure}

\begin{figure}[H]
\begin{centering}
\includegraphics[scale=0.4]{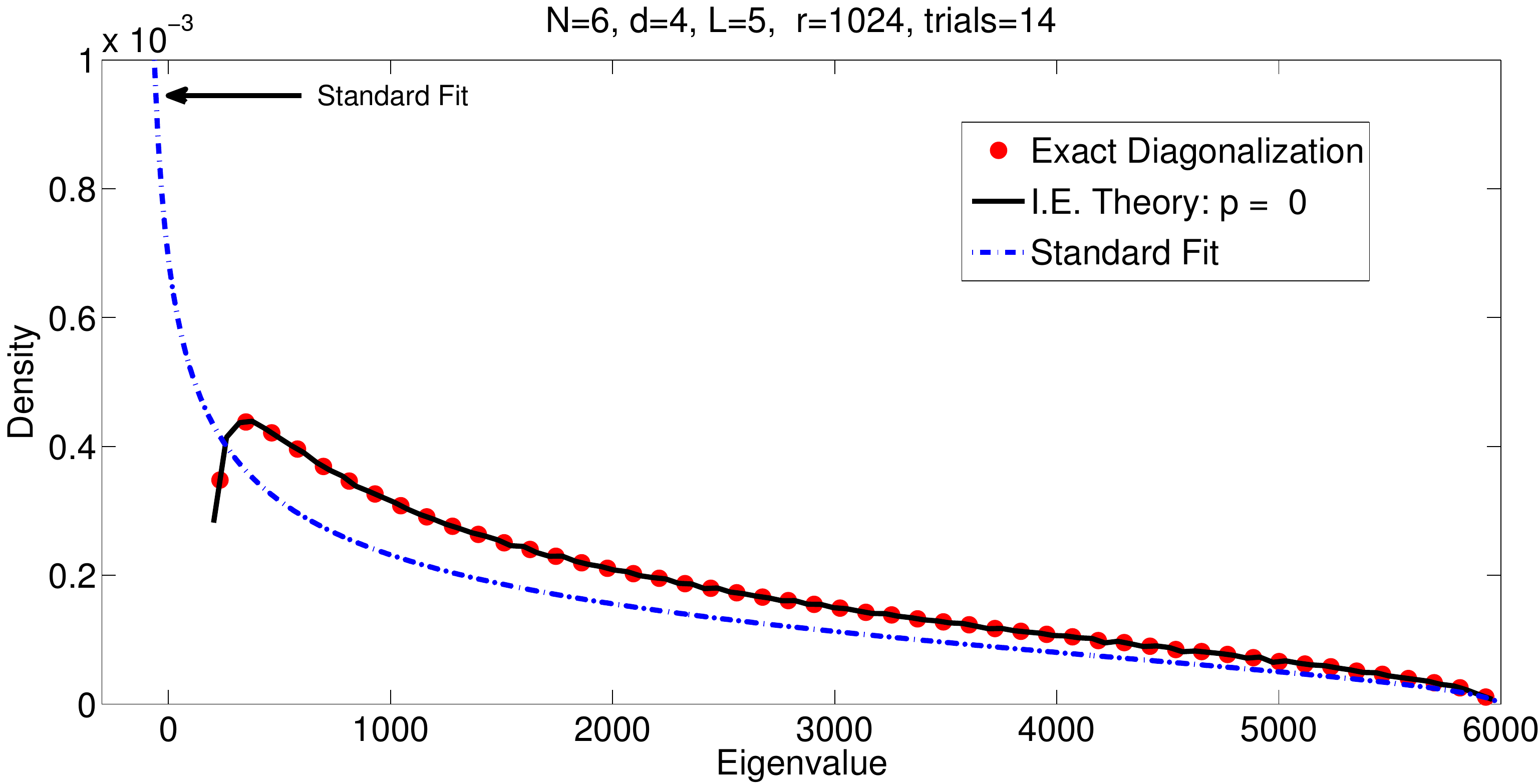}
\par\end{centering}

\caption{\label{fig:L=00003D5}IE method approximates the quantum spectrum
by $H^{IE}=\sum_{l=1}^{2}Q_{l}^{T}H_{l,\cdots,l+4}Q_{l}$}
\end{figure}

Most distributions built solely from the first four moments, would
give smooth curves. Roughly speaking, the mean indicates the center
of the distribution, variance its width, skewness its bending away
from the center and kurtosis how tall and skinny versus how short
and fat the distribution is. Therefore, it is hard to imagine that
the kinks, cusps and local extrema of the quantum problem (as seen
in some of our examples and in particular Figure in \ref{fig:binomial})
could be captured by fitting only the first four moments of the QMBS
Hamiltonian to a known distribution. It is remarkable that a one parameter
(i.e., $p$) interpolation between the isotropic and classical suffices
in capturing the richness of the spectra of QMBS.
\end{onehalfspace}

\begin{onehalfspace}

\section{Conjectures and Open Problems}
\end{onehalfspace}

\begin{onehalfspace}
In this paper we have offered a method that successfully captures
the density of states of QMBS with generic local interaction with
an accuracy higher than one expects solely from the first four moments.
We would like to direct the reader's attention to open problems that
we believe are within reach.
\end{onehalfspace}
\begin{enumerate}
\begin{onehalfspace}
\item We conjecture that the higher moments may be analyzed for their significance.
For example, one can show that the fraction of departing terms in
the expansion of the higher moments (e.g. analogous to bold faced
and underlined terms in Eqs. \ref{eq:fourthMoments},\ref{eq:fifthmoments}
but for higher moments) is asymptotically upper bounded by $1/N^{3}$.
In Section\ref{sub:More-Than-Four} we conjectured that their expectation
values would not change the moments significantly. It would be of
interest to know if 
\end{onehalfspace}

\begin{onehalfspace}
\begin{eqnarray*}
\mathbb{E}\textrm{Tr}\left\{ \ldots Q^{-1}B^{\ge1}QA^{\ge1}Q^{-1}B^{\ge1}Q\ldots\right\}  & \le\\
\mathbb{E}\textrm{Tr}\left\{ \ldots Q_{q}^{-1}B^{\ge1}Q_{q}A^{\ge1}Q_{q}^{-1}B^{\ge1}Q_{q}\ldots\right\}  & \le\\
\mathbb{E}\textrm{Tr}\left\{ \ldots\Pi^{-1}B^{\ge1}\Pi A^{\ge1}\Pi^{-1}B^{\ge1}\Pi\ldots\right\}  & .
\end{eqnarray*}
For example, one wonders if

\[
\mathbb{E}\textrm{Tr}\left\{ \left(AQ^{-1}BQ\right)^{k}\right\} \le\mathbb{E}\textrm{Tr}\left\{ \left(AQ_{q}^{-1}BQ_{q}\right)^{k}\right\} \le\mathbb{E}\textrm{Tr}\left\{ \left(A\Pi^{-1}B\Pi\right)^{k}\right\} 
\]
for $k>2$; we have proved that the inequality becomes an equality
for $k=1$ (Departure Theorem) and holds for $k=2$ (Slider Theorem). 
\end{onehalfspace}

\begin{onehalfspace}
\item Though we focus on a decomposition for s, we believe that the main
theorems may generalize to higher dimensional graphs. Further rigorous
and numerical work in higher dimensions would be of interest. 
\item At the end of Section \ref{sub:Isotropic-theory} we proposed that
more general local terms might be treated by explicitly including
the extra terms (Type I terms). 
\item Application of this method to slightly disordered systems would be
of interest in CMP. In this case, the assumption of fully random local
terms needs to be relaxed. 
\item In our numerical work, we see that the method gives accurate answers
in the presence of an external field. It would be nice to formally
extend the results and calculate thermodynamical quantities.
\item We derived our results for general $\beta$ but numerically tested
$\beta=1,2$. We acknowledge that general $\beta$ remains an abstraction. 
\item Readers may wonder whether it is better to consider ``iso'' or the
infinite limit which is ``free''. We have not fully investigated
these choices, and it is indeed possible that one or the other is
better suited for various purposes.
\item A grander goal would be to apply the ideas of this paper to very general
sums of matrices. \end{onehalfspace}

\end{enumerate}

\section{Appendix}

\begin{onehalfspace}
To help the reader with the random quantities that appear in this
paper, we provide explanations of the exact nature of the random variables
that are being averaged. A common assumption is that we either assume
a uniformly randomly chosen eigenvalue from a random matrix or we
assume a collection of eigenvalues that may be randomly ordered, Random
ordering can be imposed or a direct result of the eigenvector matrix
having the right property. Calculating each of the terms separately
and then subtracting gives the same results.

\begin{equation}
\frac{1}{m}\mathbb{E}\left[\textrm{Tr}\left(AQ^{T}BQ\right)^{2}\right]=\frac{1}{m}\mathbb{E}\left\{ \sum_{1\leq i_{1},i_{2},j_{1},j_{2}\leq m}a_{i_{1}}a_{i_{2}}b_{j_{1}}b_{j_{2}}\left(q_{i_{1}j_{1}}q_{i_{1}j_{2}}q_{i_{2}j_{1}}q_{i_{2}j_{2}}\right)\right\} ,\label{eq:isoStart}
\end{equation}
where $a_{i}$ and $b_{j}$ are elements of matrices $A$ and$B$
respectively. The right hand side of Eq. \ref{eq:isoStart} can have
terms with two collisions (i.e., $i_{1}=i_{2}$ and $j_{1}=j_{2}$),
one collision (i.e. $i_{1}\neq i_{2}$ exclusive-or $j_{1}\neq j_{2}$),
or no collisions (i.e., $i_{1}\neq i_{2}$ and $j_{1}\neq j_{2}$).
Our goal now is to group terms based on the number of collisions.
The pre-factors for two, one and no collisions along with the counts
are summarized in Table \ref{tab:HaarExp}. Using the latter we can
sum the three types of contributions, to get the expectation

\begin{equation}
\begin{array}{c}
\frac{1}{m}\mathbb{E}\left[\textrm{Tr}\left(AQ^{T}BQ\right)^{2}\right]=\frac{\left(\beta+2\right)}{\left(m\beta+2\right)}\mathbb{E}\left(a^{2}\right)\mathbb{E}\left(b^{2}\right)+\\
\frac{\beta\left(m-1\right)}{\left(m\beta+2\right)}\left[\mathbb{E}\left(b^{2}\right)\mathbb{E}\left(a_{1}a_{2}\right)+\mathbb{E}\left(a^{2}\right)\mathbb{E}\left(b_{1}b_{2}\right)\right]-\frac{\beta\left(m-1\right)}{\left(m\beta+2\right)}\mathbb{E}\left(a_{1}a_{2}\right)\mathbb{E}\left(b_{1}b_{2}\right).
\end{array}\label{eq:Isotropic}
\end{equation}

If we take the local terms to be from the same distribution we can
further express the foregoing equation

\begin{equation}
\frac{1}{m}\mathbb{E}\left[\textrm{Tr}\left(AQ^{T}BQ\right)^{2}\right]=\frac{1}{\left(m\beta+2\right)}\left[\left(\beta+2\right)m_{2}^{2}+\beta\left(m-1\right)\mathbb{E}\left(a_{1}a_{2}\right)\left\{ 2m_{2}-\mathbb{E}\left(a_{1}a_{2}\right)\right\} \right].\label{eq:isotropicGeneral}
\end{equation}

The quantity of interest is the difference of the classical and the
iso (see Eq. \ref{eq:denom}),

\begin{equation}
\begin{array}{c}
\frac{1}{m}\mbox{\ensuremath{\mathbb{E}}Tr}\left(A\Pi^{T}B\Pi\right)^{2}-\frac{1}{m}\mbox{\ensuremath{\mathbb{E}}Tr}\left(AQ^{T}BQ\right)^{2}=\\
\frac{\beta\left(m-1\right)}{\left(m\beta+2\right)}\left\{ \mathbb{E}\left(a^{2}\right)\mathbb{E}\left(b^{2}\right)-\mathbb{E}\left(b^{2}\right)\mathbb{E}\left(a_{1}a_{2}\right)-\mathbb{E}\left(a^{2}\right)\mathbb{E}\left(b_{1}b_{2}\right)+\mathbb{E}\left(a_{1}a_{2}\right)\mathbb{E}\left(b_{1}b_{2}\right)\right\} =\\
\frac{\beta\left(m-1\right)}{\left(m\beta+2\right)}\left\{ m_{2}^{(A)}m_{2}^{(B)}-m_{2}^{(B)}m_{1,1}^{(A)}-m_{2}^{(A)}m_{1,1}^{(B)}+m_{1,1}^{(A)}m_{1,1}^{(B)}\right\} =\\
\frac{\beta\left(m-1\right)}{\left(m\beta+2\right)}\left(m_{2}^{(A)}-m_{1,1}^{(A)}\right)\left(m_{2}^{(B)}-m_{1,1}^{(B)}\right)=
\end{array}\label{eq:denom_General}
\end{equation}

If we assume that the local terms have the same distribution: $m_{2}\equiv m_{2}^{(A)}=m_{2}^{\left(B\right)}$,
$m_{11}\equiv m_{1,1}^{(A)}=m_{1,1}^{\left(B\right)}$ as in Eq. \ref{eq:isotropicGeneral},
the foregoing equation simplifies to

\[
\frac{1}{m}\mbox{\ensuremath{\mathbb{E}}Tr}\left(A\Pi^{T}B\Pi\right)^{2}-\frac{1}{m}\mbox{\ensuremath{\mathbb{E}}Tr}\left(AQ^{T}BQ\right)^{2}=\frac{\beta\left(m-1\right)}{\left(m\beta+2\right)}\left(m_{2}-m_{1,1}\right)^{2}.
\]

In the example of Wishart matrices as local terms we have

\begin{equation}
m_{1,1}\equiv\begin{array}{c}
m_{2}\equiv\mathbb{E}\left(a^{2}\right)=rk\left(rk+n+1\right)\\
\mathbb{E}\left(a_{1}a_{2}\right)=k\left(k-1\right)r^{2}+\frac{kr}{m-1}\left\{ \left(tn^{k-1}-1\right)(n+r+1)+tn^{k-1}\left(n-1\right)(r-1)\right\} \\
=k\left(k-1\right)r^{2}+\frac{kr}{m-1}\left\{ tn^{k-1}(nr+2)-n-r-1\right\} .
\end{array}\label{eq:elem}
\end{equation}
\end{onehalfspace}
%
%
%

\chapter{Calculating the Density of States in Disordered Systems Using Free Probability}

In the previous chapter we saw that when a decomposition into two
commuting subsets is possible, one can with a high accuracy obtian
the density of states of quantum spin chains. Spin chains that we
discussed do not possess transport properties such as hopping of electrons.
In this chapter we like to extent the ideas of decomposing the Hamiltonian
to 'easier' pieces and treat different systems. In particular we will
focus on one particle hopping random Schrödinger operator. What follow
in the rest of this chapter also apprears in \cite{RamisAnderson}.\vspace{0.5cm}

\section{Introduction}

Disordered materials have long been of interest for their unique physics
such as localization~\cite{Thouless1974,Evers:2008p706}, anomalous
diffusion~\cite{Bouchaud1990,Shlesinger1993} and ergodicity breaking~\cite{Palmer1982}.
Their properties have been exploited for applications as diverse as
quantum dots~\cite{Barkai2004,Stefani2009}, magnetic nanostructures~\cite{Hernando1999},
disordered metals~\cite{Dyre2000,Dugdale2005}, and bulk heterojunction
photovoltaics~\cite{Peet2009,Difley2010,Yost2011}. Despite this,
theoretical studies have been complicated by the need to calculate
the electronic structure of the respective systems in the presence
of random external potentials. Conventional electronic structure theories
can only be used in conjunction with explicit sampling of thermodynamically
accessible regions of phase space, which make such calculations enormously
more expensive than usual single-point calculations~\cite{Kollman1993}.

Alternatively, we aim to characterize the ensemble of electronic Hamiltonians
that arise from statistical sampling directly using random matrix
theory; this would in principle allow us to sidestep the cost of explicit
statistical sampling. This naturally raises the question of whether
accurate approximations can be made for various characteristics of
random Hamiltonians such as their densities of state (DOSs). We use
techniques from free probability theory, which allow the computation
of eigenvalues of sums of certain matrices~\cite{Voiculescu1991}.
While this has been proposed as a tool applicable to general random
matrices~\cite{Biane1998} and has been used for similar purposes
in quantum chromodynamics~\cite{Zee1996a}, we are not aware of any
quantification of the accuracy of this approximation in practice.
We provide herein a general framework for quantitatively estimating
the error in such situations.

\section{Quantifying the error in approximating a PDF using free probability}

We propose to quantify the deviation between two PDFs using moment
expansions. Such expansions are widely used to describe corrections
to the central limit theorem and deviations from normality, and are
often applied in the form of Gram--Charlier and Edgeworth series~\cite{Stuart1994,Blinnikov1998a}.
Similarly, deviations from non-Gaussian reference PDFs can be quantified
using generalized moment expansions. For two PDFs $w\left(\xi\right)$
and $\tilde{w}\left(\xi\right)$ with finite cumulants $\kappa_{1},\kappa_{2},\dots$
and $\tilde{\kappa}_{1},\tilde{\kappa}_{2},\dots$, and moments $\mu_{1},\mu_{2},\dots$
and $\tilde{\mu}_{1},\tilde{\mu}_{2},\dots$ respectively, we can
define a formal differential operator which transforms $\tilde{w}$
into $w$ and is given by~\cite{Wallace1958,Stuart1994} 
\begin{equation}
w\left(\xi\right)=\exp\left[\sum_{n=1}^{\infty}\frac{\kappa_{n}-\tilde{\kappa}_{n}}{n!}\left(-\frac{d}{d\xi}\right)^{n}\right]\tilde{w}\left(\xi\right).\label{eq:formal-expansion}
\end{equation}
This operator is parameterized completely by the cumulants of both
distributions.

The first $k$ for which the cumulants $\kappa_{k}$ and $\tilde{\kappa}_{k}$
differ then allows us to define a degree to which the approximation
$w\approx\tilde{w}$ is valid. Expanding the exponential and using
the well-known relationship between cumulants and moments allows us
to state that if the first $k-1$ cumulants agree, but the $k$th
cumulants differ, that this is equivalent to specifying that
\begin{equation}
w\left(\xi\right)=\tilde{w}\left(\xi\right)+\frac{\mu_{k}-\tilde{\mu}_{k}}{k!}\left(-1\right)^{k}\tilde{w}^{\left(k\right)}\left(\xi\right)+O\left(\tilde{w}^{\left(k+1\right)}\right).\label{eq:error-term}
\end{equation}

At this point we make no claim on the convergence of the series defined
by the expansion of \eqref{eq:formal-expansion}, but use it as a
justification for calculating the error term defined in \eqref{eq:error-term}.
We will examine this claim later.

\section{The free convolution}

We now take the PDFs to be DOSs of random matrices. For a random matrix
$Z$, the DOS is defined in terms of the eigenvalues $\left\{ \lambda_{n}^{\left(m\right)}\right\} $
of the $M$ samples $Z_{1},\ldots,Z_{m},\dots,Z_{M}$ according to
\begin{equation}
\rho^{\left(Z\right)}\left(\xi\right)=\lim_{M\rightarrow\infty}\frac{1}{M}\sum_{m=1}^{M}\frac{1}{N}\sum_{n=1}^{N}\delta\left(\xi-\lambda_{n}^{\left(m\right)}\right).
\end{equation}
The central idea to our approximation scheme is to split the Hamiltonian
$H=A+B$ into two matrices $A$ and $B$ whose DOSs, $\rho^{\left(A\right)}$
and $\rho^{\left(B\right)}$ respectively, can be determined easily.
The eigenvalues of the sum is in general not the sum of the eigenvalues.
Instead, we propose to approximate the exact DOS with the free convolution
$A\boxplus B$, i.e.~$\rho^{\left(H\right)}\approx\rho^{\left(A\boxplus B\right)}$,
a particular kind of ``sum'' which can be calculated without exact
diagonalization of $H$. The moment expansion presented above quantifies
the error of this approximation in terms of the onset of discrepancies
between the $k$th moment of the exact DOS, $\mu_{k}^{\left(H\right)}$,
and that for the free approximant $\mu_{k}^{\left(A\boxplus B\right)}$.
By definition, the exact moments are 
\begin{equation}
\mu_{k}^{\left(H\right)}=\mu_{k}^{\left(A+B\right)}=\left\langle \left(A+B\right)^{k}\right\rangle ,
\end{equation}
where $\left\langle Z\right\rangle =\mathbb{E}\mbox{ Tr }\left(Z\right)/N$
denotes the normalized expected trace (NET) of the $N\times N$ matrix
$Z$. The $k$th moment can be expanded using the (noncommutative)
binomial expansion of $\left(A+B\right)^{k}$; each resulting term
will have the form of a joint moment $\left\langle A^{n_{1}}B^{m_{1}}\cdots A^{n_{r}}B^{n_{r}}\right\rangle $
with each exponent $n_{s},m_{s}$ being a positive integer such that
$\sum_{s=1}^{r}\left(n_{s}+m_{s}\right)=k$. The free convolution
$\tilde{\mu}_{k}$ is defined similarly, except that $A$ and $B$
are assumed to be freely independent, and therefore that each term
must obey, by definition~\cite{Nica2006a}, relations of the form

\begin{subequations}
\begin{align}
0 & =\left\langle \Pi_{s=1}^{r}\left(A^{n_{s}}-\left\langle A^{n_{s}}\right\rangle \right)\left(B^{m_{s}}-\left\langle B^{m_{s}}\right\rangle \right)\right\rangle \label{eq:centered-joint-moment}\\
 & =\left\langle \Pi_{s=1}^{r}A^{n_{s}}B^{m_{s}}\right\rangle +\mbox{lower order terms},\label{eq:free-indep-matrices}
\end{align}
\end{subequations}where the degree $k$ is the sum of exponents $n_{s}$,
$m_{s}$ and the second equality is formed by expanding the first
line using linearity of the NET\@. Note for $k\le3$ that this is
identical to the statement of (classical) independence~\cite{Nica2006a}.
Testing for $\mu_{k}^{\left(A+B\right)}\ne\mu_{k}^{\left(A\boxplus B\right)}$
then reduces to testing whether each centered joint moment of the
form in \eqref{eq:centered-joint-moment} is statistically nonzero.
The cyclic permutation invariance of the NET means that the enumeration
of all the centered joint moments of degree $k$ is equivalent to
the combinatorial problem of generating all binary necklaces of length
$k$, for which efficient algorithms exist~\cite{Sawada2001}.

The procedure we have described above allows us to ascribe a degree
$k$ to the approximation $\rho^{\left(H\right)}\approx\rho^{\left(A\boxplus B\right)}$
given the splitting $H=A+B$. For each positive integer $n$, we generate
all unique centered joint moments of degree $n$, and test if they
are statistically nonzero. The lowest such $n$ for which there exists
at least one such term is the degree of approximation $k$. We expect
that $k\ge4$ in most situations, as the first three moments of the
exact and free PDFs match under very general conditions~\cite{Movassagh2010}.
However, we have found examples, as described in the next section,
where it is possible to do considerably better than degree 4.

\section{Decomposition of the Anderson Hamiltonian}

As an illustration of the general method, we focus on Hamiltonians
of the form
\begin{equation}
H=\left(\begin{array}{cccc}
h_{1} & J\\
J & h_{2} & \ddots\\
 & \ddots & \ddots & J\\
 &  & J & h_{N}
\end{array}\right),\label{eq:anderson}
\end{equation}
where $J$ is constant and the diagonal elements $h_{i}$ are identically
and independently distributed (iid) random variables with probability
density function (PDF) $p_{h}\left(\xi\right)$. This is a real, symmetric
tridiagonal matrix with circulant (periodic) boundary conditions on
a one-dimensional chain. Unless otherwise stated, we assume herein
that $h_{i}$ are normally distributed with mean $0$ and variance
$\sigma^{2}$. We note that $\sigma/J$ gives us a dimensionless order
parameter to quantify the strength of disorder. 

So far, we have made no restrictions on the decomposition scheme $H=A+B$
other than $\rho^{\left(A\right)}$ and $\rho^{\left(B\right)}$ being
easily computable. A natural question to pose is whether certain choices
of decompositions are intrinsically superior to others. For the Anderson
Hamiltonian, we consider two reasonable partitioning schemes:\begin{subequations}
\begin{equation}
H=A_{1}+B_{1}=\left(\begin{array}{cccc}
h_{1}\\
 & h_{2}\\
 &  & h_{3}\\
 &  &  & \ddots
\end{array}\right)+\left(\begin{array}{cccc}
0 & J\\
J & 0 & J\\
 & J & 0 & \ddots\\
 &  & \ddots & \ddots
\end{array}\right)\label{eq:scheme-i}
\end{equation}
\begin{equation}
H=A_{2}+B_{2}=\left(\begin{array}{ccccc}
h_{1} & J\\
J & 0\\
 &  & h_{3} & J\\
 &  & J & 0\\
 &  &  &  & \ddots
\end{array}\right)+\left(\begin{array}{ccccc}
0\\
 & h_{2} & J\\
 & J & 0\\
 &  &  & h_{4} & \cdots\\
 &  &  & \vdots & \ddots
\end{array}\right).\label{eq:scheme-ii}
\end{equation}
\end{subequations}We refer to these as Scheme~I and II respectively.
For both schemes, each fragment matrix on the right hand side has
a DOS that is easy to determine. In Scheme~I, we have $\rho_{A_{1}}=p_{h}$.
$B_{1}$ is simply $J$ multiplied by the adjacency matrix of a one-dimensional
chain, and therefore has eigenvalues $\lambda_{n}=2J\cos\left(2n\pi/N\right)$~\cite{Strang1999}.
Then the DOS of $B_{1}$ is $\rho_{B_{1}}\left(\xi\right)=\sum_{n=1}^{N}\delta\left(\xi-\lambda_{n}\right)$
which converges as $N\rightarrow\infty$ to the arcsine distribution
with PDF $p_{AS}\left(\xi\right)=1/\left(\pi\sqrt{4J^{2}-\xi^{2}}\right)$
on the interval $\left[-2\left|J\right|,2\left|J\right|\right]$.
In Scheme~II, we have that $\rho_{A_{2}}=\rho_{B_{2}}=\rho_{X}$
where $\rho_{X}$ is the DOS of $X=\left(\begin{array}{cc}
h_{1} & J\\
J & 0
\end{array}\right)$. Since $X$ has eigenvalues $\epsilon_{\pm}\left(\xi\right)=h_{1}\left(\xi\right)/2\pm\sqrt{h_{1}^{2}\left(\xi\right)/4+J^{2}}$,
their distribution can be calculated to be 
\begin{align}
\rho_{X}\left(\xi\right) & =\left(1+\frac{J^{2}}{\xi^{2}}\right)p_{h}\left(\xi-\frac{J^{2}}{\xi}\right).
\end{align}

\section{Numerical free convolution}

We now calculate the free convolution $A\boxplus B$ numerically by
sampling the distributions of $A$ and $B$. We define $\rho^{\left(A\boxplus B\right)}$
as simply the average DOS of the free approximant $Z=A+Q^{-1}BQ$,
where $Q$ is a $N\times N$ random matrix of Haar measure. For real
symmetric Hamiltonians it is sufficient to consider orthogonal matrices
$Q$, which can be generated from the $QR$ decomposition of a Gaussian
orthogonal matrix~\cite{Diaconis2005}. (This can be generalized
readily to unitary and symplectic matrices for complex and quaternionic
Hamiltonians respectively.) In the $N\rightarrow\infty$ limit, this
converges to the free convolution $A\mbox{\ensuremath{\boxplus}}B$~\cite{Voiculescu1991,Voiculescu1994}.

The exact DOS $\rho^{\left(A+B\right)}$ and free approximant $\rho^{\left(A\boxplus B\right)}$
are plotted in Figure~\ref{fig:dos}(a)--(c) for both schemes for
low, moderate and high noise regimes ($\sigma/J=$0.1, 1, 10 respectively).

\begin{figure}
\begin{centering}
\includegraphics[width=12.6cm]{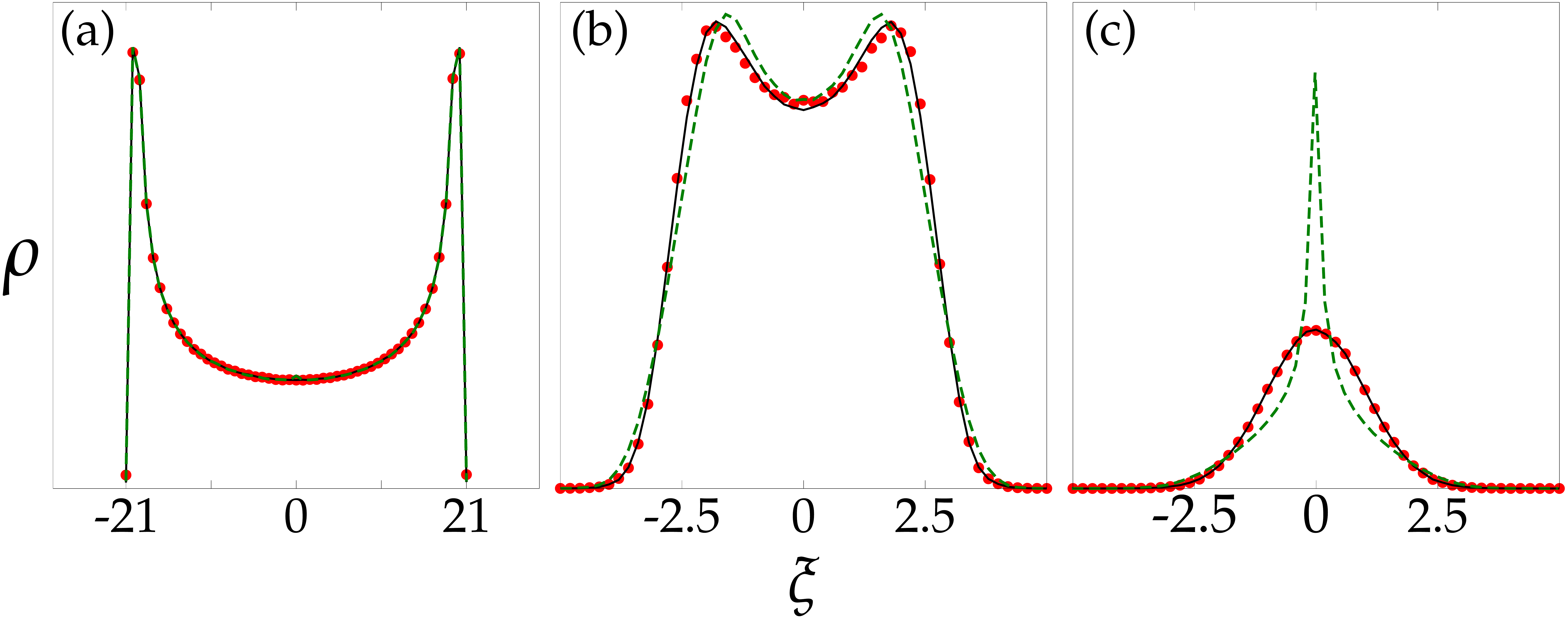}
\par\end{centering}

\caption{\label{fig:dos}Calculation of the DOS, $\rho(\xi)$, of the Hamiltonian
$H$ of (\ref{eq:anderson}) with $M=5000$ samples of $2000\times2000$
matrices for (a) low, (b) moderate and (c) high noise ($\sigma/J$=0.1,
1 and 10 respectively with $\sigma=1$). For each figure we show the
results of free convolution defined in Scheme I ($\rho^{\left(A_{1}\boxplus B_{1}\right)}$;
black solid line), Scheme II ($\rho^{\left(A_{2}\boxplus B_{2}\right)}$;
green dashed line) and exact diagonalization ($\rho^{\left(H\right)}$;
red dotted line).}
\end{figure}
We observe that for Scheme~I we have excellent agreement between
$\rho^{\left(H\right)}$ and $\rho^{\left(A_{1}\boxplus B_{1}\right)}$
across all values of $\sigma/J$, which is evident from visual inspection;
in contrast, Scheme~II shows variable quality of fit.

We can understand the starkly different behaviors of the two partitioning
schemes using the procedure outlined above to analyze the accuracy
of the approximations $\rho^{\left(H\right)}\approx\rho^{\left(A_{1}\boxplus B_{1}\right)}$
and $\rho^{\left(H\right)}\approx\rho^{\left(A_{2}\boxplus B_{2}\right)}$.
For Scheme~I, we observe that the approximation \eqref{eq:error-term}
is of degree $k=8$; the discrepancy lies solely in the term $\left\langle \left(A_{1}B_{1}\right)^{4}\right\rangle $~\cite{Popescu2011}.
Free probability expects this term to vanish, since both $A_{1}$
and $B_{1}$ are centered (i.e. $\left\langle A_{1}\right\rangle =\left\langle B_{1}\right\rangle =0$)
and hence must satisfy \eqref{eq:free-indep-matrices} with $n_{1}=m_{1}=\cdots=n_{4}=m_{4}=1$.
In contrast, we can calculate its true value from the definitions
of $A_{1}$ and $B_{1}$. By definition of the NET $\left\langle \cdot\right\rangle $,
only closed paths contribute to the term. Hence, only two types of
terms can contribute to $\left\langle \left(A_{1}B_{1}\right)^{4}\right\rangle $;
these are expressed diagrammatically in Figure~\ref{fig:hopping}.
The matrix $A_{1}$ weights each path by a factor of $h$, while $B_{1}$
weights each path by $J$, and in addition forces the path to hop
to an adjacent site.

\begin{figure}
\begin{centering}
\includegraphics{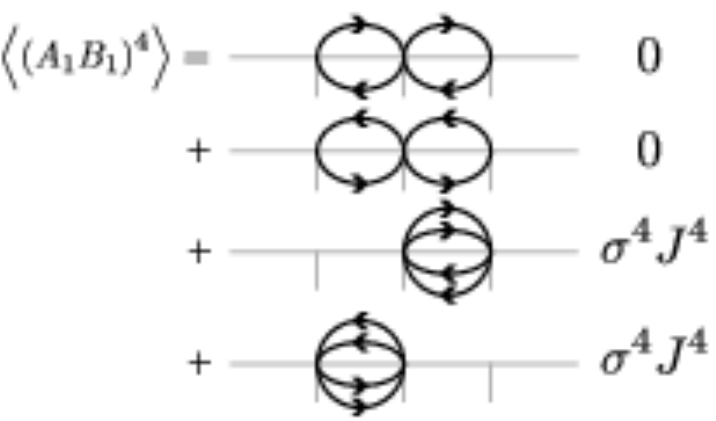}
\par\end{centering}

\caption{\label{fig:hopping}Diagrammatic expansion of the term $\left\langle A_{1}B_{1}A_{1}B_{1}A_{1}B_{1}A_{1}B_{1}\right\rangle $
in terms of allowed paths dictated by the matrix elements of $A_{1}$
and $B_{1}$ of Scheme~I in~\eqref{eq:scheme-i}.}
\end{figure}
 Consequently, we can write explicitly
\begin{align}
\left\langle \left(A_{1}B_{1}\right)^{4}\right\rangle = & \frac{1}{N}\sum_{i}\mathbb{E}\left(h_{i}Jh_{i-1}Jh_{i}Jh_{i+1}J\right)\nonumber \\
 & +\frac{1}{N}\sum_{i}\mathbb{E}\left(h_{i}Jh_{i+1}Jh_{i}Jh_{i-1}J\right)\nonumber \\
 & +\frac{1}{N}\sum_{i}\mathbb{E}\left(h_{i}Jh_{i-1}Jh_{i}Jh_{i-1}J\right)\nonumber \\
 & +\frac{1}{N}\sum_{i}\mathbb{E}\left(h_{i}Jh_{i+1}Jh_{i}Jh_{i+1}J\right)\nonumber \\
= & 2J^{4}\mathbb{E}\left(h_{i}\right)^{2}\mathbb{E}\left(h_{i}^{2}\right)+2J^{4}\mathbb{E}\left(h_{i}^{2}\right)^{2}=0+2J^{4}\sigma^{4},
\end{align}
where the second equality follows from the independence of the $h_{i}$'s.
As this is the only source of discrepancy at the eighth moment, this
explains why the agreement between the free and exact PDFs is so good,
as the leading order correction is in the eighth derivative of $\rho^{\left(A_{1}\boxplus B_{1}\right)}$
with coefficient $2\sigma^{4}J^{4}/8!=\left(\sigma J\right)^{4}/20160$.
In contrast, we observe for Scheme~II that the leading order correction
is at $k=4$, where the discrepancy lies in $\left\langle A_{2}^{2}B_{2}^{2}\right\rangle $.
Free probability expects this to be equal to $\left\langle A_{2}^{2}B_{2}^{2}\right\rangle =\left\langle A_{2}^{2}\right\rangle \left\langle B_{2}^{2}\right\rangle =\left\langle X^{2}\right\rangle ^{2}=\left(J^{2}+\sigma^{2}/2\right)^{2}$,
whereas the exact value of this term is $J^{2}\left(J^{2}+\sigma^{2}\right)$.
Therefore the discrepancy is in the fourth derivative of $\rho^{\left(A\boxplus B\right)}$
with coefficient $\left(-\sigma^{4}/4\right)/4!=-\sigma^{4}/96$.

\section{Analytic free convolution}

Free probability allows us also to calculate the limiting distributions
of $\rho^{\left(A\boxplus B\right)}$ in the macroscopic limit of
infinite matrix sizes $N\rightarrow\infty$ and infinite samples $M\rightarrow\infty$.
In this limit, the DOS $\rho^{\left(A\boxplus B\right)}$ is given
as a particular type of integral convolution of $\rho^{\left(A\right)}$
and $\rho^{\left(B\right)}$. We now calculate the free convolution
analytically in the macroscopic limit for the two partitioning schemes
discussed above, thus sidestepping the cost of sampling and matrix
diagonalization altogether.

For our first example, we take $A$ and $B$ as in Scheme~I, but
with each iid $h_{i}$ following a Wigner semicircle distribution
with PDF $p_{W}\left(\xi\right)=\sqrt{4-\xi^{2}}/4\pi$ on the interval
$\left[-2,2\right]$. (Using semicircular noise instead of Gaussian
noise simplifies the analytic calculation considerably.) Then, $\rho^{\left(A\right)}=p_{W}$
and $\rho^{\left(B\right)}=p_{AS}$. The key tool to performing the
free convolution analytically is the $R$-transform $r\left(w\right)=g^{-1}\left(w\right)-w^{-1}$~\cite{Voiculescu1985},
where $g^{-1}$ is defined implicitly via the Cauchy transform
\begin{equation}
w=\int_{\mathbb{R}}\frac{\rho\left(\xi\right)}{g^{-1}\left(w\right)-\xi}d\xi.
\end{equation}
For freely independent $A$ and $B$, the $R$-transforms linearize
the free convolution, i.e. $r^{\left(A\boxplus B\right)}\left(w\right)=r^{\left(A\right)}\left(w\right)+r^{\left(B\right)}\left(w\right)$,
and that the PDF can be recovered from the Plemelj--Sokhotsky inversion
formula by\begin{subequations} 
\begin{align}
\rho^{\left(A\boxplus B\right)}\left(\xi\right) & =\frac{1}{\pi}\mbox{Im}\left(\left(g^{\left(A\boxplus B\right)}\right)^{-1}\left(\xi\right)\right)\\
g^{\left(A\boxplus B\right)}\left(w\right) & =r^{\left(A\boxplus B\right)}\left(w\right)+w^{-1}.
\end{align}
\end{subequations}Applying this to Scheme~I, we have $r^{\left(A\right)}\left(w\right)=w$
and $r^{\left(B\right)}\left(w\right)=\left(-1+\sqrt{4J^{2}+w^{2}}\right)/w$,
so that $g^{\left(A\boxplus B\right)}\left(w\right)=w+\left(\sqrt{4J^{2}+w^{2}}\right)/w$.
The need to calculate the functional inverse $\left(g^{\left(A\boxplus B\right)}\right)^{-1}$
in this procedure unfortunately precludes our ability to write $\rho^{\left(A\boxplus B\right)}\left(\xi\right)$
in a compact closed form; nevertheless, the inversion can be calculated
numerically. We present calculations of the DOS as a function of noise
strength $\sigma/J$ in Figure~\ref{fig:dos-semicircle}.

\begin{figure}
\begin{centering}
\includegraphics[width=12.6cm]{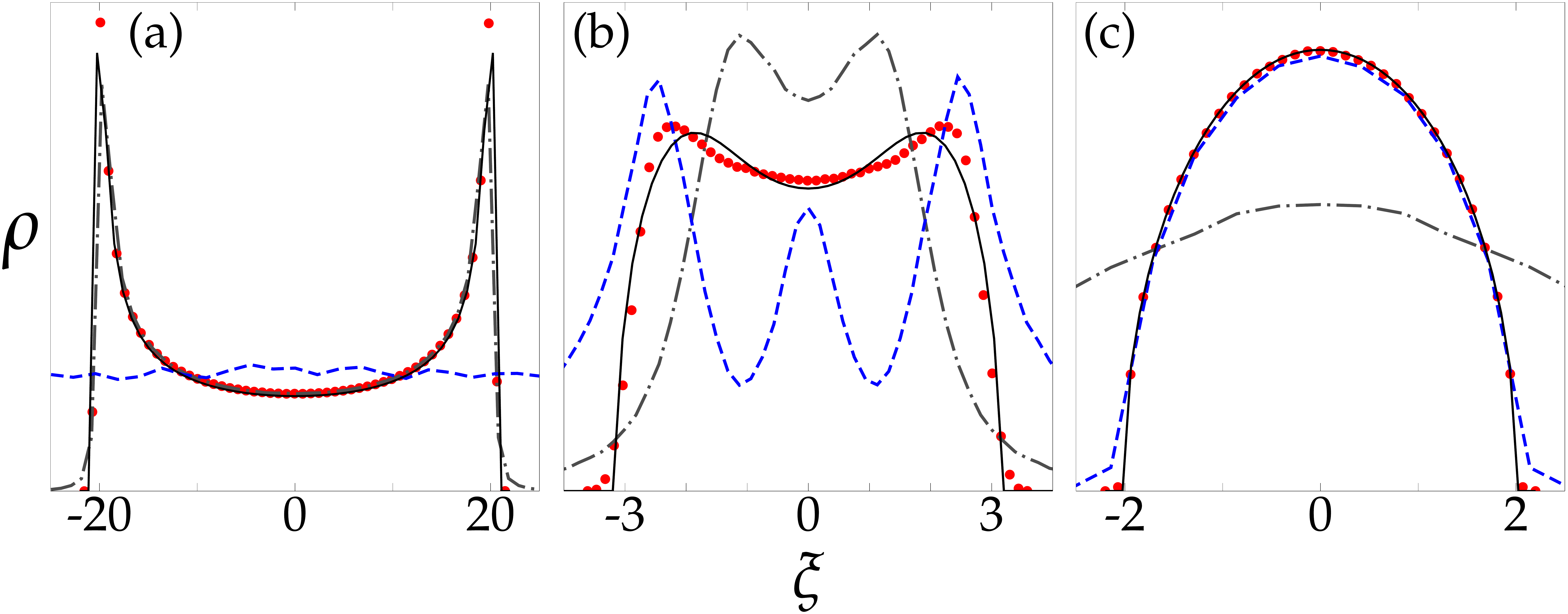}
\par\end{centering}

\caption{\label{fig:dos-semicircle}DOS, $\rho(\xi)$, of the Hamiltonian~\eqref{eq:anderson}
with $M=5000$ samples of $2000\times2000$ matrices with (a) low,
(b) moderate and (c) high semicircular on-site noise ($\sigma/J$=0.1,
1 and 10 respectively with $\sigma=1$), as calculated with exact
diagonalization (red dotted line), free convolution (black solid line),
and perturbation theory with $A_{1}$ as reference (blue dashed line)
and $B_{1}$ as reference (gray dash-dotted line). The partitioning
scheme is Scheme~I of~\eqref{eq:scheme-i}.}
\end{figure}

\section{Comparison with other approximations}

For comparative purposes, we also performed calculations using standard
second-order matrix perturbation theory~\cite{Horn1990} for both
partitioning schemes. The results are also shown in Figure~\ref{fig:dos-semicircle}.
Unsurprisingly, perturbation theory produces results that vary strongly
with $\sigma/J$, and that the different series, based on whether
$A$ is considered a perturbation of $B$ or vice versa, have different
regimes of applicability. Furthermore it is clear even from visual
inspection that the second moment of the DOS calculated using second-order
perturbation theory is in general incorrect. In contrast, the free
convolution produces results with a more uniform level of accuracy
across the entire range of $\sigma/J$, and that we have at least
the first three moments being correct~\cite{Movassagh2010}.

It is also natural to ask what mean-field theory, another standard
tool, would predict. Interestingly, the limiting behavior of Scheme~I
as $N\rightarrow\infty$ is equivalent to a form of mean-field theory
known as the coherent potential approximation (CPA)~\cite{Neu1994,Neu1995a,Neu1995b}
in condensed matter physics, and is equivalent to the Blue's function
formalism in quantum chromodynamics for calculating one-particle irreducible
self-energies~\cite{Zee1996a}. The breakdown in the CPA in the term
$\left\langle \left(A_{1}B_{1}\right)^{4}\right\rangle $ is known~\cite{Blackman1971,Thouless1974};
however, to our knowledge, the magnitude of the deviation was not
explained. In contrast, our error analysis framework affords us such
a quantitative explanation.

Finally, we discuss the predictions of isotropic entanglement theory,
which proposes a linear interpolation between the classical convolution
\[
\rho^{\left(A*B\right)}\left(\xi\right)=\int_{-\infty}^{\infty}\rho^{\left(A\right)}\left(\xi\right)\rho^{\left(B\right)}\left(x-\xi\right)dx
\]
and the free convolution $\rho^{\left(A\boxplus B\right)}\left(\xi\right)$
in the fourth cumulant~\cite{Movassagh2010,Movassagh2011a}. The
classical convolution can be calculated directly from the random matrices
$A$ and $B$; by diagonalizing the matrices as $A=Q_{A}^{-1}\Lambda_{A}Q_{A}$
and $B=Q_{B}^{-1}\Lambda_{B}Q_{B}$, the classical convolution $\rho^{\left(A*B\right)}\left(\xi\right)$
can be computed from the eigenvalues of random matrices of the form
$Z_{cl}=\Lambda_{A}+\Pi^{-1}\Lambda_{B}\Pi$ where $\Pi$ is a $N\times N$
random permutation matrix. It is instructive to compare this with
the free convolution, which can be sampled from matrices of the form
$Z^{\prime}=\Lambda_{A}+Q^{-1}\Lambda_{B}Q$, which can be shown by
orthogonal invariance of the Haar measure random matrices $Q$ to
be equivalent to sampling matrices of the form $Z=A+Q^{-1}BQ$ described
previously.

As discussed previously, the lowest three moments of $Z$ and $H$
are identical; this turns out to be true also for $Z_{cl}$~\cite{Movassagh2010}.
Therefore IE proposes to interpolate via the fourth cumulant, with
interpolation parameter $p$ defined as
\begin{equation}
p=\frac{\kappa_{4}^{\left(H\right)}-\kappa_{4}^{\left(A\boxplus B\right)}}{\kappa_{4}^{\left(A*B\right)}-\kappa_{4}^{\left(A\boxplus B\right)}}
\end{equation}

We observe that IE appears to always favor the free convolution limit
($p=0$) as opposed to the classical limit ($p=1$); this is not surprising
as we know from our previous analysis that $\kappa_{4}^{\left(H\right)}=\kappa_{4}^{\left(A\boxplus B\right)}$,
and that the agreement with the exact diagonalization result is excellent
regardless of $\sigma/J$. In Scheme~II, we observe the unexpected
result that $p$ can sometimes be negative and that the agreement
varies with $\sigma/J$. From the moment expansion we understand why;
we have that the first three moments match while $\kappa_{4}^{\left(A_{2}+B_{2}\right)}-\kappa_{4}^{\left(A_{2}\boxplus B_{2}\right)}=-\sigma^{4}/4$.
The discrepancy lies in the term $\left\langle A_{2}^{2}B_{2}^{2}\right\rangle $
which is expected to have the value $\left\langle A_{2}^{2}\right\rangle \left\langle B_{2}^{2}\right\rangle =\left(J^{2}+\sigma^{2}/2\right)^{2}$
in free probability but instead has the exact value $J^{2}\left(J^{2}+\sigma^{2}\right)$.
Furthermore, we have that $\kappa_{4}^{\left(A_{2}*B_{2}\right)}\ne\kappa_{4}^{\left(A_{2}\boxplus B_{2}\right)}$where
the only discrepancy lies is in the so-called departing term $\left\langle A_{2}B_{2}A_{2}B_{2}\right\rangle $~\cite{Movassagh2010,Movassagh2011a}.
This term contributes 0 to $\kappa_{4}^{\left(A\boxplus B\right)}$
but has value $\left\langle A_{2}^{2}\right\rangle \left\langle B_{2}^{2}\right\rangle =\left(J^{2}+\sigma^{2}/2\right)^{2}$
in $\kappa_{4}^{\left(A_{2}*B_{2}\right)}$, since for the classical
convolution we have that 
\begin{equation}
\left\langle \Pi_{s=1}^{r}\left(A_{2}^{n_{s}}B_{2}^{m_{s}}\right)\right\rangle =\left\langle A_{2}^{\sum_{s=1}^{r}n_{s}}\right\rangle \left\langle B_{2}^{\sum_{s=1}^{r}m_{s}}\right\rangle .
\end{equation}
This therefore explains why we observe a negative $p$, as this calculation
shows that
\begin{equation}
p=\frac{\kappa_{4}^{\left(A_{2}+B_{2}\right)}-\kappa_{4}^{\left(A_{2}\boxplus B_{2}\right)}}{\kappa_{4}^{\left(A_{2}*B_{2}\right)}-\kappa_{4}^{\left(A_{2}\boxplus B_{2}\right)}}=-2\left(2\left(\frac{\sigma}{J}\right)^{-2}+1\right)^{-2}
\end{equation}
which is manifestly negative.

In conclusion, we have demonstrated that the free probability of random
matrices can provide unexpectedly accurate approximations for the
DOS of disordered Hamiltonians, both for finite dimensional systems
and in the macroscopic limit $N\rightarrow\infty$. Our results illustrate
variable accuracies of the approximations predicated on our particular
choices of partitioning schemes, which can be quantitatively estimated
using moment expansions. These results represent an optimistic beginning
to addressing the electronic structure of disordered condensed matter
systems using the tools of random matrix theory. We are currently
investigating the predictions of free probability on the localization
properties of eigenvectors, as well as studying the general applicability
of free probability to lattices in higher dimensions, as well as systems
with off-diagonal and correlated disorder.

\part{Eigenvectors}
%
%
%
%
%
%
%
\chapter{Generic Quantum Spin Chains} 

In this chapter I discuss frustration free condition, evolution in
time and imaginary time within MPS representation along with numerical
methods. I then discuss the degeneracy and frustration free condition
for generic quantum spin chains with local interactions. I leave the
discussion of the entanglement of the ground states for the next chapter.

\section{Frustration Free Condition}

An $L-$local Hamiltonian where each interaction acts nontrivially
on $L$ particles can be written as

\[
H=\sum_{i=1}^{M}H_{i}
\]
where $i$ indexes $M$ groups of $L$ spins. By spectral decomposition
any such Hamiltonian on $N$ spins each with $d-$states can be written
as 

\[
H=m_{g}\lambda_{g}\left|g\left\rangle \right\langle g\right|+\sum_{\alpha=1}^{\alpha_{max}}\lambda_{\alpha}\left|\alpha\left\rangle \right\langle \alpha\right|,\mbox{ with\,\ensuremath{\alpha_{max}+m_{g}=d^{N}}}
\]
where $g$ refers to the ground state quantities, $m_{g}$ the possible
multiplicity due to degeneracy and $\alpha$'s label the excited states.
Similarly for each of the local terms

\begin{equation}
H_{i}=m_{g}^{i}\lambda_{g}^{i}\left|g_{i}\left\rangle \right\langle g_{i}\right|+\sum_{p=1}^{r}\lambda_{p}^{i}\left|v_{i}^{p}\left\rangle \right\langle v_{i}^{p}\right|\mbox{ with \ensuremath{r+m_{g}^{i}=d^{L}}}.\label{eq:LocalGeneral}
\end{equation}

Energy is meaningful in a relative sense, hence one can shift the
local Hamiltonian such that the ground state has energy zero. This
defines an effective Hamiltonian $H_{i}'$

\[
H'_{i}\equiv H_{i}-\lambda_{g}^{i}\mathbb{I}_{d^{L}}=\sum_{p=1}^{r}\left(\lambda_{p}^{i}-\lambda_{g}^{i}\right)\left|v_{i}^{p}\left\rangle \right\langle v_{i}^{p}\right|
\]

The effective local Hamiltonian has zero energy ground states and
eigenstates corresponding to the excited states with positive eigenvalues.
By definition, $\lambda_{p}^{i}-\lambda_{g}^{i}>0$ for all $p$,
where we have not made as of yet any assumptions on the possible numerical
values they can take. The span and Kernel of $H'_{s,t}$ is therefore
the same as 
\begin{equation}
H"_{i}=\sum_{p=1}^{r}\left|v_{i}^{p}\left\rangle \right\langle v_{i}^{p}\right|,\label{eq:sumsProj}
\end{equation}
where we replaced all $\lambda_{p}^{i}-\lambda_{g}^{i}>0$ by ones
to get projectors of rank $r$ as our local terms. As discussed at
the end of this section, there are advantages in doing so. Locally
there are $d^{2}-r$ zero energy states; however, the projectors are
not mutually exclusive. This makes the problem of counting the kernel
of $H$ non-trivial, namely we cannot assert that the number of ground
states is $\left(d^{2}-r\right)^{M}$. In particular the ground state
of the global Hamiltonian may not be a local ground state of some
local term(s): frustrated.
\begin{lem}
Suppose $E_{gr}^{1}$ and $E_{gr}^{2}$ are the smallest eigenvalues
of Hamiltonians $H^{1}$ and $H^{2}$ respectively and $E_{gr}^{1,2}$
that of $H^{1}+H^{2}$. Further, let $E_{max}^{1}$ and $E_{max}^{2}$
be the largest eigenvalues of Hamiltonians $H^{1}$ and $H^{2}$ respectively
and $E_{max}^{1,2}$ that of $H^{1}+H^{2}$ then

\begin{eqnarray*}
E_{gr}^{1}+E_{gr}^{2} & \leq & E_{gr}^{1,2}\\
E_{max}^{1}+E_{max}^{2} & \ge & E_{max}^{1,2}.
\end{eqnarray*}
In the first (second) case equality holds if $H^{1}$ and $H^{2}$
have the same eigenstate for their smallest (largest) eigenvalues.\end{lem}
\begin{proof}
The proof follows from the fact that the average is greater than the
least eigenvalue and smaller than the largest eigenvalue. Suppose
we are in the states in which $H^{1}+H^{2}$ takes the eigenvalue
$E_{gr}^{1,2}$, then $\left\langle H^{1}+H^{2}\right\rangle =E_{gr}^{1,2}$
, but $\left\langle H^{1}+H^{2}\right\rangle =\left\langle H^{1}\right\rangle +\left\langle H^{2}\right\rangle \geq E_{gr}^{1}+E_{gr}^{2}$.
The equality holds if $H^{1}$ and $H^{2}$ have the same eigenstate
for their smallest eigenvalues. Similarly if we are in the state in
which $H^{1}+H^{2}$ takes the eigenvalue $E_{max}^{1,2}$, then $\left\langle H^{1}+H^{2}\right\rangle =E_{max}^{1,2}$
, but $\left\langle H^{1}+H^{2}\right\rangle =\left\langle H^{1}\right\rangle +\left\langle H^{2}\right\rangle \le E_{max}^{1}+E_{max}^{2}$.
The equality holds if $H^{1}$ and $H^{2}$ have the same eigenstate
for their largest eigenvalues.
\end{proof}
One can prove this using Rayleigh quotient or min-max theorems as
well.

Therefore, the summation in the first Eq. \ref{eq:sumsProj} can result
in ``lifting'' the global ground state relative to that of the local
ones. We seek conditions under which global ground state has zero
energy, i.e., remains ``unlifted'' (see Section \ref{sec:Why-Care-AboutFF}
for motivation).

It is necessary and sufficient for a state to be the ground state
if it is orthogonal to all the local terms $H'_{i}$ because all the
terms in the Hamiltonian will be zero. If the local terms commute
by the foregoing lemma the global and local ground states will be
the same, i.e., $E_{gr}^{1,2}=E_{gr}^{1}+E_{gr}^{2}=0+0=0.$ This
is too strong a requirement. It could be that the lowest eigenvectors
of each summand are aligned, in which case the system is Frustration
Free (FF) or unfrustrated. A classical analogue would be the ferromagnet.
\begin{defn}
The ground state is unfrustrated if it is also a common ground state
of all of the local terms $H_{i}$.
\end{defn}
Let us go back to Eq. \ref{eq:hamiltonian} and investigate chains
of $d$-dimensional quantum spins (qudits) on a line with nearest-neighbor
interactions. The Hamiltonian of the system, 
\begin{equation}
H=\sum_{k=1}^{N-1}H_{k,k+1}\label{eq:hamiltonian}
\end{equation}
is 2-local (each $H_{k,k+1}$ acts non-trivially only on two neighboring
qudits). Our goal is to find the necessary and sufficient conditions
for the quantum system to be {\em unfrustrated} 

The local terms can be written as 
\begin{eqnarray}
H_{k,k+1}=E_{0}^{(k)}P_{k,k+1}^{(0)}+\sum_{p}E_{p}^{(k)}P_{k,k+1}^{(p)},\label{excited}
\end{eqnarray}
 where $E_{0}^{(k)}$ is the ground state energy of $H_{k,k+1}$ and
each $P_{k,k+1}^{(p)}$ is a projector onto the subspace spanned by
the eigenstates of $H_{k,k+1}$ with energy $E_{p}^{(k)}$. The question
of existence of a common ground state of all the local terms is equivalent
to asking the same question for a Hamiltonian whose interaction terms
are 
\begin{eqnarray}
H_{k,k+1}' & = & \mathbb{I}_{1,\dots,k-1}\otimes P_{k,k+1}\otimes\mathbb{I}_{k+2,\dots,N},\label{eq:projham}
\end{eqnarray}
 with $P_{k,k+1}=\sum_{p=1}^{r}P_{k,k+1}^{(p)}$ projecting onto the
excited states of each original interaction term $H_{k,k+1}$. When
this modified system is unfrustrated, its ground state energy is zero
(all the terms are positive semi-definite). The unfrustrated ground
state belongs to the intersection of the ground state subspaces of
each original $H_{k,k+1}$ and is annihilated by all the projector
terms.

As far as the question of (un)frustration and count of the ground
states are concerned, Eq. \ref{eq:projham} yields the same result
as Eq. \ref{excited}.

\section{\label{sec:Why-Care-AboutFF}Why Care About Frustration Free Systems?}

There are many models such as the Heisenberg ferromagnetic chain,
AKLT, Parent Hamiltonians that are frustration free (FF) \cite{Koma95,AKLT87,spinGlass,cirac}.
Besides such models and the mathematical convenience of working with
projectors as local terms, what is the significance of FF systems?
In particular, do FF systems describe systems that can be realized
in nature? Some answers can be given:
\begin{enumerate}
\item It has been proved in \cite{hastings06} all gapped Hamiltonians can
be approximated by frustration free Hamiltonians if one allows for
the range of interaction to be $\mathcal{O}\mathcal{\left(\log N\right)}$.
It is believed that any type of gapped ground state is adequately
described by a frustration free model \cite{spinGlass}. 
\item The ground state is stable against variation of the Hamiltonian against
perturbations \cite{kraus_Markov} 
\[
H\left(g\right)=\sum_{k}g_{k}H_{k,k+1},\quad g_{k}\ge0
\]
as the kernels of the local terms remain invariant.
\item In quantum complexity theory the classical SAT problem is generalized
to the so called qSAT \textbf{\cite{bravyi}}. The qSAT problem is:
Given a collection of $k-$local projectors on $n$ qubits, is there
a state $\psi\rangle$ that is annihilated by all the projectors?
Namely, is the system frustration free?
\item Ground states of frustration free Hamiltonians, namely MPS can be
prepared by dissipation \textbf{\small \cite{VerstraeteNature}.}{\small \par}
\end{enumerate}

\section{Why Does Imaginary Time Evolution Work?}

By imaginary time evolution we mean $it\rightarrow\tau$. The Hamiltonian
evolution becomes 

\[
e^{-itH}|\psi\rangle\rightarrow e^{-\tau H}|\psi\rangle.
\]
Intuitively one might like to see imaginary time evolution as dissipation
of energy such that for sufficiently long time the system relaxes
to its lowest energy state. Mathematically, one can do a spectral
decomposition $H=\sum_{\alpha}E_{\alpha}|\alpha\rangle\langle\alpha|$
where $E_{\alpha}$ are energies associated with states $|\alpha\rangle$.
The imaginary time evolution

\[
e^{-\tau H}|\psi\rangle=e^{-\tau\sum E_{\alpha}|\alpha\rangle\langle\alpha|}|\psi\rangle,
\]
which implies an exponential suppression of the overlap of $|\psi\rangle$
with states that have energies higher than that of the ground state.
My numerical implementation of the imaginary time evolution is the
same as described in \cite{daniel}.

\section{Numerical Study of Quantum Spin Chains Using MPS}

Suppose we want to evolve the MPS representation of the quantum spin
system, (to be explicit we include the $\Lambda$'s see Eq. \ref{eq:RamisBrief})
starting at time $t$

\[
\psi_{t}\rangle={\displaystyle \mathcal{P}}\left\{ \bigotimes_{p=1}^{N-1}\Gamma_{t}\left(i_{p}\right)\Lambda_{t}^{\left(p\right)}|i_{p}\rangle\right\} \quad\quad\quad\mbox{OBC}.
\]

The quantum mechanical time evolution in $\Delta t$ is given by

\[
\psi_{t+\Delta t}\rangle=\left\{ \bigotimes_{p=1}^{N-1}\exp\left(-i\Delta tH_{p,p+1}\right)\right\} |\psi_{t}\rangle.
\]

As in Chapter 2, we once again decompose the Hamiltonian into two
pieces $H\equiv H_{1}+H_{2}$, where $H_{1}\equiv\sum_{p=1,3,\cdots}H_{p,p+1}$
and $H_{2}\equiv\sum_{p=2,4,\cdots}H_{p,p+1}$ are made up of terms
that all commute with one another. One can show \cite[Exercise 4.47]{NC} 

\begin{eqnarray}
e^{-itH_{1}} & = & e^{-itH_{1,2}}e^{-itH_{3,4}}\cdots e^{-itH_{N-2,N-1}}\label{eq:Trott1}\\
e^{-itH_{2}} & = & e^{-itH_{2,3}}e^{-itH_{4,5}}\cdots e^{-itH_{N-1,N}}.\nonumber 
\end{eqnarray}

In order to evolve the system in time we need to make use of Trotter's
formula \cite[Thm 4.3]{NC}
\begin{thm*}
(Trotter formula) Let $A$ and $B$ be Hermitian operators. Then for
any real $t$, 

\begin{eqnarray}
e^{i\left(A+B\right)t} & = & \lim_{n\rightarrow\infty}\left(e^{iAt/n}e^{iBt/n}\right)^{n}.\label{eq:TrotterFormula}
\end{eqnarray}

\end{thm*}
Among other things, one can use the proof of this theorem to show
\cite[see Eq. 4.103]{NC},

\begin{equation}
e^{iH\Delta t}=e^{iH_{1}\Delta t}e^{iH_{2}\Delta t}+\mathcal{O}\left(\Delta t^{2}\right)\label{eq:TrotterUSE}
\end{equation}
where we used $A\equiv H_{1}$ and $B\equiv H_{2}$.

Using Eqs. \ref{eq:Trott1} and \ref{eq:TrotterUSE} 

\[
\psi_{t+\Delta t}\rangle=\left\{ \bigotimes_{p\mbox{ odd}}\exp\left(-i\Delta tH_{p,p+1}\right)\bigotimes_{p\mbox{ even}}\exp\left(-i\Delta tH_{p,p+1}\right)\right\} \psi_{t}\rangle+\mathcal{O}\left(\Delta t^{2}\right)
\]
which implies that Trotterization and evolution in time can be implemented
by evolving the even terms first and then the odd terms for small
time intervals. 

An advantage of MPS is that one can locally update the state \cite[Lemma 2]{vidal1},
i.e., apply the operator on the local terms $n$ and $n+1$, we first
define

\[
\Theta_{t}^{i_{n},i_{n+1}}\equiv\Gamma_{t}\left(i_{n}\right)\Lambda_{t}^{\left(n\right)}\Gamma_{t}\left(i_{n+1}\right)
\]
Let the imaginary time evolution operator be

\[
V_{i'_{n},i'_{n+1}}^{i_{n},i_{n+1}}\equiv\exp\left(-\Delta tH_{i'_{n},i'_{n+1}}^{i_{n},i_{n+1}}\right),
\]
which after the imaginary time evolution update reads

\begin{equation}
\Theta_{t+\Delta t}^{i_{n},i_{n+1}}=V_{i'_{n},i'_{n+1}}^{i_{n},i_{n+1}}\Theta_{t}^{i'_{n},i'_{n+1}}\label{eq:ThetaUpdate}
\end{equation}

To get the updated Matrix Products we use SVD

\begin{equation}
\Gamma_{t+\Delta t}\left(i_{n}\right)\Lambda_{t+\Delta t}^{\left(n\right)}\Gamma_{t+\Delta t}\left(i_{n+1}\right)\equiv\mbox{SVD}\left(\Theta_{t+\Delta t}^{i_{n},i_{n+1}}\right)\label{eq:updatedGammas}
\end{equation}

The state after imaginary time evolution will in general not be normalized.
As in \cite{daniel}, I normalize the state after every update step.
The work in \cite{daniel} was confined to translationally invariant
chains with $d=2$ where $\Gamma$'s and $\Lambda$'s could be taken
to be the same at every site and bond respectively. I generalized
the numerical code to implement the updates for any $d$ and regardless
of whether there was translational invariance. A major hinderance
was implementation of $\Theta$ update (Eq. \ref{eq:ThetaUpdate})
in an efficient manner. I overcame this by first reshaping the matrices
and applying the update rules followed by undoing of the reshaping.
I then took the SVD to update the chain and normalized the state.
This was done for every pair of nearest neighbors qudits in $H_{1}$
followed by updates on every pair of nearest neighbor qudits in $H_{2}$
(see the code in Algorithm \ref{alg:This-code-updates}). See the
section \ref{sec:mpsnumerics} for further discussion and numerical
results.

\begin{algorithm}
\begin{raggedright}
\includegraphics[scale=0.9]{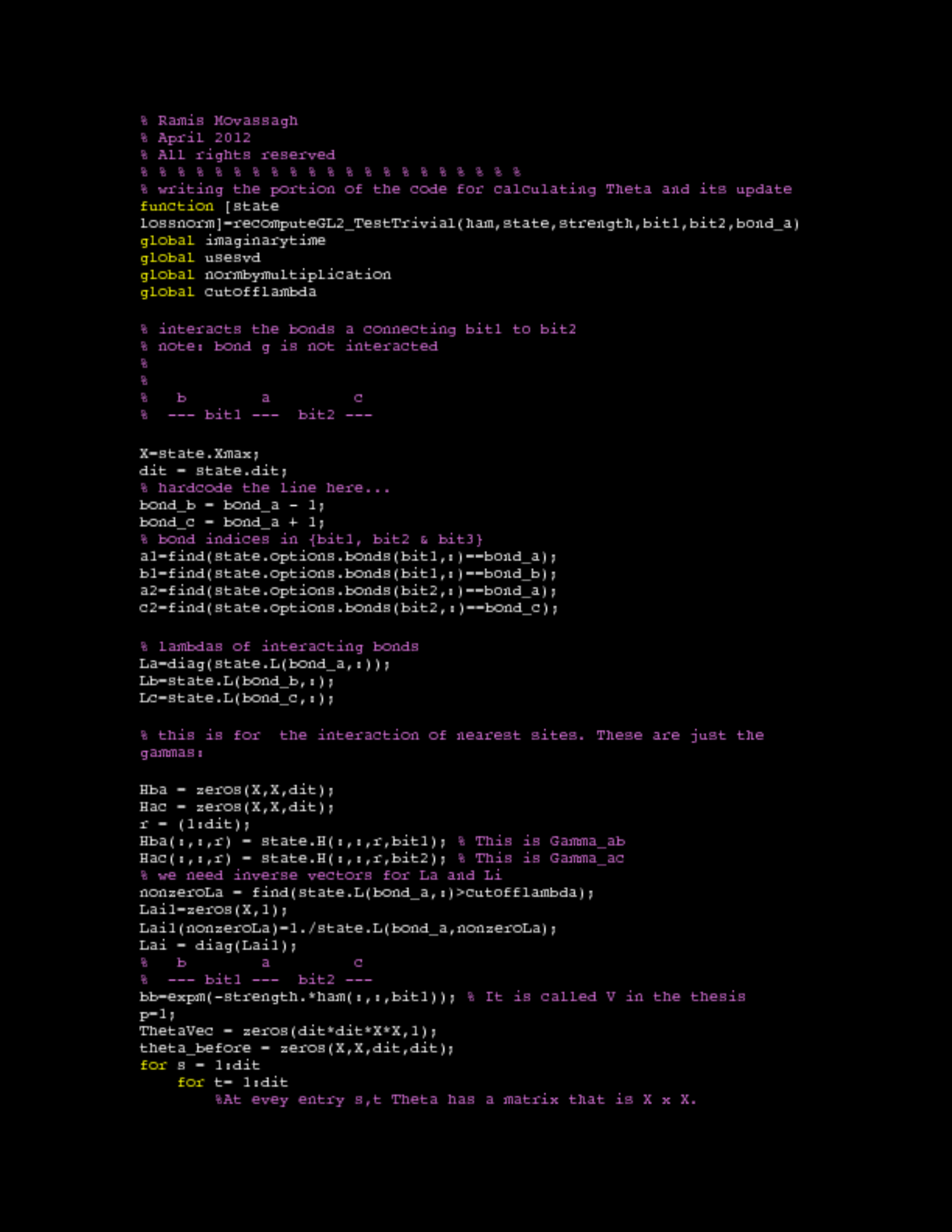}
\par\end{raggedright}

\caption{\label{alg:This-code-updates}This code updates Eqs.\ref{eq:ThetaUpdate}
and\ref{eq:updatedGammas} }
\end{algorithm}

\begin{flushleft}
\includegraphics[scale=0.9]{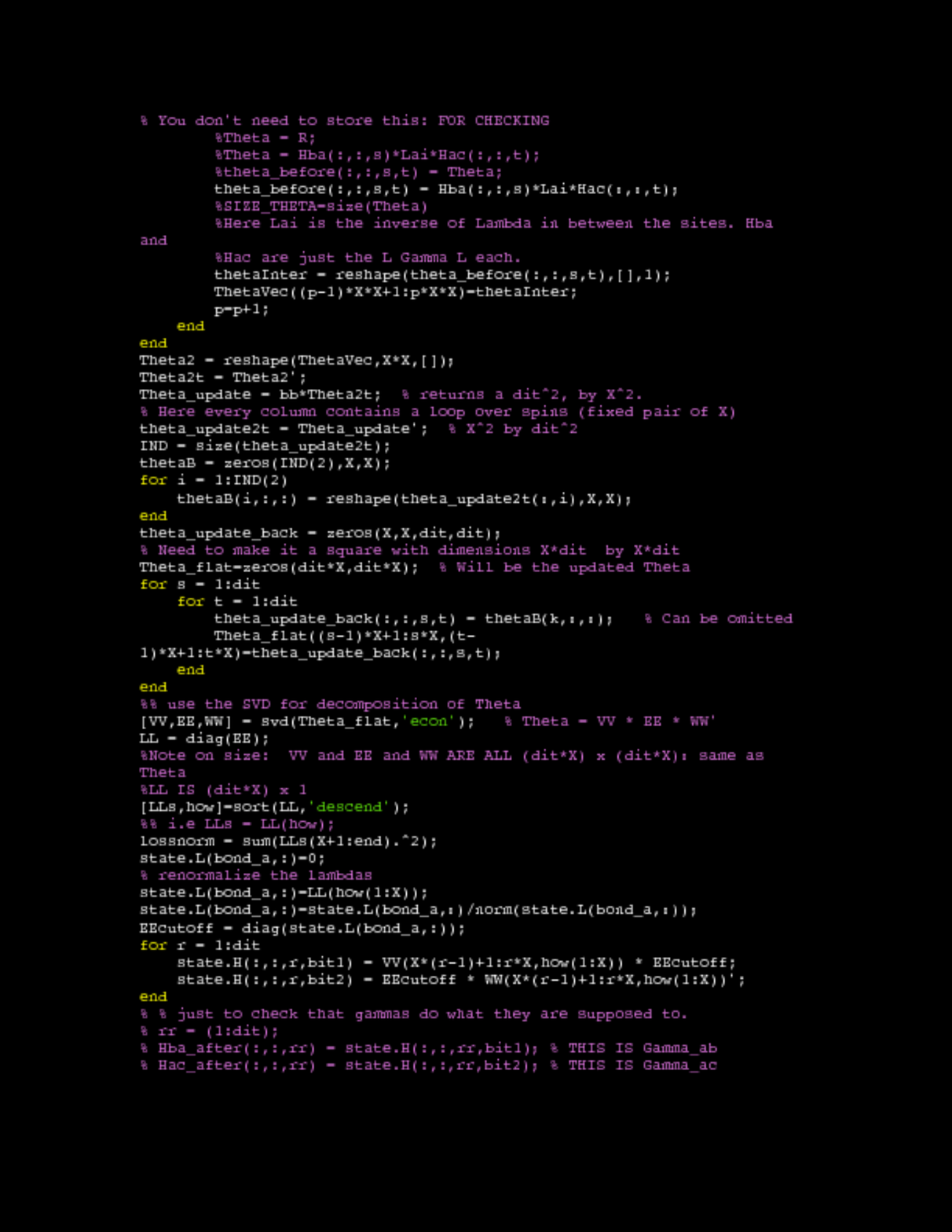}
\par\end{flushleft}

\section{Degeneracy and Non-Frustration Condition for Generic Local Terms}

We choose to investigate chains of $d$-dimensional quantum spins
(qudits) with 2-local nearest-neighbor interactions %
\footnote{The rest of this chapter is based on \cite{MFGNOS}; however, I have
corrected the erroneous assertion we made regarding choosing solutions. %
}. Our first result is an analytic derivation of the necessary and
sufficient conditions for such quantum systems to be unfrustrated.
Second, we look at their ground state properties and find a range
of parameters where we conjecture that these states are highly entangled
and thus may be difficult to find computationally. We then corroborate
this by a numerical investigation using a Matrix Product State (MPS)
method.

While the MPS formulation has been shown to work very well numerically
for most one-dimensional particle systems, complexity theory issues
seem to show there must be exceptions to this rule. Finding the ground-state
energy of a one-dimensional qudit chain with $d=11$ has been shown
to be as hard as any problem in QMA \cite{q2satline,nagajthesis}.
It is not believed that classical computers can efficiently solve
problems in QMA. However, to our knowledge until now there have not
been any concrete examples (except at phase transitions) for which
MPS methods do not appear to work reasonably well. This research was
undertaken to try to discover natural examples of Hamiltonians for
which MPS cannot efficiently find or approximate the ground states.

In Section \ref{sec:zeroenergy} we show that the question of non-frustration
for qudit chain Hamiltonians with general nearest-neighbor interactions
can be simplified to only Hamiltonians that are sums of projector
terms\cite{bravyi}. We then analytically show under what conditions
zero energy ground states for this system exist. Second, in Section
\ref{sec:mpsnumerics} we show how to search for and approximate the
ground states numerically and analyze the efficiency of finding the
required MPS. We identify an interesting class of unfrustrated qudit
chain Hamiltonians, on which our MPS methods do not work well. Led
by our numerical work, we conjecture that these ground states are
highly entangled. Finally, we summarize our results and conclude with
an outlook to further work in Section \ref{sec:conclusions}.


\section{Generic Interactions}

\label{sec:zeroenergy}

We now choose to focus on a class of Hamiltonians whose local terms
have generic eigenvectors. This implies 
\begin{eqnarray}
P_{k,k+1}=\sum_{p=1}^{r}|v_{k,k+1}^{p}\rangle\langle v_{k,k+1}^{p}|\label{eq:definev}
\end{eqnarray}
is a {\em random} rank $r$ projector acting on a $d^{2}$-dimensional
Hilbert space of two qudits, chosen by picking an orthonormal set
of $r$ random vectors (a different set for every qudit pair -- we
are not assuming translational invariance). 

We now find conditions governing the existence of zero energy ground
states (from now on, called solutions in short). We do so by counting
the number of solutions possible for a subset of the chain, and then
adding another site and imposing the constraints given by the Hamiltonian.

Suppose we have a set of $D_{n}$ linearly independent solutions for
the first $n$ sites of the chain in the form 
\begin{equation}
\psi_{\alpha_{n}}^{i_{1},\dots,i_{n}}=\Gamma_{\alpha_{1}}^{i_{1},[1]}\Gamma_{\alpha_{1}\alpha_{2}}^{i_{2},[2]}\dots\Gamma_{\alpha_{n-2},\alpha_{n-1}}^{i_{n-1},[n-1]}\Gamma_{\alpha_{n-1},\alpha_{n}}^{i_{n},[n]}
\end{equation}
 similar to MPS, with $i_{k}=1,\dots,d$ and $\alpha_{k}=1,\dots,D_{k}$;
here and below all the repeated indices are summed over. The $\Gamma$'s
satisfy the linear independence conditions%
\footnote{Note that this is not the standard MPS form, which also requires linear
independence in the other direction, i.e. 
\[
y_{\alpha_{k-1}}\Gamma_{\alpha_{k-1},\alpha_{k}}^{i_{k},[k]}=0,\forall i_{k},\alpha_{k}\Longleftrightarrow y_{\alpha_{k-1}}=0,\;\forall\alpha_{k-1}
\]
 In this case $s_{k}$ would be the Schmidt rank for the partition
of the qudits into $(1,\dots,k)$ and $(k+1,\dots,n).$ %
} 
\[
\Gamma_{\alpha_{k-1},\alpha_{k}}^{i_{k},[k]}x_{\alpha_{k}}=0,\forall i_{k},\alpha_{k-1}\Longleftrightarrow x_{\alpha_{k}}=0,\;\forall\alpha_{k}.
\]
 We now add one more site to the chain, impose the constraint $P_{n,n+1}$
and look for the 
zero-energy ground states for $n+1$ sites in the form 
\begin{equation}
\psi_{\alpha_{n+1}}^{i_{1},\dots,i_{n+1}}=\psi_{\alpha_{n-1}}^{i_{1},\dots,i_{n-1}}\Gamma_{\alpha_{n-1},\alpha_{n}}^{i_{n},[n]}\Gamma_{\alpha_{n},\alpha_{n+1}}^{i_{n+1},[n+1]}.
\end{equation}
 The unknown matrix $\Gamma_{\alpha_{n},\alpha_{n+1}}^{i_{n+1},[n+1]}$
must satisfy 
\begin{eqnarray}
\langle v_{n,n+1}^{p}|i_{n}i_{n+1}\rangle\Gamma_{\alpha_{n-1},\alpha_{n}}^{i_{n},[n]}\Gamma_{\alpha_{n},\alpha_{n+1}}^{i_{n+1},[n+1]}=0
\end{eqnarray}
 for all values of $\alpha_{n-1},\alpha_{n+1}$ and $p$, with $|v_{n,n+1}^{p}\rangle$
vectors defined in \eqref{eq:definev}. This results in a system of
linear equations 
\begin{eqnarray}
C_{p\alpha_{n-1},i_{n+1}\alpha_{n}}\Gamma_{\alpha_{n},\alpha_{n+1}}^{i_{n+1},[n+1]}=0,\label{eq:linearsystem}
\end{eqnarray}
 with $C_{p\alpha_{n-1},i_{n+1}\alpha_{n}}=\langle v_{n,n+1}^{p}|i_{n}i_{n+1}\rangle\Gamma_{\alpha_{n-1},\alpha_{n}}^{i_{n},[n]}$
a matrix with dimensions $rD_{n-1}\times dD_{n}$. If $dD_{n}\geq rD_{n-1}$
and the matrix $C$ has rank $rD_{n-1}$, the conditions given by
Eq. \ref{eq:linearsystem} are independent and we can construct $dD_{n}-rD_{n-1}$
linearly independent $\Gamma_{\alpha_{n},\alpha_{n+1}}^{i_{n+1},[n+1]}$,
corresponding to solutions for the $n+1$ qudit chain (see the appendix
for a proof of $\mbox{rank}\left(C\right)=rD_{n-1}$). The freedom
we have now is to use only a subset of them for constructing solutions
(see the next chapter). Previously in our work \cite{MFGNOS} we asserted
that the choice made implies 
\begin{equation}
s_{n+1}\leq ds_{n}-rs_{n-1},\label{eq:theequation}
\end{equation}
valid for all $n$. In order for the forgoing inequality to hold,
one needs to prove that the possible dependences resulting from making
choices do not affect the inequality (see next chapter). For example,
excluding a subspace at a given step could break the full rankness
of $C$ at a later step.

\begin{figure}
\begin{centering}
\includegraphics[width=8cm]{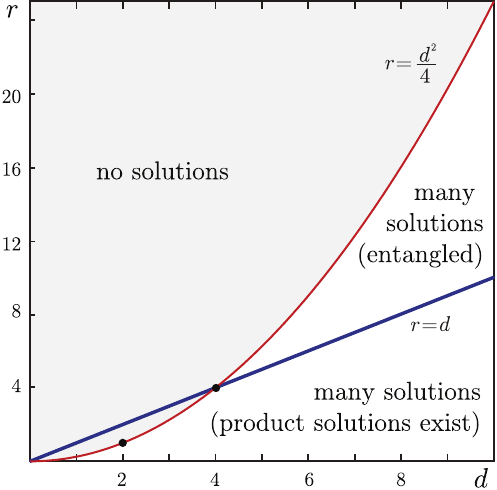}
\par\end{centering}

\caption{\label{fig:The-existence-solutions}The existence of zero energy ground
states for a qudit chain with $d$-dimensional qudits and $r$ projectors
per pair. We highlight two notable cases: $d=2,r=1$ and $d=4,r=4$.}
\end{figure}

Let $D_{0}=1$ and $D_{1}=d$ as the only constraint on $\Gamma_{\alpha_{1}}^{[1],i_{1}}$
is linear independence. The recursion relation above at each gives
$D_{n}$ linearly independent zero energy states, where 
\begin{eqnarray}
D_{n} & = & dD_{n-1}-rD_{n-2},\label{eq:recursion}
\end{eqnarray}
for all $n$ with. The solution of this recursion relation is 
\begin{eqnarray*}
D_{n} & = & \frac{f^{n+1}-g^{n+1}}{f-g}
\end{eqnarray*}
 with $f+g=d$ and $fg=r$. Hence, 
\[
f=\frac{d}{2}+\sqrt{\frac{d^{2}}{4}-r},\quad g=\frac{d}{2}-\sqrt{\frac{d^{2}}{4}-r}.
\]
 There are three interesting regimes for $r$ and $d$ which yield
different behaviors of $D_{n}$ (Figure \ref{fig:The-existence-solutions}): 
\begin{enumerate}
\item $r>\frac{d^{2}}{4}$ gives $D_{n}=r^{\frac{n}{2}}\frac{\sin(n+1)\theta}{\sin\theta}$
with $\cos\theta=\frac{d}{2\sqrt{r}}$. $D_{n}$ becomes negative
when $n+1>\frac{\pi}{\theta}$ and thus no zero energy states can
be constructed for a long chain if $r>\frac{d^{2}}{4}$. 
\item $r=\frac{d^{2}}{4}$ results in $D_{n}=\left(\frac{d}{2}\right)^{n}(n+1)$,
an exponential growth in $n$ (except when $d=2$, which gives linear
growth).

\item $r<\frac{d^{2}}{4}$ implies $f>\frac{d}{2}$ and $f>g$ so for large
$n$, $D_{n}\sim f^{n}\left(1-\frac{g}{f}\right)^{-1}$ and the number
of zero energy states grows exponentially. 
\end{enumerate}

\section{Numerical investigation using Matrix Product States}

\label{sec:mpsnumerics}

\begin{figure}
\begin{centering}
\includegraphics[width=12cm]{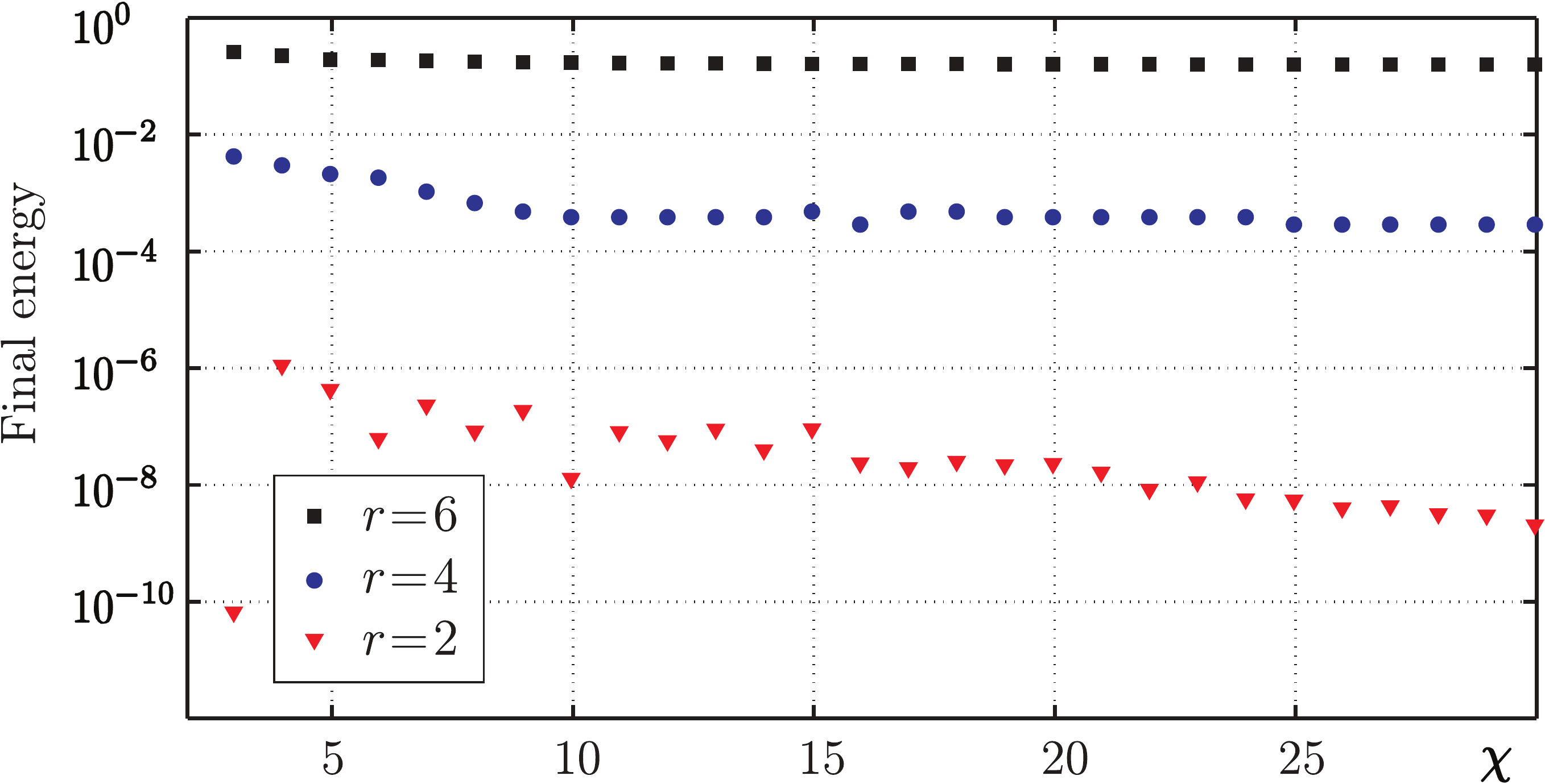} 
\par\end{centering}

\caption{\label{fig:MPSplot1}(Color online) Ground state energy from imaginary
time evolution vs. $\chi$ for different ranks of the Hamiltonian.
This is a plot for d = 4, and projector ranks of 2, 4, 6. Exact description
would require $\chi=d^{\frac{N}{2}}=2^{20}.$}
\end{figure}

In this Section we use the methods described above to numerically
search for the ground states of our class of random projector Hamiltonians
\eqref{eq:projham}. We probe the relations obtained in the previous
Section, and see how well the energy coming from our small-$\chi$
MPS imaginary time evolution converges to zero. The numerical technique
we use is similar to Vidal's \cite{vidal1,vidal2}. We use imaginary
time evolution to bring the system from a known state to its ground
state: $|\Psi_{\textrm{grd}}\rangle=\lim_{\tau\rightarrow\infty}\frac{e^{-H\tau}|\Psi_{0}\rangle}{||e^{-H\tau}|\Psi_{0}\rangle||}$.
Our experimentation with the parameters for a linear chain of length
$N=20$ is shown in Figures \ref{fig:MPSplot1},\ref{fig:MPSplot2}
and \ref{fig:MPSplot3}; all the plots are on semi-log scale and the
quantities being plotted are dimensionless.

We see that for $r<d$ the final energy converges to the zero energy
ground state relatively fast with $\chi\ll d^{N/2}$. This can be
seen in all the figures by the lowest curves (marked by triangles).
As can be seen the final energy obtained from imaginary time evolution
tends toward zero with a steep slope, indicating that the ground state
can be approximated efficiently with a small $\chi$ in MPS ansatz. 

\begin{figure}
\begin{centering}
\includegraphics[width=12cm]{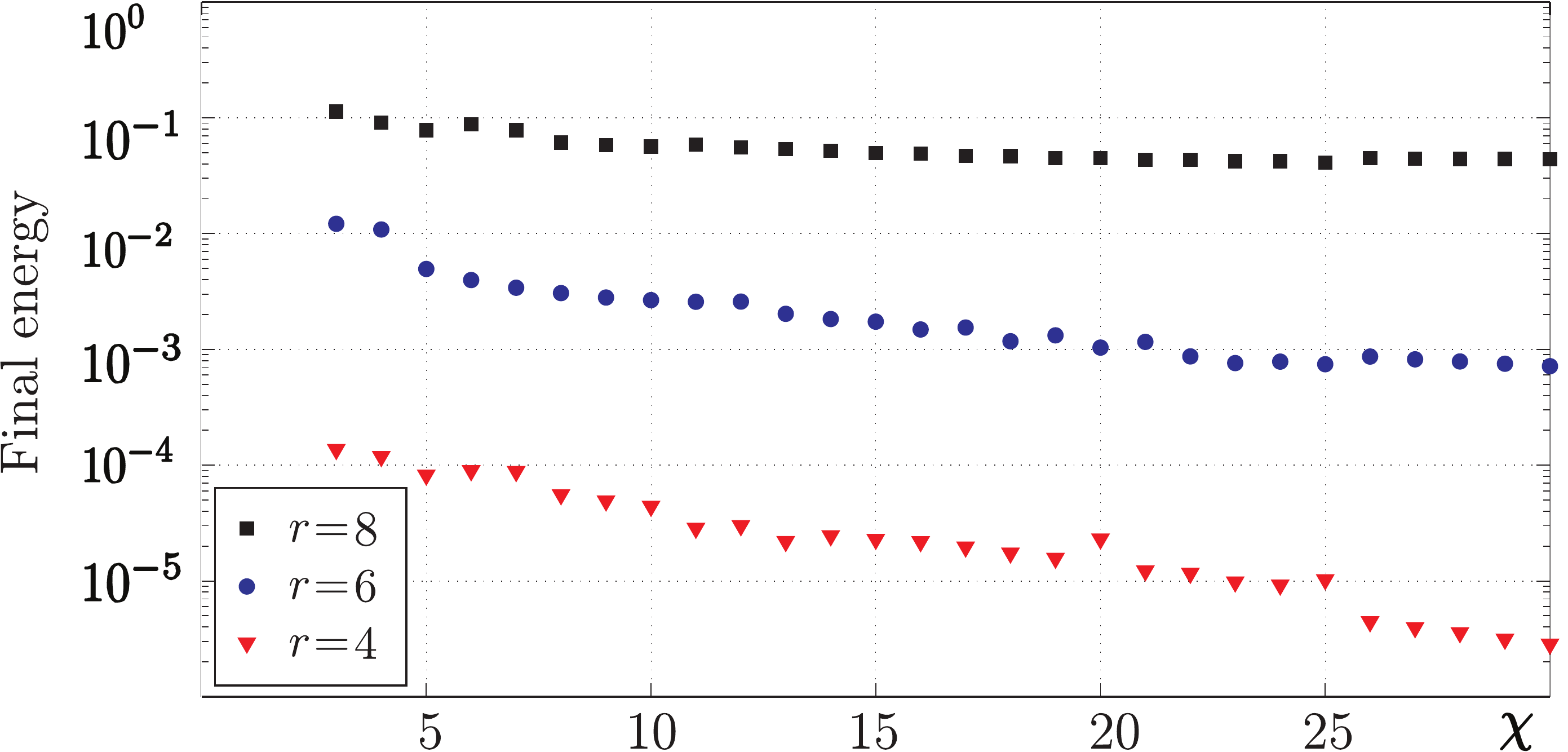} 
\par\end{centering}

\caption{\label{fig:MPSplot2}This is a plot for d = 5, and projector ranks
of 4, 6, 8. Exact description would require $\chi=5^{10}$.}
\end{figure}

\begin{figure}
\begin{centering}
\includegraphics[width=12cm]{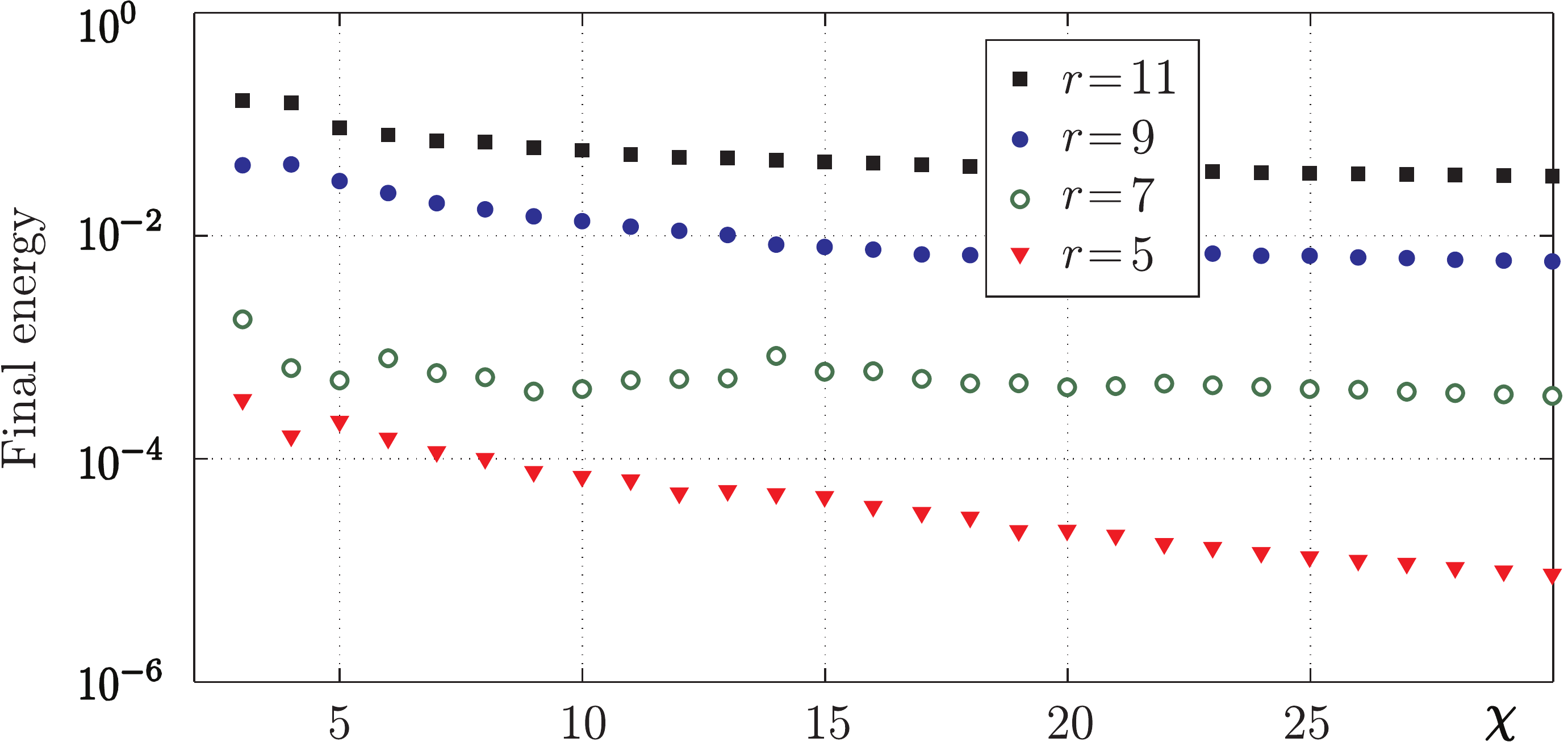}
\par\end{centering}

\caption{\label{fig:MPSplot3}The case of d = 6, and projector ranks: 5, 7,
9, 11. Exact description in general would require $\chi=6^{10}$.}
\end{figure}

The \textit{$r>d^{2}/4$ }case, marked by squares, is shown by the
top curves in all the figures. One sees that the final energy plateaus
relatively fast in all three cases. This shows that the numerics have
converged to a nonzero value and that increasing $\chi$ will not
yield a lower value of energy. Therefore, the numerical results suggest
that there are no ground states with zero energy. 

In the previous section we analytically showed that when $d\leq r\leq d^{2}/4$
there are many zero energy ground states. However, when we try to
numerically find these states we see that the final energy converges
to zero slowly. This is shown in all the Figures by the curves marked
by circles. Out of these there are the critical cases, where $r=\frac{d^{2}}{4}$.
These correspond to the curves marked by closed circles in Figures
3 and 5. The numerical investigation of the case $d\leq r\leq d^{2}/4$
is interesting because it suggests that for large number of spins
finding the ground state with small $\chi$, tractable on a normal
computer, is very hard. We interpret this as high amount of entanglement
among the zero energy ground states and leave the analytical proof
of this statement for a follow up paper.\\


\section{Summary}

\label{sec:conclusions}

We have investigated the no-frustration conditions for a system of
qudits on a line with $d$ states per site and random rank $r$ local
projector Hamiltonians acting between the nearest neighbor sites.
We proved that there are no ground states with zero energy for $r>\frac{d^{2}}{4}$
and sufficiently large $N$. The system is not frustrated for $r\leq\frac{d^{2}}{4}$.
This second parameter region further splits into two. For $d\leq r\leq\frac{d^{2}}{4}$,
many entangled zero energy ground states exist. On the other hand,
for $r<d$ we can also construct separable zero-energy ground states
(see the next chapter and also Figure \ref{fig:The-existence-solutions}).

We have verified the above numerically, in particular we have seen
that when $d\leq r\leq d^{2}/4$ approximating the ground state energy
(finding the ground states) is hard as the states seem to be highly
entangled. 


\section{Appendix}

$\;$

We would like to say that for random $|v\rangle$ the rank of $C$
is generically the maximum rank allowed, $\min(rD_{n-1},dD_{n})$.
The full rankness of $C$ in Eq. \ref{eq:linearsystem} is not obviously
true. In particular in the regime $r\leq\frac{d^{2}}{4}$, to which
we restrict ourselves from now on, $D_{n}$ grows exponentially in
$n$, while the number of parameters in the $|v_{k,k+1}^{p}\rangle$
on which $C$ depends only grows linearly. Thus $C$ is far from a
generic matrix of its size, but we now prove that its rank is indeed
$rD_{n-1}.$

The argument used by Laumann et al \cite{laumann} to prove their
``geometrization theorem'' also applies to our problem. It shows
that for a chain of N qudits with random $|v_{k,k+1}^{p}\rangle$,
i.e. for a Hamiltonian $H$ as in equations (2), (4) and (5) the number
of zero-energy states, i.e. $\dim(\ker(H))$, is with probability
one (which is what we mean by generic) equal to its minimum value.
The calculation leading to the recursion relation Eq. \ref{eq:recursion}
and its solution, shows that in the regime $r\leq d^{2}/4$ this minimum
is $\geq D_{N},$ since if the rank of the $rD_{k-1}\times dD_{k}$
matrix $C$ is ever less than $rD_{k-1}$ we can choose $D_{k+1}$
of them. Hence it is sufficient to find a single set of $|v_{k,k+1}^{p}\rangle$
for which $\dim(\ker(H))=D_{N}$ to prove that $D_{N}$ is the generic
value, i.e. that greater values occur with probability zero. This
implies that the rank of each $C$ is generically $rD_{k-1}$, since
otherwise at the first $k$ where $C$ had smaller rank we could construct
more than $D_{k+1}$ solutions for a chain of length $k+1$. 

\begin{singlespace}
We construct $|v_{k,k+1}^{p}\rangle$ with the property $\langle v_{k,k+1}^{p}|i_{k}i_{k+1}\rangle=0$
unless $i_{k}\leq\frac{d}{2}$ and $i_{k+1}>\frac{d}{2}$. This can
be done for $r$ linearly independent $|v^{p}\rangle$ if $r\leq\frac{d^{2}}{4}$.
We now assume $d$ is even; the modifications for $d$ odd are obvious.
We proceed by induction on $n$. Assume that in each $\Gamma_{\alpha_{k-1},\alpha_{k}}^{i_{k},[k]}$
with $k\leq n$, $\alpha_{k}$ runs from $1$ to $D_{k}$. From the
definition of $C$ (following Eq. \ref{eq:linearsystem}) and the
special choice of $|v\rangle$, $C_{p\alpha_{n-1},i_{n+1}\alpha_{n}}=0$
for $i_{n+1}\leq\frac{d}{2}$ and so from Eq. \ref{eq:linearsystem}
$\Gamma_{\alpha_{n},\alpha_{n+1}}^{i_{n+1},[n+1]}$ is unconstrained
for $i_{n+1}\leq\frac{d}{2}$. This allows us to choose, for $1\leq\alpha_{n+1}\leq\frac{d}{2}D_{n}$,
\end{singlespace}

\noindent $\Gamma_{\alpha_{n},\alpha_{n+1}}^{i_{n+1},[n+1]}=1$ when
$\alpha_{n+1}=\frac{d}{2}(\alpha_{n}-1)+i_{n+1}$ 

\noindent with $1\leq\alpha_{n}\leq D_{n}$, $1\leq i_{n+1}\leq\frac{d}{2}$

\noindent $\quad\Gamma_{\alpha_{n},\alpha_{n+1}}^{i_{n+1},[n+1]}=0$
otherwise.

\begin{singlespace}
As part of our induction, we assume that for $1\leq\alpha_{n}\leq\frac{d}{2}D_{n-1}$,
\end{singlespace}

\noindent $\quad\Gamma_{\alpha_{n-1},\alpha_{n}}^{i_{n},[n]}=1$ when
$\alpha_{n}=\frac{d}{2}(\alpha_{n-1}-1)+i_{n}$ 

\noindent with $1\leq\alpha_{n-1}\leq D_{n-1}$, $1\leq i_{n}\leq\frac{d}{2}$ 

\noindent $\quad\Gamma_{\alpha_{n-1},\alpha_{n}}^{i_{n},[n]}=0$ otherwise.

Now we can show that the rows of $C_{p\alpha_{n-1},i_{n+1}\alpha_{n}}$
are linearly independent. For if, $\sum_{p,\alpha_{n-1}}y_{p,\alpha_{n-1}}C_{p\alpha_{n-1},i_{n+1}\alpha_{n}}=0$
for all $i_{n+1},\alpha_{n}$ this is true in particular for all $i_{n+1}>d/2$,
$\alpha_{n}\leq\frac{d}{2}D_{n-1}$, when it becomes $\sum_{p}y_{p\alpha_{n-1}}\langle v_{n,n+1}^{p}|i_{n}i_{n+1}\rangle=0$
for all $i_{n}\leq\frac{d}{2}$, $i_{n+1}>\frac{d}{2}$, and $\alpha_{n-1}\leq D_{n-1}$.
Since the $|v^{p}\rangle$ are linearly independent, this is only
true if $y_{p\alpha_{n-1}}=0$ for all $p$, $\alpha_{n-1}$. Hence
the rank of $C$ is $rD_{n-1}$ and $\alpha_{n+1}$ can take altogether
$dD_{n}-rD_{n-1}=D_{n+1}$ values, which is what we wanted to prove.

\chapter{Entanglement of The Ground States}

The next natural question to ask is: how entangled are the ground
states of qudit chains with generic local interaction? In this chapter
I summarize my efforts in proving various results, which have not
appeared elsewhere. As of now we have not succeeded in proving the
main theorem (see the conjecture below); I include the partial results
with the hope that they inspire future progress. In the previous chapter
we showed that there are exponentially many ground states when $d\le r\le d^{2}/4$
(see Table). 

\begin{center}
\begin{tabular}{|c|c|c|}
\hline 
Parameter range & Number of ground states & Frustrated?\tabularnewline
\hline 
\hline 
$d\leq r\leq d^{2}/4$ & $\sim exp\left(n\right)$ & No\tabularnewline
\hline 
$r>d^{2}/4$ & ---------- & Yes\tabularnewline
\hline 
\end{tabular}
\par\end{center}

We will see that when $r<d$ the ground states can be product states.
How entangled are the ground states in the regime $d\le r\le d^{2}/4$
when $d\ge4$? Simple theorems of algebraic geometry tell us that
among the many ground states there are highly entangled ground states.
Can there be product states? If not, are \textit{all} the ground states
highly entangled (see the conjecture below)? By highly entangled we
mean their Schmidt rank is exponentially lower bounded. 

There are two methods that can yield the bounds needed: 1) Choosing
solutions as we build the ground states marching from one end of the
chain and proving lower bounds on the number of solutions that need
to be kept to build any state.2) Relating the construction of the
previous chapter to the Schmidt rank by working from both ends and
matching solutions.

\section{Set Up}

\subsection{Genericity of $C^{[n]}$}

Due to generic local interactions, the complete set of eigenstates
has parameter count equal to a polynomial in the number of spins,
though the dimensionality is exponential. This implies that $C^{[n]}$
is non-generic despite its entries being functions of polynomial random
variables. However, as long as the variables are continuous, statements
about the rank of $C^{[n]}$ can be made in a 'generic' sense. That
is if the rank is full for some choice of random variables, the full
rankness holds with probability one for random choice of those variables.

\subsection{Basis for Solutions and $C^{[n]}$ }

We want to understand the Kernel of $C^{[n]}$ as give by

\begin{equation}
C_{p\alpha_{n-1},i_{n+1}\alpha_{n}}^{[n]}\equiv\langle v_{n,n+1}^{p}|i_{n}i_{n+1}\rangle\Gamma_{\alpha_{n-1},\alpha_{n}}^{[n],i_{n}}\label{eq:cOriginal}
\end{equation}

\begin{flushleft}
Further,
\par\end{flushleft}

\begin{equation}
|v_{n,n+1}^{p}\rangle=\beta_{i_{n},i_{n+1}}^{p}|i_{n}i_{n+1}\rangle\label{eq:v}
\end{equation}
 with $\beta's$ drawn randomly, say from a Gaussian distribution.
$C^{[n]}$ solves:

\begin{equation}
C_{p\alpha_{n-1},i_{n+1}\alpha_{n}}^{[n]}\Gamma_{\alpha_{n},\alpha_{n+1}}^{[n+1],i_{n+1}}=0.\label{eq:constraint-1}
\end{equation}
Inserting Eq.\ref{eq:v} in Eq.\ref{eq:cOriginal} ($\beta$'s are
real) we get

\begin{equation}
C_{p\alpha_{n-1},i_{n+1}\alpha_{n}}^{[n]}=\beta_{i_{n},i_{n+1}}^{p}\Gamma_{\alpha_{n-1},\alpha_{n}}^{[n],i_{n}}.\label{eq:c_index}
\end{equation}

Let $B$ be the $r\times d^{2}$ matrix of $\beta$'s which we can
write as $d$ blocks of matrices ($B_{k}$) of size $r\times d$ put
next to one another (not multiplied): 

\begin{equation}
B\equiv\left[B_{1}B_{2}\cdots B_{d}\right],
\end{equation}
with, $B_{k}=\left[\begin{array}{ccc}
\beta_{1k}^{1} & \cdots & \beta_{d,k}^{1}\\
\vdots & \ddots & \vdots\\
\beta_{1k}^{r} & \cdots & \beta_{d,k}^{r}
\end{array}\right],\quad k=1,\ldots,d.$ The matrix $B$ has rank $r$; whereas $\mbox{rank}\left(B_{k}\right)=\mbox{min}\left(r,d\right)$
(both with probability one).

In matrix notation,

\begin{equation}
C^{[n]}=\left[\begin{array}{cccc}
\left(\mathbb{I}_{D_{n-1}}\otimes B_{1}\right)\Gamma^{[n]} & \left(\mathbb{I}_{D_{n-1}}\otimes B_{2}\right)\Gamma^{[n]} & \cdots & \left(\mathbb{I}_{D_{n-1}}\otimes B_{d}\right)\Gamma^{[n]}\end{array}\right].\label{eq:injectivity}
\end{equation}

\begin{flushleft}
Remark: In what follows, by the word \textit{generic} we mean with
probability one. 
\par\end{flushleft}
\begin{lem}
For $r<d$ there are product states, i.e., there exists $\chi=1$
solutions.\end{lem}
\begin{proof}
We show that the states can be satisfied by product states. In this
case, taking the Schmidt rank to be one, we propose that $\Gamma^{[n]}=\overrightarrow{\gamma_{n}}$
is a vector of size $d$ corresponding to the physical index $i_{n}$.
The constraint matrix $C^{[n]}$ takes the form $C_{p,i_{n+1}}^{[n]}=\left[\begin{array}{cccc}
B_{1}\overrightarrow{\gamma_{n}} & B_{2}\overrightarrow{\gamma_{n}} & \cdots & B_{d}\overrightarrow{\gamma_{n}}\end{array}\right]$, which is a $r\times d$ matrix. It clearly has full rank equal to
$r$. The question becomes, can we build $\overrightarrow{\gamma_{n+1}}$
to be a vector of size $d$? The system of $r<d$ equations with random
coefficients can always be satisfied. Therefore indeed we can take
$\Gamma^{[n+1]}=\overrightarrow{\gamma_{n+1}}$. \end{proof}
\begin{cor}
When $r\geq d$, generically there are no product states.\end{cor}
\begin{proof}
The only solutions is the trivial solution $\overrightarrow{\gamma_{n+1}}=0$.
The problem is over specified; $r$ constraints and $d$ variables
generically cannot be satisfied for $r\geq d$. 
\end{proof}
Comment: When $r\geq d$ , $\left(\mathbb{I}\otimes B_{1}\right)$
is a full rank injective map. $\Gamma^{[n]}$ has $D_{n}$ independent
columns and is a tall rectangular matrix and is therefore an injective
map. Their composition $\left(\mathbb{I}\otimes B_{1}\right)\Gamma^{[n]}$
is yet another full rank injective map, whose rank is the number of
columns. It is however not clear that when we concatenate this map
by $\left(\mathbb{I}\otimes B_{k}\right)\Gamma^{[n]}$ with $1<k\leq d$
to obtain Eq. \ref{eq:injectivity}, the resulting matrix has rank
equal to the number of rows. 

\begin{flushleft}
Comment: In our previous work we proved that there are many zero energy
ground states when $d\leq r\leq d^{2}/4$ and none when $r>d^{2}/4$. 
\par\end{flushleft}

\subsection{\label{sub:Constructing-Solutions-for}Constructing Solutions for
$d\leq r\leq d^{2}/4$}

In Eq. \ref{eq:noChoice} taking $D_{0}=1$ and $D_{1}=d$, we build
the solutions recursively by marching along the chain. The very first
set of solutions $\Gamma_{D_{0},D_{1}}^{[1],i_{1}}$, is a $d\times d$
diagonal matrix. From Eq. \ref{eq:injectivity} we have $C^{[1]}=\left(B_{1}\Gamma^{[1]}\;\cdots\; B_{d}\Gamma^{[1]}\right)\equiv B$.
To find a basis for the kernel we row reduce $C^{[1]}$ to put it
in the row echelon form: 
\[
C_{\mbox{echel}}^{[1]}{\scriptstyle \equiv\left[\begin{array}{cccccc}
1 &  &  & \tilde{\beta}_{d^{2}-r}^{1} & \cdots & \tilde{\beta}_{d^{2}}^{1}\\
 & \ddots &  & \vdots &  & \vdots\\
 &  & 1 & \tilde{\beta}_{d^{2}-r}^{r} &  & \tilde{\beta}_{d^{2}}^{r}
\end{array}\right]\equiv\left[\begin{array}{cccccc}
 &  &  & \uparrow &  & \uparrow\\
 & \mathbb{I}_{r} &  & \tilde{\beta}_{1} & \cdots & \tilde{\beta}_{d^{2}-r}\\
 &  &  & \downarrow &  & \downarrow
\end{array}\right]},
\]
Recall that $\Gamma^{[2]}$ is the set of solutions in the Kernel
of $C^{[1]}$. The null space of $C^{[1]}$ is

\[
\mathrm{\mathcal{N}}\left(C^{[1]}\right)=\mbox{span}\left(\begin{array}{c}
\uparrow\\
-\tilde{\beta}_{j}\\
\downarrow\\
\hline \\
e_{j}\\
\\
\end{array}\right),\mbox{ }1\leq j\leq d^{2}-r\mbox{ }
\]
the horizontal line depicts the partition of the $\tilde{\beta}$'s
from the unit vectors and does not signify any mathematical operation. 

Consequently, $\Gamma^{[2]}$ can be expressed as

\[
\Gamma^{[2]}=\left(\begin{array}{ccc}
\uparrow &  & \uparrow\\
-\tilde{\beta}_{1} & \cdots & -\tilde{\beta}_{d^{2}-r}\\
\downarrow &  & \downarrow\\
\hline \\
 & \mathbb{I}_{d^{2}-r}\\
\\
\end{array}\right).
\]

This will in turn define the $rd\times d^{2}$ matrix $C^{[2]}$ 

\[
C^{[2]}=\left[\begin{array}{cccc}
\left(\mathbb{I}_{D_{1}}\otimes B_{1}\right)\Gamma^{[2]} & \left(\mathbb{I}_{D_{1}}\otimes B_{2}\right)\Gamma^{[2]} & \cdots & \left(\mathbb{I}_{D_{1}}\otimes B_{d}\right)\Gamma^{[2]}\end{array}\right],
\]
whose Kernel defines $\Gamma^{[3]}$. Continuing this way we arrive
at 

\[
\Gamma^{[n]}=\left(\begin{array}{ccc}
\uparrow &  & \uparrow\\
-\tilde{\beta}_{1} & \cdots & -\tilde{\beta}_{D_{n}}\\
\downarrow &  & \downarrow\\
\hline \\
 & \mathbb{I}_{D_{n}}\\
\\
\end{array}\right),\mbox{ giving }
\]

\begin{equation}
C^{[n]}=\left(\begin{array}{ccccccc}
\uparrow &  & \uparrow &  & \uparrow &  & \uparrow\\
\gamma_{1}^{\left(1\right)} & \cdots & \gamma_{D_{n}}^{\left(1\right)} & \centerdot\centerdot\centerdot & \gamma_{1}^{\left(d\right)} &  & \gamma_{D_{n}}^{\left(d\right)}\\
\downarrow &  & \downarrow &  & \downarrow &  & \downarrow\\
\hline B_{1} &  &  &  & B_{d}\\
 & \ddots &  & \centerdot\centerdot\centerdot &  & \ddots\\
 &  & B_{1} &  &  &  & B_{d}
\end{array}\right)\label{eq:Cmatrix}
\end{equation}
 where 

\[
\left(\begin{array}{c}
\uparrow\\
\gamma_{h}^{\left(k\right)}\\
\downarrow
\end{array}\right)\equiv\mathbb{I}\otimes B_{k}\left(\begin{array}{c}
\uparrow\\
-\tilde{\beta}_{h}\\
\downarrow
\end{array}\right);\mbox{ }1\leq h\leq D_{n}
\]

Therefore these constraints give 

\begin{equation}
D_{n+1}\geq dD_{n}-rD_{n-1};\label{eq:FirstInequ}
\end{equation}
solutions with equality for $C^{[n]}$ being full rank. In our previous
work we proved that $C^{[n]}$ is full rank:

\begin{equation}
D_{n+1}=dD_{n}-rD_{n-1}.\label{eq:noChoice}
\end{equation}

The main theorem of this chapter remains to be proved: 
\begin{conjecture}
Generically all the ground states of the qudit chain with generic
local interactions are all highly entangled in the regime $d\le r\le d^{2}/4$.
\end{conjecture}
In order to obtain entanglement bounds one can take at least two different
approaches.

\section{First Method: Choosing Solutions}

Let the $n^{\mbox{th}}$ step be the first step in which we make a
choice by throwing away $u_{n}$ solutions and keeping $s_{n}\equiv D_{n}-u_{n}$
solutions. For example if we throw out the second and the last column
of $\Gamma^{[n]}$ we obtain

\begin{equation}
\Gamma^{[n]}=\left(\begin{array}{ccccc}
\uparrow & \uparrow & \uparrow &  & \uparrow\\
-\tilde{\beta}_{1} & -\tilde{\beta}_{3} & -\tilde{\beta}_{4} & \cdots & -\tilde{\beta}_{D_{n}-1}\\
\downarrow & \downarrow & \downarrow &  & \downarrow\\
\hline 1\\
0 & 0\\
 & 1\\
 &  & 1\\
 &  &  & \ddots\\
 &  &  &  & 1\\
 &  &  &  & 0
\end{array}\right).\label{eq:ExampleChoice}
\end{equation}
This in turn defines the $rD_{n-1}\times ds_{n}$ matrix $C^{[n]}$

\[
C^{[n]}=\left[\begin{array}{cccc}
\left(\mathbb{I}_{s_{n}}\otimes B_{1}\right)\Gamma^{[n]} & \left(\mathbb{I}_{s_{n}}\otimes B_{2}\right)\Gamma^{[n]} & \cdots & \left(\mathbb{I}_{s_{n}}\otimes B_{d}\right)\Gamma^{[n]}\end{array}\right],
\]
whose dimension of Kernel dictates the number of independent solutions
we can build, i.e., $\Gamma^{[n+1]}$. We can keep making choices
by excluding solutions to build some arbitrary state, i.e., $s_{k}=D_{k}-u_{k}$.
\begin{lem}
\label{lem:lower-bound}A lower bound on $s_{k}$ is a lower bound
on $\chi_{k}$.\end{lem}
\begin{proof}
$\chi_{k}$ can be obtained from $s_{k}$ by applying the canonicality
condition \cite[condition 2, in Theorem 1]{cirac} to the $\Gamma^{[k]}$;
equivalently a further constraint of linear independence from right
to left (see the previous chapter) needs to be imposed on the solutions.
Since solutions with $\chi_{n}$ are contained in solutions with $s_{n}$
one can take the $\chi_{k}=s_{k}^{\mbox{lb}}$, (lb denotes lower
bound) to allow construction of any state.
\end{proof}

\subsection{Rank of $C^{[n]}$}

Suppose at some step $n$ we throw away $u_{n}$ and keep $s_{n}\equiv D_{n}-u_{n}$
of the solutions. In our previous work we proved that for $u_{n}=0$
for all $n$, $C^{[n]}$ is full rank. The full rankness however does
not generally hold if we throw away solutions. When a choice is made,
as in the example shown in Eq. \eqref{eq:ExampleChoice}, the matrix
of constraints becomes

\begin{equation}
C^{[n]}=\left(\begin{array}{ccccccc}
\uparrow &  & \uparrow &  & \uparrow &  & \uparrow\\
\gamma_{1}^{\left(1\right)} & \cdots & \gamma_{s_{n}}^{\left(1\right)} & \centerdot\centerdot\centerdot & \gamma_{1}^{\left(d\right)} &  & \gamma_{s_{n}}^{\left(d\right)}\\
\downarrow &  & \downarrow &  & \downarrow &  & \downarrow\\
\hline B_{1}^{1} &  &  &  & B_{d}^{1}\\
 & \ddots &  & \centerdot\centerdot\centerdot &  & \ddots\\
 &  & B_{1}^{\frac{s_{d}}{d}} &  &  &  & B_{d}^{\frac{s_{d}}{d}}
\end{array}\right)\equiv\left(\begin{array}{c}
\mathbb{R}^{[n]}\\
\hline \mathbb{B}^{[n]}
\end{array}\right).\label{eq:Cmatrix_choice}
\end{equation}
Each solution that is excluded will reduce the columns of $C^{[n]}$
$d-$fold and the superscripts on $B_{k}$'s remind us that there
may be some columns missing as a consequence of having made a choice.
It is easy to see how this comes about. Suppose we exclude the first
column of $\Gamma^{[n]}$, then each of the top foremost $B_{k}$'s
in Eq. \eqref{eq:Cmatrix} loses its first column. $ $ In particular
if we throw away the last $d$ solutions in $\Gamma^{[n]}$, one can
see that $C^{[n]}$ above would have $r$ rows all zeros as its last
rows. %

\subsection{Towards Entanglement Bounds}

Suppose we march along a semi-infinite line and choose solutions randomly
then it is plausible to assume the recursion $s_{n+1}\le ds_{n}-rs_{n-1}$
(i.e., no rank deficiency) to hold. In this case, the entanglement
bounds can nicely be obtained
\begin{lem}
Suppose $s_{n+1}\le ds_{n}-rs_{n-1}$ holds everywhere, then $s_{n+1}\ge qs_{n}$
where $q\equiv\frac{r}{d}\left(1+\frac{r}{d^{2}-r}\right)$, which
for $d\ge4$ and $d\le r\le d^{2}/4$ implies exponentially large
lower bound on $\chi_{n+1}$.\end{lem}
\begin{proof}
We apply the set of inequalities, $s_{n+2}\le ds_{n+1}-rs_{n}$ and
existence of solutions $s_{k}>0$ for all $k$. Positivity of $s_{n+2}\Rightarrow s_{n+1}>\frac{r}{d}s_{n}$.
Let $s_{n+1}\equiv\frac{r}{d}s_{n}+u_{n}$ which implies $s_{n+2}\le du_{n}$.
We now bound $u_{n}$ using $s_{n+3}\le ds_{n+2}-rs_{n+1}\Rightarrow s_{n+3}\le d^{2}u_{n}-rs_{n+1}=d^{2}u_{n}-\frac{r^{2}}{d}s_{n}-ru_{n}=u_{n}\left(d^{2}-r\right)-\frac{r^{2}}{d}s_{n}$.
Now $s_{n+3}>0$ implies $u_{n}>\frac{r^{2}}{d\left(d^{2}-r\right)}s_{n}.$
Combining this with $s_{n+1}\equiv\frac{r}{d}s_{n}+u_{n}$ we get

\[
s_{n+1}>\frac{r}{d}\left(1+\frac{r}{d^{2}-r}\right)s_{n}\equiv qs_{n}=\begin{cases}
\begin{array}{c}
\frac{d}{d-1}s_{n}\\
\\
\frac{d}{3}s_{n}
\end{array} & \begin{array}{c}
r=d\\
\\
r=\frac{d^{2}}{4}.
\end{array}\end{cases}
\]
which proves $s_{n}=q^{n}$ are lower bounds on solutions that grow
exponentially with $n$ for $d\ge4$. Using Lemma \ref{lem:lower-bound}
we conclude $\chi_{n}$ is exponentially lower bounded. 
\end{proof}

Comment: The formulation above does not take into account the finiteness
of the chain. Further it ignores the affect of possible rank deficiency
due to choosing solutions by assuming $s_{n+1}\le ds_{n}-rs_{n-1}$
holds at every step. In particular, for a finite chain of length $N$,
$\chi_{N-1}\le d$ whereas, $q^{N-1}\sim\exp\left(N-1\right)$.

\section{Second Method: Matching Solutions}

One can approach the problem by marching along both ways and match
solutions in between the two ends and ask what the lower bound on
the Schmidt rank must be. It is notationally convenient to index the
marching along from left and right differently; let $N=m+n$. Suppose
we march along from left $n$ sites building the $\Gamma$ matrices
up to and including $\Gamma_{\alpha_{n-1},\alpha_{n}}^{i_{n},[n]}$
and suppose we we march along from right $m$ sites with $\Gamma_{\alpha_{m},\alpha_{m-1}}^{i_{m},[m]}$
. Recall that marching along from left (right) only imposed linear
independence of the solutions from left (right). To build any state
we need to match the solutions and apply the last constraint between
sites $n$ and $m$,

\begin{eqnarray}
\beta_{i_{n},i_{m}}^{p}\Gamma_{\alpha_{n-1},\alpha_{n}}^{i_{n},[n]}X_{\alpha_{n},\alpha_{m}}\Gamma_{\alpha_{m},\alpha_{m-1}}^{i_{m},[m]} & = & 0\label{eq:matching}\\
\Leftrightarrow C_{p\alpha_{n-1},\alpha_{m-1};\alpha_{n},\alpha_{m}}X_{\alpha_{n},\alpha_{m}} & = & 0,\nonumber 
\end{eqnarray}
where $C_{p\alpha_{n-1},\alpha_{m-1};\alpha_{n},\alpha_{m}}\equiv\beta_{i_{n},i_{m}}^{p}\Gamma_{\alpha_{n-1},\alpha_{n}}^{i_{n},[n]}\Gamma_{\alpha_{m},\alpha_{m-1}}^{i_{m},[m]}$
defines $rD_{n-1}D_{m-1}$ constraints on $D_{n}D_{m}$ variables--entries
of $X$. It can be checked that $D_{n+m}=D_{n}D_{m}-rD_{n-1}D_{m-1}$
as expected. Let us make a crisp problem definition. Given,
\begin{enumerate}
\item Eq. \ref{eq:matching} has $D_{n+m}\equiv D_{N}$ solutions. Let us
call the space of solutions $\mathcal{S}$
\item $\Gamma^{[n]}$ has $D_{n}$ independent columns: If $\Gamma_{\alpha_{n-1},\alpha_{n}}^{i_{n},[n]}b_{\alpha_{n}}=0,\quad\forall i_{n},\alpha_{n-1}\Rightarrow b_{\alpha_{n}}=0$
\item $\Gamma^{[m]}$ has $D_{m}$ independent columns: If $c_{\alpha_{m}}\Gamma_{\alpha_{m},\alpha_{m-1}}^{i_{m},[m]}=0,\quad\forall i_{m},\alpha_{m-1}\Rightarrow c_{\alpha_{m}}=0$
\item $\beta_{i_{n},i_{m}}^{p}$ are generic;
\end{enumerate}
prove that 
\[
\left\{ \mathcal{S}\right\} \cap\left\{ \mbox{determinental variety of \ensuremath{D_{n}\times D_{m}}matrices with rank \ensuremath{\chi\le\chi_{0}}}\right\} =\emptyset
\]
with probability one for some large $\chi_{0}$. Using matrix notation,
one wants to bound the rank of matrix $X$ that satisfies

\[
\sum_{i_{m}=1}^{d}\left(\mathbb{I}_{D_{n-1}}\otimes B_{i_{m}}\right)\Gamma\left(i_{n}\right)X\Gamma\left(i_{m}\right)=0
\]
where each $B_{i_{m}}$ is a $r\times d$ generic matrix as discussed
above and $\Gamma\left(i_{n}\right)$ is $dD_{n-1}\times D_{n}$ with
$D_{n}$ independent columns and each $\Gamma\left(i_{m}\right)$
for a fixed $i_{m}$ is $D_{m}\times D_{m-1}$. 

How do the entries of $X$ depend on the local constraints that were
imposed at previous steps? The entries are polynomials of very high
degree%
\footnote{The degrees of polynomials, shown below, were obtained by Jeffrey
Goldstone.%
} -- much larger than the size of the matrix. One can see this by using
the following basis for building solutions 

\[
\left[\begin{array}{ccc}
X & | & Y\end{array}\right]\left[\begin{array}{c}
\left(\mbox{adj}X\right)^{-1}\\
-I
\end{array}\right]
\]

\[
\begin{array}{ccccccccc}
\begin{array}{c}
1\\
\centerdot
\end{array} & \cdots & \begin{array}{c}
n-1\\
\centerdot
\end{array} & \begin{array}{c}
n\\
\centerdot
\end{array} & X & \begin{array}{c}
m\\
\centerdot
\end{array} & \begin{array}{c}
m-1\\
\centerdot
\end{array} & \cdots & \begin{array}{c}
1\\
\centerdot
\end{array}\end{array}
\]
The elements of $X$ have entries that are homogeneous polynomial
functions of local terms at the previous steps 

\begin{eqnarray}
X_{ij} & = & \mbox{poly}\left\{ \left[\left(\beta_{12}\right)^{r^{n}D_{0}D_{1}\cdots D_{n-1}}\cdots\left(\beta_{n-1,n}\right)^{r^{2}D_{n-2}D_{n-1}}\right]^{D_{n-1}}\right.\label{eq:polyDegree}\\
 &  & \left.\beta_{nm}^{rD_{n-1}D_{m-1}}\left[\left(w_{m,m-1}\right)^{r^{2}D_{m-1}D_{m-2}}\cdots w_{21}^{r^{m}D_{m-1}\cdots D_{0}}\right]^{D_{m-1}}\right\} ,\nonumber 
\end{eqnarray}
where we denote the matrix of the local terms between sites $k$ and
$k+1$ by $\beta_{k,k+1}$ if we are marching from the left and by
$w_{k,k+1}$ if we are marching from the right. For example $\left(\beta_{12}\right)^{r^{n}D_{0}D_{1}\cdots D_{n-1}}$
means that the polynomial dependence of entries of $X$ on the first
set of local terms from left are homogeneous of degree $r^{n}D_{0}D_{1}\cdots D_{n-1}$
with respect to the elements of the random local terms between the
first and the second sites.

The conjecture can be answered if the lower bound on $\mbox{rank}\left(X\right)$
is found. In particular, can one use the fact that the entries are
very high powers of random parameters to argue in favor of an effective
genericity for the matrix $X$.

\section{What Does Algebraic Geometry Buy You?}

Though the Hamiltonian does have generic local interactions, it is
highly non-generic. In particular, it has to obey the restrictive
tensor product structure. Moreover, the number of random parameters
scale linearly with the system's size and are no more than $rd^{2}N$,
whereas the size of the Hamiltonian is $d^{N}$. This prevents one
from utilizing the techniques of Algebraic Geometry in any direct
way to prove statements about the entanglement of all the ground states.
However, one can use basic ideas of Algebraic Geometry \cite{Hayden_Generic,cubitt}
to show that among the many ground states there is at least a highly
entangled state \cite{cubitt}. We can consider the space of all the
ground states and use the following proposition given in \cite[prop. 10]{cubitt}. 
\begin{prop*}
Every bipartite system $\mathbb{C}^{D_{n}}\otimes\mathbb{C}^{D_{m}}$
has a subspace $\mathcal{S}$ of Schmidt rank $\ge\chi$, and of dimension
$\mbox{dim}\left(\mathcal{S}\right)=(D_{n}\lyxmathsym{\textminus}\chi+1)(D_{m}\lyxmathsym{\textminus}\chi+1)$
\end{prop*}
It is straighforward to calculate $\chi$ needed for the space of
solutions with $\mbox{dim}\left(\mathcal{S}\right)=D_{N}$ to have
an intersection when a cut is made in the middle $D_{m}=D_{n}=D_{\frac{N}{2}}$.
In the previous chapter we obtained the functional dependence of $D_{n}$
on $n$

\begin{eqnarray*}
D_{n} & \sim & f^{n}\left(1-\frac{g}{f}\right)^{-1}\qquad r<\frac{d^{2}}{4}\\
D_{n} & = & \left(\frac{d}{2}\right)^{n}(n+1)\qquad r=\frac{d^{2}}{4}
\end{eqnarray*}
with

\[
f=\frac{d}{2}+\sqrt{\frac{d^{2}}{4}-r},\quad g=\frac{d}{2}-\sqrt{\frac{d^{2}}{4}-r}.
\]
For large $N$ one gets $\chi\sim f^{N/2}$ when $r<\frac{d^{2}}{4}$
and $\chi\sim N\left(\frac{d}{2}\right)^{N/2}$ when $r=\frac{d^{2}}{4}$.
It is not surprising to see that there is at least one solutions with
a high Schmidt rank; the conjecture requires a stronger result. 

\chapter{Examples of Quantum 2-SAT and Combinatorial Techniques}

Here I describe two examples of quantum $2$-SAT (both of which are
due to Sergey Bravyi \cite{Bravyi_Notes}) on a chain of length $2n$
with three-dimensional ($d=3$) and four-dimensional ($d=4$) qudits,
both of which have unique highly entangled ground states. In the $d=3$
case, the ground state has a Schmidt rank that grows linearly with
the number of sites and in $d=4$ case, the ground state has Schmidt
rank $\chi=2^{n+1}-1$. Below I show that the entanglement entropies
for $d=3$ and $d=4$ cases are $H=\frac{1}{2}\log n+0.645$ and $H=\left(\sqrt{2}-1\right)n+\frac{1}{2}\log_{2}n+\frac{1}{2}\log_{2}\left(\frac{\sqrt{2}\pi}{3+2\sqrt{2}}\right)$
respectively. I provide numerical simulations to verify these formulas.
In the next chapter we prove that the gap closes polynomially in the
$d=3$ case. As far as we know the technique we use for proving the
gap in this case is new. The $d=3$ example below gives the combinatorial
background for the next chapter and $d=4$ example does not appear
elsewhere.

\section{Quantum 2-SAT for balanced parentheses $(d=3)$}

We describe an example of a frustration-free $2$-local Hamiltonian
on a chain of $n$ qutrits which has a unique highly entangled ground
state $\psi_{0}$. More precisely, if one cuts the chain in the middle,
the Schmidt rank of $\psi_{0}$ is $\chi\approx n/2$, while the entanglement
entropy is $S\approx(1/2)\log_{2}{n}$. The Hamiltonian is likely
to have a polynomial spectral gap.

Define a $3$-letter alphabet (see Fig. \ref{fig:qutritStates })
\[
\Sigma=\{l,r,0\}.
\]

\begin{figure}
\begin{centering}
\includegraphics[scale=0.4]{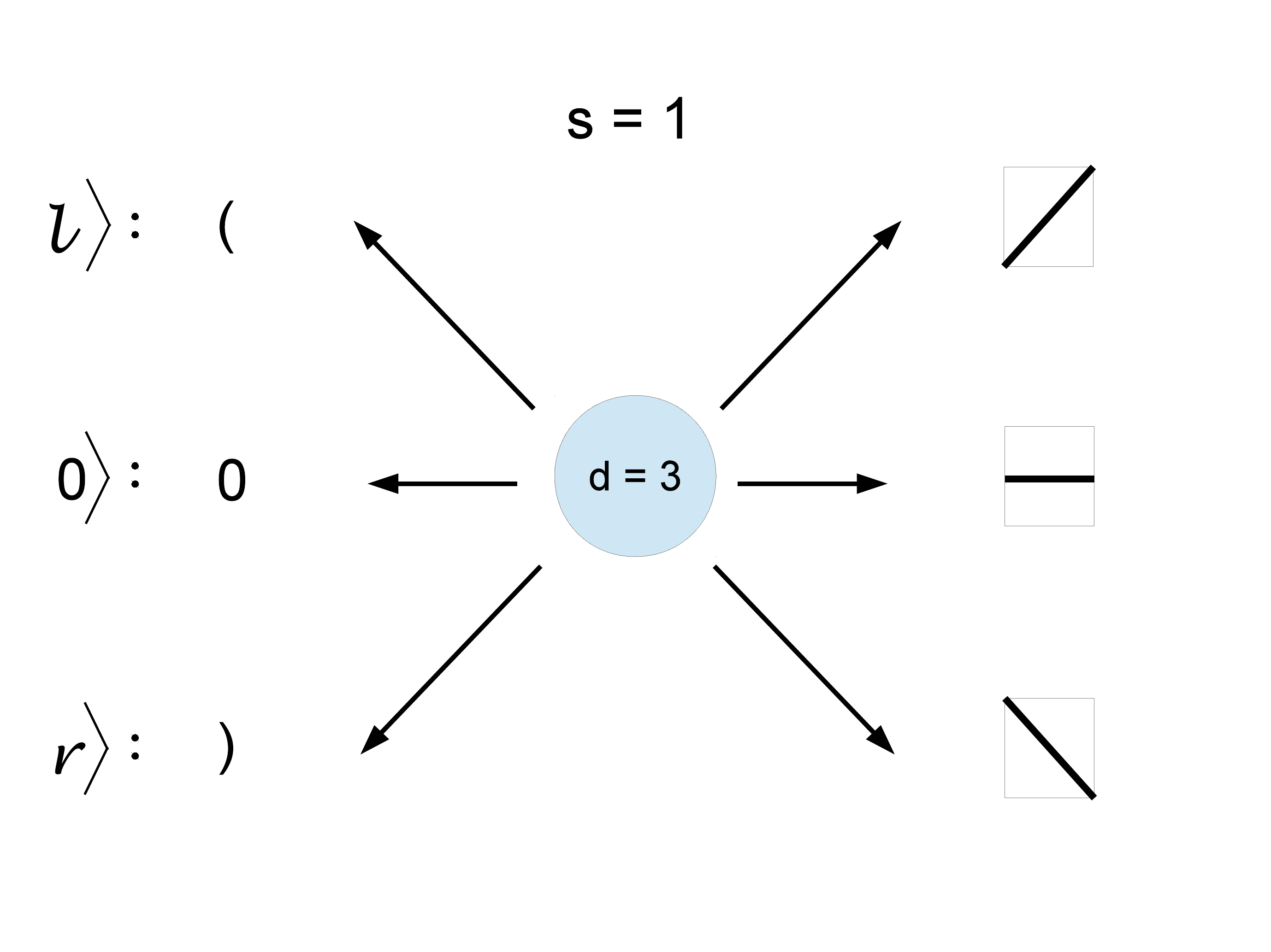}
\par\end{centering}

\caption{\label{fig:qutritStates }The states of the qudit.}

\end{figure}

We shall identify $l$ and $r$ with the left and right brackets respectively,
that is, $l\equiv[$ and $r\equiv]$. Let us say that a string $s\in\Sigma^{n}$
is {\em balanced} iff after removing all zeroes from $s$ one gets
a balanced sequence of brackets. More formally, for any $s\in\Sigma^{n}$
let $L_{i}(s)$ and $R_{i}(s)$ be the number of left and right brackets
in $s$ located in the interval $1,\ldots,i$.
\begin{defn}
A string $s\in\Sigma^{n}$ is called balanced iff $L_{i}(s)\ge R_{i}(s)$
for all $i=1,\ldots,n-1$ and $L_{n}(s)=R_{n}(s)$.
\end{defn}
For example, for $n=2$ there are only two balanced strings: $00$
and $lr$. For $n=3$ there are four balanced strings: $000$, $0lr$,
$lr0$, and $l0r$. For $n=4$ there are nine balanced strings:

\[
\begin{array}{cc}
0000 & l00r\\
00lr & l0r0\\
0l0r & lr00\\
0lr0 & llrr\\
 & lrlr
\end{array}
\]
 We would like to construct a Hamiltonian whose unique ground state
is the uniform superposition of all balanced strings,

\[
|\psi\rangle=\sum_{s\in\mathcal{B}}|s\rangle.
\]

First we need to find a more local description of balanced strings.
We shall say that a pair of strings $s,t\in\Sigma^{n}$ are equivalent,
$s\sim t$, if one can obtain $s$ from $t$ by a sequence of local
moves

\begin{equation}
00\longleftrightarrow lr,\quad0l\longleftrightarrow l0,\quad0r\longleftrightarrow r0.\label{eq:moves}
\end{equation}
applied to pairs of consecutive letters. For any integers $p,q\ge0$
let $u_{p,q}\in\Sigma^{n}$ be the string that has $p$ leading $r$'s
and $q$ tailing $l$'s, that is, 
\[
u_{p,q}\equiv\underbrace{r\ldots r}_{p}\underbrace{0\ldots0}_{n-p-q}\underbrace{l\ldots l}_{q}.
\]
 In particular, $u_{0,0}\equiv0^{n}$. 
\begin{prop*}
A string $s\in\Sigma^{n}$ is balanced iff it is equivalent to the
all-zeros string, $s\sim0^{n}$. Any string $s\in\Sigma^{n}$ is equivalent
to one and only one string $u_{p,q}$ for some integers $p,q\ge0$. \end{prop*}
\begin{proof}
Indeed, applying the local moves Eq.~(\ref{eq:moves}) we can make
sure that $s$ does not contain substrings $lr$ and $l0\ldots0r$.
It means that if $s$ contains at least one $l$, then all letters
on the right of $l$ are $l$ or $0$. Similarly, if $s$ contains
at least one $r$, then all letters on the left of $r$ are $r$ or
$0$. Since we can swap $0$ with any other letter by the local moves,
$s$ is equivalent to $u_{p,q}$ for some $p,q$. It remains to show
that different strings $u_{p,q}$ are not equivalent to each other.
Indeed, suppose $u_{p,q}\sim u_{p',q'}$ such that $p\ge p'$. Then
$R_{p}(s)-L_{p}(s)\le p'$ for any string $s$ equivalent to $u_{p',q'}$.
This is a contradiction unless $p=p'$. Similarly one shows that $q=q'$.
See Figure 
\end{proof}
\begin{figure}
\begin{centering}
\includegraphics[scale=0.4]{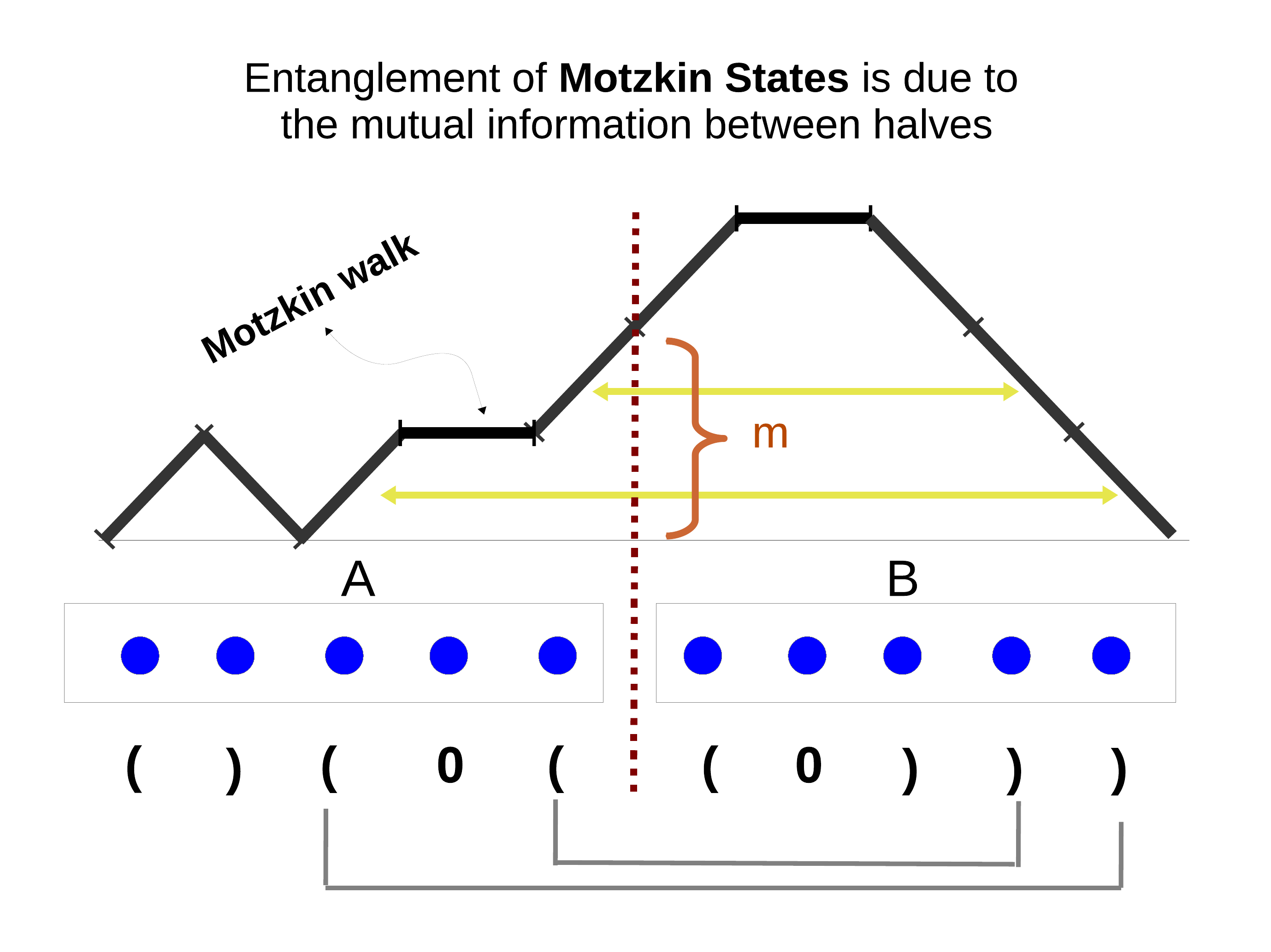}
\par\end{centering}

\caption{A state in the ground state for $d=3$ example. The high amount of
entanglement is due to the high mutual information between the two
halves of the chain.}

\end{figure}

It follows that the set of all strings $\Sigma^{n}$ is a disjoint
union of the equivalence classes $[u_{p,q}]$. We shall now introduce
projectors $Q$ that ``implement'' the local moves Eq. (\ref{eq:moves})
and a frustration-free Hamiltonian

\begin{equation}
H^{\mbox{prop}}=\sum_{j=1}^{n-1}Q_{j,j+1}\label{eq:Hprop}
\end{equation}
such that ground states of $H^{\mbox{prop}}$ are

\begin{equation}
|\psi_{p,q}\rangle=\sum_{s\sim u_{p,q}}|s\rangle\label{eq:psipq}
\end{equation}

Define quantum states $|A\rangle,|B\rangle,|C\rangle\in\mathbb{C}^{3}\otimes\mathbb{C}^{3}$
as 
\[
|A\rangle=|00\rangle-|lr\rangle,
\]
 
\[
|B\rangle=|0r\rangle-|r0\rangle,
\]
 and 
\[
|C\rangle=|0l\rangle-|l0\rangle.
\]
Define a projector

\[
Q=\frac{1}{2}\left(|A\rangle\langle A|+|B\rangle\langle B|+|C\rangle\langle C|\right)
\]
If a state $|\psi\rangle$ obeys $Q_{j,j+1}\,|\psi\rangle=0$ for
all $1\le j\le n-1$ then $\langle s|\psi\rangle=\langle s'|\psi\rangle$
for any pair of equivalent strings $s,s'$. Hence $H^{\mbox{prop}}$
is indeed frustration free and its ground subspace is spanned by the
states $\psi_{p,q}$.

How can we exclude the unwanted ground states $\psi_{p,q}$ with $p\ne0$
and/or $q\ne0$ ? The key observation is that the equivalence class
$\mathcal{B}_{n}=[u_{0,0}]$ is the only class in which every string
$s$ satisfies $s_{1}\ne r$ and $s_{n}\ne l$. Hence we can modify
our Hamiltonian as 

\[
H=H^{\mbox{prop}}+|r\rangle\langle r|_{1}+|l\rangle\langle l|_{n}.
\]
Now $H$ is a frustration-free Hamiltonian with the unique ground
state $|\psi_{0}\rangle$.

Let us now show that the Schmidt rank of $|\psi_{n}\rangle$ grows
linearly with $n$. Consider a bipartition $\{1,\ldots,n\}=AB$ where
$A$ and $B$ is the left and the right halves of the chain (we assume
for simplicity that $n$ is even). For any string $s\in\Sigma^{n}$
let $s_{A}$ and $s_{B}$ be the restrictions of $s$ onto $A$ and
$B$. If $s$ is a balanced string, one must have $s_{A}\sim u_{0,p}$
and $s_{B}\sim u_{p,0}$ for some $0\le p\le n/2$, since each unbalanced
left bracket in $A$ must have a matching unbalanced right bracket
in $B$. It follows that the Schmidt decomposition of $\psi_{0}$
can be written (ignoring the normalization) as

\[
|\psi_{0}\rangle=\sum_{p=0}^{n/2}|\psi_{0,p}\rangle_{A}\otimes|\psi_{p,0}\rangle_{B},
\]
where the states $\psi_{p,q}$ are defined in Eq.~(\ref{eq:psipq}).
Therefore, the reduced density matrix of $A$ has rank $\chi=1+n/2$.

Numerical simulation shows that the entanglement entropy of $A$ grows
logarithmically, $S(A)\approx(1/2)\log_{2}(n)$, while the spectral
gap of $H$ decays polynomially, $\Delta\sim1/n^{3}$ (see next chapter).

\subsection{Entanglement entropy }

To have a balanced string, we need to have an even number of slots
available for $r$'s and $l$'s after $0$'s have been removed. First,
recall that the Catalan numbers are 
\begin{eqnarray*}
C_{k} & = & \frac{1}{k+1}\left(\begin{array}{c}
2k\\
k
\end{array}\right),\quad k=0,1,\cdots\\
 & = & \left\{ 1,1,2,5,14,\cdots\right\} ,
\end{eqnarray*}
which among many other things, count the number of ways one can deposite
and withdraw a dollar a day such that after $k$ days one starts and
ends with zero dollars without ever going negative. It is immediate
to see that a string that has $2n-2k$ zeros has $w_{k}$ number of
configurations

\begin{eqnarray}
w_{k} & = & C_{k}\left(\begin{array}{c}
2n\\
2k
\end{array}\right)=\frac{1}{k+1}\left(\begin{array}{c}
2k\\
k
\end{array}\right)\left(\begin{array}{c}
2n\\
2k
\end{array}\right),\mbox{ with}\label{eq:qutrit}\\
k & = & \left\{ 1,\cdots,n\right\} \nonumber 
\end{eqnarray}
This count (with zeros taken into account) is also known as \textit{Motzkin
Numbers}. Let $Z_{2n}=\sum_{k=1}^{n}w_{k}$, the ground state reads

\[
|\psi_{g}\rangle=\frac{1}{\sqrt{Z_{2n}}}\sum_{i=1}^{Z_{2n}}|i\rangle,
\]
where $i$ is any of the walks that start and end with height zero
without ever going negative on $2n$ qutrits allowing only $\nearrow$,
$\searrow$, $\longrightarrow$ moves. 

In order to calculate the entanglement entropy, we need to consider
the correlation between the left and the right when a cut is made
at some arbitrary bond. For simplicity we put the cut in the middle,
i.e., at $n$. Let us define $\mathcal{L}_{m}$ to be the space of
all states with $m$ excess left parentheses on the first $n$ qutrits;
moreover $\mathcal{R}_{m}$ would be the space of all states with
$m$ excess right parentheses on the remaining $n$ qutrits: 

\begin{eqnarray*}
\mathcal{L}_{m} & = & \left\{ |s_{m}\rangle_{1\cdots n}|s_{m}\in\left\{ 0,r,l\right\} ^{n}\mbox{ with }m\mbox{ excess }l\mbox{'s}\right\} \\
\mathcal{R}_{m} & = & \left\{ |s_{m}\rangle_{n+1\cdots2n}|s_{m}\in\left\{ 0,r,l\right\} ^{n}\mbox{ with }m\mbox{ excess }r\mbox{'s}\right\} 
\end{eqnarray*}

We want the states of the form

\[
|\psi\rangle=\frac{1}{\sqrt{N}}\sum_{m=0}^{n}\left(\sum_{|x\rangle\in\mathcal{L}_{m}}|x\rangle\right)\left(\sum_{|y\rangle\in\mathcal{R}_{m}}|y\rangle\right).
\]
Tracing over the right $n$ qutrits we have

\[
\mbox{Tr}_{R}|\psi\rangle\langle\psi|=\frac{1}{N}\sum_{m=0}^{n}M_{m,n}\left(\sum_{|x\rangle\in\mathcal{L}_{m}}|x\rangle\right)\left(\sum_{|x\rangle\in\mathcal{L}_{m}}\langle x|\right).
\]

Recalling that we can have $0$ or $l$ and $r$, the number of states
in $\mathcal{L}_{m}$ denoted by $M_{m}$ become 

\begin{eqnarray}
M_{m,n} & = & \sum_{k=0}^{n-m}\left(\begin{array}{c}
n\\
k
\end{array}\right)\left(\mbox{\# nn walks of height \ensuremath{m\;}on remaining \ensuremath{\left(\ensuremath{n-k}\right)} qutrits}\right)\label{eq:MiWord}\\
N & = & \sum_{m=0}^{n}M_{m,n}^{2}.
\end{eqnarray}
where nn means none-negative. The number of walks such that we start
with zero and end with zero without going negative is given by the
Catalan numbers. The same problem but ending with a positive height
is given by a theorem, originally due to D. André (1887), the so called
Ballot problem \cite[p. 8]{Naryana}:
\begin{thm*}
(D. André 1887) Let $a,b$ be integers satisfying $1\leq b\leq a$.
The number of lattice paths $\mathcal{N}\left(p\right)$ joining the
origina $O$ to the point $\left(a,b\right)$ and not touching the
diagonal $x=y$ except at O is given by 

\[
\mathcal{N}\left(p\right)=\frac{a-b}{a+b}\left(\begin{array}{c}
a+b\\
b
\end{array}\right)
\]
 In other words, given a ballot at the end of which candidates $P$,
$Q$ obtain $a$, $b$ votes respectively, the probability that $P$
leads $Q$ throughout the counting of votes is $\frac{a-b}{a+b}$
.
\end{thm*}
First note that $a=b+1$ gives the Catalan numbers. For us $a+b=m-k+1$
and $a-b=m+1$. Using this in Eq. \ref{eq:MiWord}$ $ and taking
care of the parity

\begin{eqnarray}
M_{n,m} & = & \sum_{k=0}^{n-m}\frac{m+1}{n-k+1}\left(\begin{array}{c}
n\\
k
\end{array}\right)\left(\begin{array}{c}
n-k+1\\
\frac{1}{2}\left(n-k-m\right)
\end{array}\right)\nonumber \\
 & \underset{=}{k\rightarrow n-m-2i} & \sum_{i=0}^{\left(n-m\right)/2}\binom{n}{2i+m}\cdot\left\{ \binom{2i+m}{i}-\binom{2i+m}{i-1}\right\} \label{eq:Mm_d=00003D3_start}\\
 & = & \sum_{i\geq0}\binom{n}{2i+m}\cdot\left\{ \binom{2i+m}{i}-\binom{2i+m}{i-1}\right\} \\
 & = & n!\left(m+1\right)\sum_{i\geq0}\frac{1}{\left(i+m+1\right)!i!\left(n-2i-m\right)!}\label{eq:Mm_d=00003D3}\\
\end{eqnarray}

To check Eq. \ref{eq:Mm_d=00003D3_start}, we see that $M_{n}=1$,
corresponding to all left parentheses and 

\[
M_{0}=\sum_{k=0}^{n}\frac{1}{n-k+1}\left(\begin{array}{c}
n\\
k
\end{array}\right)\left(\begin{array}{c}
n-k+1\\
\frac{1}{2}\left(n-k\right)
\end{array}\right)=\sum_{k=0}^{n}\left(\begin{array}{c}
n\\
k
\end{array}\right)C_{\frac{1}{2}\left(n-k\right)}
\]
as expected. This count is also known as \textit{Motzkin triangles}
\cite[p. 4]{motzkin_tri}. Consequently the Schmidt numbers become
$p_{m}\equiv\frac{M_{m}^{2}}{N}$ and the entanglement entropy is
$H\left(n\right)=-\sum_{m=0}^{n}p_{m}\log_{2}p_{m}$. Before making
approximations, using Maple one can express

\[
M_{m,n}=\frac{n!\left(m+1\right)\mbox{ }_{2}F_{1}\left([-\frac{1}{2}\left(n-m\right),-\frac{1}{2}\left(n-m-1\right)],[m+2],4\right)}{\Gamma\left(m+1\right)\Gamma\left(n-m+1\right)}
\]
where $\mbox{}_{2}F_{1}$ denotes Hypergeometric function. 

For simplicity we have made the cut in the middle and will confine
to this restriction below. However, more generally, one can place
the cut at site $1\le h\le2n-1$ then 

\[
|\psi\rangle=\frac{1}{\sqrt{N}}\sum_{m=0}^{\min\left(h,2n-h\right)}\left(\sum_{|x\rangle\in\mathcal{L}_{m}}|x\rangle\right)\left(\sum_{|y\rangle\in\mathcal{R}_{m}}|y\rangle\right)
\]
where

\begin{eqnarray*}
\mathcal{L}_{m} & = & \left\{ |s_{m}\rangle_{1\cdots h}|s_{m}\in\left\{ 0,r,l\right\} ^{h}\mbox{ with }m\mbox{ excess }l\mbox{'s}\right\} \\
\mathcal{R}_{m} & = & \left\{ |s_{m}\rangle_{h+1\cdots2n}|s_{m}\in\left\{ 0,r,l\right\} ^{2n-h}\mbox{ with }m\mbox{ excess }r\mbox{'s}\right\} 
\end{eqnarray*}
and Schmidt numbers would become $p_{m}=\frac{M_{h,m}M_{2n-h,m}}{N}$,
where $M_{h,m}$ is defined as above but of height $m$ on $h$ sites,
similarly for $M_{2n-h,m}$. The normalization being $N=\sum_{m=0}^{\min\left(h,2n-h\right)}M_{h,m}M_{2n-h,m}$.

Now we analyze the sum given by Eq. \ref{eq:Mm_d=00003D3} carefully,
\[
M_{n,m}=(m+1)\sum_{i\ge0}\frac{n!}{(i+m+1)!i!(n-2i-m)!}
\]
First, though, let's do a little calculation. We will analyze a trinomial
coefficient, where $x+y+z=0$.

We first use Stirling's formula to get 
\begin{eqnarray*}
{n \choose \frac{n}{3}+x\ \frac{n}{3}+y\ \frac{n}{3}+z} & = & \left(\frac{54\pi n}{8\pi^{3}(n+3x)(n+3y)(n+3z)}\right)^{1/2}\cdot\\
 &  & \left(\frac{n}{n+3x}\right)^{n/3+x}\left(\frac{n}{n+3y}\right)^{n/3+y}\left(\frac{n}{n+3z}\right)^{n/3+z}3^{n}.
\end{eqnarray*}

Let's expand this by saying 
\begin{eqnarray*}
\left(\frac{n}{n+3x}\right)^{n/3+x} & = & \exp\left(-\left(\frac{n}{3}+x\right)\ln\left(1+3\frac{x}{n}\right)\right)\\
 & \approx & \exp\left(-\frac{n}{3}\left(3\frac{x}{n}-\frac{1}{2}\cdot9\frac{x^{2}}{n^{2}}\right)-3\frac{x^{2}}{n}\right)\\
 & = & \exp\left(-x-\frac{3}{2}\frac{x^{2}}{n}\right).
\end{eqnarray*}
 Thus, since $x+y+z=0$, we get 
\[
{n \choose \frac{n}{3}+x\ \ \frac{n}{3}+y\ \ \frac{n}{3}+z}\approx\frac{3\sqrt{3}}{2\pi}\sqrt{\frac{n}{(n+3x)(n+3y)(n+3z)}}\exp\left(-\frac{3}{2}\frac{x^{2}+y^{2}+z^{2}}{n}\right)3^{n}.
\]

Now, we take the formula for $M_{n}$, let $m=\alpha\sqrt{n}$ and
$i=\frac{n}{3}+\beta\sqrt{n}$, and use this approximation. The term
inside the square root is approximately $1/n^{2}$, so we make this
substitution to get 
\begin{eqnarray*}
M_{n,m,i} & = & \frac{(m+1)}{n+1}{n+1 \choose i+m+1\ \ i\ \ n-2i-m}\\
 & \approx & \frac{3\sqrt{3}}{2\pi n}\frac{\alpha\sqrt{n}}{n}\exp\left[\frac{-3}{2}\left(\left(\alpha+\beta\right)^{2}+\beta^{2}+\left(\alpha+2\beta\right)^{2}\right)\right]3^{n+1}\\
 & = & \frac{3\sqrt{3}}{2\pi n^{3/2}}\alpha\exp\left[\frac{-3}{2}\left(2\alpha^{2}+6\alpha\beta+6\beta^{2}\right)\right]3^{n+1}\\
 & = & \frac{3\sqrt{3}}{2\pi n^{3/2}}3^{n+1}\alpha\exp\left(-3\alpha^{2}-9\alpha\beta-9\beta^{2}\right).
\end{eqnarray*}

We need to evaluate the sum of $M_{n,m,i}$ from $i=0$ to $i=n$.
We approximate this by integrating over $i$. Since we have $i=\frac{n}{3}+\beta\sqrt{n}$,
we get $di=\sqrt{n}d\beta$, Since the maximum is near $i=\frac{n}{3}$,
we can turn this sum into an integral from $-\infty$ to $\infty$.
The integral we need to evaluate is thus 
\begin{eqnarray*}
M_{n,m} & \approx & \frac{3\sqrt{3}}{2\pi n^{3/2}}3^{n+1}\alpha\int_{-\infty}^{\infty}\exp\left(-3\alpha^{2}-9\alpha\beta-9\beta^{2}\right)\sqrt{n}d\beta\\
 & = & \frac{3\sqrt{3}}{2\pi n}3^{n+1}\alpha\int_{-\infty}^{\infty}\exp\left(-9\left(\beta-\alpha/2\right)^{2}-\frac{3}{4}\alpha^{2}\right)\sqrt{n}d\beta\\
 & = & \frac{\sqrt{3}}{2\sqrt{\pi}n}3^{n+1}\alpha\exp\left(-\frac{3}{4}\alpha^{2}\right).
\end{eqnarray*}
 This is maximized when 
\[
\frac{d}{d\alpha}\alpha\exp\left(-\frac{3}{4}\alpha^{2}\right)=0,
\]
 or $\alpha=\sqrt{2/3}$.

Now, we need to figure out the entropy of the probability distribution
proportional to $M_{n,m}^{2}$. Recalling that $m=\alpha\sqrt{n}$,
and noticing that the normalization factor cancels, the entropy is
\[
H(\{M_{n,m}^{2}\})\approx-\frac{1}{T}\sum_{m=0}^{n}\frac{m^{2}}{n}\exp\left(-\frac{3}{2}\frac{m^{2}}{n}\right)\log\left[\frac{1}{T}\frac{m^{2}}{n}\exp\left(-\frac{3}{2}\frac{m^{2}}{n}\right)\right]
\]
 where 
\[
T=\sum_{m=0}^{n}\frac{m^{2}}{n}\exp\left(-\frac{3}{2}\frac{m^{2}}{n}\right).
\]
 We can approximate the sum with an integral. Since the integrand
goes to 0 rapidly, we can extend the upper limit of integration to
$\infty$, getting 
\[
H(\{M_{n,m}^{2}\})\approx-\frac{1}{T'}\int_{0}^{\infty}\frac{m^{2}}{n}\exp\left(-\frac{3}{2}\frac{m^{2}}{n}\right)\log\left[\frac{1}{T'}\frac{m^{2}}{n}\exp\left(-\frac{3}{2}\frac{m^{2}}{n}\right)\right]dm,
\]
 where 
\[
T'=\int_{m=0}^{\infty}\frac{m^{2}}{n}\exp\left(-\frac{3}{2}\frac{m^{2}}{n}\right)dm.
\]
 Now, we can replace $m/\sqrt{n}$ by $\alpha$ again. This gives
\[
H(\{M_{n,m}^{2}\})\approx\log\sqrt{n}-\frac{1}{T''}\int_{0}^{\infty}\alpha^{2}\exp(-\frac{3}{2}\alpha^{2})\log\left[\frac{1}{T''}\alpha^{2}\exp\left(-\frac{3}{2}\alpha^{2}\right)\right]d\alpha,
\]
 where 
\[
T''=\int_{0}^{\infty}\alpha^{2}\exp(-\frac{3}{2}\alpha^{2})d\alpha.
\]
 Here the extra $\log\sqrt{n}$ term comes from the fact that $T'=\sqrt{n}T''$.
The integral is just a constant, so we can evaluate it to get 
\begin{eqnarray}
H(\{M_{n,m}^{2}\}) & \approx & \frac{1}{2}\log n-\frac{1}{2}+\gamma+\frac{1}{2}(\log2+\log\pi-\log3)\mathrm{\ \ nats}\label{eq:qutrit_theory}\\
 & \approx & \frac{1}{2}\log_{2}n+0.64466547\mathrm{\ \ bits},\nonumber 
\end{eqnarray}
 where $\gamma$ is Euler's constant. In Figure \ref{fig:Hd=00003D3}
we compare Eq. \ref{eq:qutrit_theory} with numerical evaluation of
Eq. \ref{eq:Mm_d=00003D3_start}.

\begin{figure}[H]
\begin{centering}
\includegraphics[scale=0.45]{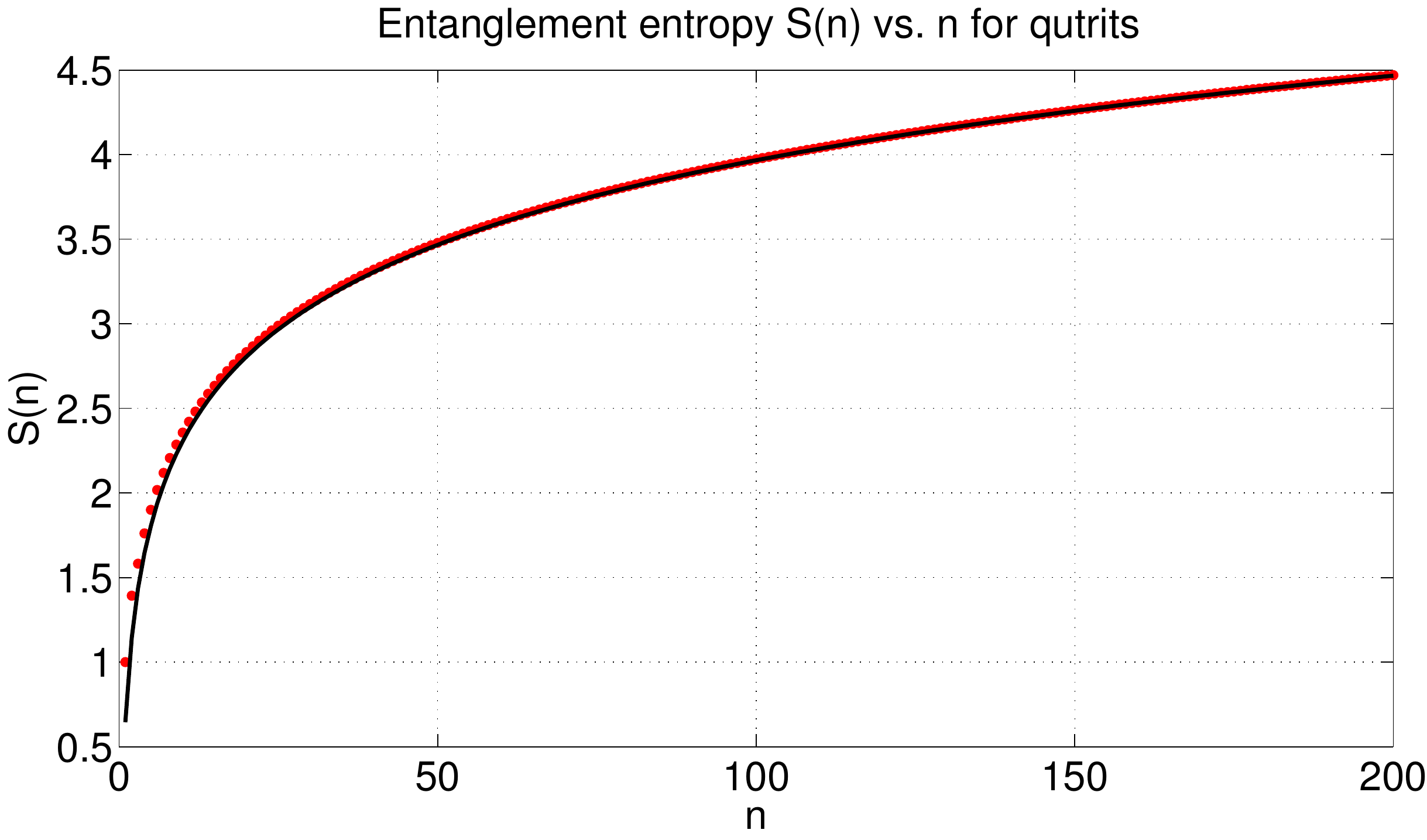}
\par\end{centering}

\centering{}\includegraphics[scale=0.48]{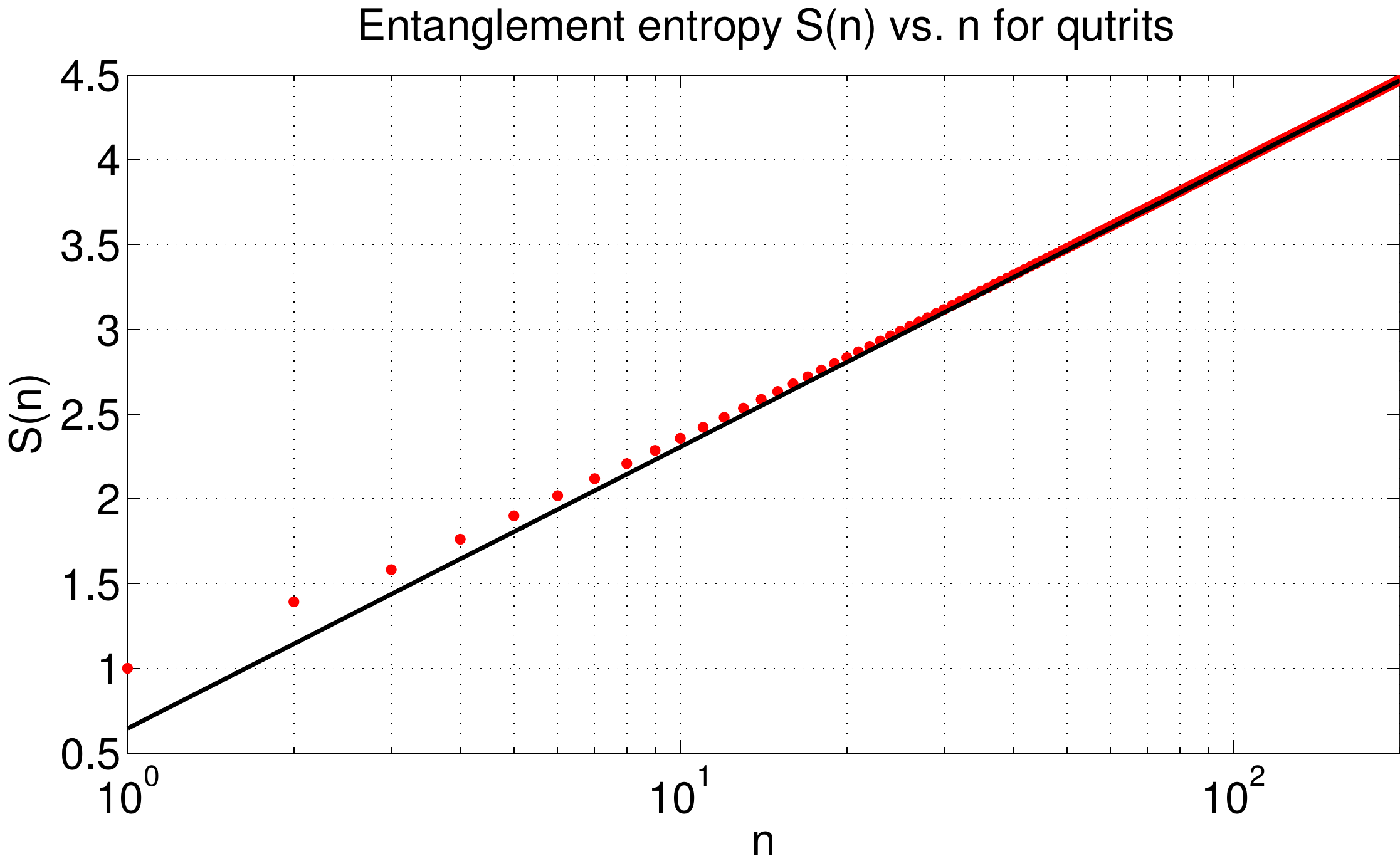}\caption{\label{fig:Hd=00003D3}Entanglement entropy for the case of qutrits}
\end{figure}

\section{Quantum 2-SAT for Mirror Symmetric States $(d=4)$}

We describe an example of a quantum $2$-SAT on a chain of $2n$ four-dimensional
($d=4$) qudits which has a unique satisfying state with the Schmidt
rank growing exponentially with $n$ \cite{Bravyi_Notes}.

Let $A$ and $B$ be the left and the right halves of the chain containing
$n$ qudits each, that is, 
\[
A=\{1,2,\ldots,n\}\quad\mbox{and}\quad B=\{n+1,n+2,\ldots,2n\}.
\]
 Basis states of each qudit will be labeled using the alphabet 
\[
\Sigma=\{0,\alpha,\beta,\gamma\}.
\]
 Let us first informally describe the idea behind the construction.
The letters $\alpha,\beta,\gamma$ represent three `particle types'
while the state $0$ represents the `vacuum'. The particles can propagate
freely through the vacuum, although they cannot pass through each
other. Furthermore, the boundary between $A$ and $B$ is impenetrable
for $\alpha$ and $\beta$ particles, while $\gamma$ particles can
propagate freely across the boundary. Let us first describe the role
of $\alpha$ and $\beta$ particles. The only place where $\alpha,\beta$
can be created or annihilated is the boundary between $A$ and $B$.
Specifically, one can create/annihilate pairs $\alpha\alpha$ or $\beta\beta$
from the vacuum at qudits $(n,n+1)$. This will create a `gas' of
$\alpha$ and $\beta$ particles such that the gas contained in $A$
is the `mirror image' of the gas contained in $B$ if one ignores
all zeroes. For example, $\alpha\alpha000\beta:0\beta\alpha0\alpha0$
represents an admissible gas of particles for $n=6$ (here : represents
the boundary between $A$ and $B$). All possible admissible configurations
of the gas will appear in superposition in the ground state. The mirror
symmetry between $A$ and $B$ will be responsible for the exponentially
large Schmidt rank. To maintain the mirror symmetry will will forbid
pairs $\alpha\beta$ and $\beta\alpha$ on the boundary between $A$
and $B$. This however is not sufficient by itself because it does
not guarantee that $A$ and $B$ contain the same number of particles.
For example, a string $000:\alpha\beta\alpha$ does not have forbidden
pairs on the boundary but it cannot propagate to any other string.
Such strings could give rise to unwanted ground states with no entanglement
between $A$ and $B$. This is where $\gamma$ particles come to play.
The rules for creation/annihilation of $\gamma$ particles are as
follows: 
\begin{itemize}
\item Every $\alpha$ or $\beta$ particle located in $A$ can emit/absorb
$\gamma$ particle on its \textit{right}, that is, $\alpha0\leftrightarrow\alpha\gamma$
and $\beta0\leftrightarrow\beta\gamma$. 
\item Every $\alpha$ or $\beta$ particle located in $B$ can emit/absorb
$\gamma$ particle on its \textit{left}, that is, $0\alpha\leftrightarrow\gamma\alpha$
and $0\beta\leftrightarrow\gamma\beta$. 
\item $\gamma$ particles are forbidden at qudits $1$ and $2n$ 
\end{itemize}
Note that $\gamma$-particles cannot tell the difference between $\alpha$
and $\beta$ particles, so they cannot maintain the mirror symmetry
between $A$ and $B$ by themselves. The purpose of $\gamma$-particles
is to ensure that the total number of $\alpha$ and $\beta$ particles
is the same in $A$ and $B$. In the above example a string $000:\alpha\beta\alpha$
can now propagate to a forbidden string: the leftmost $\alpha$-particle
in $B$ emits $\gamma$-particle obtaining $00\gamma:\alpha\beta\alpha$
which now can propagate to a forbidden string $\gamma00:\alpha\beta\alpha$
since $\gamma$-particles can move freely through the vacuum. On the
other hand, a balanced string like $00\alpha0:00\alpha0$ cannot propagate
to a forbidden string since $\gamma$-particles are confined to the
interval between the two $\alpha$ particles. We will prove below
(see Lemma~\ref{lemma:good}) that the only strings that cannot propagate
to a forbidden string are those representing a gas of $\alpha,\beta$
particles where the intervals between adjacent particles may be filled
by $0$'s and $\gamma$'s and the gas contained in $A$ is the mirror
image of the gas contained in $B$ (if one ignores all $0$'s and
all $\gamma$'s).

Let us now describe this construction more formally. We shall say
that a pair of strings $s,t\in\Sigma^{2n}$ is equivalent, $s\sim t$,
iff one can obtain $s$ from $t$ by a sequence of local moves listed
below. These moves can be applied to some pairs of consecutive qudits
$(j,j+1)$. We say that the pair is inside $A$ iff $1\le j\le n-1$.
We say that the pair is inside $B$ iff $n+1\le j\le2n-1$. We say
that the pair is on the boundary iff $j=n$.

\noindent \begin{center}
\textbf{}%
\parbox[t]{12cm}{%
\textbf{Move 1:} $\quad0\alpha\leftrightarrow\alpha0$, $0\beta\leftrightarrow\beta0$,
$0\gamma\leftrightarrow\gamma0$ (inside $A$ or inside $B$)%
}\textbf{ }\\
\textbf{}%
\parbox[t]{12cm}{%
\textbf{Move 2:} $\quad\alpha0\leftrightarrow\alpha\gamma$, $\beta0\leftrightarrow\beta\gamma$
(inside $A$ or on the boundary)%
}\textbf{ }\\
\textbf{}%
\parbox[t]{12cm}{%
\textbf{Move 3:} $\quad0\alpha\leftrightarrow\gamma\alpha$, $0\beta\leftrightarrow\gamma\beta$
(inside $B$ or on the boundary) %
}\textbf{ }\\
\textbf{}%
\parbox[t]{12cm}{%
\textbf{Move 4: }$\quad00\leftrightarrow\alpha\alpha$, $00\leftrightarrow\beta\beta$,
$0\gamma\leftrightarrow\gamma0$ (on the boundary) %
}\textbf{ }\\

\par\end{center}

In addition to these moves we shall impose several constraints:

\noindent \begin{center}
\textbf{}%
\parbox[t]{12cm}{%
\textbf{Constraint 1:} $\quad$Pairs $\alpha\beta$, $\beta\alpha$
are forbidden on the boundary.%
}\textbf{ }\\
\textbf{}%
\parbox[t]{12cm}{%
\textbf{Constraint 2:} $\quad$The first qudit of $A$ is not $\gamma$.%
}\textbf{ }\\
\textbf{}%
\parbox[t]{12cm}{%
\textbf{Constraint 3:} $\quad$The last qudit of $B$ is not $\gamma$.%
}\textbf{ }\\

\par\end{center}
\begin{defn}
A string $s\in\Sigma^{2n}$ is called good iff all strings in the
equivalence class of $s$ obey Constraints~1,2,3. Otherwise a string
$s$ is called bad. 
\end{defn}
Given a pair of strings $s,s'\in\{\alpha,\beta\}^{m}$, we shall say
that $s'$ is the mirror image of $s$ iff $s'_{i}=s_{m-i+1}$ for
all $i=1,\ldots,m$. For any string $s\in\Sigma^{2n}$ let us denote
$s_{A}$ and $s_{B}$ the restrictions of $s$ onto $A$ and $B$.
\begin{defn}
A string $s=(s_{A},s_{B})\in\Sigma^{2n}$ has mirror symmetry iff
after removing all zeroes and all $\gamma$'s the strings $s_{A}$
and $s_{B}$ become mirror images of each other. \end{defn}
\begin{lem}
\label{lemma:good} A string is good iff it has mirror symmetry, the
leftmost particle in $A$ (if any) is $\alpha$ or $\beta$, and the
rightmost particle in $B$ (if any) is $\alpha$ or $\beta$. Any
good string is equivalent to the all-zeroes string. \end{lem}
\begin{proof}
Let $s=(s_{A},s_{B})\in\Sigma^{2n}$ be any good string. It is clear
that the leftmost particle in $A$ cannot be $\gamma$ since otherwise
Move~1 would propagate $\gamma$ to the first qudit of $A$ violating
Constraint~2. By the same reason the rightmost particle in $B$ cannot
be $\gamma$. Let us show that $s$ has mirror symmetry. If both $s_{A},s_{B}$
are all-zeroes strings we are done, so let us assume that $s_{A}$
contains at least one non-zero. Suppose $s_{A}$ contains at least
one $\gamma$. Consider the left-most $\gamma$ and push it to the
left until it gets absorbed by $\alpha$ or $\beta$ (Move~2). Applying
this to every $\gamma$-particle in $s_{A}$ we can assume that $s_{A}$
contains only $0$, $\alpha$, and $\beta$. Applying Move~1 we can
transform $s_{A}$ to the following canonical form 
\[
s_{A}=(\underbrace{0\ldots0}_{n-m},x_{1},\ldots,x_{m})\quad\mbox{where \ensuremath{x=(x_{1},\ldots,x_{m})\in\{\alpha,\beta\}^{m}}}
\]
 for some $m>0$. If $s_{B}$ is all-zeroes string, we can apply Move~2
to the rightmost particle in $A$ (which is $x_{m}\in\{\alpha,\beta\}$)
to emit $\gamma$-particle, $x_{m}0\to x_{m}\gamma$. Propagating
this $\gamma$-particle to the last qudit of $B$ we violate Constraint~3.
It shows that $s_{B}$ must contain at least one $\alpha$ or $\beta$.
Using the same arguments as above, we can apply Moves~1,3 to transform
$s_{B}$ into canonical form 
\[
s_{B}=(y_{k},\ldots,y_{1},\underbrace{0\ldots0}_{n-k})\quad\mbox{where \ensuremath{y=(y_{k},\ldots,y_{1})\in\{\alpha,\beta\}^{k}}}
\]
 for some $k>0$. Using Move~4 and keeping in mind that $s$ satisfies
Constraint~1, we can consecutively annihilate all pairs $x_{i}y_{i}$
until we arrive at $s_{A}=0$ or $s_{B}=0$. However the same arguments
as above show that if $s_{A}=0$ then the sting $s_{B}$ cannot contain
$\alpha$ or $\beta$, that is, $s_{B}=0$ (and vice verse). Thus
we proved that any good string is equivalent to the all-zeroes string.
Since all moves used above preserve the mirror symmetry, we also proved
that any good string has mirror symmetry.

Conversely, suppose a string $s\in\Sigma^{2n}$ has mirror symmetry,
the leftmost particle in $A$ (if any) is $\alpha$ or $\beta$, and
the rightmost particle in $B$ (if any) is $\alpha$ or $\beta$.
Let $t$ be any string equivalent to $s$. We have to show that $t$
obeys Constraints~1,2,3. Since all Moves~1,2,3 preserve the order
of $\alpha$, $\beta$ particles in $A$ and $B$, it clear that $t$
satisfies Constraint~1. Since $\alpha$ or $\beta$ particles located
in $A$ can emit $\gamma$ particles only on their right (Move~2),
no $\gamma$ particle emitted in $A$ can violate Constraint~2. A
$\gamma$-particle emitted in $B$ or on the boundary (Move~3) can
violate Constraint~2 only if $B$ contains at least one $\alpha$
or $\beta$ particle while $A$ does not. However this contradicts
to the mirror symmetry. Thus $t$ satisfies Constraint 2. The same
arguments show that $t$ satisfies Constraint 3.
\end{proof}
Consider a state $|\psi_{2n}\rangle\in(\mathbb{C}^{4})^{\otimes2n}$
defined as the uniform superposition of all good strings, 
\[
|\psi_{2n}\rangle=\sum_{\substack{s\in\Sigma^{2n}\\
s\,\,\mathrm{is\,\, good}
}
}\;|s\rangle.
\]

\begin{lem}
The state $|\psi_{2n}\rangle$ considered as a bipartite state shared
by $A$ and $B$ has Schmidt rank 
\begin{equation}
\chi_{n}=2^{n+1}-1.\label{eq:rank}
\end{equation}
\end{lem}
\begin{proof}
Indeed, the Schmidt basis of $|\psi_{2n}\rangle$ can be easily constructed
using Lemma~\ref{lemma:good}. Choose any integer $m\in[0,n]$ and
any string $x\in\{\alpha,\beta\}^{m}$. Let $|A(m,x)\rangle$ be the
uniform superposition of all strings $s\in\Sigma^{n}$ of the form
\[
s=(Z_{0},x_{1},Z_{1},x_{2},Z_{2},\ldots,x_{m},Z_{m}),
\]
 where $Z_{0}$ is a string of zeroes, and $Z_{1},\ldots,Z_{m}$ are
arbitrary strings of zeroes and $\gamma$'s. Any of the strings $Z_{0},\ldots,Z_{m}$
can be empty. Similarly, let $|B(m,x)\rangle$ be the uniform superposition
of all strings $s\in\Sigma^{n}$ of the form 
\[
s=(Z_{m},x_{m},Z_{m-1},x_{m-1},\ldots,Z_{1},x_{1},Z_{0}),
\]
 where $Z_{0}$ is a string of zeroes, and $Z_{1},\ldots,Z_{m}$ are
arbitrary strings of zeroes and $\gamma$'s. Any of the strings $Z_{0},\ldots,Z_{m}$
can be empty. Using the characterization of good strings given by
Lemma~\ref{lemma:good} we conclude that 
\[
|\psi_{2n}\rangle=\sum_{m=0}^{n}\;\sum_{x\in\{\alpha,\beta\}^{m}}\;|A(m,x)\rangle\otimes|B(m,x)\rangle
\]
 is the Schmidt decomposition of $|\psi_{2n}\rangle$ (up to normalization
of the Schmidt basis vectors). It immediately implies Eq. (\ref{eq:rank}). 
\end{proof}
Since the set of good strings is specified by $2$-local moves and
constraints, we can specify the state $|\psi_{2n}\rangle$ by $2$-local
projectors acting on nearest-neighbor qudits. Define auxiliary states

\[
\begin{array}{c}
|M_{\alpha}\rangle\sim|0\alpha\rangle-|\alpha0\rangle\\
|M_{\beta}\rangle\sim|0\beta\rangle-|\beta0\rangle\\
|M_{\gamma}\rangle\sim|0\gamma\rangle-|\gamma0\rangle\\
|-\rangle\sim|0\rangle-|\gamma\rangle\\
|C_{\alpha}\rangle\sim|00\rangle-|\alpha\alpha\rangle\\
|C_{\beta}\rangle\sim|00\rangle-|\beta\beta\rangle
\end{array}
\]

We assume that all above states are normalized. Define a propagation
Hamiltonian $H^{prop,A}$ responsible for `implementing' Moves~1,2
for consecutive pairs of qudits inside $A$, namely

\begin{eqnarray*}
H^{\mbox{prop,}A} & = & |M_{\alpha}\rangle\langle M_{\alpha}|+|M_{\beta}\rangle\langle M_{\beta}|+|M_{\gamma}\rangle\langle M_{\gamma}|\\
 &  & +|\alpha\rangle\langle\alpha|\otimes|-\rangle\langle-|+|\beta\rangle\langle\beta|\otimes|-\rangle\langle-|
\end{eqnarray*}

Define a propagation Hamiltonian $H^{prop,B}$ responsible for `implementing'
Moves~1,3 for consecutive pairs of qudits inside $B$, namely

\begin{eqnarray*}
H^{\mbox{prop,}B} & = & |M_{\alpha}\rangle\langle M_{\alpha}|+|M_{\beta}\rangle\langle M_{\beta}|+|M_{\gamma}\rangle\langle M_{\gamma}|\\
 &  & +|-\rangle\langle-|\otimes|\alpha\rangle\langle\alpha|+|-\rangle\langle-|\otimes|\beta\rangle\langle\beta|
\end{eqnarray*}

Define a propagation Hamiltonian $H^{prop,AB}$ responsible for `implementing'
Moves~2,3,4 on the boundary, namely

\begin{eqnarray*}
H^{\mbox{prop,}AB} & = & |M_{\gamma}\rangle\langle M_{\gamma}|+|\alpha\rangle\langle\alpha|\otimes|-\rangle\langle-|+|\beta\rangle\langle\beta|\otimes|-\rangle\langle-|\\
 &  & +|-\rangle\langle-|\otimes|\alpha\rangle\langle\alpha|+|-\rangle\langle-|\otimes|\beta\rangle\langle\beta|+|C_{\alpha}\rangle\langle C_{\alpha}|+|C_{\beta}\rangle\langle C_{\beta}|
\end{eqnarray*}

This Hamiltonian acts on the pair of qudits $(n,n+1)$. Finally, define
Hamiltonians imposing Constraints~1,2,3, namely, 

\[
H^{\mbox{con,}A}=|\gamma\rangle\langle\gamma|_{1},\quad H^{\mbox{con,}B}=|\gamma\rangle\langle\gamma|_{2n},\quad H^{\mbox{con,}AB}=|\alpha\beta\rangle\langle\alpha\beta|+|\beta\alpha\rangle\langle\beta\alpha|.
\]

Here $H^{con,AB}$ acts on the pair of qudits $(n,n+1)$. 
\begin{lem}
The state $|\psi_{2n}\rangle$ is the unique state annihilated by
all the Hamiltonians $H^{prop,A}$, $H^{prop,B}$, $H^{prop,AB}$,
$H^{con,A}$, $H^{con,B}$, and $H^{con,AB}$. \end{lem}
\begin{proof}
Indeed let $H$ be Hamiltonian defined as the sum of all above Hamiltonians.
It is clear that $|\psi_{2n}\rangle$ is annihilated by $H$. Since
$H$ is a stoquastic Hamiltonian, it suffices to consider ground states
$|\psi\rangle$ with real non-negative amplitudes. If $|\psi\rangle$
has a positive amplitude on some string $s$, the propagation Hamiltonians
ensure that $|\psi\rangle$ has the same amplitude on any string equivalent
to $s$. The Hamiltonians implementing the constraints then ensure
that only good strings can appear in $|\psi\rangle$. Lemma~\ref{lemma:good}
implies that there is only one equivalence class of good strings.
Hence $H$ has unique ground state $|\psi_{2n}\rangle$. 
\end{proof}

\subsection{Entanglement entropy}

Recall that we impose three \textbf{constraints }on the states of
the $2n$ qudits:
\begin{enumerate}
\item Pairs $\alpha\beta$ and $\beta\alpha$ are forbidden at the boundary.
\item The first qudit of $A$ is not $\gamma.$
\item The last qudit of $B$ is not $\gamma.$
\end{enumerate}
The particles $\gamma$ can propagate freely through the the vacuum
state given by $0$ states.

We want to count the number of mirror symmetric states that obey the
constraints 2 and 3 (the first constraint is implied by mirror symmetry).
First let us ask: how many strings can there be in $A$ alone? Well
out of the $4^{n}$ possible strings the ones that violate constraint
2 need to be excluded (we are not worrying about $B$ yet). The complete
list of the excluded states is (each row represents a forbidden string
in $s_{A}$)

\begin{equation}
\begin{array}{ccccccc}
\gamma & \# & \# & \# & \# & : & 4^{n-1}\\
0 & \gamma & \# & \# & \# & : & 4^{n-2}\\
0 & 0 & \gamma & \# & \# & : & 4^{n-3}\\
 &  &  & \ddots\\
0 & 0 & 0 & 0 & \gamma & : & 1
\end{array}\label{eq:exclude}
\end{equation}
where $\#\in\left\{ \alpha,\beta,\gamma,0\right\} $ denotes any state
and the counts are written to the right. Therefore, the total number
of possible strings in $s_{A}$ is, 

\[
\mbox{Number of allowed strings in }s_{A}=4^{n}-\Sigma_{k=0}^{n-1}4^{k}.
\]

\begin{defn*}
($m-$dense string) A string of size $m$ is $m-$dense if it has
no $0$'s or $\gamma$'s. 
\end{defn*}
We wish to find the number of symmetric states where every string
on $A$ is $m-$dense (i.e., there are $m$ qudits on $A$ that are
not $\gamma$ or $0$). The number of mirror symmetric states becomes

\begin{eqnarray}
\left(\mbox{Count of }n-m\mbox{ particles of type }0,\gamma\mbox{ in }A\right)^{2} & {\scriptstyle \left\{ \left(\begin{array}{c}
m\\
0
\end{array}\right)+\left(\begin{array}{c}
m\\
1
\end{array}\right)+\cdots+\left(\begin{array}{c}
m\\
m
\end{array}\right)\right\} }\label{eq:essence}\\
=\left(\mbox{Count of }n-m\mbox{ particles of type }0,\gamma\mbox{ in }A\right)^{2} & 2^{m} & .
\end{eqnarray}
where as before $\left\{ \left(\begin{array}{c}
m\\
0
\end{array}\right)+\left(\begin{array}{c}
m\\
1
\end{array}\right)+\cdots+\left(\begin{array}{c}
m\\
m
\end{array}\right)\right\} $ is the number of ways that $\alpha$ and $\beta$ particles can be
positioned in $m$ slots. To find the number of allowed $m-$dense
strings in $A$ we first count all possible (unconstrained) ways of
putting $n-m$ of $0$ or $\gamma$ particles and $m$ of $\alpha$
and $\beta$ particles on the $n$ qudits. We then subtract from it
the forbidden states. The total number of ways one can have an $m-$dense
chain is (without imposing the constraints)

\begin{eqnarray*}
\left(\mbox{Number of ways to choose }n-m\mbox{ slots for }0\mbox{ and }\gamma\right) & \times\\
\left(\mbox{Number of ways to place }0,\gamma\mbox{ on the }n-m\mbox{ qudits}\right) & \times\\
\left(\mbox{Number of ways to place }\alpha,\beta\mbox{ on the remaining }m\right) & .
\end{eqnarray*}
Mathematically

\[
\mbox{Number of unconstraint }m-\mbox{dense chains}=\left(\begin{array}{c}
n\\
m
\end{array}\right)2^{n-m}2^{m}=\left(\begin{array}{c}
n\\
m
\end{array}\right)2^{n}
\]

The number of states that we need to exclude in $A$ are

\begin{equation}
{\scriptstyle \begin{array}{cccccccc}
\gamma_{1} & \# & \# & \# & \# & \cdots & \#: & 1.\left(\begin{array}{c}
n-1\\
n-m-1
\end{array}\right)\left\{ 2^{n-m-1}\right\} 2^{m}=\left(\begin{array}{c}
n-1\\
m
\end{array}\right)2^{n-1}\\
0 & \gamma_{2} & \# & \# & \# & \cdots & \#: & 1.\left(\begin{array}{c}
n-2\\
n-m-2
\end{array}\right)\left\{ 2^{n-m-2}\right\} 2^{m}=\left(\begin{array}{c}
n-2\\
m
\end{array}\right)2^{n-2}\\
 &  & \ddots &  &  &  &  & \vdots\\
0 & 0 & \gamma_{k} & \# & \# & \cdots & \#: & 1.\left(\begin{array}{c}
n-k\\
n-m-k
\end{array}\right)\left\{ 2^{n-m-k}\right\} 2^{m}=\left(\begin{array}{c}
n-k\\
m
\end{array}\right)2^{n-k}\\
 &  & \vdots &  &  &  & \vdots & \vdots\\
0 & 0 & 0 & 0 & \gamma_{n-m} & \cdots & \#: & 1.\left(\begin{array}{c}
m\\
0
\end{array}\right)\left(\begin{array}{c}
0\\
0
\end{array}\right)2^{m}=2^{m}
\end{array}}\label{eq:listExclude}
\end{equation}
where, $1\leq k\leq n-m$ is the first $k-1$ zeros followed by a
$\gamma$; $\#$ is means \textit{it can be} any state as long as
we have a total of $n-m$ of $0$'s and $\gamma$'s and $m$ of $\alpha,\beta$.
The number of states that need to be excluded are therefore

\begin{equation}
\mbox{number of }m-\mbox{dense states to exclude}=\sum_{k=1}^{n-m}\left(\begin{array}{c}
n-k\\
m
\end{array}\right)2^{n-k}.\label{eq:NumExclude}
\end{equation}
Comment: Mathematica erroneously expresses the foregoing equation
in terms of a Hypergeometric function, that has poles for integer
$n$.

\textit{In summary the number of allowed $m-$dense state on the $n$
qudits are }

\begin{equation}
\left(\begin{array}{c}
n\\
m
\end{array}\right)2^{n}-\sum_{k=1}^{n-m}\left(\begin{array}{c}
n-k\\
m
\end{array}\right)2^{n-k}.\label{eq:-1}
\end{equation}

The ground states are

\[
|\psi\rangle=\sum_{m=0}^{n}\sum_{x\in\left\{ \alpha,\beta\right\} ^{m}}|A\left(m,x\right)\rangle\otimes|B\left(m,x\right)\rangle
\]

The number of ways that $0,\gamma$ can be put on $n$ qudits are

\begin{equation}
M_{m,n}\equiv2^{n-m}\left\{ \left(\begin{array}{c}
n\\
m
\end{array}\right)-\sum_{k=1}^{n-m}\left(\begin{array}{c}
n-k\\
m
\end{array}\right)2^{-k}\right\} .\label{eq:Mm4}
\end{equation}
Consequently, the Schmidt numbers are $p_{m,n}\equiv\frac{M_{m,n}^{2}}{N}$
where $M_{m,n}$ is given by Eq. \ref{eq:Mm4} and $N\equiv\sum_{m=0}^{n}2^{m}M_{m,n}^{2}$
is the normalization constant. The entanglement entropy becomes

\begin{equation}
H\left(\left\{ p_{m,n}\right\} \right)=-\sum_{m=0}^{n}2^{m}p_{m,n}\log_{2}p_{m,n}.\label{eq:Sd=00003D4}
\end{equation}

Let us rewrite Eq. \ref{eq:Mm4} as

\begin{eqnarray}
M_{m,n} & = & 2^{n-m}\left(\begin{array}{c}
n\\
m
\end{array}\right)\label{eq:Mmd4Exact}\\
 & \times & \left\{ 1-\frac{1}{2}\left(1-\frac{m}{n}\right)-\cdots-\frac{1}{2^{n-m}}\left(1-\frac{m}{n}\right)\cdots\left(1-\frac{m}{m+1}\right)\right\} \\
 & \simeq & 2^{n-m}\left(\begin{array}{c}
n\\
m
\end{array}\right)\left\{ 1-\sum_{k=1}^{n-m}r^{k}\right\} =2^{n-m+1}\left(\begin{array}{c}
n\\
m
\end{array}\right)\frac{m}{n+m}\nonumber 
\end{eqnarray}
where $r=\frac{1}{2}\left(1-\frac{m}{n}\right)$. Next we use Stirling's
approximation $n!\sim\left(n/e\right)^{n}\sqrt{2\pi n}$ , similarly
for $m!$ and $\left(n-m\right)!$, to obtain (below all the logrithms
are in base $2$ unless stated otherwise)

\begin{eqnarray}
M_{m,n}^{2} & = & \frac{2nm}{\pi\left(n-m\right)}\frac{\exp\left[f\left(m,n\right)\right]}{\left(n+m\right)^{2}},\label{eq:Mmd=00003D4}\\
f\left(n,m\right) & \equiv & 2\left[\left(n-m\right)\log\left(2\right)+n\log n-m\log m-\left(n-m\right)\log\left(n-m\right)\right]\nonumber 
\end{eqnarray}
Let $m=\alpha n$, giving 

\begin{eqnarray}
M_{\alpha,n}^{2} & = & \frac{2\alpha}{\pi n\left(1-\alpha\right)}\frac{\exp\left[nf\left(\alpha\right)\right]}{\left(1+\alpha\right)^{2}},\nonumber \\
f\left(\alpha\right) & \equiv & 2\left[\left(1-\alpha\right)\log\left(2\right)+\log n-\alpha\log\alpha n-\left(1-\alpha\right)\log n\left(1-\alpha\right)\right]\label{eq:Mmd4-alpha}\\
 & = & 2\left[\left(\alpha-1\right)\log\left(1-\alpha\right)-\alpha\log\alpha+\left(1-\alpha\right)\log\left(2\right)\right]
\end{eqnarray}

\begin{eqnarray*}
N & = & \sum_{m=0}^{n}2^{m}M_{m,n}^{2}=\sum_{m=0}^{n}\frac{2nm\exp\left[g\left(n,m\right)\right]}{\pi\left(n-m\right)\left(n+m\right)^{2}}\\
g\left(n,m\right) & \equiv & f\left(n,m\right)+m\log2.
\end{eqnarray*}

We can approximate this sum with an integral over $\alpha$

\begin{eqnarray}
N & \simeq & \int_{0}^{1}d\alpha\frac{2\alpha\exp\left[ng\left(\alpha\right)\right]}{\pi\left(1-\alpha\right)\left(1+\alpha\right)^{2}},\label{eq:N(alpha)}\\
g\left(\alpha\right) & \equiv & f\left(\alpha\right)+\alpha\log2\nonumber \\
 & = & 2\left[\left(\alpha-1\right)\log\left(1-\alpha\right)-\alpha\log\alpha-\frac{\alpha}{2}\log2+\log\left(2\right)\right],\nonumber 
\end{eqnarray}
note that the factor of $n$ cancelled because of change of variables
from $m$ to $\alpha$. In anticipation of the steepest descent approximation
to the entanglement entropy, we evaluate

\begin{eqnarray*}
g' & \equiv & \frac{\partial g\left(\alpha\right)}{\partial\alpha}=2\left[\log\left(1-\alpha\right)-\log\alpha-\frac{1}{2}\log2\right]\\
g'' & \equiv & \frac{\partial^{2}g\left(\alpha\right)}{\partial\alpha^{2}}=2\left[\frac{1}{\alpha\left(\alpha-1\right)}\right].
\end{eqnarray*}
$g'=0\Rightarrow\alpha_{0}=\frac{1}{1+\sqrt{2}}=\sqrt{2}-1$ and $g''\left(\alpha_{0}\right)=-\sqrt{2}\left(3+2\sqrt{2}\right)$,
which implies $\alpha_{0}$ is a maximum. Let us proceed in calculating
the entanglement entropy given by $H\left(m,n\right)=-\sum_{m=0}^{n}2^{m}\frac{M_{m}^{2}}{N}\log_{2}\frac{M_{m}^{2}}{N}$
by first approximating the sum with an integral over $\alpha$ and
then performing the steepest descent approximation (in nats)

\begin{eqnarray*}
H\left(\left\{ M\left(\alpha,n\right)\right\} \right) & \simeq & -\frac{1}{N}\int_{0}^{1}d\alpha\frac{2\alpha\exp\left[ng\left(\alpha\right)\right]}{\pi\left(1-\alpha\right)\left(1+\alpha\right)^{2}}\log\frac{M^{2}\left(\alpha,n\right)}{N}\\
 & \simeq & -\log\frac{M^{2}\left(\alpha_{0},n\right)}{N}\left\{ \frac{1}{N}\int_{0}^{1}d\alpha\frac{2\alpha\exp\left[ng\left(\alpha\right)\right]}{\pi\left(1-\alpha\right)\left(1+\alpha\right)^{2}}\right\} \mbox{ }\\
 & \simeq & -\log\frac{M^{2}\left(\alpha_{0},n\right)}{N},
\end{eqnarray*}
by the definition of $N$. It remains to calculate $-\log\frac{M^{2}\left(\alpha_{0},n\right)}{N}$

\begin{eqnarray*}
-\log\frac{M^{2}\left(\alpha_{0},n\right)}{N} & = & -\log\frac{\frac{1}{n}\exp\left[nf\left(\alpha_{0}\right)\right]}{\int_{0}^{1}d\alpha\exp\left[ng\left(\alpha\right)\right]}\\
 & \simeq & -\log\frac{\frac{1}{n}\exp\left[nf\left(\alpha_{0}\right)\right]}{\int_{-\infty}^{+\infty}d\alpha\exp\left\{ n\left[g\left(\alpha_{0}\right)+\frac{1}{2}g''\left(\alpha_{0}\right)\left(\alpha-\alpha_{0}\right)^{2}\right]\right\} }\\
 & = & -\log\left\{ \frac{1}{n}\exp\left[nf\left(\alpha_{0}\right)-ng\left(\alpha_{0}\right)\right]\sqrt{\frac{n\left|g''\left(\alpha_{0}\right)\right|}{2\pi}}\right\} \\
 & = & -\log\left\{ \exp\left[-n\alpha_{0}\log2\right]\sqrt{\frac{\left(3+2\sqrt{2}\right)}{n\sqrt{2}\pi}}\right\} .
\end{eqnarray*}

Expanding this we obtain

\begin{eqnarray*}
H\left(\left\{ M\left(\alpha,n\right)\right\} \right) & \simeq & \left(\sqrt{2}-1\right)n\log2+\frac{1}{2}\log n+\frac{1}{2}\log\left(\frac{\sqrt{2}\pi}{3+2\sqrt{2}}\right)\quad\mbox{nats}\\
 & = & \left(\sqrt{2}-1\right)n+\frac{1}{2}\log_{2}n+\frac{1}{2}\log_{2}\left(\frac{\sqrt{2}\pi}{3+2\sqrt{2}}\right)\quad\mbox{bits.}
\end{eqnarray*}

\begin{figure}[H]
\centering{}\includegraphics[scale=0.4]{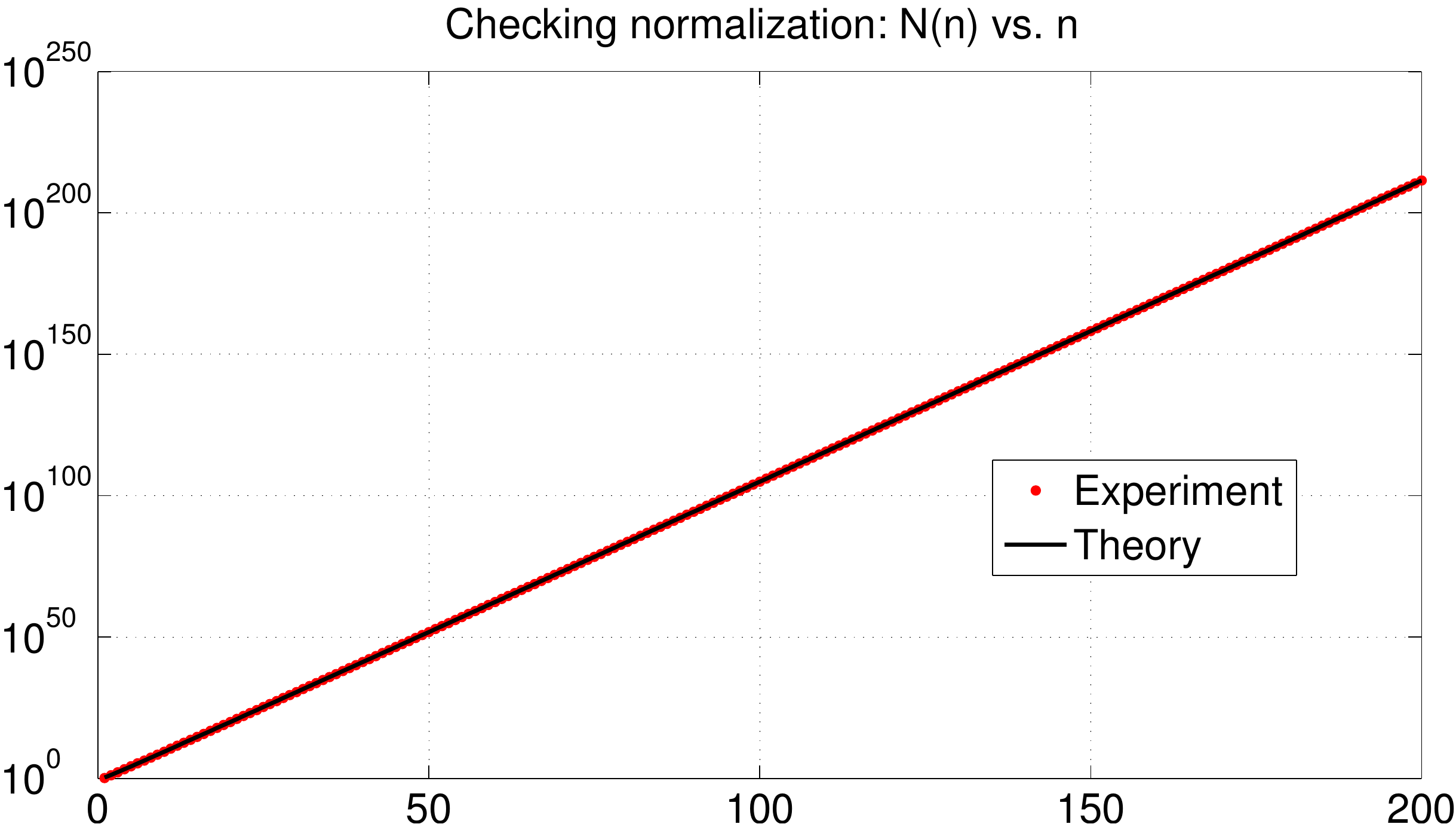}\caption{\label{fig:Nd=00003D4}Normalization as a function of $n$. }
\end{figure}

\begin{figure}[H]
\centering{}\includegraphics[scale=0.4]{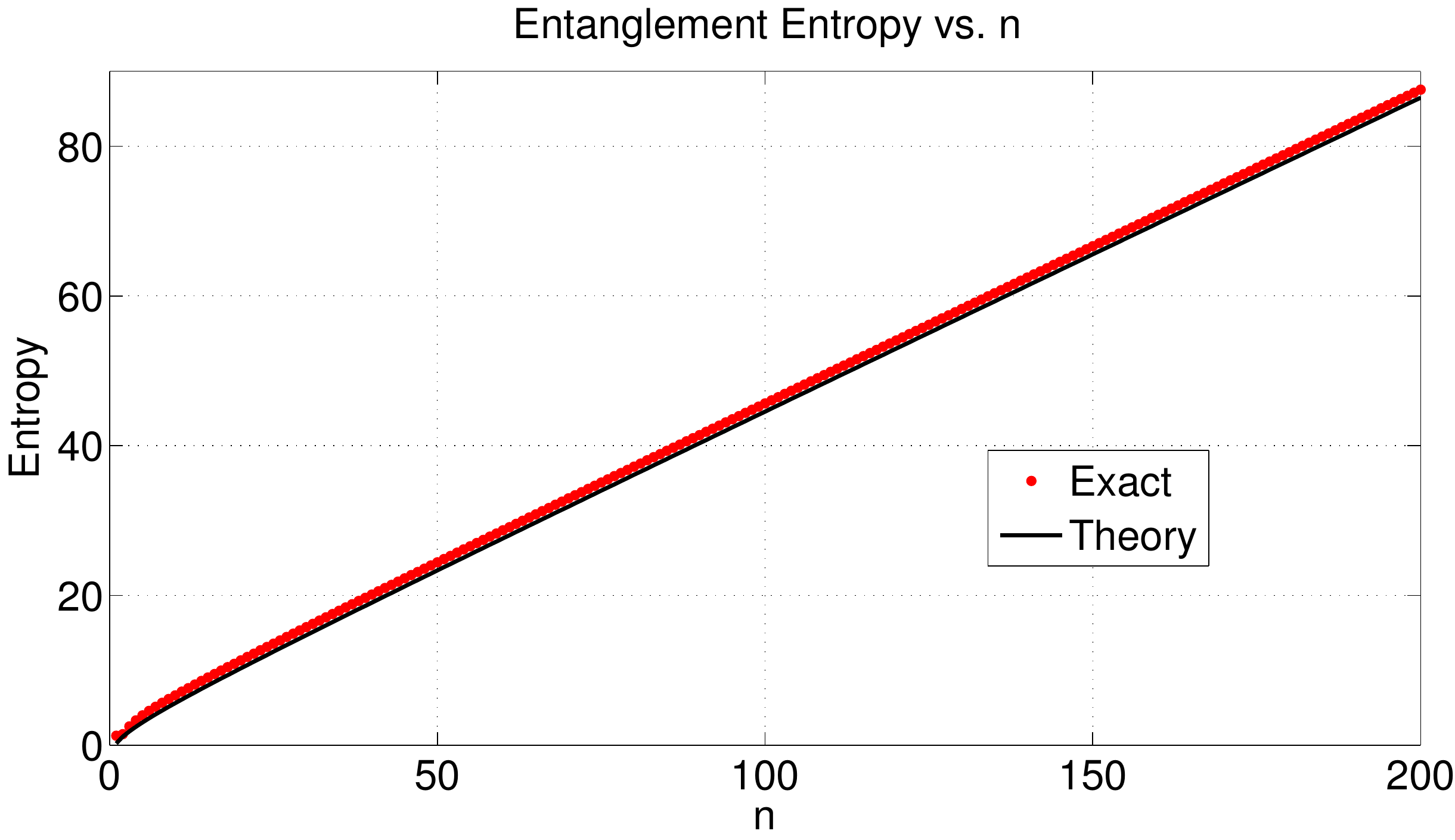}\caption{\label{fig:Entropyd=00003D4}Entropy $H\left(\left\{ M_{m}^{2}\right\} \right)$
vs. $n$ for $d=4$ case. We include the figure on left to demonstrate
the closeness of the approximation.}
\end{figure}


\chapter{Criticality Without Frustration for Quantum Spin-$1$ Chains }

In the previous chapter we showed two examples of FF qudit chains with high entanglement and introduced the mathematical techniques needed for calculating their entanglement entropies. Here we elaborate on the $d=3$ model- balanced parenthesis model. While FF spin-1/2 chains are known to have unentangled ground states, the case s=1 remains less explored. We propose the first example of a FF translation-invariant spin-1 chain that has a unique highly entangled ground state and exhibits some signatures of a critical behavior.  The rest of this chapter also appears in ~\cite{RamisMotzkin}.

\section{Motivation}
The presence of long-range entanglement in the ground states of critical spin chains
with only short-range interactions
is one of the most fascinating discoveries in the theory of quantum phase transitions~\cite{sachdev,Vidal03,Korepin04}.
It can be quantified by the scaling
law $S(L)\sim \log{L}$, where
$S(L)$ is the entanglement entropy of
a block of $L$ spins. In contrast, non-critical
spin chains characterized by a non-vanishing energy gap
obey an area law~\cite{Hastings07,Eisert08,Arad11}
asserting that $S(L)$ has a constant upper bound independent of $L$.

One can ask how stable is the long-range ground state entanglement
against small variations of Hamiltonian parameters?  The scaling theory predicts~\cite{Vidal03,Latorre03} that
a chain whose Hamiltonian is controlled by some parameter $g$
follows the  law $S(L)\sim \log{L}$
only if $L$ does not exceed the correlation length
$\xi\sim |g-g_c|^{-\nu}$, where $\nu>0$ is the critical exponent
and $g_c$ is the critical point. For larger $L$ the entropy $S(L)$
saturates at a constant value. Hence achieving the scaling $S(L)\sim \log{L}$
requires fine-tuning of the parameter $g$ with
precision scaling polynomially with $1/L$ posing a serious experimental challenge.

The stringent precision requirement described above can be partially avoided
for spin chains described by {\em frustration-free} Hamiltonians.
Well-known (non-critical) examples of such Hamiltonians are the
Heisenberg ferromagnetic chain~\cite{Koma95}, the AKLT model~\cite{AKLT87},
and parent Hamiltonians of matrix product states~\cite{spinGlass,PEPS07}.
More generally, we consider Hamiltonians of a form  $H=\sum_j g_j \Pi_{j,j+1}$, where
$\Pi_{j,j+1}$ is a projector acting on spins $j,j+1$ and
$g_j>0$ are some coefficients. The Hamiltonian is called frustration-free (FF)
if the projectors $\Pi_{j,j+1}$ have a common zero eigenvector $\psi$.
Such zero eigenvectors $\psi$ span the ground subspace of $H$.
Clearly, the ground subspace  does not depend on the
coefficients $g_j$ as long as they remain positive.
This inherent stability against variations of the Hamiltonian parameters
motivates a question  of whether FF Hamiltonians can describe critical spin chains.

In this Letter  we propose a  toy model describing a FF translation-invariant spin-$1$ chain
with open boundary conditions
 that has a unique ground state with a logarithmic
scaling of entanglement entropy and a polynomial energy gap.
Thus our FF model reproduces some of the main signatures of critical spin chains.
In contrast, it was recently shown by Chen et al~\cite{Chen04} that any FF spin-$1/2$ chain has an unentangled ground state.
Our work may also offer valuable insights for the problem of realizing
long-range entanglement in open quantum systems with an engineered dissipation.
Indeed, it was shown by Kraus et al~\cite{Kraus08} and Verstraete et al~\cite{Verstraete08}
that the ground state of a FF Hamiltonian can be represented as a unique steady state of a
dissipative process described by the Lindblad equation with local quantum jump operators. A proposal for realizing
such dissipative processes in cold atom systems has been made by Diehl et al~\cite{Diehl08}.

{\em Main results. \hspace{2mm}}
We begin by describing the ground state of our model.
The three basis states of a single spin will be identified
with a left bracket $l\equiv [$, right bracket $r\equiv \, ]$, and
an empty space represented by $0$. Hence a state of a single
spin can be written as $\alpha |0\ra + \beta |l\ra + \gamma |r\ra$
for some complex coefficients $\alpha,\beta,\gamma$.
For a chain of $n$ spins, basis states $|s\ra$ correspond to strings
$s\in \{0,l,r\}^n$.
A string $s$ is called a {\em Motzkin path}~\cite{Motzkin2,*Motzkin1}
iff (i) any initial  segment of $s$ contains at least as many $l$'s as $r$'s,
and (ii) the total number of $l$'s is equal to the total number of $r$'s.
For example, a string $lllr0rl0rr$ is a Motzkin path while $l0lrrrllr$ is not
since its initial segment $l0lrrr$ has more $r$'s than $l$'s.
By ignoring all $0$'s one can view Motzkin paths as  balanced strings of
left and right brackets.  We shall be interested in the {\em Motzkin state}
$|\calM_n\ra$ which is the uniform superposition of all
Motzkin paths of length $n$.
For example, $|\calM_2\ra \sim |00\ra+|lr\ra$,
$|\calM_3\ra \sim |000\ra+|lr0\ra + |l0r\ra + |0lr\ra$, and
\bea
|\calM_4\ra &\sim & |0000\ra + |00lr\ra + |0l0r\ra + |l00r\ra \nn \\
&& + |0lr0\ra + |l0r0\ra + |lr00\ra + |llrr\ra + |lrlr\ra. \nn
\eea
Let us first ask how entangled is the Motzkin state.
For a contiguous block of spins $A$,
let $\rho_A =\trace_{j\notin A} |\calM_n\ra\la \calM_n|$
be the reduced density matrix of $A$.
Two important measures of entanglement are
the Schmidt rank $\chi(A)$ equal to the number of non-zero
eigenvalues of $\rho_A$, and the
 entanglement entropy
$S(A)=-\trace \rho_A \log_2{\rho_A}$.
We will choose $A$ as the left half of the chain,
$A=\{1,\ldots,n/2\}$. We show that
\be
\label{ent}
\chi(A)=1+n/2 \quad \mbox{and} \quad S(A)= \frac12 \log_2{n} + c_n
\ee
where $\lim_{n\to \infty} c_n= 0.14(5)$.
The linear scaling of the Schmidt rank stems from the presence of locally unmatched
left brackets in $A$ whose matching right brackets belong to the complementary region $B=[1,n]\backslash A$.
The number of the locally unmatched brackets $m$ can vary from $0$ to $n/2$ and must
be the same in $A$ and $B$ leading to long-range
entanglement between the two halves of the chain.

Although the definition of Motzkin paths may seem very non-local,
we will show that the state $|\calM_n\ra$ can be specified by
imposing  local constraints on nearest-neighbor spins.
Let $\Pi$ be a projector onto the three-dimensional
subspace of $\CC^3\otimes \CC^3$ spanned by
states $|0l\ra-|l0\ra$, $|0r\ra-|r0\ra$, and $|00\ra-|lr\ra$.
Our main result is the following.
\begin{theorem}
\label{thm:main}
The Motzkin state $|\calM_n\ra$ is a unique
ground state with zero energy of a  frustration-free Hamiltonian
\be
\label{H}
H=|r\ra\la r|_1 + |l\ra\la l|_n+\sum_{j=1}^{n-1} \Pi_{j,j+1},
\ee
where  subscripts indicate spins acted upon by a projector.
The spectral gap\footnote{Here and below the spectral gap of a Hamiltonian
means the difference between the smallest and the second smallest eigenvalue.}
 of $H$ scales  polynomially with $1/n$.
 \end{theorem}
The theorem remains true if $H$
is modified by introducing arbitrary weights $g_j\ge 1$
for every projector in Eq.~(\ref{H}).
A polynomial lower bound on the spectral  gap of $H$
is, by far, the most difficult  part of Theorem~\ref{thm:main}.
Our proof consists of several steps.
First, we use a perturbation theory to relate the
spectrum of $H$ to the one of
an effective Hamiltonian $H_{\eff}$ acting on
Dyck paths --- balanced strings of left and right brackets~\footnote{One can regard Dyck paths
as a special case of Motzkin paths in which no `$0$' symbols are allowed.}.
This step involves successive applications of the Projection Lemma due to Kempe et al~\cite{KKR04}.
Secondly, we map $H_{\eff}$ to a stochastic matrix $P$ describing a random
walk on Dyck paths in which transitions correspond to insertions/removals of
consecutive $lr$ pairs.
The key step of the proof is to show that the random walk on Dyck paths is rapidly mixing.
Our method of proving the desired rapid mixing property employs the polyhedral description of matchings in bipartite
graphs~\cite{Schrijver}. This method appears to be new and might be interesting on its own right.
Exact diagonalization performed for short chains suggests that
the spectral gap of $H$ scales as $\Delta\sim 1/n^3$, see Fig.~\ref{fig:gap}.
Our proof gives an upper bound $\Delta=O(n^{-1/2})$ and a lower bound $\Delta=\Omega(n^{-c})$
for some $c\gg 1$.

\begin{center}
\begin{figure}[t]
\centerline{
\mbox{
 \includegraphics[height=5cm]{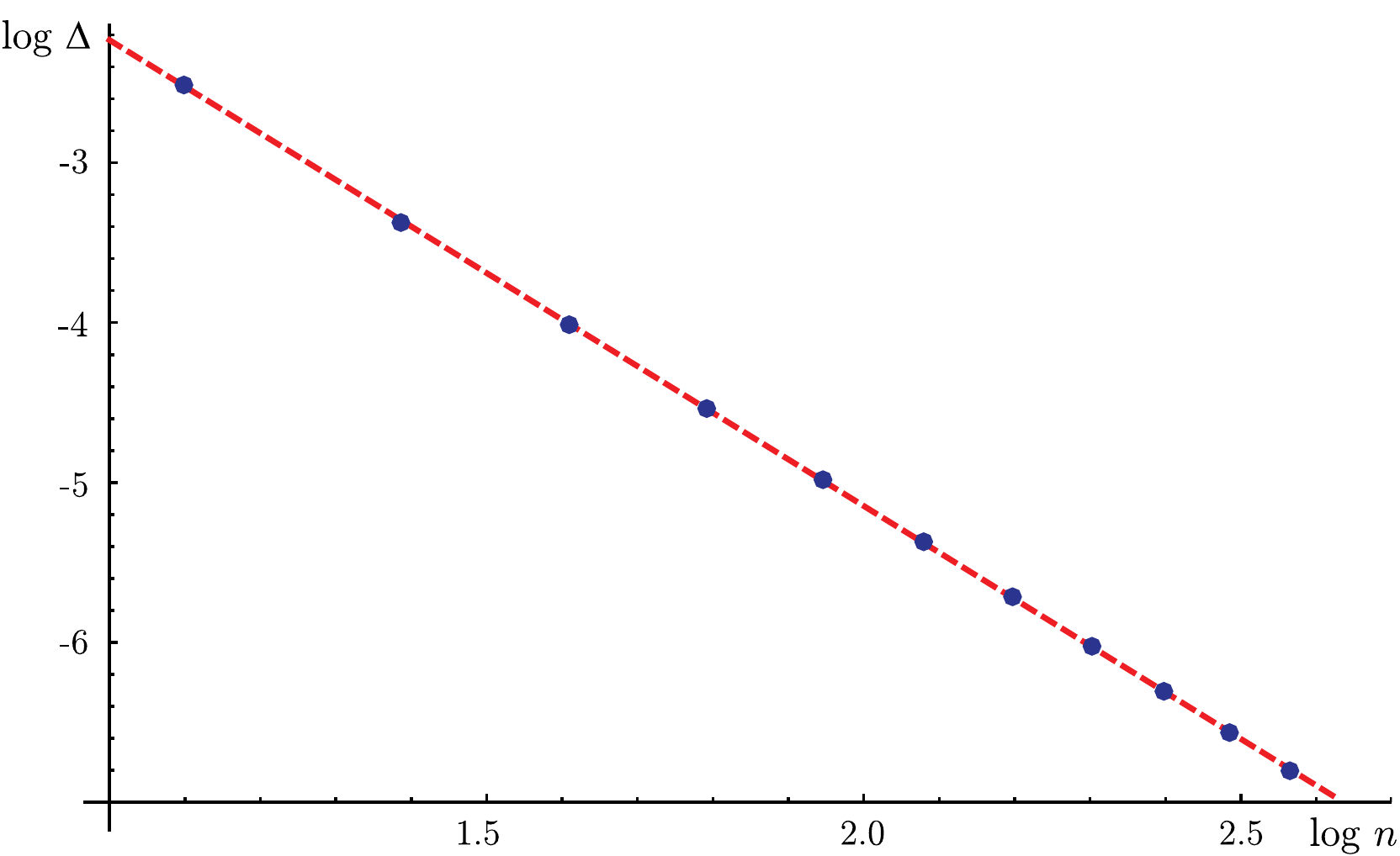}
 }}
\caption{The spectral gap $\Delta$ of the Hamiltonian $H$ defined in Eq.~(\ref{H})
for $3\le n\le 13$ obtained by the exact diagonalization. The dashed line shows a linear fit $\log{\Delta} = 0.68-2.91 \log{n}$.
Our numerics suggests that the first excited state of $H$ belongs to the subspace spanned by
strings with exactly one unmatched bracket.
}
 \label{fig:gap}
\end{figure}
\end{center}

{\em Previous work.  \hspace{2mm}}
Examples of spin chain Hamiltonians with highly entangled ground states
have been constructed by Gottesman and Hastings~\cite{Gottesman09},
and Irani~\cite{Irani09} for local dimension $d=9$ and $d=21$ respectively
(here and below $d\equiv 2s+1$).
 These models exhibit a linear scaling of the entropy $S(L)$ for some blocks of spins
while the spectral gap is polynomial in $1/n$. The model found in~\cite{Irani09}
is FF and translation-invariant.
Ref.~\cite{MFGNOS} focused on `generic' spin chains with a Hamiltonian $H=\sum_j \Pi_{j,j+1}$
where the projectors  $\Pi_{j,j+1}$ are chosen randomly with a fixed rank $r$~\footnote{Though the results
of Ref.~\cite{MFGNOS} are applicable to more general Hamiltonians, the
convenient restriction to random projectors is sufficient for
addressing the degeneracy and frustration condition.}.
The authors of~\cite{MFGNOS} identified three important regimes: (i) frustrated chains, $r >d^2/4$,
(ii) FF chains, $d\le r\le d^2/4$, and (iii) FF chains with product ground states, $r<d$.
It was conjectured in~\cite{MFGNOS} that generic FF chains in the regime $d\le r\le d^2/4$
have only highly entangled ground states with probability one.
This regime however requires local dimension $d\ge 4$.
The new model based on Motzkin paths corresponds to the case $d=r=3$ (ignoring the boundary terms)
and thus it can be frustrated by arbitrarily small deformations of the projectors making them generic.
In addition, results of~\cite{MFGNOS} imply that
examples of FF spin-$1$ chains with highly entangled ground states have measure zero
in the parameter space.
The question of whether matrix product states specified by FF parent Hamiltonians can exhibit
quantum phase transitions has been studied by Wolf et al~\cite{Wolf06}.
However, the models studied
in~\cite{Wolf06} have bounded entanglement entropy, $S(L)=O(1)$.

{\em Hamiltonian.  \hspace{2mm}}
Let us now construct a FF Hamiltonian $H$ whose unique ground state is $|\calM_n\ra$.
First we need to find a more local description of Motzkin paths.
Let $\Sigma=\{0,l,r\}$. We will say that a pair of strings $s,t\in \Sigma^n$
is equivalent, $s\sim t$, if $s$ can be obtained from $t$ by a sequence of local
moves
\be
\label{moves}
00\leftrightarrow lr, \quad 0l\leftrightarrow l0, \quad 0r\leftrightarrow r0.
\ee
These moves can be applied to any consecutive pair of letters.
For any integers $p,q\ge 0$ such that $p+q\le n$ define a string
\[
c_{p,q}\equiv \underbrace{r\ldots r}_p \underbrace{0\ldots 0}_{n-p-q} \underbrace{l\ldots l}_q.
\]
\begin{lemma}
Any string $s\in \Sigma^n$ is equivalent to one and only one string $c_{p,q}$.
A string $s\in \Sigma^n$ is a Motzkin path iff it is equivalent to the all-zeros string,
$s\sim c_{0,0}$.
\end{lemma}
\begin{proof}
Indeed, applying the local moves Eq.~(\ref{moves}) one can make sure that $s$
does not contain substrings $lr$ or $l0\ldots 0r$. If this is the case and $s$ contains at least one $l$,
then all letters to the right of $l$ are $l$ or $0$. Similarly, if $s$ contains at least one $r$, then all
letters to the left of $r$ are $r$ or $0$. Since we can swap $0$ with any other letter by the local moves,
$s$ is equivalent to $c_{p,q}$ for some $p,q$.
It remains to show that different strings $c_{p,q}$ are not equivalent to each other.
Let $L_j(s)$ and $R_j(s)$ be the number of $l$'s and $r$'s among the first $j$ letters of $s$.
Suppose $c_{p,q}\sim c_{p',q'}$ and $p\ge p'$.
Then $R_p(s)-L_p(s)\le p'$ for any string $s$ equivalent to $c_{p',q'}$.
This is a contradiction unless $p=p'$.  Similarly one shows that $q=q'$.
\end{proof}
The lemma shows that the set of all strings $\Sigma^n$ can be partitioned
into equivalence classes $C_{p,q}$, such that $C_{p,q}$ includes
all strings equivalent to $c_{p,q}$.
In other words, $s\in C_{p,q}$ iff $s$ has $p$ unmacthed right brackets
and $q$ unmatched left brackets. Accordingly, the set of Motzkin paths
$\calM_n$ coincides with the equivalence class $C_{0,0}$.

Let us now define projectors `implementing' the local moves in Eq.~(\ref{moves}).
Define normalized states
\[
|\phi\ra\sim |00\ra-|lr\ra, \quad |\psi^l\ra\sim |0l\ra -|l0\ra, \quad |\psi^r\ra\sim |0r\ra-|r0\ra
\]
and a projector $\Pi=|\phi\ra\la \phi|+|\psi^l\ra \la \psi^l| + |\psi^r\ra \la \psi^r|$.
Application of $\Pi$ to  a pair of spins $j,j+1$ will be denoted $\Pi_{j,j+1}$.
If some state $\psi$ is annihilated by every projector $\Pi_{j,j+1}$,
it must have the same amplitude on any pair of equivalent strings, that is,
$\la s|\psi\ra=\la t|\psi\ra$ whenever $s\sim t$. It follows that a Hamiltonian
$H_\sim=\sum_{j=1}^{n-1} \Pi_{j,j+1}$ is FF and the ground subspace of $H_\sim$
is spanned by pairwise orthogonal states
$|C_{p,q}\ra$, where $|C_{p,q}\ra$ is the uniform superposition of all
strings in $C_{p,q}$.
The desired Motzkin state $|\calM_n\ra=|C_{0,0}\ra$ is thus a ground state of $H_\sim$.
(It is worth mentioning that not all states $|C_{p,q}\ra$ are highly entangled.
For example, $|C_{n,0}\ra=|r\ra^{\otimes n}$ is a product state.)
How can we exclude the unwanted ground states $|C_{p,q}\ra$ with $p\ne 0$ or $q\ne 0$?
We note that $C_{0,0}$ is the only
class in which strings never start from $r$ and never end with $l$.
 Hence a modified Hamiltonian
$H=|r\ra\la r|_1 + |l\ra\la l|_n + H_{\sim}$
that penalizes strings starting from $r$ or ending with $l$ has
a unique ground state $|C_{0,0}\ra$. This proves the first part of Theorem~\ref{thm:main}.
We can also consider weighted Hamiltonians $H_\sim(g) =\sum_{j=1}^{n-1} g_j \Pi_{j,j+1}$
and $H(g)=g_0 |r\ra\la r|_1 + g_n |l\ra\la l|_n + H_{\sim}(g)$, where $g_0,\ldots,g_n\ge 1$
are arbitrary coefficients. One can easily check that the ground state of $H(g)$ does not
depend on $g$ and $H(g)\ge H$. It implies that the spectral gap of $H(g)$ is lower bounded
by the one of $H$.

{\em Entanglement entropy.  \hspace{2mm}}
We can now construct the Schmidt decomposition of the Motzkin state.
Let $A=\{1,\ldots,n/2\}$ and $B=\{n/2+1,\ldots,n\}$ be the two halves of the chain
(we assume that $n$ is even). For any string $s\in \Sigma^n$ let $s_A$ and $s_B$
be the restrictions of $s$ onto $A$ and $B$. We claim that $s$ is a Motzkin path iff
$s_A\sim c_{0,m}$ and $s_B\sim c_{m,0}$ for some $0\le m \le n/2$.
Indeed, $s_A$ ($s_B$) cannot have
unmatched right (left) brackets, while each unmatched left bracket in $s_A$ must be matched
with some unmatched right bracket in $s_B$. It follows that the Schmidt decomposition
of $|\calM_n\ra$ can be written as
\be
\label{Schmidt}
|\calM_n\ra = \sum_{m=0}^{n/2} \sqrt{p_m} \, |\hat{C}_{0,m}\ra_A  \otimes |\hat{C}_{m,0}\ra_B,
\ee
where $|\hat{C}_{p,q}\ra$ is the normalized uniform superposition
of all strings in $C_{p,q}$ and $p_m$
are the Schmidt coefficients  defined by
\be
\label{Schmidt1}
p_m=\frac{|C_{0,m}(n/2)|^2}{|C_{0,0}(n)|}.
\ee
Here we added an explicit dependence of the classes $C_{p,q}$ on $n$.
For large $n$ and $m$ one can use an
approximation  $p_m\sim m^2 \exp{(-3m^2/n)}$, see the Supplementary Material
for the proof. Note that $p_m$ achieves its maximum at $m^*\approx \sqrt{n/3}$.
Approximating the sum $\sum_m p_m=1$ by an integral over $\alpha=m/\sqrt{n}$
one gets  $p_m\approx n^{-1/2} q_{\alpha(m)}$, where $q_\alpha$ is a normalized pdf defined as
\[
q_\alpha=Z^{-1} \alpha^2 e^{-3\alpha^2}, \quad
Z\equiv \int_0^\infty d\alpha \alpha^2 e^{-3\alpha^2}=\frac{\sqrt{\pi}}{4\cdot 3^{3/2}}.
\]
It gives
\[
S(A)=-\sum_m p_m \log_2{p_m} \approx \log{\sqrt{n}} -\int_0^\infty d\alpha \, q_\alpha \log_2{q_\alpha}.
\]
Evaluating the integral over $\alpha$ yields Eq.~(\ref{ent}).
The approximation $p_m\approx n^{-1/2} q_{\alpha(m)}$ also implies
that $\max_m p_m=O(n^{-1/2})$. This bound will be used below in our spectral gap analysis.
We conjecture  that one can achieve a power law
scaling of $S(A)$ in Eq.~(\ref{ent}) by introducing two types of brackets, say $l\equiv [$, $r\equiv ]$,
$l'\equiv \{$, and $r'\equiv \}$, such that bracket pairs $lr$ and $l'r'$ are created
from the `vacuum' $00$ in a maximally entangled state $(|lr\ra+|l'r'\ra)/\sqrt{2}$.
The local moves Eq.~(\ref{moves}) must be modified as
$0x \leftrightarrow x0$, where $x$ can be either of $l,r,l',r'$, and
$00\leftrightarrow (lr+l'r')/\sqrt{2}$.
We expect the modified model with two types of brackets to obey a scaling
$S(A)\sim \sqrt{n}$, while its gap will remain lower bounded by an inverse  polynomial.

{\em Spectral gap: upper bound. \hspace{2mm}} Let $\lambda_2>0$ be the smallest non-zero eigenvalue of the Hamiltonian defined in Eq.~(\ref{H}).
We shall use the fact that the ground state $|\calM_n\ra$ is highly entangled to
prove an upper bound  $\lambda_2\le O(n^{-1/2})$.
Fix any $k\in [0,n/2]$ and define a `twisted' version of the ground state:
\[
|\phi\ra =\sum_{m=0}^{n/2} e^{i\theta_m} \, \sqrt{p_m} \, |\hat{C}_{0,m}\ra_A  \otimes |\hat{C}_{m,0}\ra_B,
\]
where $\theta_m=0$ for $0\le m\le k$ and $\theta_m=\pi$ otherwise.
Note that $|\phi\ra$ and $|\calM_n\ra$ have the same reduced density matrices on $A$ and $B$.
Hence $\la \phi|H|\phi\ra =\la \phi|\Pi_{cut}|\phi\ra$, where $\Pi_{cut}\equiv \Pi_{n/2,n/2+1}$.
Since $\max_m p_m =O(n^{-1/2})$ and $\sum_m p_m=1$,  there must exist $k\in [0,n/2]$
such that $\sum_{0\le m\le k} p_m =1/2 \pm \epsilon$ for some $\epsilon=O(n^{-1/2})$.
This choice of $k$ ensures that
 $\la \phi|\calM_n\ra = \sum_m p_m e^{i\theta_m} \le 2\epsilon$,
 that is $\phi$ is almost orthogonal to the ground state.
Define a normalized state $|\tilde{\phi}\ra\sim |\phi\ra- \la \calM_n|\phi\ra \cdot |\calM_n\ra$.
Then $\la \tilde{\phi}|\calM_n\ra=0$ and
$\la \tilde{\phi}|H|\tilde{\phi}\ra =\la \tilde{\phi}|\Pi_{cut}|\tilde{\phi}\ra \le \la \phi|\Pi_{cut}|\phi\ra+O(\epsilon)$.
The difference $\la \phi|\Pi_{cut}|\phi\ra-\la \calM_n|\Pi_{cut}|\calM_n\ra$ gets contributions
only from the terms $m=k,k\pm 1$  in the Schmidt decomposition,  since $\Pi_{cut}$ can change the number of unmatched brackets in $A$
and $B$ at most by one. Since  $\la \calM_n|\Pi_{cut}|\calM_n\ra=0$, we get
\[
\la \phi|\Pi_{cut}|\phi\ra \le O(p_k+p_{k-1}+p_{k+1})=O(n^{-1/2}).
\]
We arrive at $\la \tilde{\phi}|H|\tilde{\phi}\ra=O(n^{-1/2})$.
Therefore $\lambda_2$ is at most $O(n^{-1/2})$.

{\em Spectral gap: lower bound.  \hspace{2mm}}
It remains to prove a lower bound $\lambda_2\ge n^{-O(1)}$.
Let $\calH_{p,q}$ be the subspace spanned by strings $s\in C_{p,q}$
and $\calH_M\equiv \calH_{0,0}$ be the Motzkin space spanned by Motzkin paths.
Note that $H$ preserves any subspace $\calH_{p,q}$ and the unique
ground state of $H$ belongs to $\calH_M$. Therefore it suffices
to derive a lower bound $n^{-O(1)}$ for two quantities: (i) the gap of $H$ inside the Motzkin space $\calH_M$,
and (ii) the ground state energy of $H$ inside any `unbalanced' subspace $\calH_{p,q}$
with $p\ne 0$ or $q\ne 0$. Below we shall focus on part~(i)
since it allows us to introduce all essential ideas. The proof of part~(ii) can be found in the Supplementary Material.

Recall that a string $s\in \{l,r\}^{2m}$
is called a {\em Dyck path} iff any initial  segment of $s$ contains at least as many $l$'s as $r$'s,
and the total number of $l$'s is equal to the total number of $r$'s.
For example, Dyck paths of length $6$ are $lllrrr$, $llrlrr$, $llrrlr$, $lrlrlr$, and $lrllrr$.
The proof of part~(i) consists of the following steps:\\
{\em Step~1.} Map the original Hamiltonian $H$ acting on Motzkin paths
to an effective Hamiltonian $H_{\eff}$ acting on Dyck paths using perturbation
theory. \\
{\em Step~2.} Map $H_{\eff}$ to a stochastic matrix $P$ describing a random walk on Dyck paths
in which transitions correspond to insertions or removals of consecutive $lr$ pairs.\\
{\em Step~3.} Bound the spectral gap of $P$
using the canonical paths method~\cite{Diaconis91,*Sinclair92}.

To construct a good family of canonical paths in Step~3 we will
organize Dyck paths into a rooted tree in which
level-$m$ nodes represent Dyck paths of length $2m$,
edges correspond to insertion of $lr$ pairs, and each node has at most four children.
Existence of such a tree will be proved using  the fractional matching method~\cite{Schrijver}.

Let $\calD_m$ be the set of Dyck paths of length $2m$,
$\calD$ be the union of all $\calD_m$ with $2m\le n$,
and $\calM_n$ be the set of Motzkin paths of length $n$.
Define a Dyck space $\calH_D$ whose basis vectors are Dyck paths $s\in \calD$.
Given a Motzkin path $u$ with $2m$ brackets, let $\mathrm{Dyck}(u)\in \calD_{m}$
be the Dyck path obtained from $u$ by removing zeros. We shall use an embedding
$V\, : \, \calH_D \to \calH_{M}$ defined as
\[
V\, |s\ra = \frac1{\sqrt{{n \choose 2m}}}\, \sum_{\substack{u\in \calM_n\\ \mathrm{Dyck}(u)=s\\}}\, |u\ra, \quad \quad s\in
\calD \cap \calD_m.
\]
One can easily check that $V^\dag V=I$, that is, $V$ is an isometry.
For any Hamiltonian $H$, let $\lambda_2(H)$ be the second smallest
eigenvalue of $H$.

\noindent
{\em Step~1.} The restriction of the Hamiltonian Eq.~(\ref{H}) onto the Motzkin space $\calH_M$
can be written as $H=H_{move}+H_{int}$,
where $H_{move}$ describes freely moving left and right brackets, while
$H_{int}$ is an `interaction term' responsible for pairs creation.
More formally, $H_{move}=\sum_{j=1}^{n-1} \Pi_{j,j+1}^{move}$ and
$H_{int}=\sum_{j=1}^{n-1} \Pi_{j,j+1}^{int}$, where
$\Pi^{move}$ projects onto the subspace spanned by $|0l\ra-|l0\ra$
and $|0r\ra - |r0\ra$, while $\Pi^{int}$
projects onto the state $|00\ra-|lr\ra$.
Note that the boundary terms in Eq.~(\ref{H}) vanish on $\calH_M$.
 We shall treat $H_{int}$ as a small
perturbation of $H_{move}$.  To this end define a modified FF Hamiltonian
$H_{\epsilon}=H_{move} + \epsilon H_{int}$,
where $0<\epsilon\le 1$ will be chosen later.
One can easily check that $|\calM_n\ra$ is the unique ground state of $H_\epsilon$
and $\lambda_2(H)\ge \lambda_2(H_\epsilon)$ (use the operator inequality $H\ge H_\epsilon$).
Note that $H_{move}\psi=0$ iff $\psi$ is symmetric under the moves $0l\leftrightarrow l0$
and $0r\leftrightarrow r0$. It follows that the ground subspace of $H_{move}$
is spanned by states $V\, |s\ra$ with $s\in \calD$.
To compute the spectrum of $H_{move}$, we can ignore the difference between
$l$'s and $r$'s since $H_{move}$ is only capable of swapping zeros with
non-zero letters. It follows that the spectrum of $H_{move}$
must coincide with the spectrum of the Heisenberg ferromagnetic spin-$1/2$ chain,
that is, $\Pi^{move}$ can be replaced by the projector onto the singlet state $|01\ra-|10\ra$,
where $1$ represents either $l$ or $r$. Using the exact formula for
the spectral gap of the Heisenberg chain found by Koma and Nachtergaele~\cite{Koma95} we arrive at
$\lambda_2(H_{move})=1-\cos{\left(\frac{\pi}{n}\right)}=\Omega(n^{-2})$.
Let
\[
H_{\eff}=V^\dag H_{int} V
\]
be the first-order effective Hamiltonian acting on the Dyck space $\calH_D$.
Applying the Projection Lemma of~\cite{KKR04} to the orthogonal complement of $|\calM_n\ra$ in $\calH_M$ we infer that
\be
\label{KKR}
\lambda_2(H_\epsilon)\ge \epsilon \lambda_2(H_{\eff})-\frac{O(\epsilon^2) \|H_{int}\|^2}{\lambda_2(H_{move})-2\epsilon\|H_{int}\| }.
\ee
Choosing $\epsilon\ll n^{-3}$ guarantees that $\epsilon\|H_{int}\|$ is small compared with $\lambda_2(H_{move})$.
For this choice of $\epsilon$ one gets
\be
\label{KKR1}
\lambda_2(H)\ge \lambda_2(H_\epsilon)\ge \epsilon \lambda_2(H_{\eff})- O(\epsilon^2 n^4).
\ee
Hence it suffices to prove that $\lambda_2(H_{\eff})\ge n^{-O(1)}$.

\noindent
{\em Step~2.} Recall that $H_{\eff}=V^\dag H_{int} V$ acts on the Dyck space $\calH_D$. Its unique
ground state $|\calD\ra\in \calH_D$ can be found by solving $|\calM_n\ra=V\,|\calD\ra$. It yields
\be
\label{Dstate}
|\calD\ra=\frac1{\sqrt{|\calM_n|}} \sum_{2m\le n} \sqrt{{n \choose 2m}} \sum_{s\in \calD_m} |s\ra.
\ee
Let $\pi(s)=\la s|\calD\ra^2$ be the induced probability distribution on $\calD$.
Given a pair of Dyck paths $s,t\in \calD$, define
\be
\label{walk}
P(s,t)=\delta_{s,t}-\frac1{n}\la s |H_{\eff} |t\ra \sqrt{\frac{\pi(t)}{\pi(s)}}.
\ee
We claim that $P$ describes a random walk on the set of Dyck paths $\calD$
such that $P(s,t)$ is a transition probability from $s$ to $t$,
and $\pi(s)$ is the unique steady state of $P$.
Indeed, since $\sqrt{\pi(t)}$ is a zero eigenvector of $H_{\eff}$, one has $\sum_t P(s,t)=1$
and $\sum_s \pi(s) P(s,t)=\pi(t)$.
Off-diagonal matrix elements $\la s|H_{\eff}|t\ra$ get contributions only from
terms $-|00\ra\la lr|$ and $-|lr\ra\la 00|$ in $H_{int}$. It implies that $\la s|H_{\eff}|t\ra\le 0$
for $s\ne t$ and hence $P(s,t)\ge 0$. Furthermore, $P(s,s)\ge 1/2$ since $\la s|H_{\eff}|s\ra \le n/2$.
In the Supplementary Material we shall prove the following.
\begin{lemma}
\label{lemma:transitions}
Let $s,t\in \calD$ be any Dyck paths such that
$t$ can be obtained from $s$  by adding or removing a single $lr$ pair. Then
$P(s,t)=\Omega(1/n^3)$. Otherwise $P(s,t)=0$.
\end{lemma}
Let $\lambda_2(P)$ be the second largest eigenvalue of $P$. From Eq.~(\ref{walk})
one gets $\lambda_2(H_{\eff})=n(1-\lambda_2(P))$.
Hence it suffices to prove that the random walk
$P$ has a polynomial spectral gap, that is, $1-\lambda_2(P)\ge n^{-O(1)}$.

\noindent
{\em Step~3.} Lemma~\ref{lemma:transitions} tells us that $P$ describes
a random walk on a graph $G=(\calD,E)$ where two Dyck
paths are connected by an edge, $(s,t)\in E$,
iff  $s$ and $t$ are related by insertion/removal of a single $lr$ pair.
To bound the spectral gap of $P$ we shall connect any pair
of Dyck paths $s,t\in \calD$ by a {\em canonical path} $\gamma(s,t)$,
that is, a sequence $s_0,s_1,\ldots,s_l\in \calD$ such that
$s_0=s$, $s_l=t$, and $(s_i,s_{i+1})\in E$ for all $i$.
The canonical paths theorem~\cite{Diaconis91,Sinclair92} shows that
$1-\lambda_2(P)\ge 1/(\rho l)$, where $l$ is the maximum length
of a canonical path and $\rho$ is the maximum edge load defined as
\be
\label{edge_load}
\rho=\max_{(a,b)\in E} \; \frac1{\pi(a) P(a,b)} \sum_{s,t\, : \, (a,b)\in \gamma(s,t)} \; \; \pi(s) \pi(t).
\ee
The key new result that allows us to choose a good family of canonical paths is the following.
\begin{lemma}
\label{lemma:map}
Let $\calD_k$ be the set of Dyck paths of length $2k$.
For any $k\ge 1$ there exists a map $f\, : \, \calD_k \to \calD_{k-1}$ such that
(i)  the image of any path $s\in \calD_k$ can be obtained from $s$ by
removing a single $lr$ pair,
(ii) any path $t\in \calD_{k-1}$ has at least one  pre-image in $\calD_k$, and
(iii) any path $t\in \calD_{k-1}$ has at most four  pre-images in $\calD_k$.
\end{lemma}
The lemma allows one to organize the set of all Dyck paths $\calD$ into a
{\em supertree}  $\calT$ such that the root of $\calT$ represents the empty path and such that children
of any node $s$ are elements of  $f^{-1}(s)$. The properties of $f$ imply that Dyck paths of length $2m$
coincide with level-$m$ nodes of $\calT$,
any step away from the root on $\calT$
corresponds to insertion of a single $lr$ pair,  and any node of $\calT$ has at most four children.
Hence the lemma provides a recipe for growing long Dyck paths
from short ones without overusing any intermediate Dyck paths. It should be noted
that restricting the maximum number of children to four is optimal
since $|\calD_k|=C_k\approx 4^k/\sqrt{\pi} k^{3/2}$, where $C_k$ is the $k$-th Catalan number.
Our proof of Lemma~\ref{lemma:map} based on the fractional matching method can be found in
the Supplementary Material.
Five lowest levels of the supertree $\calT$ are shown on Fig.~\ref{fig:supertree}.

\begin{center}
\begin{figure*}[t]
\centerline{
\mbox{
 \includegraphics[height=4.5cm]{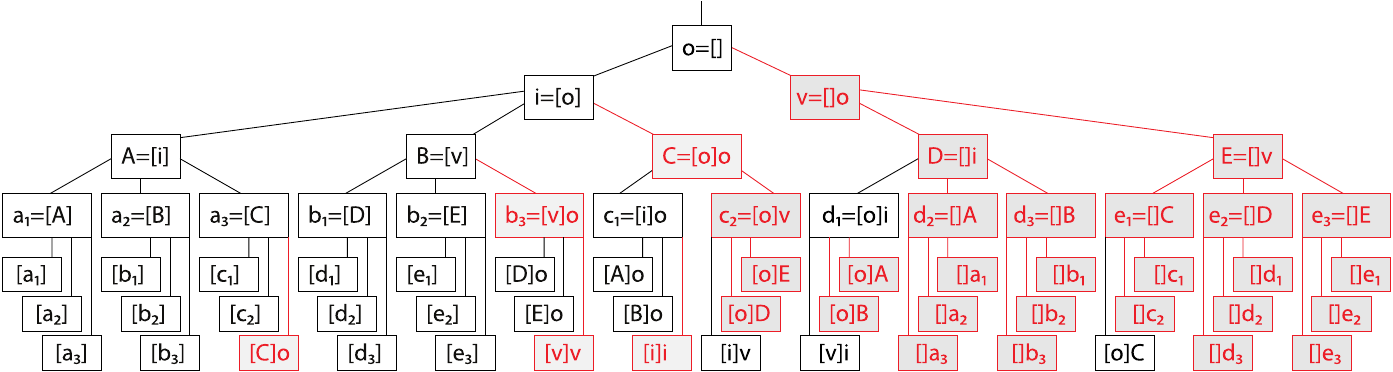}
 }}
\caption{(Color Online) Five lowest levels of the supertree $\calT$.
Nodes of $\calT$ are Dyck paths --- balanced strings of left and right brackets.
Depth-$k$ nodes are in one-to-one correspondence with Dyck paths
of length $2k$ (the set $\calD_k$). Any step towards the root requires removal of a consecutive $[\, ]$ pair.
Any node has at most four children.
The supertree can be described a family of maps $f\, :\, \calD_k\to \calD_{k-1}$
such that $f(s)$ is the parent of $s$. The maps $f$ are defined inductively such that
$f([X])=[f(X)]$, $f([\, ]Y)=[\, ]f(Y)$ for any node,  $f([X]Y)=[f(X)]Y$ for black nodes,
and $f([X]Y)=[X]f(Y)$ for red (shaded) nodes.
See the proof of Lemma~\ref{lemma:map} in the Supplementary Material for more details.
}
 \label{fig:supertree}
\end{figure*}
\end{center}

We can now define the canonical path $\gamma(s,t)$ from
$s\in \calD_m$ to $t\in \calD_k$.
Any intermediate state in $\gamma(s,t)$ will be represented as $uv$ where $u\in \calD_{l'}$ is
an ancestor of $s$ in the supertree  and $v\in \calD_{l''}$ is an ancestor of $t$.
The canonical path starts from $u=s$, $v=\emptyset$ and
alternates between shrinking $u$ and growing $v$ by making steps towards the root  (shrink)
and away from the root (grow) on the supertree. The path terminates as soon as $u=\emptyset$ and $v=t$.
The shrinking steps are skipped whenever $u=\emptyset$, while the growing steps
are skipped whenever $v=t$.
Note that any intermediate state $uv$ obeys
\be
\label{constraint}
\min{(|s|,|t|)} \le |u|+|v| \le \max{(|s|,|t|)}.
\ee
Since any path $\gamma(s,t)$ has length at most $2n$, it suffices to bound
the maximum edge load $\rho$.
Fix the edge $(a,b)\in E$ with the maximum load.
Let $\rho(m,k,l',l'')$ be the contribution to $\rho$ that comes from canonical paths
$\gamma(s,t)$ such that $a=uv\in \calD_{l'+l''}$, where
\[
s\in \calD_{m}, \quad t\in \calD_{k}, \quad u\in \calD_{l'}, \quad v\in \calD_{l''},
\]
and such that $b$ is obtained from $a$ by growing $v$ (the case when $b$ is obtained from
$a$ by shrinking  $u$ is analogous).
The number of possible source strings $s\in \calD_m$ contributing to $\rho(m,k,l',l'')$
is at most $4^{m-l'}$ since $s$ must be a descendant of $u$ on the supertree.
The number of possible target strings $t\in \calD_k$ contributing to $\rho(m,k,l',l'')$
is at most $4^{k-l''}$ since $t$ must be a descendant of $v$ on the supertree.
Taking into account that $\pi(s)$ and $\pi(t)$ are the same for all $s\in \calD_m$
and $t\in \calD_k$ we arrive at
\[
\rho(m,k,l',l'') \le 4^{m+k-l'-l''} \frac{\pi(s)\pi(t)}{\pi(a) P(a,b)}=\frac{\pi_m \pi_k}{\pi_{l'+l''} P(a,b)}
\]
with $\pi_l=4^l {n \choose 2l}/|\calM_n|$.
 Here we used the identity
$\pi(w)=\la w|\calD\ra^2$ and Eq.~(\ref{Dstate}).
Lemma~\ref{lemma:transitions} implies that $1/P(a,b)\le n^{O(1)}$.
Furthermore, the fraction of Motzkin paths of length $n$
that have exactly $2l$ brackets is
$\sigma_l=C_l {n \choose 2l}/|\calM_n|$.
However $C_l\approx 4^l/\sqrt{\pi} l^{3/2}$
coincides with $4^l$ modulo factors polynomial in $1/n$. Hence
\[
\rho(m,k,l',l'') \le n^{O(1)} \cdot \frac{\sigma_m \sigma_k}{\sigma_{l'+l''}}.
\]
By definition, $\sigma_l\le 1$ for all $l$.
Also, one can easily check that $\sigma_l$ as a function of $l$ has a unique
maximum at $l\approx n/3$ and decays monotonically away from the maximum.
Consider two cases. {\em Case~(1):}  $l'+l''$ is on the left from the maximum of $\sigma_l$.
From Eq.~(\ref{constraint}) one gets
$\min{(m,k)}\le l'+l''$ and thus
$\sigma_m\sigma_k\le \sigma_{\min{(m,k)}} \le \sigma_{l'+l''}$.
{\em Case~(2):} $l'+l''$ is on the right from the maximum of $\sigma_l$.
From Eq.~(\ref{constraint}) one gets
$\max{(m,k)}\ge l'+l''$ and thus
$\sigma_m\sigma_k\le \sigma_{\max{(m,k)}} \le \sigma_{l'+l''}$.
In both cases we get a bound $\rho(m,k,l',l'')\le n^{O(1)}$.
Since the number of choices for $m,k,l',l''$ is at most $n^{O(1)}$,
we conclude  that $\rho\le n^{O(1)}$ and thus $1-\lambda_2(P)\ge n^{-O(1)}$.

{\em Open problems.  \hspace{2mm}}
Our work raises several questions. First, one can ask
what is the upper bound on the ground state entanglement
of FF spin-$1$ chains and whether the Motzkin state achieves this bound.
For example, if the Schmidt rank $\chi(L)$ for a block of $L$ spins can only grow
polynomially with $L$, as it is the case for the Motzkin state, ground states
of FF spin-$1$ chains could be efficiently represented by Matrix Product States~\cite{Verstraete05}
(although finding such representation might be a computationally hard problem~\cite{Schuch08}).
One drawback of the model based on Motzkin paths is the need for boundary conditions
and the lack of the thermodynamic limit. It would be interesting to find examples of FF spin-$1$
chains with highly entangled ground states that are free from this drawback.
We also leave open the question of whether our model can indeed be regarded as a critical spin chain
in the sense that its continuous limit can be described by a conformal field theory.
Finally, an intriguing open question is whether long-range ground state entanglement (or steady-state
entanglement in the case of dissipative processes)
in 1D spin chains can be stable against more general local perturbations, such as external magnetic fields.

\section{Supplementary Material}

\subsection{Schmidt coefficients of the Motzkin state}

In this section we compute the Schmidt coefficients $p_m$ defined in Eq.~(\ref{Schmidt1})
and show that for large $n$ and $m$ one can use an approximation
\be
\label{approx}
p_m\sim m^2 \exp{(-3m^2/n)}.
\ee
Let $\calD_{n,k}\subseteq \{l,r\}^{2n+k}$ be the set of balanced
strings of left and right brackets of length $2n+k$ with $k$ extra left
brackets. More formally, $s\in \calD_{n,k}$ iff
any initial  segment of $s$ contains at least as many $l$'s as $r$'s,
and the total number of $l$'s is equal to $n+k$.
\begin{lemma}[\bf Andr\'e's reflection method]
\label{lemma:reflection}
The total number of strings in $\calD_{n,k}$ is
\[
D_{n,k}=\frac{k+1}{n+k+1} {2n+k \choose n}.
\]
\end{lemma}
\begin{proof}
For any bracket string $s$ let $L(s)$ and $R(s)$ be the number of left and right brackets in $s$.
Any $s\in  \{l,r\}^{2n+k}$ such that $s\notin \calD_{n,k}$ can be uniquely
represented as $s=u r v$, where $r$ corresponds to the first unmatched right bracket
in $s$, while $u$ is a balanced string (Dyck path). Let $v'$ be a string obtained from $v$ by a reflection $r\leftrightarrow l$
and $s'=urv'$. Then
\[
L(s')=L(u)+L(v')=R(u)+R(v)=R(s)-1=n-1
\]
and
\bea
R(s')&=& R(u)+1+R(v')=L(u)+1+L(v)=L(s)+1 \nn \\
&=& n+k+1. \nn
\eea
Furthermore, any string $s'$ with $n-1$ left brackets and $n+k+1$ right bracket
can be uniquely represented as $s'=urv'$. Hence the number of strings in $\calD_{n,k}$ is
\[
D_{n,k}={2n+k \choose n} -{2n+k \choose n-1}=\frac{k+1}{n+k+1} {2n+k \choose n}.
\]
\end{proof}
One can easily check that $|C_{p,q}(n)|=|C_{0,p+q}(n)|$ since the identity of unmatched brackets
does not matter for the counting.
Let
\[
M_{n,m}\equiv |C_{0,m}(n)|.
\]
Lemma~\ref{lemma:reflection} implies that
\[
M_{n,m}=
\sum_{\substack{i\ge 0 \\ 2i+m\le n\\}} \;  \frac{m+1}{i+m+1} {n \choose 2i+m} {2i+m \choose i}.
\]
It can be rewritten as
\be
\label{Mnm}
M_{n,m}=\sum_{\substack{i\ge 0 \\ 2i+m\le n\\}} M_{n,m,i},
\ee
where
\be
M_{n,m,i}=\frac{(m+1)\cdot n!}{(i+m+1)!i!(n-2i-m)!}.
\ee
Let $\alpha=m/\sqrt{n}$ and $\beta=(i-n/3)/\sqrt{n}$.
Using Stirling's formula one can get
\be
M_{n,m,i}\approx
\frac{3\sqrt{3}}{2\pi n^{3/2}}3^{n+1}\alpha\exp\left(-3\alpha^{2}-9\alpha\beta-9\beta^{2}\right).
\ee
We approximate the sum in Eq.~(\ref{Mnm}) by integrating over $i$. Since  $i=\frac{n}{3}+\beta\sqrt{n}$,
we get $di=\sqrt{n}d\beta$, Since the maximum is near $i=\frac{n}{3}$,
we can turn this sum into an integral from $-\infty$ to $\infty$.
The integral we need to evaluate is thus
\begin{eqnarray*}
M_{n,m} & \approx & \frac{3\sqrt{3}}{2\pi n^{3/2}}3^{n+1}\alpha\int_{-\infty}^{\infty}e^{-3\alpha^{2}-9\alpha\beta-9\beta^{2}}\sqrt{n}d\beta\\
 & = & \frac{3\sqrt{3}}{2\pi n}3^{n+1}\alpha\int_{-\infty}^{\infty} e^{-9\left(\beta-\alpha/2\right)^{2}-\frac{3}{4}\alpha^{2}}\sqrt{n}d\beta\\
 & = & \frac{\sqrt{3}}{2\sqrt{\pi}n}3^{n+1}\alpha e^{-\frac{3}{4}\alpha^{2}}
 \sim m \exp{(-3m^2/4n)}.
\end{eqnarray*}
Recalling that
\be
p_m = \frac{|C_{0,m}(n/2)|^2}{|C_{0,0}(n)|} \sim M_{n/2,m}^2
\ee
 we arrive at Eq.~(\ref{approx}).

\subsection{Proof of Lemma~\ref{lemma:transitions}}
Suppose $s\in \calD_k$ and $t\in \calD_{k\pm1}$.
Using the definition of $P(s,t)$ one can easily get
\[
P(s,t)=-\frac1{n}{n \choose 2k}^{-1} \sum_{u\in \calM_n[s]}\; \;  \sum_{v\in \calM_n[t]} \la u|H_{int} |v\ra.
\]
Here $\calM_n[s]=\{ u\in \calM_n\, : \,  \mathrm{Dyck}(u)=s\}$.
Note that $\la u|H_{int} |v\ra=-1/2$ if $u$ and $v$ differ exactly at two consecutive positions
where $u$ and $v$ contain $00$ and $lr$ respectively or vice versa.
In all other cases one has $\la u|H_{int} |v\ra=0$.

Suppose $t\in \calD_{k+1}$ and $P(s,t)>0$.
Let us fix some $j\in [0,2k]$ such that
$t$ can be obtained from $s$ by inserting a pair $lr$ between $s_j$ and $s_{j+1}$.
For any string $u\in \calM_n[s]$ in which $s_j$ and $s_{j+1}$ are separated by at least two
zeros one can find at least one $v\in \calM_n[t]$ such that
$\la u|H_{int} |v\ra=-1/2$. The fraction of strings $u\in \calM_n[s]$
in which $s_j$ and $s_{j+1}$ are separated by two or more
zeros is at least $1/n^2$ which implies
\[
P(s,t)\ge \frac1{2n^3}.
\]
Suppose now that $t\in \calD_{k-1}$.
Let us fix some $j\in [1,2k-1]$ such that $t$ can be obtained from $s$
by removing the pair $s_j s_{j+1}=lr$. For any string $u\in \calM_n[s]$ in which $s_j$ and $s_{j+1}$
are not separated by zeros one can find at least one $v\in \calM_n[t]$ such that
$\la u|H_{int} |v\ra=-1/2$. The fraction of strings $u\in \calM_n[s]$
in which $s_j$ and $s_{j+1}$ are not separated by
zeros is at least $1/n$ which implies
\[
P(s,t)\ge \frac1{2n^2}.
\]

\subsection{Proof of Lemma~\ref{lemma:map}}

Let us first prove a simple result concerning fractional matchings.
Consider a bipartite graph $G=(A\cup B,E)$.
Let $x=\{x_e\}_{e\in E}$ be a vector of real variables
associated with edges of the graph. For any vertex $u$ let $\delta(u)$ be the set
of edges incident to $u$.
Define a {\em matching polytope}~\cite{Schrijver}
\bea
\calP&=& \{ x\, : \, x_e\ge 0 \quad \mbox{for all $e\in E$}, \nn \\
&& 1\le \sum_{e\in \delta(a)} x_e \le 4,  \quad
\sum_{e\in \delta(b)} x_e=1 \nn \\
&& \mbox{for all $a\in A$ and $b\in B$}\}. \nn
\eea
\begin{lemma}
\label{lemma:matching}
Suppose $\calP$ is non-empty. Then there exists a map $f\, : \, B\to A$ such that
(i) $f(b)=a$ implies $(a,b)\in E$, (ii) any vertex $a\in A$ has at least one pre-image in $B$,
and (iii) any vertex $a\in A$ has at most four pre-images in $B$.
\end{lemma}
\begin{proof}
Since $\calP$ is non-empty, it must have at least one extremal point $x^*\in \calP$.
Let $E^*\subseteq E$ be the set of edges such that $x^*_e>0$. We claim that $E^*$ is a forest
(disjoint union of trees). Indeed,
suppose $E^*$ contains a cycle $Z$ (a closed path). Then $x^*_{a,b}<1$ for
all $(a,b)\in Z$ since otherwise the cycle would terminate at $b$.
Hence $0<x^*_e<1$ for all $e\in Z$.
Since the graph is bipartite, one can label edges of $Z$ as even and odd in alternating order.
There exists $\epsilon\ne 0$ such that $x^*$ can
be shifted by $\pm \epsilon$ on even and odd edges of $Z$  respectively without leaving
$\calP$. Hence $x^*$ is a convex mixture of two distinct vectors from $\calP$.
This is a  contradiction since $x^*$ is an extreme point. Hence $E^*$ contains no cycles, that is,
$E^*$ is a forest. We claim that $x_e^*=1$ for all $e\in E^*$. Indeed, let $T\subseteq E^*$
be the subset of edges with $0<x_e^*<1$. Obviously, $T$ itself is a forest. Degree-$1$ nodes
of $T$ must be in $A$ and there must exist a path $\gamma\subseteq T$ starting
and ending at degree-$1$ nodes $u,u'\in A$. Since $0<x^*_e<1$ for all $e\in \gamma$,
there exists $\epsilon\ne 0$ such that $x^*$ can be shifted by $\pm \epsilon$ on
even and odd edges of $\gamma$ respectively without leaving $\calP$.
This is a  contradiction since $x^*$ is an extreme point. Hence $x_e^*=1$ for all $e\in E^*$.
We conclude that $x^*_e\in \{0,1\}$ for all edges of $G$.
The desired map can now be defined as $f(b)=a$ iff $x^*_{a,b}=1$.
\end{proof}

We are interested in the case where
\[
A=\calD_{n-1} \quad \mbox{and} \quad  B=\calD_n
\]
are Dyck paths of semilength $n-1$ and $n$ respectively. Paths $a\in \calD_{n-1}$ and $b\in \calD_n$ are
connected by an edge iff $a$ can be obtained from $b$ by removing a single $lr$ pair.
Our goal is to construct a map $f\,: \, \calD_n\to \calD_{n-1}$
with the properties (i),(ii),(iii) stated in Lemma~\ref{lemma:matching}.
According to the lemma, it suffices to choose
$f$ as a stochastic map. Namely, for any $b\in \calD_n$ we shall define a random
variable $f(b)\in \calD_{n-1}$ with some normalized probability distribution.
It suffices  to satisfy two conditions:
\be
\label{goal1}
\mathrm{Pr}[f(b)=a]>0 \quad \parbox[t]{5cm}{only if $a$ can be obtained from $b$ by removing a single $lr$ pair,}
\ee
and
\be
\label{goal2}
\sum_{b\in \calD_n} \mathrm{Pr}[f(b)=a]=X_n \quad \quad \mbox{for all $a\in \calD_{n-1}$}.
\ee
Here $1\le X_n\le 4$ is some function of $n$ that we shall choose later.
We shall define $f$ using induction in $n$.
\begin{lemma}
\label{lemma:small}
For every $n\ge 1$ there exists $1\le X_n\le 4$ and a stochastic map $f\, : \, \calD_n \to \calD_{n-1}$
satisfying Eqs.~(\ref{goal1},\ref{goal2}).
\end{lemma}
\begin{proof}
Any Dyck path $b\in \calD_n$ can be uniquely represented as $b=lsrt$
for some $s\in \calD_i$, $t\in \calD_{n-i-1}$, and $i\in [0,n-1]$.
We shall specify the map $f\, : \, \calD_n\to \calD_{n-1}$ by the following rules:
\begin{center}
\begin{tabular}{c|c|c}
$b\in \calD_n$ & $f(b)\in \calD_{n-1}$ & probability \\
\hline
$lsrt$, $s\in \calD_i$, $1\le i\le n-2$ & $lf(s) r t$ & $p_i$ \\
\hline
$lsrt$, $s\in \calD_i$, $1\le i\le n-2$ & $lsr f(t)$ & $1-p_i$ \\
\hline
$lrt$, $t\in \calD_{n-1}$ & $lrf(t)$ & $1$ \\
\hline
$lsr$, $s\in \calD_{n-1}$ & $lf(s)r$ & $1$ \\
\hline
\end{tabular}
\end{center}
Here we assumed that $f$ has been already defined for strings of semilength up to $n-1$
such that Eqs.~(\ref{goal1},\ref{goal2}) are satisfied.
By abuse of notation, we ignore the index $n$ in $f$, so we regard $f$ as a family of maps defined for all $n$.
It is clear that our inductive definition of $f$ on $\calD_n$ satisfies Eq.~(\ref{goal1}).
The probabilities $p_1,\ldots,p_{n-2}\in [0,1]$ are free parameters that must be chosen to satisfy
Eq.~(\ref{goal2}). Note that these probabilities also implicitly depend on $n$.
The  choices of $f(b)$ in the first two lines of the above table are represented by
black and red nodes in the example shown on Fig.~\ref{fig:supertree}.
Consider three cases:

{\em Case~1:} $a=lrt'$ for some $t'\in \calD_{n-2}$. Then $f(b)=a$ iff
$b=llrrt'$ or $b=lrt$ for some $t\in \calD_{n-1}$ such that $f(t)=t'$.
These possibilities are mutually exclusive.
Hence
\bea
 \sum_{b\in \calD_n} \mathrm{Pr}[f(b)=lrt'] &=& p_1 + \sum_{t\in \calD_{n-1}} \; \mathrm{Pr}[f(t)=t'] \nn   \\
&=&   p_1+X_{n-1}. \nn
\eea
Substituting it into Eq.~(\ref{goal2}) gives a constraint
\be
\label{C1}
p_1= X_n- X_{n-1}.
\ee

{\em Case~2:} $a=ls'r$ for some $s'\in \calD_{n-2}$. Then $f(b)=a$ iff
$b=ls'rlr$ or $b=lsr$ for some $s\in \calD_{n-1}$ such that $f(s)=s'$.
These possibilities are mutually exclusive. Hence
\bea
\sum_{b\in \calD_n} \mathrm{Pr}[f(b)=ls'r] &=&  1-p_{n-2} +
\sum_{s\in \calD_{n-1}} \mathrm{Pr}[f(s)=s']    \nn \\
&= &   1- p_{n-2} + X_{n-1}.  \nn
\eea
Substituting it into Eq.~(\ref{goal2}) gives a constraint
\be
\label{C2}
p_{n-2}=1 -(X_n-X_{n-1}).
\ee
It says that $X_n$ must be a non-decreasing sequence.

{\em Case~3:} $a=ls'rt'$ for some $s'\in \calD_i$, $t'\in \calD_{n-i-2}$, and $i=1,\ldots,n-3$.
In other words, both $s'$ and $t'$ must be non-empty. Then $f(b)=a$ iff
$b=lsrt'$ for some $s\in \calD_{i+1}$ such that $f(s)=s'$,
or $b=ls'rt$ for some $t\in \calD_{n-i-1}$ such that $f(t)=t'$.
These possibilities are mutually exclusive. Hence
\bea
\sum_{b\in \calD_n} \mathrm{Pr}[f(b)=ls'rt'] &=&
p_{i+1}\sum_{s\in \calD_{i+1}} \mathrm{Pr}[f(s)=s'] \nn \\
&& + (1-p_i) \sum_{t\in \calD_{n-i-1}} \mathrm{Pr}[f(t)=t'] \nn \\
&= &   p_{i+1}\, X_{i+1}  +  (1-p_i) X_{n-i-1}.\nn
\eea
Substituting it into Eq.~(\ref{goal2}) gives a constraint
\be
\label{C3}
 p_{i+1} \, X_{i+1}  +  (1-p_i)  X_{n-i-1} =X_n
\ee
for each $i=1,\ldots,n-3$.
Let us choose
\be
\label{Xi}
X_i=\frac{C_i}{C_{i-1}} =\frac{4(i-1/2)}{i+1}.
\ee
Combining Eqs.~(\ref{C1},\ref{C2},\ref{C3})
we obtain a linear system with unknown variables $p_1,\ldots,p_{n-2}\in [0,1]$.
We shall look for a solution $\{p_i\}$ having an extra symmetry
\be
\label{symmetry}
p_i+p_{n-i-1}=1 \quad \mbox{for $i=1,\ldots,n-2$}.
\ee
One can check that the system defined by Eqs.~(\ref{C1},\ref{C2},\ref{C3},\ref{symmetry})
has a solution
\be
p_i=\frac{i(i+1)(3n-2i-1)}{n(n+1)(n-1)}, \quad i=1,\ldots,n-2.
\ee
Hence we have defined the desired stochastic map $f\, :\, \calD_n\to \calD_{n-1}$.
This proves the induction hypothesis.

It remains to note that for $n=1,2$ the map $f$ is uniquely specified by Eqs.~(\ref{goal1},\ref{goal2})
and our choice of $X_n$. Indeed, one has $\calD_2=\{llrr,lrlr\}$, $\calD_1=\{lr\}$, and $\calD_0=\emptyset$.
To satisfy Eq.~(\ref{goal1}),  we have to choose $f(llrr)=f(lrlr)=lr$ for $n=2$ and $f(lr)=\emptyset$ for $n=1$.
It also satisfies Eq.~(\ref{goal2}) since $X_2=2$ and $X_1=1$.
This proves the base of induction.
\end{proof}

\subsection{Ground state energy for unbalanced subspaces}

Recall that the unbalanced subspace $\calH_{p,q}$ is spanned by strings
$s\in C_{p,q}$ that have $p$ unmatched right and $q$ unmatched left brackets.
Our goal is to prove that the restriction of $H$ onto any subspace $\calH_{p,q}$
with $p>0$ or $q>0$ has ground state energy at least $n^{-O(1)}$.
By the symmetry, it suffices to consider the case $p>0$.
To simplify the analysis we shall omit the boundary term $|l\ra\la l|_n$.
Note that such omission can only decrease the ground state energy. Accordingly, our
simplified  Hamiltonian becomes
\be
\label{Hsimplified}
H=|r\ra \la r|_1 + \sum_{j=1}^{n-1} \Pi_{j,j+1}.
\ee
Recall that $\Pi$ is a projector  onto the subspace spanned by states
$|00\ra -|lr\ra$, $|0l\ra-|l0\ra$, and $|0r\ra-|r0\ra$.
Let $\lambda_1(H)$ be the ground state energy of $H$.

Any string $s\in C_{p,q}$ can be uniquely represented as
\[
s=u_0 r u_1 r u_2 \ldots r u_p l v_1 l v_2 \ldots l v_q
\]
where $u_i$ and $v_j$ are Motzkin paths (balanced strings of brackets).
The remaining $p$ right and $q$ left brackets are unmatched
and never participate in the move $00\leftrightarrow lr$.
It follows that the  unmatched brackets can be regarded
as ``solid walls'' that can be swapped with $0$'s but otherwise
do not participate in any interactions. In particular, the
spectrum of $H$ restricted to $\calH_{p,q}$ depends only on $p+q$ as long as $p>0$.
This allows us to focus on the case $q=0$, i.e.
assume that all unmatched brackets are right.

Given a string $s\in C_{p,0}$, let $\tilde{s}\in  \{0,l,r,x,y\}^n$ be the string
obtained from $s$ by the following operations: (i) replace the first unmatched right
bracket in $s$ by `$x$', and (ii) replace all other unmatched brackets in $s$ (if any) by `$y$'.
Define a new Hilbert space $\tilde{\calH}_p$ whose basis vectors are $|\tilde{s}\ra$, $s\in C_{p,0}$.
Consider a Hamiltonian
\be
\label{Htilde}
\tilde{H}=|x\ra\la x|_1 + \sum_{j=1}^{n-1} \Pi_{j,j+1} + \Theta^x_{j,j+1} + \Theta^y_{j,j+1},
\ee
where
$\Theta^x$ and $\Theta^y$ are  projectors onto the states
$|0x\ra-|x0\ra$ and $|0y\ra-|y0\ra$ respectively (with a proper normalization).
One can easily check that $\la s|H|t\ra =\la \tilde{s}|\tilde{H}|\tilde{t}\ra$ for any $s,t\in C_{p,0}$.
Hence the spectrum of $H$ on $\calH_{p,0}$ coincides with the spectrum of $\tilde{H}$
on $\tilde{\calH}_p$. Furthermore, if we omit all the terms $\Theta^y_{j,j+1}$ in $\tilde{H}$, the ground state energy
can only decrease. Hence it suffices to consider a simplified Hamiltonian
\be
\label{Hx}
H^x=|x\ra\la x|_1 + \sum_{j=1}^{n-1} \Pi_{j,j+1} + \Theta^x_{j,j+1}
\ee
which acts on $\tilde{\calH}_p$.
Note that positions of $y$-particles are integrals of motion for $H^x$.
Moreover, for fixed positions of $y$-particles,
any term in $H^x$ touching a $y$-particle vanishes. Hence
$H^x$ can be analyzed separately on each  interval between consecutive $y$-particles.
Since our goal is to get a lower bound on the ground state energy,
we can only analyze the interval between $1$ and the first $y$-particle.
Equivalently, we can redefine $n$ and focus on the case $p=1$, $q=0$, that is,
assume that there is only one unmatched right bracket.
The relevant Hilbert space $\tilde{\calH}_1$ is now spanned by states
\[
|s\ra \otimes |x\ra \otimes |t\ra, \quad \mbox{where} \quad s\in \calM_{j-1}, \quad t\in \calM_{n-j}.
\]
Recall that $\calM_k$ is the set of Motzkin paths (balanced strings of left and right brackets)
of length $k$.

We would like to treat the terms responsible for the motion and detection of
the $x$-particle as a small perturbation. To this end, choose any $0<\epsilon\le 1$ and define
the Hamiltonian
\[
H^x_\epsilon = \sum_{j=1}^{n-1} \Pi_{j,j+1} + \epsilon |x\ra\la x|_1 + \epsilon \sum_{j=1}^{n-1} \Theta^x_{j,j+1}.
\]
Clearly, $H^x_\epsilon\le H^x$, so it suffices to get a lower bound
on the ground state energy of $H^x_\epsilon$.

Let us first find the ground subspace and the spectral gap of the unperturbed Hamiltonian $H^x_0=\sum_{j=1}^{n-1} \Pi_{j,j+1}$.
Note that the position of the $x$-particle $j$ is an invariant of motion for $H^x_0$.
Moreover, any projector $\Pi_{i,i+1}$ touching the $x$-particle vanishes.
Hence we can analyze $H^x_0$ separately on the two disjoint intervals $A=[1,j-1]$ and
$B=[j+1,n]$. It follows that the ground subspace of $H^x_0$ is spanned by normalized states
\be
\label{fixed_j}
|\psi_j\ra = |\calM_{j-1}\ra \otimes |x\ra \otimes |\calM_{n-j}\ra, \quad j=1,\ldots,n.
\ee
The spectral gap of  $H^x_0$ can also be computed separately in $A$ and $B$.
Since we have already shown that the original Hamiltonian
Eq.~(\ref{H}) has a polynomial gap inside the Motzkin subspace, we conclude that
$\lambda_2(H^x_0)\ge n^{-O(1)}$.

Let us now turn on the perturbation.
The first-order effective Hamiltonian acting on the ground subspace
spanned by $\psi_1,\ldots,\psi_n$ describes a hopping of the $x$-particle
on a chain of length $n$ with a delta-like repulsive potential applied at site $j=1$.
Parameters of the hopping Hamiltonian can be found by calculating the matrix elements
\[
\la \psi_j|\Theta^x_{j,j+1} |\psi_j\ra =\frac{M_{n-j-1}}{2M_{n-j}}\equiv \alpha_j^2,
\]
\[
\la \psi_{j+1}|\Theta^x_{j,j+1} |\psi_{j+1}\ra =\frac{M_{j-1}}{2M_{j}}\equiv \beta_j^2,
\]
and
\[
\la \psi_j| \Theta^x_{j,j+1} |\psi_{j+1}\ra =-\frac12 \sqrt{\frac{M_{n-j-1}}{M_{n-j}} \cdot \frac{M_{j-1}}{M_{j}}}=-\alpha_j \beta_j,
\]
where $M_k=|\calM_k|$ is the $k$-th Motzkin number.
We arrive at the effective hopping Hamiltonian acting on $\CC^n$, namely,
\be
\label{Heff1D}
H_{\eff}=|1\ra\la 1| + \sum_{j=1}^{n-1} \Gamma_{j,j+1},
\ee
where
\bea
\Gamma_{j,j+1} &=& \alpha_j^2 \, |j\ra\la j| + \beta_j^2\,  |j+1\ra\la j+1| \nn \\
&& - \alpha_j\beta_j (|j\ra\la j+1| + |j+1\ra\la j|)
\eea
is a rank-$1$ projector.
Applying the Projection Lemma of~\cite{KKR04} we infer that
\[
\lambda_1(H^x_\epsilon)\ge \epsilon \lambda_1(H_{\eff}) - \frac{O(\epsilon^2) \| V\|^2}{\lambda_2(H^x_0) - 2\epsilon \|V\|},
\]
where $V=|x\ra\la x|_1+\sum_{j=1}^{n-1} \Theta^x_{j,j+1}$ is the perturbation operator. Since $\lambda_2(H^x_0) \ge n^{-O(1)}$,
we can choose $\epsilon$ polynomial in $1/n$ such that $2\epsilon \|V\|$ is small compared with $\lambda_2(H^x_0)$.
For this choice of $\epsilon$ one gets
\[
\lambda_1(H^x_\epsilon)\ge \epsilon \lambda_1(H_{\eff}) - O(\epsilon^2) n^{O(1)}.
\]
Hence it suffices to show that $\lambda_1(H_{\eff})\ge n^{-O(1)}$,
where $H_{\eff}$ is now the single $x$-particle hopping Hamiltonian Eq.~(\ref{Heff1D}).

Let us first focus on the hopping Hamiltonian without the repulsive potential:
\[
H_{move}\equiv \sum_{j=1}^{n-1} \Gamma_{j,j+1}.
\]
This Hamiltonian is FF and its unique ground state is
\be
\label{g}
|g\ra \sim \sum_{j=1}^n \sqrt{M_{j-1} M_{n-j}} \, |j\ra.
\ee
Our strategy will be to bound the spectral gap of $H_{move}$
and apply the Projection Lemma to $H_{\eff}$ by treating the repulsive potential $|1\ra\la 1|$
as a perturbation of $H_{move}$. First let us
map $H_{move}$ to a stochastic matrix describing a random
walk on the interval $[1,n]$ with the steady state $\pi(j)=\la j|g\ra^2$.
For any $a,b\in [1,n]$ define
\be
\label{walk2}
P(j,k)=\delta_{j,k} - \la j|H_{move}|k\ra \sqrt{\frac{\pi(k)}{\pi(j)}}.
\ee
Since $\sqrt{\pi(j)}$ is a zero eigenvector of $H_{move}$, we
infer that $\sum_k P(j,k)=1$ and $\sum_j \pi(j) P(j,k)=\pi(k)$.
A simple algebra shows that
\[
P(j,j+1)=\frac{M_{n-j-1}}{2M_{n-j}} \quad \mbox{and} \quad P(j+1,j)=\frac{M_{j-1}}{2M_{j}}
\]
are the only non-zero off-diagonal matrix elements of $P$.
We shall use the following property of the Motzkin numbers.
\begin{lemma}
For any $n\ge 1$ one has $1/3 \le M_n/M_{n+1}\le 1$.
Furthermore, for large $n$ one can use an approximation
\be
M_n \approx c\frac{3^n}{n^{3/2}}
\ee
where $c\approx 1.46$.
\end{lemma}
The lemma implies that
\[
\frac16 \le P(j,j\pm 1)\le \frac12
\]
for all $j$. Hence the diagonal matrix elements $P(j,j)$ are non-negative, that is,
we indeed can regard $P(j,k)$ as a transition probability from $j$ to $k$.
Furthermore, using Eq.~(\ref{g}) and the above lemma we infer that the steady state $\pi$ is `almost uniform', that is,
\be
n^{-O(1)} \le \frac{\pi(k)}{\pi(j)} \le n^{O(1)} \quad \mbox{for all $1\le j,k\le n$}.
\ee
In particular, $\min_j{\pi(j)} \ge n^{-O(1)}$.
We can now easily bound the spectral gap of $P$.  For example, applying the canonical
paths theorem stated above we get $1-\lambda_2(P)\ge 1/(\rho l)$
where $\rho$ is defined in Eq.~(\ref{edge_load}) and the canonical path $\gamma(s,t)$ simply
moves the $x$-particle from $s$ to $t$. Since the denominator in Eq.~(\ref{edge_load})
is lower bounded by $n^{-O(1)}$, we conclude that $1-\lambda_2(P)\ge n^{-O(1)}$.
It shows that $\lambda_2(H_{move}) \ge n^{-O(1)}$.

To conclude the proof, it remains to apply the Projection Lemma to
$H_{\eff}$ defined in Eq.~(\ref{Heff1D}) by treating the
repulsive potential $|1\ra\la 1|$ as  a perturbation.
Now the effective first-order Hamiltonian will be simply a $c$-number $\la 1|g\ra^2=\pi(1)\ge n^{-O(1)}$
which proves the bound $\lambda_1(H_{\eff}) \ge n^{-O(1)}$.


\bibliography{main}
%

\bibliographystyle{plain}

\end{document}